\documentclass[floatfix,amsmath,superscriptaddress]{revtex4}
\usepackage{amssymb}
\usepackage{graphicx}
\usepackage{graphics}
\usepackage{amsmath}
\usepackage{amsthm}
\usepackage{color}
\usepackage{comment}

\def\h{{\bf h}}
\def\g{{\bf g}}
\def\v{{\bf v}}
\def\Umb{U}

\def\i{{i_1,i_2}}

\def\k{{\bf k}}

\newcommand{\bea}{\begin{eqnarray}}
\newcommand{\eea}{\end{eqnarray}}
\newcommand{\e}{\eta}

\def\bi{\begin{itemize}}
\def\ei{\end{itemize}}
\def\bc{\begin{center}}
\def\ec{\end{center}}

\def\C{\hbox{$\mit I$\kern-.7em$\mit C$}}
\def\R{\hbox{$\mit I$\kern-.6em$\mit R$}}
\def\N{\hbox{$\mit I$\kern-.6em$\mit N$}}

\def\ket#1{|#1\rangle}
\newcommand{\one}{\mbox{$1 \hspace{-1.0mm}  {\bf l}$}}
\def\tr{\mathrm{tr}}
\def\ket#1{\left| #1\right>}
\def\bra#1{\left< #1\right|}

\newcommand{\proj}[1]{\ket{#1}\bra{#1}}

\newtheorem{theorem}{Theorem}
\newtheorem{corollary}[theorem]{Corollary}
\newtheorem{lemma}[theorem]{Lemma}

\newtheorem{observation}[theorem]{Observation}

\begin{document}

\author{C. Spee}
\affiliation{Institute for Theoretical Physics, University of
Innsbruck, Innsbruck, Austria}
\author{J.I. de Vicente}
\affiliation{Departamento de Matem\'aticas, Universidad Carlos III de
Madrid, Legan\'es (Madrid), Spain}
\author{B. Kraus}
\affiliation{Institute for Theoretical Physics, University of
Innsbruck, Innsbruck, Austria}
\title{The maximally entangled set of 4-qubit states}

\begin{abstract}
Entanglement is a resource to overcome the natural restriction of operations used for state manipulation to Local Operations assisted by Classical Communication (LOCC). Hence, a bipartite maximally entangled state is a state which can be transformed deterministically into any other state via LOCC. In the multipartite setting no such state exists. There, rather a whole set, the Maximally Entangled Set of states (MES), which we recently introduced, is required. This set has on the one hand the property that any state outside of this set can be obtained via LOCC from one of the states within the set and on the other hand, no state in the set can be obtained from any other state via LOCC. Recently, we studied LOCC transformations among pure multipartite states and derived the MES for three and generic four qubit states. Here, we consider the non-generic four qubit states and analyze their properties regarding local transformations. As already the most coarse grained classification, due to Stochastic LOCC (SLOCC), of four qubit states is much richer than in case of three qubits, the investigation of possible LOCC transformations is correspondingly more difficult. We prove that most SLOCC classes show a similar behavior as the generic states, however we also identify here three classes with very distinct properties. The first consists of the GHZ and W class, where any state can be transformed into some other state non--trivially. In particular, there exists no isolation. On the other hand, there also exist classes where all states are isolated. Last but not least we identify an additional class of states, whose transformation properties differ drastically from all the other classes. Although the possibility of transforming states into local-unitary inequivalent states by LOCC turns out to be very rare, we identify those states (with exception of the latter class) which are in the MES and those, which can be obtained (transformed) non-trivially from (into) other states respectively. These investigations do not only identify the most relevant classes of states for LOCC entanglement manipulation, but also reveal new insight into the similarities and differences between separable and LOCC transformations and enable the investigation of LOCC transformations among arbitrary four qubit states.
\end{abstract}
\maketitle

\tableofcontents

\section{Introduction}

Multipartite states occur in many applications of quantum
information, such as one-way quantum computing, quantum error
correction, and quantum secret sharing \cite{Gothesis97,RaBr01,secretshar}. The
existence of those practical and abstract applications is due to
the subtle properties of multipartite entangled states.
Furthermore, multipartite entanglement plays also an
important role in other fields such as many-body physics \cite{AmFa08}.  Thus, one
of the main goals in quantum information theory is to gain a better
understanding of the non-local properties of quantum states. However,
whereas the bipartite case is well understood, the multipartite
case is much more complex. Even though a big theoretical effort has
been undertaken where several entanglement measures for
multipartite states have been introduced \cite{CoKu00}, different
classes of entangled states have been identified \cite{DuViCi00},
and a normal form of multipartite states has been presented
\cite{Ves}, we are far from a full-fledged theory of multipartite entanglement \cite{HoHo07}.

The paradigm of Local Operations and Classical Communication (LOCC) plays a fundamental role in entanglement theory. First, this is the most general class of transformations spatially separated parties can implement: each party can apply locally to its share of the state the most general transformation allowed by the laws of quantum mechanics (a trace-preserving completely positive map) which can be conditioned on the other parties' outcomes through successive rounds of classical communication. Thus, the study of LOCC transformations provides the possible protocols with which entangled states can be manipulated for spatially separated parties. Second, and more importantly, entanglement is the resource to overcome the limitations due to LOCC transformation, as entanglement can only decrease under LOCC transformations. Hence, the investigation of LOCC transformations induces an operationally meaningful ordering in the set of entangled states. Indeed, if the state $|\psi\rangle$ can be transformed by LOCC into the state $|\phi\rangle$, then all tasks that can be implemented using the latter are also amenable using the former (but not necessarily the other way around). Hence, $|\psi\rangle$ is more (or equally) useful than $|\phi\rangle$ and consequently more (or equally) entangled. Obviously then, the very basic requirement for an entanglement measure is to respect the LOCC ordering. Thus, the study of LOCC transformations allows to classify which states are more useful for quantum information processing from a solid theoretical point of view. Unfortunately, the investigation of LOCC protocols is in general very difficult as their mathematical characterization is involved and plagued with technical subtleties \cite{Donald}. For instance, it has been recently shown that the set of LOCC operations is not closed \cite{LOCCnotclosed1, LOCCnotclosed2} and certain transformations can only be accomplished in the limit of infinitely many rounds of classical communication \cite{infiniteround}. In fact, considering transformations among pure states, although LOCC convertibility is characterized in the bipartite case \cite{Nielsen}, only a few classes of states have been studied in the multipartite case \cite{Turgut1,Turgut2,Tajima} prior to our work.

Due to these difficulties, other directions to classify the entanglement contained in a multipartite system have been pursued such as Local Unitary (LU) and Stochastic LOCC (SLOCC) equivalence. These approaches consider a more amenable mathematical problem which provides an operationally meaningful classification. Two $n$-partite states, $\ket{\Psi_1},\ket{\Psi_2}$ are called LU-equivalent if there exist LUs, $U_1,\ldots U_n$ such that $U_1\otimes \ldots \otimes U_n\ket{\Psi_1}=\ket{\Psi_2}$. Note that two states which are LU-equivalent are equally useful (i.\ e.\ LU is an invertible LOCC protocol) and, therefore, this classification groups states that contain exactly the same amount of entanglement. Necessary and sufficient conditions for the LU-equivalence of pure $n$-qubit states have been derived in \cite{barbaraLU}. On the other hand, the SLOCC classification questions whether two states, $\ket{\Psi_1}$ and $\ket{\Psi_2}$, can be interconverted at least probabilistically by LOCC \cite{DuViCi00,slocc4}. Mathematically this means that there exists a local invertible operation, $g$ such that $g \ket{\Psi_1}=\ket{\Psi_2}$. For three-qubit entangled states it has been shown there there exist two different truly tripartite entangled SLOCC classes, the GHZ-class and the W-class \cite{DuViCi00} \footnote{Note that we consider troughout the paper only truly multipartite entangled states}. However, for more systems, there are infinitely many SLOCC classes \cite{slocc4}. Note that the entanglement contained in states belonging to different SLOCC classes is  fundamentally different. However, as it is the case for the LU classification too, this just defines equivalence classes and does not allow to order entangled states according to their usefulness. On the other hand, it is known that any LOCC operation can be written as a separable map (SEP). These are maps of the form $\rho\rightarrow\sum_i X_i\rho X_i^\dagger$, where the $X_i$ are local operators and $\sum_iX_i^\dag X_i=\one$. In \cite{Gour} necessary and sufficient conditions for the possibility of transforming one state into another via SEP have been investigated (see also below). Contrary to the approaches mentioned above, SEP defines an ordering; however, it must be noted that LOCC is strictly contained in SEP (i.\ e.\ there exist transformations that can be implemented by SEP but not by LOCC \cite{Bennett,Eric,LOCCnotclosed1}) and SEP transformations lack an operational interpretation.

Given the enormous complexity of multipartite entangled states, it would be desirable to go beyond the classifications based on equivalence mentioned above and find classes of states which are indicated to be particularly relevant in an operationally meaningful sense as put forward by the principles of entanglement theory. To this end, we introduced recently the concept of the maximally entangled set (MES) of multipartite states \cite{MESus}, which we will explain in the following. In contrast to the simple picture we have in the bipartite case, in the multipartite case many different notions of maximally entangled states exist in the literature (see e.g. \cite{maxent, Lstate}). Most of them refer to states which optimize a certain measure of entanglement. However, from the fundamental point of view, it should be clear from our argumentations above that the most meaningful way to order the set of entangled states is by LOCC convertibility. Considering entanglement as a resource, a maximally entangled state ought to be a state which can be transformed into any other state deterministically via LOCC. In the bipartite case, it follows from the results of \cite{Nielsen} that this is the state $\ket{\Phi^+}\propto \sum_i \ket{ii}$. Note that here and in the following we consider convertibility via LOCC only among LU-equivalence classes as LU transformations can always be performed. Thus, we only consider one representative for each LU-equivalence class. As there are more than one truly $n$-partite entangled SLOCC classes for $n\geq 3$, there cannot exist a single maximally entangled state. In fact, as we showed in \cite{MESus} the generalization of the maximally entangled bipartite state leads to a set of pure states, the MES. This set, $MES_n$, has the following properties:  (i) No state in $MES_n$ can be obtained from any other $n$-partite state via LOCC (excluding LU) and (ii) for any truly
$n$-partite entangled pure state, $\ket{\Phi} \not\in MES_n$, there exists a state in $MES_n$ from which $\ket{\Phi}$ can be obtained via LOCC. Stated differently, it is the unique minimal set of states from which any other state can be reached via LOCC. Thus, our notion of maximal entanglement completely matches the idea of maximal usefulness under LOCC manipulation. The determination of the MES should then not only improve our understanding of the possibilities of multipartite LOCC protocols but, more importantly, identify the class of the most relevant entangled states in an operational way. Notice that we restrict to LOCC transformations among fully-entangled $n$-qubit  pure states. This implies, that transformations can only be done among states in the same SLOCC class. Thus, the MES identifies the most (in the above mentioned sense) relevant states in each SLOCC class.

In \cite{MESus} we provided general tools to decide LOCC convertibility in the multipartite case and characterized $MES_3$ and the generic subset of $MES_4$. Interestingly, these sets contain infinitely many states (even in the 3-qubit case where there exist just two different SLOCC classes). However, $MES_3$ is of measure
zero in the full set of $3$-qubit states and only relatively few states are maximally useful for LOCC conversions. On the contrary, $MES_4$ is of
full measure: almost all $4$-qubit
states are in the MES. However, the reason for this is that almost all states
are \textit{isolated}, i.\ e.\ they can neither be obtained from nor
transformed to any other fully entangled state by deterministic
LOCC (excluding LU). Hence, LOCC induces a trivial ordering in the set of entangled states and the possibility of LOCC conversions is very rare in the
multipartite case. This implies, that most states are useless for entanglement manipulation via LOCC.
However, in \cite{MESus} we also identified a zero-measure subset of states in $MES_4$
which are LOCC convertible. Hence, despite the fact that $MES_4$ is of full measure and therefore almost all states are required to reach any other state, the most useful states regarding entanglement manipulation are contained in the zero-measure set of LOCC convertible states in the MES. Note that this significant set of multipartite entangled states is as in the bipartite case likely to guide to new applications.
In \cite{MESus} we considered solely generic $4$-qubit
states, i.\ e.\ we excluded SLOCC classes giving rise to a subset of states of measure
zero. The aim of this paper is to complete this analysis by considering all the remaining SLOCC classes which can be grouped into nine families.
The reason for doing so is not just for the sake of completeness. It is not enough to know the MES for generic states as from the point of view of applications it is generally believed that very few states are relevant. Actually, most multipartite states highlighted in previous literature from its potential for applications and/or relevant mathematical structure such as the GHZ state \cite{ghz}, the cluster state \cite{RaBr01} or the $|L\rangle$ and $|M\rangle$ states \cite{Lstate} belong to the non-generic 4-qubit SLOCC classes.  Moreover, even if certain properties, such as isolation is a generic feature for the SLOCC classes studied in \cite{MESus}, it could happen that this is not the case for non-generic families.
In fact, as we will show here, there are three classes (each containing one or more SLOCC classes) which show a completely different behavior regarding LOCC transformations than the generic and all remaining SLOCC classes do. In order to explain the differences let us mention here that for the generic SLOCC class investigated in
\cite{MESus} the following holds:
i) most states are isolated, i.e. they can neither be reached from nor transformed to any other state via non--trivial LOCC;
ii) any state which can be reached via a separable operation can also be reached via a very simple LOCC protocol. They consist of only one party applying a measurement (and all the other parties just apply a local unitary depending on the measurement outcome) \cite{SaSc15}.
iii) All properties, such as being in MES, being convertible, being reachable (via both SEP and LOCC) can be directly read off a state given in the standard form. The standard form will always be given by $g_1\otimes g_2\otimes g_3\otimes g_4\otimes \ket{\Psi}$, where the $g_i \in SL(2)$ for any $i$ and the state $\ket{\Psi}$ is some carefully picked representative of the SLOCC class. In standard form, some conditions, which depend on the SLOCC class under consideration, on the operators $g_i$ are imposed, e.g. that they are diagonal. Given this standard form, it holds generically that a state is reachable iff there exists a local symmetry of $\ket{\Psi}$, such that three of the operators $g_i$ commute with the corresponding local terms of the symmetry, whereas the fourth does not. A state which is not reachable is necessarily in MES. All reachable states can then be shown to be reachable by a state in the MES. Moreover, interestingly all states which are reachable are also convertible.

Here, we show that for the vast variety of non-generic SLOCC classes properties i)--iii) still hold true. However, we also identify the following classes which show a very distinct behavior:
\bi
\item [(I)] Classes, where no state is isolated. That is LOCC conversions to or from any state are always possible. In fact, there are only two such classes, the one corresponding to the GHZ state, $|0000\rangle+|1111\rangle$, and that corresponding to the W state, $|0001\rangle+|0010\rangle+|0100\rangle+|1000\rangle$ (see Sec. \ref{secGHZ} and Sec. \ref{secW}). Note that this also explains why there exists no isolation in case of three qubits, as there the only genuinely three--partite entangled classes are the GHZ and the W class. Moreover, the GHZ and W class are the only classes for which the set of states in MES and these particular SLOCC class are of zero-measure. In all the other cases almost all states inside the particular class are isolated and are therefore necessarily in MES. Note further that these results hold also true in the n--qubit W and GHZ case.
\item [(II)] Classes where all states are isolated. In these classes not even separable transformation can be found which transform any of the states into a state which is not LU--equivalent to the initial state (see Sec. \ref{secLba3aneq0}, Sec. \ref{secLba3a0bneq0} and Sec. \ref{secLa4aneq0}). Note that these classes can be seen as those having exactly the opposite properties as the GHZ and W class.
\item[(III)] A class where it is no longer the case that states which can be reached via a separable operation can also be reached via a very simple LOCC protocol, where only one party measures. Moreover, it is no longer possible to infer from the standard form the SEP convertibility properties of a state. That is, the local operations do no longer need to have the properties explained in (iii) for a state to be reachable via SEP.
    The structure of the separable transformation which is possible in this case is impossible in all
    the other classes. This class is called L-state class in the literature (see Sec. \ref{cycle3}). Certain states in this class have properties resembling states in the three qutrit case \cite{HeSp15}. There we identify pairs of states, $(\ket{\Psi},\ket{\Phi})$, where the transformation from $\ket{\Psi}$ to $\ket{\Phi}$ is possible via SEP, however, no LOCC protocol exists to realize this transformation. Note that as
    the difference among these transformations has only been proven among conversions involving ensembles of states \cite{Eric, LOCCnotclosed1}, and that it is known that they are equivalent in the case of transformations among pure bipartite states \cite{nielsenandchuang}, this result is in strong contrast to all previous results of pure state transformations. It constitutes the first example of pure state transformations, which are possible via SEP, but not via LOCC. As the properties of states in this class are significantly distinct from all other classes and as they resemble the three qutrits, which we analyze in \cite{HeSp15}, the investigation of all MES in this class is beyond the scope of this paper. Note that the existence of this particular SLOCC class is likely to be the reason for why the general properties valid for all other classes (mentioned above) cannot be proven in general. This implies that basically every family of SLOCC classes has to be treated separately.
    \ei

Note that the existence of these classes is precisely the reason for investigating in detail the very involved case of four qubits. It is our intention to investigate in the future whether this different behavior might be useful for certain applications. Moreover, their existence also shows once again the notoriously difficult structure of multipartite LOCC transformations.

Apart from the identification of the classes mentioned above, we present here the characterization of all states in $MES_4$ (excluding the ones in the $L$--state class). Moreover, we characterize the LOCC convertible states in $MES_4$ for all these SLOCC classes, thus providing the most relevant subclass of states for LOCC manipulation.

On the other hand, it is of technical interest to investigate further the difference between SEP and LOCC operations.
Moreover, dealing here with all the possible 4-qubit SLOCC classes and their very different mathematical structure shows in more detail the wide applicability of the techniques we introduced in \cite{MESus} to decide LOCC and SEP convertibility. We present here all the details of the derivations as they can be used directly to study arbitrary LOCC transformations among four qubits. This is relevant to compute for instance the new entanglement measures which we recently introduced \cite{SaSc15, ScSa15}. There, the entanglement of a state is measured by how useful it is for state transformations and how difficult it is to generate the state. More precisely, the source entanglement measures the volume of states which can reach the state via LOCC, the accessible entanglement measures how many LU--inequivalent states can be reached from the state of interest via LOCC. Due to the operational character of these quantities it is easy to construct entanglement measures from them. Moreover, the standard form introduced for each SLOCC class can be utilized to decide very easily whether two states are LU--equivalent or not \footnote{Note that in some cases the seed state is not uniquely chosen. However, the results derived for those SLOCC classes apply to any choice of seed state. Moreover, the seed state can be made unique by imposing positivity and/or an order  of some of the parameters of the seed state. This unique seed state can then be used to decide LU-equivalence in a very simple way.}.

Hence, the results presented here do not only characterize the most relevant states (with the exception of the L--state class), but can be used to study general LOCC transformations among four qubit states. Moreover, they reveal novel differences and similarities between SEP and LOCC transformations and lead to a simple criterion for LU--equivalence.

The outline of the remaining of the paper is the following. In Section II we first introduce our notation. Then, we recall the results on SEP and will then show how they can be used to study LOCC transformations. As we will see, the symmetries of a state play an important role in the investigation of the possible transformations from and to a state. We will present a systematic method of determining the symmetries of an arbitrary 4-qubit state. Next, we will show how the symmetry can be used to identify a standard form up to LUs for states in a particular SLOCC class. Using the results on SEP we will then present the general formalism to determine all possible LOCC transformations. Since any state, which cannot be reached via LOCC is necessarily in the MES, any state which cannot be reached via SEP is particularly in the MES. For all classes but the L--state class, we will then show that whenever a SEP transformation exists, it is already so constrained that one can find a corresponding very simple LOCC transformation. It turns out that there is no need for complicated protocols and it suffices that each party measures just once. In fact, any reachable state can be obtained via a LOCC protocol where only one party measures and the remaining once apply, depending on the outcome, a local unitary operator, as mentioned above. In Sec. II the outline of the subsequent sections is presented and the results are summarized. In Sec. III to Sec. X  the various SLOCC classes \footnote{Note that the classification of SLOCC classes given in \cite{slocc4} was derived using related methods to the one that we use here in order to determine the symmetries. A representative of each SLOCC class is given at the beginning of each section.} are considered and the corresponding states in $MES_4$ as well as the convertible states are identified. The structure of these sections is always as outlined in Sec. II. However, unfortunately, in order to determine the possible SEP operations almost always a new proof is required.

\section{General method of determining the MES}

In this section we present the general formalism of deriving the MES and the non--isolated states contained in the MES.
The section is organized as follows. First of all, we present the notation we are going to use throughout this paper. Then we summarize the results presented in \cite{Gour} regarding possible SEP transformation and the MES for three and generic four qubit states. From then on it will be clear that the symmetries of a state will play an important role in studying possible state transformations. In Subsec \ref{SubsecSym} we will demonstrate how the symmetries of an arbitrary 4 qubit state can be determined. Finally, we will show how this allows us to identify a standard form of states within a particular SLOCC class and the MES for that class. An outline of the subsequent sections, where the MES and the non--isolated MES are determined for the various SLOCC classes, is presented and the results are summarized.

\subsection{Notation}

Throughout the paper
$\sigma_i$, where $i=0,1,2,3$  and $\one, \sigma_x, \sigma_y, \sigma_z$ denote the identity operator and the Pauli
operators. Furthermore, $W(\alpha)=\exp(i\alpha \sigma_w)$ for $w=x,y,z$. We will ignore normalization whenever it does not lead to any
confusion. By $\mathcal{G}$
we denote the set of local invertible (not necessarily determinant
$1$) operators. $g,h$ denote elements of $\mathcal{G}$, e.g. $g=g_1\otimes
\ldots \otimes g_n$, with $g_i\in GL(2)$. Two states are in the same SLOCC class (LU class) if there exists a $g\in \mathcal{G}$ ($g$ local unitary) which converts one state into the other respectively. Throughout the paper we denote by $d_i$ (and $\tilde{d}_i)$ a diagonal matrix $\in GL(2)$. Moreover, we will use the notation $D_i=d_i^\dagger d_i$.
$U$ denotes the unitary transformation that maps the computational basis into the magic basis \cite{magicbasis}, i.e. \bea \label{eq:Umb}
U=\ket{\Phi^+}\bra{00} -i\ket{\Phi^-}\bra{01}+\ket{\Psi^-}\bra{10}-i\ket{\Psi^+}\bra{11},\eea

where $\ket{\Phi^\pm}=\frac{1}{\sqrt{2}}(\ket{00}\pm \ket{11}), \ket{\Psi^\pm}=\frac{1}{\sqrt{2}}(\ket{01}\pm \ket{10})$ denote the Bell states. Whenever we consider a transformation between two states, we denote by $\ket{\Psi_1} = g
\ket{\Psi}$ the initial state and by $\ket{\Psi_2}= h
\ket{\Psi}$ the final state, where $|\Psi\rangle$ is some properly chosen representative state, which we call the seed state, for the considered SLOCC class. We will also use the notation $G_i= g_i^\dagger g_i$. These operators are strictly positive (if they would not be of full rank entanglement would be destroyed) and without loss of generality (wlog) we normalize these operators such that $\tr (G_i)= \tr (H_i) =1$. Hence, three parameters are needed in order to specify these operators. We will use two different parametrizations throughout the text as, depending on the SLOCC class, it is more convenient to work with one or the other. In the first one we will use the notation $G_i= \left(
    \begin{array}{cccc}
       & g_i^1&g_i^2 \\
     & g_i^{2 *} & g_i^3 \\
     \end{array}
  \right)$  where $g_i^3=1-g_i^1$ and $0\leq \sqrt{(g_i^1-1/2)^2+|g_i^2|^2}<1/2$ (and similarly for $H_i$). In case $g_i^1=g_i^3$ and $g_i^2$ is purely imaginary (or real) we write  $G_i^w=(g_i^w)^\dagger g_i^w=\one/2+g_i^2 \sigma_w$ for $w=y$ (or $w=x$) respectively.  For the second notation we will expand the operators in the basis $\{\one, \sigma_x, \sigma_y, \sigma_z\}$ leading to $G_i=\one/2+\sum_{j=1}^3 \bar{g}_i^j\sigma_j$ and similarly for $H_i$. We will sometimes arrange these coefficients to form vectors that we denote by $\textbf{g}_i=(\bar{g}_i^1,\bar{g}_i^2,\bar{g}_i^3)$ and that must fulfill $0\leq|\textbf{g}_i|<1/2$.
It will often happen that only one Pauli operator has to be considered. In this case we use the notation $g_i^w \in span\{\one, \sigma_w\}$ for $w\in \{x,y,z\}$.
The group generated by some operators $S_i$ will be denoted by $\langle S_i\rangle$.

\subsection{Separable transformations and $MES_3$ and $MES_4$}
\label{SubSecSepMES}
We review here the results on state transformation using separable transformation presented in
\cite{Gour} and the maximally entangled sets of three and generic four qubit states. Let us denote by $S(\ket{\Psi})$ the set of local symmetries of $\ket{\Psi}$, i.e. $S(\Psi)=\{S\in \mathcal{G}:
S\ket{\Psi}=\ket{\Psi}\}$ \footnote{Recall that the elements of $\mathcal{G}$ do not necessary have determinant one. If we would restrict ourselves here to determinate one operators, we would define $S(\Psi)=\{S\in \mathcal{G}, |S|=1 \mbox{ and } \exists q\in \C:
S\ket{\Psi}=\sqrt{q}\ket{\Psi}\}$. However, in Eq. (\ref{EqSep}) this proportionality factor has to be taken into account as long as it is not a phase.}. It has been shown in \cite{Gour} that a state $\ket{\Psi_1} = g
\ket{\Psi}$ can be transformed via SEP to $\ket{\Psi_2}= h
\ket{\Psi}$ iff there exists an integer $m$ and a set of probabilities,
$\{p_k\}_{1}^m$ ($p_k\geq 0, \sum_{k=1}^m p_k=1$) and $S_k \in
S(\Psi)$ such that \bea \label{EqSep} \sum_k p_k S_k^\dagger H S_k=
r G.\eea Here, $H=h^\dagger h \equiv \bigotimes H_i$, and
$G=g^\dagger g \equiv \bigotimes G_i$ are local operators and
$r=\frac{n_{\Psi_2}}{n_{\Psi_1}}$ with
$n_{\Psi_i}=||\ket{\Psi_i}||^2$. The local Positive Operator-Valued Measure (POVM) elements \footnote{Note that we call here the measurement operators POVM elements.}
which transform $\ket{\Psi_1}$ into
$\ket{\Psi_2}$ are given by $M_k=\frac{ \sqrt{p_k}}{\sqrt{r}} h S_k
g^{-1}$, which can be easily understood as follows. Suppose that there exists a non--trivial SEP protocol, i.e. not all POVM elements are local unitary operations, which transforms the state $\ket{\Psi_1}$ into the state $\ket{\Psi_2}$. Let us denote by $M_k$ the local POVM elements. Then we have
\bea M_k \ket{\Psi_1}=\sqrt{p_k/r} \ket{\Psi_2},\eea
where $r=\frac{n_{\Psi_2}}{n_{\Psi_1}}$. Using the fact that $\ket{\Psi_1}$ and $\ket{\Psi_2}$ must be in the same SLOCC class and hence must be of the form $g\ket{\Psi}, h\ket{\Psi}$ for some $g,h \in {\cal G}$ and the seed state, $\ket{\Psi}$, we have that
\bea M_k=\frac{ \sqrt{p_k}}{\sqrt{r}} h S_k
g^{-1}.\eea

As $M_k$ has to be local, $S_k$ has to be a local symmetry of the state $\ket{\Psi}$. Eq. (\ref{EqSep}) expresses then nothing else but the fact that the map is trace-preserving. Obviously the same holds true for LOCC transformations. Thus, in order to determine the possible transformations, one first needs to determine the local symmetries of the state. In the next section we present a systematic method of doing so.

In \cite{MESus} we utilized this result to determine the $MES_3$ and the generic set of $MES_4$. For the sake of completeness we recall here these results, which were derived as follows. First off all, it is clear that a state that cannot be reached from any other state via SEP can also not be reached via LOCC. Thus, such a state must be in the MES. For all the other states we derived an initial state and constructed a LOCC protocol, which transforms the latter to the desired state. Note that a initial state can always be chosen to be in the MES. Using the same standard form as in \cite{MESus} these investigations led to the following results \cite{MESus}:

\begin{theorem} The MES of three qubits, $MES_3$, is given by \bea \label{mes3} MES_3=\{g^x_1\otimes g^x_2\otimes
g^x_3 \ket{GHZ}, g_1 \otimes g_2 \otimes \one \ket{W}\},\eea where no $g^x_i\propto \one$ (except for the GHZ state) and $g_1$ and $g_2$ are diagonal.
\end{theorem}

The set of generic four-qubit states considered in \cite{MESus} is given by the SLOCC classes, $G_{abcd}$, whose representatives are of the form
\begin{align}\label{seed0}
|\Psi\rangle&=\frac{a+d}{2}(|0000\rangle+|1111\rangle)+\frac{a-d}{2}(|0011\rangle+|1100\rangle)+\frac{b+c}{2}(|0101\rangle+|1010\rangle)+\frac{b-c}{2}(|0110\rangle+|1001\rangle),
\end{align}
where $a,b,c,d \in \C$ with  $b^2\neq c^2\neq d^2\neq b^2$, $a^2\neq b^2, c^2, d^2$ and the parameters fulfill the condition that there exists no $q\in \C\backslash 1$ such that $\{a^2,b^2,c^2,d^2\}=\{q a^2,q b^2,q c^2,q d^2\}$. Excluding normalization and the overall phase this gives in total six parameters to describe the seed states, $\ket{\Psi}$, for the generic SLOCC classes. The values of $\{\bar{g}_i^j\}$ identify LU-inequivalent states (up to certain sign-flips) yielding 12 additional parameters to describe the states in these SLOCC classes. Hence, in total states in these classes are described by 18 parameters. We showed in \cite{MESus} that the MES of generic four qubit states is of full--measure. In particular, it is shown that

\begin{theorem} \label{Thgen4} A generic state, $h \ket{\Psi}$,
is reachable via LOCC from some other state iff (up to permutations)
 either $h= h_1\otimes h^w_2 \otimes h^w_3 \otimes h^w_4$, for
$w\in\{x,y,z\}$ where $h_1\neq h^w_1$ or $h= h_1\otimes\one^{\otimes 3}$ with $h_1\not \propto \one$ arbitrary.
\end{theorem}
All the remaining generic states are necessarily in $MES_4$, which is hence of full measure. The reason for almost all
states being in $MES_4$ is that deterministic LOCC manipulations among fully entangled
$4$-qubit states are almost never possible, which is stated in the following theorem.

\begin{theorem} \label{theoremISO}
A generic state $g \ket{\Psi}$ is convertible via LOCC to some other state iff (up to permutations) $g=g_1\otimes g^w_2
\otimes g^w_3 \otimes g^w_4$ with $w\in \{x,y,z\}$ and $g_1$ arbitrary.
\end{theorem}
Combining Theorem \ref{Thgen4} and Theorem \ref{theoremISO} we obtained the following corollary.
\begin{corollary}
The only LOCC convertible generic states in $MES_4$ are of the form $g^w_1\otimes g^w_2
\otimes g^w_3 \otimes g^w_4|\Psi\rangle$ with $w\in \{x,y,z\}$ (excluding the case where $g^w_i\not\propto\one$ for exactly one $i$).
\end{corollary}

Thus, whereas almost all four qubit states are in $MES_4$, almost none can be transformed into any other (LU--inequivalent) state via LOCC.

\subsection{Symmetries and seed states}
\label{SubsecSym}

In order to determine the local symmetries of a state, we use the isomorphism between $SL(2)\otimes SL(2)$ and $SO(4)$. It is a well--known fact that any element of $SL(2)\otimes SL(2)$ is transformed via $U^\dagger$ (see Eq. (\ref{eq:Umb})) into a special complex orthogonal matrix, i.e. for any $g_1,g_2\in SL(2)$ we have $\Umb^\dagger g_1\otimes g_2 \Umb \in SO(4)$ \footnote{ Note that this can be generalized to more systems. To do so,  define (for $n$ even) the unitary operator $U=\sqrt{-i}e^{i\pi/4\sigma_y^{\otimes n}}$. Then, $UU^T=\sigma_y^n$, $U U^\dagger=\one$ and $\forall A_i\in SL(2,\C)$, we have $ O\equiv U^\dagger \bigotimes A_i U$,
is a complex orthogonal matrix.}. In case $g_1, g_2$ are both unitary, the corresponding orthogonal matrix can be easily shown to be real (see e.g. \cite{KrCi00}).

We consider now an arbitrary state of $4$ qubits, $\ket{\Psi}$, and consider a bipartite splitting $A|B$, where both, $A$ and $B$ contain two qubits. Then, $\ket{\Psi}$, can be written as $\ket{\Psi}=\sum_\i \ket{\Psi_\i}_A\ket{\i}_B\equiv A_\Psi\otimes \one \ket{\Phi^+}_{AB}$, where $\ket{\Phi^+}_{AB}=\sum_\i \ket{\i}_A\ket{\i}_B$. Here, the states $\ket{\Psi_\i}$ are unnormalized and $A_\Psi=\sum_\i \ket{\Psi_\i} \bra{\i}.$ In order to determine the symmetries of $\ket{\Psi}$, we introduce the symmetric matrix \bea \label{eq:Z} Z_\Psi= B_\Psi B_\Psi^T, \eea
where $B_\Psi=U^\dagger  A_\Psi U$. The reason for considering this operator is the following. Let $g_A,g_B\in SL(2)\otimes SL(2)$ and consider the local transformation from $\ket{\Psi}$ to $g_A\otimes g_B \ket{\Psi}$. The corresponding $Z$--operator transforms from $Z_\Psi$ to $Z_{g_A\otimes g_B \ket{\Psi}}=O Z_\Psi O^T$, where $O \in SO(4)$ is given by $O=U^\dagger g_A U$. Similarly, the operator \bea \label{eq:tilZ} \tilde{Z}_\Psi=U^\dagger A_\Psi^T U U ^T A_\Psi U^\ast= U^\dagger\sigma_y^{\otimes 2}U B_\Psi^T B_\Psi U^T\sigma_y^{\otimes 2}U^\ast\eea transforms to  $\tilde{O} \tilde{Z}_\Psi \tilde{O}^T$, with $\tilde{O}=U^\dagger g_B U$. Thus, in order to determine the local symmetries, $S$ of a state $\ket{\Psi}$, i.e. $S \in SL(2)^{\otimes 4}$ such that $S\ket{\Psi}=\sqrt{s}\ket{\Psi}$ for some $s \in \C\backslash 0$, one may first determine the complex orthogonal matrices, $O_i$ ($\tilde{O}_i$), which leave $Z_\Psi$ ($\tilde{Z}_\Psi$) (up to a factor) invariant respectively \footnote{ Equivalently, one could determine those orthogonal matrices, $O, \tilde{O}$, that leave $B_\Psi$ invariant, i.e. $O B_\Psi U^T \sigma_y^{\otimes 2} U^\ast\tilde{O}^T U^\dagger \sigma_y^{\otimes 2}U\propto B_\Psi$.}, i.e. \bea \label{Eq:Zs} O_iZ_\Psi O_i^T&=&q Z_\Psi \\\tilde{O}_i\tilde{Z_\Psi} \tilde{O}_i^T&=&\tilde{q} \tilde{Z_\Psi},\eea  with $q\tilde{q}=s$. Since the local symmetries of the state must be of the form $ U O_i U^\dagger \otimes U \tilde{O}_j U^\dagger$, for some $i,j$, it only remains to single out those operators, which leave the state itself invariant. Let us remark here that the symmetry group is not necessarily finite. For instance, there exist infinitely many local symmetries of the GHZ-state \cite{Gour,MESus}.

As a similarity transformation cannot change the eigenvalues and the block structure of the Jordan form of $Z_\Psi$ \cite{HornJohnsonMAnalysis}, Eq. (\ref{Eq:Zs}) can only be satisfied if $q\{\lambda_i\}=\{\lambda_i\}$, where $\{\lambda_i\}$ denotes the set of eigenvalues of $Z$ (taking multiplicities into account). Hence, in most cases, $q=1$ is the only solution. However, there are instances, where a more general choice of $q$ must be taken into account, e.g. for $\{\lambda_i\}=\{0,e^{i2\pi/3},e^{i4\pi/3},1\}$ $q=e^{i2\pi/3}$ would be a solution (see Sec. \ref{cycle3}). Moreover, for $\lambda_i=0$ $\forall i$ we have that any $q\in\C\backslash 0$ is possible.

In order to present now a systematic method to determine the symmetries and also to identify the representatives of each SLOCC class,
we use the following decomposition of symmetric matrices \cite{Craven,gantmacher}.

\begin{theorem} For any symmetric matrix, $A$, there exists a symmetric block diagonal matrix, $R=\oplus_i R^{(d_i)}(\lambda_i)$, where the $d_i \times d_i$--matrices, $R^{(d_i)}(\lambda_i)$ are defined as
 \bea R^{(d_i)}(\lambda_i)= \frac{1}{2}\left[\begin{pmatrix} 0&1&0&0&0&\ldots \\ 1& 0&1&0&0&\ldots\\\ldots &1& 0&1&0&\ldots\\\vdots &\vdots& \vdots&\vdots&\vdots&\vdots\\\ldots &0& 0&0&0&1\\\ldots &0& 0&0&1&0\end{pmatrix}+ i
 \begin{pmatrix} \ldots&0&0&0&1&0 \\ \ldots&0&0& 1&0&-1\\ \ldots&0&1& 0&-1&0\\\vdots &\vdots& \vdots&\vdots&\vdots&\vdots\\ 1&0 &-1&0& 0&\ldots \\ 0 &-1& 0&0&0&\ldots\end{pmatrix}\right]+\lambda_i \one_{d_i}\eea
and a complex orthogonal matrix $O$, ($O^T O=\one$) such that
\bea \label{eq: SymJordan} A=O R O^T.\eea
\end{theorem}

Note that $R=T J T^\dagger$, where $J$ denotes the Jordan matrix of $A$ and $T$ the direct sum of $d_i \times d_i$ matrices of the form $T^{(d_i)}=\frac{1}{\sqrt{2}}(\one-i V^{(d_i)})$. Here, $V^{(d_i)}$ denotes the $d_i \times d_i$ unitary with entries given by $V_{j,k}^{(d_i)}=\delta_{j,d_i-k+1}$. Hence, the entries of the symmetric matrix $T^{(d_i)}$ are given by \bea  T_{j,k}^{(d_i)}=1/\sqrt{2}(\delta_{j,k}-i\delta_{j,d_i-k+1})\eea and for $d_i=1$ we have $T^{(1)}=1$. For instance, for a Jordan matrix with two one-dimensional blocks and one two dimensional block we have

\begin{eqnarray}
&T= \left(
    \begin{array}{cccc}
      1 & 0 & 0&0 \\
      0 & 1  & 0&0\\
 0 & 0  & \frac{1}{\sqrt{2}}&\frac{-i}{\sqrt{2}}\\
0 & 0  & \frac{-i}{\sqrt{2}}&\frac{1}{\sqrt{2}}.\\
    \end{array}
  \right).
\end{eqnarray}

The existence of such a decomposition for symmetric matrices stems from the Jordan decomposition and the fact that any two symmetric matrices, which are similar, such as $A$ and $R$ are related to each other via a complex orthogonal matrix \cite{gantmacher,HornJohnsonTopics}. In the following we will call the matrix $R$ in Eq. (\ref{eq: SymJordan}) the symmetric Jordan form of $A$ \footnote{ Note that using Takagi's factorization for symmetric matrices \cite{HornJohnsonMAnalysis}, it can be easily shown that for any symmetric $Z$ there exists a corresponding state  as $B_\Psi=W \sqrt{D}$ leads to $Z=W D W^T$.}.

Note that $O$ in Eq. (\ref{eq: SymJordan}) is complex orthogonal, but not necessarily of determinant one. However, the case, where $|O|=-1$ for some state $\Psi$ can be easily circumvented by considering
the $Z$ matrix corresponding to a state where particle $1$ and $2$ are exchanged, i.e. $P_{12} \ket{\Psi}$. That is, if $Z_{A_{\Psi}}=O R O^T$, where $|O|=-1$, then $Z_{P_{12}A_{\Psi}}=O^\prime R (O^\prime)^T$, where $|O^\prime|=1$. In the following we will consider wlog for each SLOCC class a representing state whose corresponding $Z$ and $\tilde{Z}$ matrices are in symmetric Jordan form. These representatives will be called seed states.\\

According to the results summarized in Sec. \ref{SubSecSepMES}, it remains to determine the symmetries of the seed states. To this end, we compute the special orthogonal matrices that leave $Z$ ($\tilde{Z}$) (which have symmetric Jordan form) invariant up to a factor. To do so, we first transform $Z$ ($\tilde{Z}$) into Jordan form $J$ ($\tilde{J}$) respectively.
Then we determine all operators $X$ ($\tilde{X}$) and $q \in \C$  ($\tilde{q}\in \C$) such that $X J= q J X$ ($\tilde{X}\tilde{ J}= \tilde{ q}\tilde{ J}\tilde{ X}$) and impose on $X$ ($\tilde{X}$) the condition that $O=T X T^\dagger$ ($\tilde{O}=T \tilde{X} T^\dagger$) is a  special orthogonal matrix. As mentioned before, since neither the dimension of the Jordan blocks nor the eigenvalues can be changed via a similarity transformation this restricts the possible values for q to $q=1$ in most cases.
Let us now illustrate the method presented above with the help of a simple example. We consider the SLOCC classes $G_{abcd}$ for $b^2\neq c^2\neq d^2\neq b^2$, $a^2\neq b^2, c^2, d^2$ and under the constraint that the parameters fulfill the condition that there exists no $q\in \C\backslash 1$ such that $\{a^2,b^2,c^2,d^2\}=\{q a^2,q b^2,q c^2,q d^2\}$. The seed states are given in Eq. (\ref{seed0}). The corresponding $Z$ and $\tilde{Z}$ matrices are given by $Z=\tilde{Z}=diag(a^2,d^2,c^2,b^2)$. Note here that the Jordan form and the symmetric Jordan form coincide. It is easy to verify that $q$ must be equal to one and that the only special orthogonal matrices that leave $Z$ (and $\tilde{Z}$) invariant have to be diagonal with entries $\pm 1$ and $-1$ has to occur an even number of times. Any special orthogonal matrix with diagonal entries $\pm 1$ corresponds in the computational basis to a tensor product of Pauli operators \footnote{In fact, one can show that if $P$ is a non-trivial diagonal matrix with entries $\pm 1$ then $U P U^\dagger$ is of the form $\otimes_i A_i$ only if the number of $+1$ and the number of $-1$ in the diagonal is both $2^{n-1}$. In this case $\otimes_i A_i=\pm (i) \otimes \sigma_\k$, with $\sigma_{k_i}\in \{\sigma_x,\sigma_y\}$, or $\otimes_i A_i=\pm (i) \otimes \sigma_\k$, with $\sigma_{k_i}\in \{\one, \sigma_z\}$.}. Since $-1$ occurs here an even number of times, the corresponding possible local symmetries are $\one^{\otimes 2}$, $\sigma_x^{\otimes 2}$, $\sigma_y^{\otimes 2}$ and $\sigma_z^{\otimes 2}$. Applying all possible combinations of the symmetries to the seed state one observes that the non-trivial symmetries of the states are given by $\sigma_x^{\otimes 4}$, $\sigma_y^{\otimes 4}$ and $\sigma_z^{\otimes 4}$.

Note that if for example $a^2=-d^2$ and $b^2=-c^2$ additional symmetries arise as then, the choice $q=-1$ would be possible. Whereas we did not consider these cases in \cite{MESus}, we will present the solution to this class of states in Sec. \ref{2cycles2}. The formalism presented above can be used in general to determine the local symmetries of a 4--qubit state. However, it might be much more tedious to do so for different SLOCC classes (see subsequent sections).

Let us conclude this subsection by briefly commenting on how this method can be generalized to determine the symmetry of $n$--qubit states. First, note that in the 4-qubit case $U$, as given in Eq. (\ref{eq:Umb}) could equally well be defined as $\sqrt{-i}e^{i\pi/4\sigma_y^{\otimes 2}}$, which is LU--equivalent to $U$. This definition can be generalized to the $2n$-qubit case via $U=\sqrt{-i}e^{i\pi/4\sigma_y^{\otimes 2 n}}$. Then $U U^\dagger=\one$ and $U U^T=\sigma_y^{\otimes n}$ \footnote{ Note that for an odd number of systems, $n$, this definition cannot be used as in this case $(\sigma_y^{\otimes n})^T=- \sigma_y^{\otimes n}$ and therefore $U^T\propto U^\dagger,$ which implies that $U U^T\propto \one$.}. Using the subsequent observation it follows that if in the $2n$-qubit case $Z$ and $\tilde{Z}$ are defined analogous to the 4- qubit case then we have
\bea Z_{\otimes A_i \ket{\Psi}}=OZ_{\ket{\Psi}} O^T \\
\tilde{Z}_{\otimes A_i \ket{\Psi}}=\tilde{O}\tilde{Z}_{\ket{\Psi}} \tilde{O}^T,\eea

where both, $O$ and $\tilde{O}$ are complex orthorgonal matrices.
\begin{observation} $\forall A_i\in GL(2,\C)$, we have \bea O_c\equiv U^\dagger \bigotimes A_i U\eea
is a complex matrix with $|O_c|=\Pi_i |A_i|^{2^{n-1}}$ and $O_c O_c^T=  \Pi_i |A_i| \one$. \end{observation}

\begin{proof} Using the definition of $O_c$ we have \bea O_c O_c^T=U^\dagger \bigotimes A_i U U^T \bigotimes A^T_i U^\ast
=U^\dagger \bigotimes A_i \sigma_y^{\otimes n} \bigotimes A^T_i U^\ast= \Pi_i |A_i| U^\dagger \sigma_y^{\otimes n} U^\ast=\Pi_i |A_i| \one,\eea where we used that $A\sigma_y A^T=|A|\sigma_y$ for any $2\times 2$ matrix $A$ and that $UU^T=\sigma_y^n$ implies that $U^\dagger \sigma_y^{\otimes n} U^\ast=\one$.
\end{proof}
Hence, in the $2n$-qubit case (for $n>2$) possible symmetries can still be found by determining the orthogonal matrices that leave $Z$ and $\tilde{Z}$ invariant. However, it is no longer necessarily true that any of these special orthorgonal matrices corresponds to a local operator. Note further that for a $m$--qubit state, $\ket{\Psi}$ with $m$ odd, one could consider the $m+1$--qubit state $\ket{\Psi}\ket{0}$ and determine the symmetries of this states via the method explained above.

\subsection{\label{outline}Outline of proceeding sections and results}

Let us give here an outline of the proceeding sections, which are all structured in the subsequently described way.
We consider the SLOCC classes presented in \cite{slocc4} and use the notation given in that paper to the individual families of SLOCC classes. Note that this classification has to be considered up to permutations. That is, if there exists a particle permutation, $P$, such that the state $P\ket{\Psi}_{1234}$ is in some SLOCC class we treat $\ket{\Psi}$ as an element of this particular SLOCC class.

The class $G_{abcd}$ for  $b^2\neq c^2\neq d^2\neq b^2$, $a^2\neq b^2, c^2, d^2$ and where the parameters fulfill the condition that there exists no $q\in \C\backslash 1$ such that $\{a^2,b^2,c^2,d^2\}=\{q a^2,q b^2,q c^2,q d^2\}$, which constitutes a generic set of four qubit states, has been studied in \cite{MESus} (see also Sec. \ref{SubSecSepMES}). Here, we consider all the remaining classes. We devote to each of the nine classes its own section. All but one class, the one corresponding to the $L$--state, can be treated in the same way. The $L$--state however, exhibits a different behavior as we will briefly discuss in Section \ref{cycle3}. For each class, but the one corresponding to the latter, we proceed as follows.
\bi
\item We identify the seed states in each class as those states, $\ket{\Psi}$, for which $Z_\Psi$ and $\tilde{Z}_{ \Psi}$ are in symmetric Jordan form \footnote{Note that this state is in general not unique. However, the seed state is always chosen such that it is additionally in the MES.}.
\item We determine the local symmetries of the seed states as explained in Sec \ref{SubsecSym}. This is sometimes tedious, but straightforward.
\item Using the symmetries of the seed states we introduce a standard form up to LU's of the states within each SLOCC class. More precisely, we use that $\otimes_i g_i \ket{\Psi}=\otimes_i (g_i s_i) \ket{\Psi}$ where $\otimes s_i\in S(\Psi)$, in order to identify the standard form by choosing a particular symmetry, $\otimes_i s_i$, such that $\tilde{G}_i=(g_i s_i)^\dagger(g_i s_i)$ is of a particular form. Note that there is still some ambiguity since $G_i$ defines $g_i$ up to LU's, which are then determined by choosing $g_i=\sqrt{\tilde{G}_i}$. It should be noted here that this standard form can be used to decide very easily whether two states are LU--equivalent to each other or not. In fact, as the standard form explained above is unique, two states are LU--equivalent iff their corresponding standard forms coincide, resembling the criterion for LU--equivalence in the bipartite case.
\item We identify those states which are not reachable via LOCC. This is done using the result on the existence of SEP transformation mapping one state, $\ket{\Psi_1}$ to $\ket{\Psi_2}$. Clearly all states which cannot be reached from any other state via SEP must necessarily be in $MES_4$ (it can be easily seen in all classes that the seed state is always in $MES_4$). For all the remaining states we derive a LOCC protocol which transforms one state in $MES_4$ to the desired state. That is, the condition of the existence of a SEP transformation is already so stringent that whenever such a transformation is possible there also exists a LOCC protocol achieving the same task.
\item We identify the non--isolated states in $MES_4$, i.e. those states which can be transformed non--trivially into some other state outside of $MES_4$. Note that these states are the most relevant states for state transformations.
\ei

The determination of the symmetries and the identification of the MES is in some cases very tedious. Note however, that despite the fact that the
proof methods used to identify the (non--isolated) states in $MES_4$ differ for the various SLOCC classes, the results, which we are going to summarize in the following are very similar for each class (with exception of the L--state). However, it should be noted that this apparent similarity stems from our choice of the seed states and the standard form within the SLOCC class, whose choice is not necessarily unique.

Let us in the following summarize our findings. As before, we call a state reachable if there exists a non--LU equivalent state which can be transformed into it via LOCC. We show that for any three and any four--qubit state (excluding the L--state) the following holds:
\bi
\item[(i)] Any reachable state can be reached from some non-LU-equivalent state via a two-outcome POVM, where only one party applies a measurement  \footnote{Note that this does not mean that only transformation involving 2-outcome POVMs are possible. As an example consider the generic 4-qubit states $h_1\otimes \one^{\otimes 3}\ket{\Psi}$ where $h_1\not\propto \one$ which can be reached from $\ket{\Psi}$ via a 4-outcome POVM. However, these states can also be reached from $g_1^x\otimes \one^{\otimes 3}\ket{\Psi}$ via a two-outcome POVM.}.
\item[(ii)] This two-outcome POVM stems from a unitary symmetry of the seed state. More precisely, the two elements of the POVM are of the form
$ M_1=\sqrt{p} h g^{-1}, M_2=\sqrt{1-p} h S g^{-1}$,which act non-unitary only on one party, say party 1 and $S$ is a local symmetry (for most cases $S=\sigma_{i}^1 \otimes_{j=2}^4 U_j$, with $i\in \{1,2,3\}$ and $U_j$ local unitaries).
\item[(iii)] Denoting by $\ket{\Psi}$ the seed state of a SLOCC class, a state is reachable iff its standard form (defined individually for each SLOCC class) is of the form $\bigotimes_{i=1}^4 h_i \ket{\Psi}$, where the operators $H_i$ obey the following condition. There exist a unitary $S=\otimes_{i=1}^4 s_i \in S(\ket{\Psi})$, such that exactly three out of the four operators $H_i$ commute with the corresponding operators $s_i$.\ei

It is worth stressing again that the fact that the results can be summarized in this way depends very strongly on the choice of the seed states as well as the standard form for each SLOCC class. This is also why we present in all the lemmata the states in standard form. Hence, in order to see whether a state is for instance reachable one has to consider the state in standard form. Note that the seed state is always chosen such that it is in $MES_4$, which can be easily seen within each class. Given those facts, it is very easy to read off the states in $MES_4$ as well as the non-isolated states in $MES_4$.

Before investigating the various SLOCC classes, let us make some general observations. The first one concerns a simple sufficient condition for a state to be in MES, which we present in the following.

\begin{observation} Any state, which has the property that all single qubit reduced states are completely mixed is in MES (up to LUs).\end{observation}
\begin{proof} Let $\ket{\Psi}$ be such that the single qubit reduced states, $\rho_i\equiv \tr_{\neq i}(\proj{\Psi})$, are all completely mixed, i.e. $\rho_i=\one/2$ $\forall i$. Then, since for a pure state the entropy of a single qubit reduced state measures the entanglement between this qubit and the rest, the existence of a LOCC protocol transforming a state $\ket{\Phi}$ into $\ket{\Psi}$ implies that $S(\sigma_i)\equiv S(\tr_{\neq i}(\proj{\Phi}))\geq S(\rho_i)=1$ $\forall i$. Thus, also $\ket{\Phi}$ must have the property that $\sigma_i=\one/2$. In \cite{Bennett99} it has been shown that two states, which have the same local entropies are either $LU$--equivalent, or LOCC incomparable, which proves the statement. \end{proof}

Thus, any connected Graph state as well as any error correcting code is in MES. We will use this observation in the subsequent sections.

Finally let us mention that if the symmetry of the seed state consists exclusively of unitary operators, then due to Eq. (\ref{EqSep}) we have $r=1$, which can be easily seen by taking the trace of Eq.(\ref{EqSep}). Moreover, the necessary condition given by Eq. (\ref{EqSep}) implies that
\bea \label{eq_symLU} {\cal E}_4(H)=\otimes_{i=1}^4 {\cal E}_i(H_i),\eea
where ${\cal E}_4(H)$ is given by the left hand side of Eq. (\ref{EqSep}) and \bea \label{eq_symLUsingle} {\cal E}_i(H_i)=\sum_j p_j (s_i^j)^\dagger H_i s_i^j,\eea
where we used the notation $S_j=\otimes_i s_i^j$.
This can be easily seen considering the partial traces of Eq. (\ref{EqSep}). Furthermore, due to Uhlmann's theorem on unital maps \cite{uhlmann}, the vector containing the eigenvalues of $G_i$ for each $i$, which are $1/2\pm|\textbf{g}_i|$, must be majorized by the corresponding vector for $H_i$. This implies the following observation.
\begin{observation}\label{obsmon} When the symmetries to be considered in Eq.\ (\ref{EqSep}) are all unitary, then, for each $i$, $|\textbf{g}_i|$ cannot decrease under SEP or LOCC transformations. Moreover, in this case $r$ in Eq. (\ref{EqSep}) has to be 1.\end{observation}

As we will see, the Pauli operators play an important role in characterizing reachable states and states in $MES_4$. Let us therefore briefly mention some useful ideas which can be used in case all symmetries are tensor products of Pauli operators. First note that for any tensor product of Pauli operators, $P=\bigotimes_l \sigma_{k_l}$, where $k_l\in \{0,\ldots,3\}$, which commutes with all the elements of $S(\Psi)$, Eq. (\ref{eq_symLU}) implies that
\bea \label{eq_Pauli} \tr(H P)=\tr[H (\bigotimes_{i=1}^4 {\cal E}^\dagger_i) (P)]= \tr(H P) \prod_l \eta_{k_l},\eea
where we used that $\tr({\cal E}_4 (H)P)=\tr(H P)$ and the notation ${\cal E}^\dagger(B)=\sum_i p_i A_i^\dagger B A_i$, for a CPM ${\cal E}(B)=\sum_i p_i A_iB A_i^\dagger $. Here, $\eta_{k_l}=\tr( {\cal E}^\dagger_l (\sigma_{k_l}) \sigma_{k_l})/2= \sum_j p_j (-1)^{f_{j,k_l}+1}$, where $f_{j,k_l}=1$ for $[s_l^j,\sigma_{k_l}]=0$ and $f_{j,k_l}=0$ for $\{s_l^j,\sigma_{k_l}\}=0$.
For instance, if $S(\Psi)=\{\sigma_i^{\otimes 4}\}$ then $\eta_{i}=p_0+p_i-p_k-p_l$, where $\{i,k,l\}=\{1,2,3\}$.

Eq. (\ref{eq_Pauli}) can only be fulfilled if either (i) $\tr(H P)=0$ or (ii) $\prod_l \eta_{k_l}=1$. As $\eta_{k_l}\leq 1$ for any $k_l$, (ii) implies that any $\eta_{k_l}$ occurring in Eq. (\ref{eq_Pauli}) must be one, which implies that the only non-vanishing $p_i$ occurring in $\eta_{k_l}$ must have the same sign. Obeying this condition can already be very stringent. In order to see this, let us again consider the example above, where $S(\Psi)=\{\sigma_i^{\otimes 4}\}$. There we have that, as soon as there is one $k_l\neq 0$ occurring in Eq. (\ref{eq_Pauli}), this implies that the only possible transformation to reach the state must be a two--outcome POVM, i.e. only two probabilities are not vanishing. States which are reachable via a two--outcome POVM can be easily characterized, as one can wlog assume (without restricting the standard form) that the POVM elements are of the form $\sqrt{p}g^{-1} h$ and $\sqrt{1-p} g^{-1}S_i h$, for some $i$. Thus, case (ii) can be easily solved. In the remaining case, case (i), we have that $\tr(H P)=\prod_l \bar{h}_l^{k_l}=0$ has to hold. Considering now all Pauli operators which commute with all the symmetries leads to several strong conditions. In the subsequent sections we will use these ideas to derive simple conditions for the state to be reachable in case the symmetries contain only Pauli operators.

\section{The SLOCC classes $G_{abcd}$}
Let us start with the SLOCC classes $G_{abcd}$.
The seed states are
\begin{align}
|\Psi_{adcb}\rangle&=\frac{a+d}{2}(|0000\rangle+|1111\rangle)+\frac{a-d}{2}(|0011\rangle+|1100\rangle)+\frac{b+c}{2}(|0101\rangle+|1010\rangle)+\frac{b-c}{2}(|0110\rangle+|1001\rangle),\label{seed}
\end{align}
where $a,b,c,d \in \C$. Note that the seed states can equivalently be written as
\begin{align}
|\Psi_{adcb}\rangle&=a|\Phi^+\rangle|\Phi^+\rangle+d|\Phi^-\rangle|\Phi^-\rangle\nonumber+b|\Psi^+\rangle|\Psi^+\rangle+c|\Psi^-\rangle|\Psi^-\rangle.
\end{align}
From this decomposition it is obvious that the sign of any of the parameters can be chosen independently via the application of Pauli operators and permutations, e.g. changing the sign of $a$ corresponds to applying $\one\otimes\sigma_y\otimes\one\otimes\sigma_y$, then exchanging party $1$ and $2$ and then applying $\one\otimes\sigma_y\otimes\one\otimes\sigma_y$ again. As our studies are up to permutations and LUs this shows that the sign of the parameters is irrelevant.
The operators $B_\psi$, $Z_\psi$ and $\tilde{Z}_\psi$ corresponding to the seed states are given by (see Eq. (\ref{eq:Z}) and (\ref{eq:tilZ}))
\begin{align}\label{eigenvalues}
B_\psi&=\textrm{diag}(a,d,c,b),\nonumber\\
Z_\psi=\tilde{Z}_\psi&=D=\textrm{diag}(a^2,d^2,c^2,b^2).
\end{align}
Note that the order of the eigenvalues in $Z_\psi$ and $\tilde{Z}_\psi$ can be chosen arbitrarily. The reason for this is that the order of the eigenvalues
can be changed by a permutation, $P$, (i. e. $D\rightarrow PDP^T$ ) and that one can always
choose P such that $|P| = 1$ by changing one entry
from $1$ to $-1$ if necessary. Thus, these permutation matrices correspond to local unitaries in the computational basis. To determine now the symmetries of $|\Psi_{adbc}\rangle$ we need to find the matrices $O\in SO(4)$ and $q\in \C$ such that
\begin{equation}\label{sym}
ODO^T=qD.
\end{equation}
Since the eigenvalues are preserved by similarity transformations, this can only hold if $q\{\lambda_i\} =\{\lambda_i\}$ where $\{\lambda_i\}$ denotes here the set of eigenvalues of $D$. Thus, $q=1$  is always possible, but for certain configurations of eigenvalues we also have to take other values of $q$ into account. In particular, $q\neq 1$ is possible if  subsets of eigenvalues of $D$ fulfill the cyclic property $\lambda_{i+1}=q\lambda_i$ (if the cardinality of the subset is $n$, it should be understood that $n+1=1$) and eigenvalues that are not elements of a subset with the cyclic property vanish. Note that this implies $q^{n}=1$. In the following we will call a subset of eigenvalues with this cyclic properties a cycle whose size is given by the cardinality of the subset. The generic case, which was studied in \cite{MESus}, corresponds to non-degenerate and non-cyclic eigenvalues. Here we study the other possibilities. Let us note here that the extreme case $a^2=b^2=c^2=d^2\neq 0$ corresponds to a biseparable state and is hence not investigated here.
\subsection{Cyclic eigenvalues}
\subsubsection{\label{2cycles2}Two cycles of size 2}

Here we study the case where there exist two cycles of size 2.
This is the case $D=diag(a^2,-a^2,c^2,-c^2)$ where $a^2\neq \pm c^2$ and $a,c\neq 0$, which corresponds to seed states (Eq. (\ref{seed})) satisfying $d=ia$ and $c=ib$ (for $a^2\neq \pm c^2$ and $a,c\neq 0$). Besides the symmetries one obtains for $q=1$ (as in the generic case $O$ is diagonal which leads to $\{\sigma_i^{\otimes 4}\}$), we can also have $q=-1$. Taking this into account one can show that the symmetries for these SLOCC classes are given by
\begin{align}\label{symcycle2}
S&=\{\one^{\otimes 4},\sigma_x^{\otimes 4},\sigma_y^{\otimes 4},\sigma_z^{\otimes 4},\one\otimes\sigma_z\otimes\sigma_x\otimes\sigma_y, \sigma_z\otimes\one\otimes\sigma_y\otimes\sigma_x,\sigma_x\otimes\sigma_y\otimes\one\otimes\sigma_z,\sigma_y\otimes\sigma_x\otimes\sigma_z\otimes\one\}.
\end{align}
In the following we will denote the elements of the symmetry group $S$ by $S_k=\otimes_i s_i^k$ for $k=0,\ldots, 7$. We choose the same numeration as in Eq. (\ref{symcycle2}), i.e. $S_0=\one^{\otimes 4}$, $S_1=\sigma_x^{\otimes 4}$ etc..
The standard form can be chosen similar to the generic case, i.e. imposing an order on the coefficients of the seed state and getting rid of the ambiguity due to the symmetry allows to make the standard form unique. In the following we will use the notation $G_i=g_i^\dagger g_i=1/2\one+\sum_j \bar{g}_i^j \sigma_j$. Moreover, in order to give a compact presentation of the results we will use the following unconventional notation. That is, whenever we want to give more information on the precise structure of $G_i$ we use as notation $G_i^{W}=1/2\one+\sum_j \bar{g}_i^j \sigma_j$ for $W\in\{W_1 W_2 W_3,\neq j\}$ and  $W_j \in \{j,0,A\}$ where $W_j=j$ ($W_j=0$) means $\bar{g}_i^j\neq 0$($\bar{g}_i^j= 0$) whereas $W_j=A$ denotes that $\bar{g}_i^j$ is arbitrary and in case $W$ corresponds to $\neq j$ this means that $\bar{g}_i^k\neq 0$ and/or $\bar{g}_i^l\neq 0$ where $k\neq l\neq j\neq k$ and $j,k,l\in\{1,2,3\}$. Since the symmetries in the generic case are a strict subset of the present case, all transformations achievable in the generic case are also achievable here. However, the extra symmetries make more transformations possible. We now have for example that any state of the form $h_1^{AAA}\otimes h_2^{\neq 3}\otimes h_3^{A00}\otimes h_4^{0A0}\ket{\Psi_{a(ia)c(-ic)}}$ is reachable by LOCC as it can be obtained for example from a state of the form $h_1^{AAA}\otimes g_2^{00A}\otimes h_3^{A00}\otimes h_4^{0A0}\ket{\Psi_{a(ia)c(-ic)}}$ when the second party implements a two-outcome POVM with elements proportional to $h_2^{\neq 3}(g_2^{00A})^{-1}$ and $h_2^{\neq 3}\sigma_z(g_2^{00A})^{-1}$ where $\bar{g}_2^3=\bar{h}_2^3$ (for the second outcome parties 3 and 4 have to apply $\sigma_x$ and $\sigma_y$ repectively). In the following lemma we show which states are reachable in these SLOCC classes.
\begin{lemma}\label{lemcycle2} The only states in the SLOCC classes $G_{abcd}$ where  $a^2= -d^2$, $b^2=-c^2$, $a^2\neq \pm c^2$ and $a,c\neq 0$ that are reachable via LOCC are given by $\otimes h_i \ket{\Psi_{a(ia)c(-ic)}}$ where $\otimes H_i$ obeys the following condition. There exists a symmetry  $S_k\in S$ such that exactly
three out of the four operators $H_i$ commute with the corresponding
operator $s_i^k$ and one operator $H_j$ does not commute with $s_j^k$.
\end{lemma}
Explicitely this means that the reachable states are of the form  $h_i^{\neq 0}\otimes \one^{\otimes 3}\ket{\Psi_{a(ia)c(-ic)}}$ and $h_i^{\neq m}\otimes h_j^{00W_m=A}\otimes h_k^{00W_m=A}\otimes h_l^{00W_m=A}\ket{\Psi_{a(ia)c(-ic)}}$ where $m\in\{1,2,3\}$ and $h_1^{AAA}\otimes h_2^{\neq 3}\otimes h_3^{A00}\otimes h_4^{0A0}\ket{\Psi_{a(ia)c(-ic)}}$ and $h_1^{AAA}\otimes h_2^{00A}\otimes h_3^{\neq 1}\otimes h_4^{0A0}\ket{\Psi_{a(ia)c(-ic)}}$ and $h_1^{AAA}\otimes h_2^{00A}\otimes h_3^{A00}\otimes h_4^{\neq 2}\ket{\Psi_{a(ia)c(-ic)}}$ and $h_1^{\neq 3}\otimes h_2^{AAA}\otimes h_3^{0A0}\otimes h_4^{A00}\ket{\Psi_{a(ia)c(-ic)}}$ and $h_1^{00A}\otimes h_2^{AAA}\otimes h_3^{\neq 2}\otimes h_4^{A00}\ket{\Psi_{a(ia)c(-ic)}}$ and $h_1^{00A}\otimes h_2^{AAA}\otimes h_3^{0A0}\otimes h_4^{\neq 1}\ket{\Psi_{a(ia)c(-ic)}}$ and $h_1^{\neq 1}\otimes h_2^{0A0}\otimes h_3^{AAA}\otimes h_4^{00A}\ket{\Psi_{a(ia)c(-ic)}}$ and $h_1^{A00}\otimes h_2^{\neq 2}\otimes h_3^{AAA}\otimes h_4^{00A}\ket{\Psi_{a(ia)c(-ic)}}$ and $h_1^{A00}\otimes h_2^{0A0}\otimes h_3^{AAA}\otimes h_4^{\neq 3}\ket{\Psi_{a(ia)c(-ic)}}$ and $h_1^{0A0}\otimes h_2^{A00}\otimes h_3^{\neq 3}\otimes h_4^{AAA}\ket{\Psi_{a(ia)c(-ic)}}$
and $h_1^{0A0}\otimes h_2^{\neq 1}\otimes h_3^{00A}\otimes h_4^{AAA}\ket{\Psi_{a(ia)c(-ic)}}$ and $h_1^{\neq 2}\otimes h_2^{A00}\otimes h_3^{00A}\otimes h_4^{AAA}\ket{\Psi_{a(ia)c(-ic)}}$.\\
The details of the proof can be found in Appendix A. Here we just give  a very short outline of the proof.
First we use Eq. (\ref{eq_Pauli}) for some choices of $P$ to show which states are reachable using only the symmetries $S_0$ and one $S_i$ for some $i\in\{1,2,3,4,5,6,7\}$. Then we show that all states that are reachable can be obtained via a two-outcome POVM from some other state.\\
In the following lemma we show which states in these SLOCC classes are convertible.
\begin{lemma} The only states in the SLOCC classes $G_{abcd}$ where $a^2 =-d^2$, $b^2 = -c^2$, $a^2 \neq \pm c^2$   and $a,c\neq 0$ that are convertible via LOCC are given by the reachable states and
\begin{eqnarray}
&\{h_1^{AAA}\otimes h_2^{00A}\otimes h_3^{A00}\otimes h_4^{0A0}\ket{\Psi_{a(ia)c(-ic)}}, h_1^{00A}\otimes h_2^{AAA}\otimes h_3^{0A0}\otimes h_4^{A00}\ket{\Psi_{a(ia)c(-ic)}}, \\ \nonumber
&h_1^{A00}\otimes h_2^{0A0}\otimes h_3^{AAA}\otimes h_4^{00A}\ket{\Psi_{a(ia)c(-ic)}}, h_1^{0A0}\otimes h_2^{A00}\otimes h_3^{00A}\otimes h_4^{AAA}\ket{\Psi_{a(ia)c(-ic)}}, \\ \nonumber
&h_1^{00W_m=A}\otimes h_2^{00W_m=A}\otimes h_3^{00W_m=A}\otimes h_4^{00W_m=A}\ket{\Psi_{a(ia)c(-ic)}}\},
\end{eqnarray}
where $m\in\{1,2,3\}$.
\end{lemma}
\begin{proof}
This can be shown using Eq. (\ref{EqSep}) and Lemma \ref{lemcycle2}. Note that for each reachable state it holds that $H_i=H_i^{00W_m=A}$ for at least two parties wlog we choose $i=1,2$. Tracing over all parties but party $i$ with $i\in\{1,2\}$ in Eq. (\ref{EqSep}) one obtains that $G_i$ also has to be of the form $G_i^{00W_m=A}$. In order to see that states obeying this necessary conditions are indeed convertible note first that there exists for each of them, a symmetry $S_l$ which commutes with $G_j$ on at least 3 parties, e.g. $j=1,2,3$. Thus, one can use $\{\sqrt{p} h_4\one^{\otimes 4}g_4^{-1},\sqrt{1-p} h_4S_lg_4^{-1}\}$ to convert this state to some other. This constitutes a valid POVM for $\bar{h}_4^k=\bar{g}_4^k$, $(2 p -1 )\bar{h}_4^s=\bar{g}_4^s$ and $(2 p -1 ) \bar{h}_4^t=\bar{g}_4^t$, where $s_4^l=\sigma_k$, $s\neq t$ and $s,t\neq k$.  The remaining condition is that $G_4$ ($H_4$) is a positive rank-2 operator, so it is always possible to find for any $G_4$ a $H_4$ and a value of $p$ such that these restrictions are fulfilled. Note that in case $S_l$ also commutes with $G_4$ one simply has to set $p=1/2$.
\end{proof}
In order to illustrate how convertible states can be transformed we consider the following example. States of the form $h_1^{AAA}\otimes h_2^{00A}\otimes h_3^{\neq 1}\otimes h_4^{0A0}\ket{\Psi_{a(ia)c(-ic)}}$ are only reachable from states which obey $G_2^{00A}=1/2\one + \eta_{1267} \bar{h}_2^3\sigma_z$ and $G_4^{0A0}=1/2\one + \eta_{1356} \bar{h}_4^2\sigma_y$. Here we used the notation $\eta_{ijkl}=p_0+p_m+p_n+p_s+p_t-p_i-p_j-p_k-p_l$ where $\{m,n,s,t,i,j,k,l\}=\{1,\ldots,7\}$. In order to show that states which obey $G_2=G_2^{00A}$ and $G_4=G_4^{0A0}$ are indeed convertible consider $\{\sqrt{p} h_3\one^{\otimes 4}g_3^{-1},\sqrt{1-p} h_3S_4g_3^{-1}\}$. Note that the symmetry $S_4$ commutes with $G_i$ for $i=1,2,4$ and thus for these parties $G_i =H_i$. This transformation corresponds to a valid POVM for $\bar{h}_3^1=\bar{g}_3^1$, $(2 p -1 )\bar{h}_3^2=\bar{g}_3^2$ and $(2 p -1 ) \bar{h}_3^3=\bar{g}_3^3$ and as pointed out before for any $G_3$ it is always possible to find a probability $p$ and $H_3> 0$ so that the conditions are fulfilled. Thus, states with $G_2=G_2^{00A}$ and $G_4=G_4^{0A0}$ are convertible.\\
The states in these SLOCC classes, $\otimes_k g_k\ket{\Psi_{a(ia)c(-ic)}}$, that are non--isolated and are in $MES_4$  have the following property. There exists a symmetry $S_j$ with $j\in\{1,2,3,4,5,6,7\}$ such that $\otimes_k G_k$ commutes with $S_j$ and there exists no symmetry $S_l$ with  $l\in\{1,2,3,4,5,6,7\}$ such that exactly three out of the four operators $G_k$ commute with the corresponding operator $s_k^l$ and one does not commute.

\subsubsection{A cycle of size 4}

This is the case $D=diag(\lambda,i\lambda,-\lambda,-i\lambda)$, which corresponds to seed states satisfying $d=\exp(i\pi/4)a$, $c=ia$ and $b=\exp(i3\pi/4)a$. As $a$ is determined by the normalization condition this reduces to a single SLOCC class.
If one considers the 13 versus 24-splitting, $D$ in this case is no longer a 4-cycle but two two-cycles. In fact permuting party 2 and 3 leads to a seed state that is LU-equivalent to the one considered in the previous section and therefore this case has been studied there.

\subsubsection{A cycle of size 3\label{cycle3}}

In this case we have $D=diag(\exp(i2\pi/3)\lambda,\lambda,0,\exp(i4\pi/3)\lambda)$ for which the seed state fulfills $a=\exp(i\pi/3)d$, $b=\exp(i2\pi/3)d$ and $c=0$, i.e the seed state corresponds to
\bea
|\Psi\rangle&=d(|\Phi^-\rangle|\Phi^-\rangle+e^{i\pi/3}|\Phi^+\rangle|\Phi^+\rangle
+e^{i2\pi/3}|\Psi^+\rangle|\Psi^+\rangle).
\eea
 As before, this corresponds then to a single SLOCC class. Note that this state has been denoted as the $\ket{L}$-state in \cite{Lstate}, where it has been shown that this state maximizes the average Tsallis $\alpha$-entropy of entanglement for $\alpha> 2$, which also shows the particularity of this state. \\
Taking into account that $q\in\{1,\exp(i\pi/3),\exp(i2\pi/3)\}$ in Eq. (\ref{sym}) one obtains the following symmetry group
\bea
S(\Psi)=\{\one^{\otimes 4},\sigma_x^{\otimes 4},\sigma_y^{\otimes 4},\sigma_z^{\otimes 4},[Y(\pm \pi/4)X(\pm \pi/4)]^{\otimes 4},[Z(\pm \pi/4)X(\pm \pi/4)]^{\otimes 4}\}.
\eea
These symmetries allow to choose a unique standard form.
As mentioned before, this class resembles specific SLOCC classes of three-qutrit states, which have very different properties than all other classes of four qubits. These properties and their consequences, whose discussion are beyond the scope of this paper, are presented in \cite{HeSp15}. Let us note here, however, that similar results as in the other SLOCC classes are obtained in case more than 2 SLOCC operations are not proportional to the identity, e.g. in this case any reachable state $h\ket{\Psi}$ has the property that there exists a symmetry $S_k= s_k^{\otimes 4} \in S(\Psi) $ for some $k\in\{1, \ldots,12\}$ such that exactly
three out of the four operators $H_i$ commute with the corresponding operators $s_k$ and one operator $H_j$ does not commute with $s_k$. Moreover, the same techniques, in particular Lemma \ref{lemcoplanar1} and Lemma \ref{lemcoplanar2} (see below) can be used to derive these results.\\ \\
We now move to study families with degenerate eigenvalues. There, we will also consider the 2-cyclic case with a double-degenerate zero eigenvalue (see section \ref{secE}).  Notice  that the cases of a single four-degenerate eigenvalue and of a triple-degenerate null eigenvalue need not be considered because these correspond to biseparable states.

\subsection{A non-zero triple-degenerate eigenvalue}\label{family1}

In this case the corresponding $Z$- and $\tilde{Z}$-matrices can be chosen as $D=\textrm{diag}(a^2,a^2,c^2,a^2)$ where $a^2\neq c^2$ and $a\neq 0$. Thus, clearly Eq.\ (\ref{sym}) can only hold for $q=1$. Notice that our arranging of eigenvalues is compatible with $a=-d=-b\neq\pm c$ in Eq.\ (\ref{eigenvalues}), which means that the seed states can be chosen to be of the form

\begin{equation}
|\Psi\rangle=\alpha|\Psi^-\rangle_{12}|\Psi^-\rangle_{34}+\beta|\Psi^-\rangle_{13}|\Psi^-\rangle_{24}.
\end{equation}
The symmetries of these seed states are given by $S=X^{\otimes4}$ where $X\in$ SL(2). Alternatively, the elements of the symmetry group can also be written in their singular value decomposition as $S=(UP_zV)^{\otimes4}$, where
\begin{equation}\label{pz}
P_z=\left(
      \begin{array}{cc}
        z & 0 \\
        0 & 1/z \\
      \end{array}
    \right),\quad z>0,
\end{equation}
and $U,V\in$ SU(2).

As before, any state in these SLOCC classes can be written as $g_1\otimes g_2\otimes g_3\otimes g_4|\Psi\rangle$. In order to fix the standard form, let us first assume that none of the operators $G_i=g_i^\dagger g_i$, $i\in \{1,2,3,4\}$ coincide. Then $U$ and $P_z$ are chosen such that $P_zU^\dag G_4UP_z\propto\one$ and $V$ is chosen such that
$V^\dag G_1V\in\textrm{span}\{\one,\sigma_z\}$, which determines $V$
up to a rotation $Z(\alpha)$, which can then be used to choose $G_2\in\textrm{span}\{\one,\sigma_x,\sigma_z\}$. Therefore, each state in the SLOCC class can be written as $g_1^z\otimes g_2^{x,z}\otimes g_3\otimes\one|\Psi\rangle$, where here and in the following we make use of the notation $g_i^{v,w}$ for $g_i$ in case $\textrm{tr}[\sigma_u(g_i^{v,w})^\dagger g_i^{v,w}]=0$ for $\{u,v,w\}=\{x,y,z\}$. We will refer to a state by specifying the three vectors $\{\g_{1},\g{_2},\g_{3}\}$ using the parametrization $G_i=\one/2+\sum_{j=1}^3 \bar{g}_i^j\sigma_j$. In case two or more operators are the same, i.e. $G_i=G_j$ for some $i,j\in \{1,2,3,4\}$ we define the standard form as the one where these operators are mapped to the identity and proceed as above. It should be noted here that a state $\one \otimes h^{\otimes 3}\ket{\Psi}$ is LU--equivalent to the state $\tilde{h}\otimes \one^{\otimes 3}\ket{\Psi}$ (which is the standard form of this state as defined here), a property resembling the bipartite case. The fact that this holds for an arbitrary $h$ is in contrast to any other SLOCC class.  As we will see, states of this form can be reached unless $h\propto\one$.
We need to consider one additional, very special case, for which the standard form is chosen in a particular form. In this case where after the transformation mentioned before, it holds that $ \bar{g}_1^3\neq 0,  \bar{g}_2^1\neq 0, \bar{g}_3^1\neq 0$, $\bar{g}_3^2= 0$ and $(\bar{g}_1^3-\bar{g}_2^3)/\bar{g}_2^1=(\bar{g}_1^3-\bar{g}_3^3)/\bar{g}_3^1$. In this case we will show in the following that the standard form can be chosen as $\one \otimes g_2^z\otimes g_3^{z}\otimes g_4^{x,z} |\Psi\rangle$. In order to see that, note that in this case there exists a $\alpha$ and $z$ such that the Pauli vectors corresponding to the transformed operators $X^\dagger G_i X$ with $X^\dagger=P_z e^{i\alpha \sigma_y}$ are all parallel for $i\in \{1,2,3\}$ with vanishing $y$ component \footnote{This can be seen as follows. Let us denote by $\tilde{\g}_i$ the Pauli vectors corresponding to $X^\dagger G_i X$ (after normalization), with $X^\dagger=P_z e^{i\alpha \sigma_y}$. We choose $z\in\R\backslash 0$ and $\alpha$ such that $2\alpha\neq k\pi$ for $k\in\N_0$. Then, as $\tilde{\g}_1\neq 0$, as otherwise $G_1\propto \one$, which is excluded here, all Pauli vectors are parallel iff $\tilde{\g}_1|| \tilde{\g}_2$ and $\tilde{\g}_1|| \tilde{\g}_3$, which is equivalent to $\tilde{\g}_1  \times \tilde{\g}_2=\tilde{\g}_1  \times \tilde{\g}_3=0$. As all vectors have vanishing $y$--components, the previous conditions are equivalent to $\tilde{\g}_1  \times \tilde{\g}_i$ having a vanishing $y$--component for $i=2,3$. Due to the fact that $(\bar{g}_1^3-\bar{g}_2^3)/\bar{g}_2^1=(\bar{g}_1^3-\bar{g}_3^3)/\bar{g}_3^1$ it can be easily seen that these conditions are fulfilled iff there exists a $\alpha$ and $z$ such that $[-2 \bar{g}_1^3 + \cos (2 \alpha) +(\bar{g}_1^3- \bar{g}_3^3)/\bar{g}_3^1 \sin (2 \alpha) ] /[
2 \bar{g}_1^3+ \cos (2 \alpha)+ (\bar{g}_1^3- \bar{g}_3^3)/\bar{g}_3^1 \sin (2 \alpha)]=z^4$. It is straightforward to see that there always exists a $\alpha$ such that the left--hand side of the previous equation is positive, which implies the existence of a real $z$ for which the equation is satisfied. Hence, there always exists an operator $X=(P_z e^{i\alpha \sigma_y})^\dagger$ such that the transformed operators have parallel Pauli vectors.}. Due to the fact that $G_4=\one$, $\tilde{G}_4=X^\dagger G_4 X=G_4^z$, i.e. the corresponding Pauli vector has a non--vanishing $z$--component. Transforming now with $e^{i\beta \sigma_y}$ all the parallel Pauli vectors, i.e. $\tilde{\g}_i$ for $i\in \{1,2,3\}$ into vectors along the $z$--axis for a properly chosen $\beta$, and applying $P_{z^\prime}$ to map the first operator into the identity leads to desired standard form \footnote{It should be noted here that the choice of $\alpha$ and $z$ in the first step of the transformation is not unique. However, it can be easily seen that the standard form is at the end unique.}. Note that the uniqueness of the standard form (for all discussed cases) is then insured by choosing the signs of some of the Pauli components of $G$. Note further that up to permutations of the operators, the standard form is always of the form $g_1^z\otimes g_2^{x,z}\otimes g_3\otimes\one|\Psi\rangle$. As permutations will not alter any of the following arguments, we will denote now by $\{\g_{1},\g{_2},\g_{3}\}$ the Pauli operators of the three operators occurring in the standard form which are not proportional to the identity.  With all that, we will now investigate the reachable states and the states in the MES for this SLOCC class. In order to do so we first fix the standard form in  Eq. (\ref{EqSep}) as described above. Depending on whether the identity acts on system $1$ or $4$ we consider this equation for $G$ having the identity at the same position. Note that this can be done without loss of generality as the state $\one\otimes g_2\otimes g_3\otimes g_4|\Psi\rangle$ is LU-equivalent to the state $\tilde{g}_1\otimes \tilde{g}_2\otimes \tilde{g}_3\otimes\one|\Psi\rangle$ for some properly chosen $\tilde{g}_i$. We will first of all show that a necessary condition for  Eq. (\ref{EqSep}) to be fulfilled using this assignment is that $r=1$ and will then solve these equations.

Due to Eq. (\ref{EqSep}) we have to find all $p_x, H_i$ and $G_i$ obeying the condition
\begin{equation}\label{rulefamily1}
\sum p_x (X^\dag)^{\otimes4}H_1\otimes H_2\otimes H_3 \otimes \one X^{\otimes 4}=r G_1\otimes G_2\otimes G_3\otimes \one.
\end{equation}

Throughout we will use that $A \sigma_y A^T =|A|\sigma_y$ for any $2\times 2$ matrix $A$, where here and in the following $|A|$ denotes the determinant of an operator $A$. This implies that $\ket{\Psi^-}$ is invariant under $X \otimes X$ for any $X \in SL(2)$ and that $ A\otimes B \ket{\Psi^-}=  |B| A B^{-1}  \otimes \one \ket{\Psi^-}$. The latter equation implies in particular that $\bra{\Psi^-} A\otimes B \ket{\Psi^-}=1/4-v_A\cdot v_B$, for any $A,B$ with $\tr(A)=\tr(B)=1$. Here, $v_A$ ($v_B$) denotes the Bloch vector of $A$ ($B$) respectively and $v_A\cdot v_B$ denotes the scalar product between $v_A$ and $v_B$. Computing the overlap of the right and the left hand side of Eq. (\ref{rulefamily1}) with $\ket{\Psi^-}_{ij}\ket{\Psi^-}_{k4}$, where $\{i,j,k\}=\{1,2,3\}$ leads to
\bea \label{eq_scP}
1-4 \h_{i}\cdot\h_{j}=r ( 1-4 \g_{i}\cdot\g_{j}),\eea
for $i,j\in \{1,2,3\}, i\neq j$. Using now that the seed state, $\ket{\Psi}$ is invariant under $X^{\otimes 4}$ Eq. (\ref{rulefamily1}) implies that
\bea \bra{\Psi} H_1\otimes H_2\otimes H_3 \otimes \one \ket{\Psi}=r \bra{\Psi} G_1\otimes G_2\otimes G_3 \otimes \one \ket{\Psi}.
\eea

Note that this equation can be easily rewritten as
\bea \label{Eq_r}\bra{\vec{\alpha}}M_H \ket{ \vec{\alpha}}=r \bra{\vec{\alpha}}M_G \ket{ \vec{\alpha}},\eea where $\ket{ \vec{\alpha}}=(\alpha,\beta)$ and $M_H$ ($M_G$) is a $2 \times 2 $ matrix with entries depending on the operators $H_i$ ($G_i$) respectively. Note that Eq. (\ref{Eq_r}) has to hold for any choice of the seed parameters $\alpha,\beta$, which is equivalent to $M_H-r M_G=0$. Using Eq. (\ref{eq_scP}) it is easy to see that the latter holds iff
\bea \label{eq_Spat}
\h_{1}(\h_{2}\times \h_{3})=r \g_{1}(\g_{2}\times \g_{3}).\eea

Let us now derive some necessary operator equations. Computing the overlap of Eq. (\ref{rulefamily1}) with $\ket{\Psi^-}_{k4}$, with $k\in \{1,2,3\}$ leads to the equations
\bea
\label{eq_HiHj}
\sum p_x (X^\dag)^{\otimes 2}H_i\otimes H_j X^{\otimes 2}=r G_i\otimes G_j,
\eea
with $\{i,j,k\}=\{1,2,3\}$. Computing now the action of the right and left hand side of this equation on $\ket{\Psi^-}$, using the symmetry properties of this state and the fact that $\one \otimes A \ket{\Psi^-}=\one \otimes B \ket{\Psi^-}$ iff $A=B$, we obtain the following necessary condition
\bea \label{eq_MHiHj}{\cal M} (|H_j| H_i H_j^{-1})=r |G_j| G_i G_j^{-1},\eea for $i,j\in \{1,2,3\}, i\neq j$. Here, ${\cal M}(A)=\sum p_x (X^\dag) A (X^\dagger)^{-1}$ which is in general not a completely positive map. Similarly one obtains by computing the overlap with $\ket{\Psi^-}_{jk}$, with $j,k \in \{1,2,3\}$ and using Eq. (\ref{eq_scP}) the following conditions
\bea
\label{eq_oneHi}
\sum p_x (X^\dag)^{\otimes 2}H_i\otimes \one X^{\otimes 2}= G_i\otimes \one \quad \forall i\in \{1,2,3\},
\eea
which implies [using the same techniques as in deriving Eq. (\ref{eq_MHiHj})] that
\bea
\label{eq_Hi}
{\cal M} (H_i)=G_i \quad \forall i\in \{1,2,3\}.
\eea

Interestingly Eq. (\ref{eq_oneHi}) is exactly the equation one would obtain when investigating SEP transforming the bipartite state $g_i\otimes \one \ket{\Psi^-}$ into the state $h_i\otimes \one \ket{\Psi^-}$. The reason for that is that first off all the norms of the last two states coincide and second that the symmetry of state $\ket{\Psi^-}$ is precisely $X^{\otimes 2}$, for $X \in SL(2)$. Hence, Eq. (\ref{EqSep}) for this state transformation coincides with Eq. (\ref{eq_Hi}). In the case investigated here, the transformation must not only be possible via SEP, but via LOCC. Hence we can use the criterion of possible LOCC transformations among bipartite states given in \cite{Nielsen}. Applied to the case considered here, we obtain the following necessary condition for LOCC convertibility
\bea |\g_{i}|\leq |\h_{i}|\quad \forall i\in \{1,2,3\}.
\eea

Next, we use the standard form explained above in order to derive the action of the linear map ${\cal M}$ on the Pauli operators from Eqs. (\ref{eq_Hi}). We obtain

\bea \label{eq_MPauli} {\cal M}(\one)&=&\one \\ \label{eq_MPaulinnontriv}
\bar{h}_1^3 \bar{h}_2^1 {\cal M}(\sigma_x)&=&\bar{h}_1^3 \bar{g}_2^1 \sigma_x +(\bar{h}_1^3 \bar{g}_2^3- \bar{h}_2^3 \bar{g}_1^3)\sigma_z\\ \nonumber
\bar{h}_1^3 \bar{h}_2^1 \bar{h}_3^2 {\cal M}(\sigma_y)&=&\bar{h}_1^3(\bar{h}_2^1 \bar{g}_3^1-\bar{h}_3^1 \bar{g}_2^1)\sigma_x +\bar{h}_1^3 \bar{h}_2^1 \bar{g}_3^2\sigma_y +[\bar{h}_2^1(\bar{h}_1^3 \bar{g}_3^3-\bar{h}_3^3 \bar{g}_1^3)-\bar{h}_3^1(\bar{h}_1^3 \bar{g}_2^3-\bar{h}_2^3 \bar{g}_1^3)] \sigma_z\\ \nonumber
\bar{h}_1^3 {\cal M}(\sigma_z)&=&\bar{g}_1^3 \sigma_z \nonumber.\eea

Combining Eq. (\ref{eq_MHiHj}) and Eq. (\ref{eq_Hi}) it follows that if there exists an $i\in \{1,2,3\}$ such that $H_i=\one/2$, $r=1$ must be fulfilled. As this case will be treated below, we can therefore assume in the following that there exists no $i\in \{1,2,3\}$ such that $H_i=\one/2$. Note that due to the choice of the standard form, this implies in particular, that we can assume in the following that there exists no $i,j\in \{1,2,3\}$ such that $H_i=H_j$, as otherwise the standard form would contain two identities.

Note that if the prefactors in all equations given in Eq. (\ref{eq_MPaulinnontriv}) are non--vanishing the map ${\cal M}$ is uniquely defined and we can associate to it a matrix $M=\sum_{i,j=1}^3 M_{ji} \ket{i}\bra{j}$, where ${\cal M}(\sigma_i)=\sum_j M_{ij} \sigma_j$ \footnote{Note that we do not consider here the action of ${\cal M}$ on the identity as this equation does not lead to any new condition.}. Using that the Pauli vector corresponding to $|H_j|H_i H_j^{-1}$ is given by $\h_{i}-\h_{j}-2i(\h_{i}\times \h_{j})$, Eq. (\ref{eq_MHiHj}) read as
\bea \label{eq_MMhihj}
M[\h_{i}-\h_{j}-2i(\h_{i}\times \h_{j})]=r [\g_{i}-\g_{j}-2i(\g_{i}\times \g_{j})], \quad \forall i\neq j \eea
and Eq. (\ref{eq_Hi}) read as
\bea \label{eq_MMhi} M(\h_{i})=\g_{i}.\eea

As the matrix $M$ has real entries [see Eq. (\ref{eq_MPaulinnontriv})] we obtain from considering the real and the imaginary part of Eq. (\ref{eq_MMhihj}) the necessary conditions
\bea
\label{eq_MMhihjImRe}
M(\h_{i}-\h_{j})=r (\g_{i}-\g_{j}), \\
M(\h_{i}\times \h_{j})=r \g_{i}\times \g_{j} \quad \forall i\neq j. \eea

Combining the first of these equations with Eq. (\ref{eq_MMhi}), we obtain as necessary condition that $r=1$ unless $\g_{i}-\g_{j}=0$ for any pair $i,j$. The latter case is however not possible as Eq. (\ref{eq_Spat}) would imply that $\bar{h}_1^3\bar{h}_2^1\bar{h}_3^2=0$, which is excluded here.

Let us now consider the remaining cases, where one of the prefactors in Eqs. (\ref{eq_MPaulinnontriv}) vanishes. If $\bar{h}_1^3=0$ we have that $H_1\propto \one$, which implies, as argued above that $r=1$. Hence, we can assume in the following that $\bar{h}_1^3\neq 0$. If now $\bar{h}_2^1=0$ we obtain from the first equation in Eq. (\ref{eq_MPaulinnontriv}) that $\bar{g}_2^1=0$. Of course we can also in this case associate a matrix $M$ to the linear map ${\cal M}$. Considering now Eq. (\ref{eq_MMhihj}) for $i=1$ and $j=2$ and Eq. (\ref{eq_MMhi}) for $i=1$ and $i=2$ (for the corresponding matrix $M$) and taking into account that both, $\h_{1}\times \h_{2}$ and $\g_{1}\times \g_{2}$ vanish leads to $r=1$ unless $\g_1=\g_2$. In the latter case Eq. (\ref{eq_MMhi}) implies that either $H_1=H_2$, which leads to $r=1$ or $G_1=G_2\propto \one$. In this case Eqs. (\ref{eq_scP}) lead to either $H_1=H_2$ or one of the $H_i$ for $i\in\{1,2,3\}$ is proportional to the identity. In both cases this leads to $r=1$.

Hence, it remains to consider the case $\bar{h}_3^2=0$ (and $\bar{h}_1^3, \bar{h}_2^1\neq 0$). Note that in this case all equations are invariant under the exchange of $H_2$ with $H_3$. Hence, we can assume that $\bar{h}_3^1\neq 0$. As before one can associate to the linear map $\mathcal{M}$ a matrix $M$ (which is, however, not real). It is straightforward to see that for the corresponding $M$ Eqs. (\ref{eq_MMhi}) and Eqs. (\ref{eq_MMhihj}) are fulfilled iff one of the following three conditions are fulfilled: (i) $(\bar{h}_1^3-\bar{h}_2^3)/\bar{h}_2^1=(\bar{h}_1^3-\bar{h}_3^3)/\bar{h}_3^1$ or (ii) $G_i\propto \one$ $\forall i$ or (iii) r=1. Case (ii) will be treated below. Let us now consider case (i). Eq. (\ref{eq_MMhihj}) and Eq. (\ref{eq_MMhi}) imply that either the components of $g_i$ obey the same conditions or $\bar{g}_2^1=\bar{g}_3^1=0$. In the first case recall that the standard form of $g|\Psi\rangle$ and $h|\Psi\rangle$ is $\one \otimes h_2^z\otimes h_3^{z}\otimes h_4^{x,z} |\Psi\rangle$. In the second case the standard form of $h|\Psi\rangle$ is the same as before and therefore due to our assignment $g|\Psi\rangle$ has to be transformed into $\one\otimes g_2^z(g_1^z)^{-1}\otimes g_3^z(g_1^z)^{-1}\otimes g_4^z(g_1^z)^{-1}|\Psi\rangle$. Hence, in both cases Eq. (\ref{EqSep}) has to be considered for $\one/2 \otimes H_2^z\otimes H_3^{z}\otimes H_4^{x,z}$ and similarly for $G$. Up to permutations this corresponds to $H_2^z \otimes H_3^{z}\otimes H_4^{x,z} \otimes\one/2$. Then, either $H_2$ or $H_3$ are proportional to $\one$, which implies that $r=1$ (see above), or one of the $z$--components in non--vanishing, but the $x$--component of the second operator is vanishing, i.e. the case $\bar{h}_1^3\neq 0$ and $\bar{h}_2^1=0$ from before. As permutations of the particles do not alter the argument from above, we have that these states are reachable only if $r=1$. Hence, it remains to consider case (ii). As $\g_i=0$ the equations in Eq. (\ref{eq_scP}) imply that (using that $\bar{h}_1^3\neq 0$) $\bar{h}_2^3=\bar{h}_3^3$. We have to distinguish here the two cases, $\bar{h}_2^3=0$ and $\bar{h}_2^3\neq 0$. In the first case Eq. (\ref{eq_scP}) implies that $r=1$. In the second case Eq. (\ref{eq_scP}) implies that $\bar{h}_1^3=\bar{h}_2^3=\bar{h}_3^3$. Note that this implies that the condition in (i) is satisfied, which leads to the necessary condition that $r=1$.

Hence, we have shown that given the standard form defined at the beginning of this subsection a state is reachable only if Eq. (\ref{EqSep}) is satisfied for $r=1$. Let us now consider this case and derive the necessary and sufficient conditions for states to be reachable and convertible. From Eq. (\ref{eq_scP}) and Eq. (\ref{eq_Spat}) we have the following necessary conditions,
\begin{equation}\label{scalar}
\g_{i}\cdot\g_{j}=\h_{i}\cdot\h_{j},
\end{equation}
 and
\begin{equation}
\g_{1}\cdot(\g_{2}\times\g_{3})=\h_{1}\cdot(\h_{2}\times\h_{3}).\label{vector}
\end{equation}

Moreover, it must hold that
\bea |\g_{i}|\leq |\h_{i}|\quad \forall i\in \{1,2,3\}.
\eea

We will make use now of the following lemmata.

\begin{lemma}\label{lemcoplanar1} Let $\{\v_1,\v_2,\v_3\}$ be a triplet of coplanar pairwise-nonparallel vectors with scalar products $c_{ij}=\v_i\cdot\v_j$. Then, the only other triplets of coplanar vectors $\{\v'_1,\v'_2,\v'_3\}$ with the same scalar products ($\v'_i\cdot\v'_j=c_{ij}$) fulfilling $||\v'_i||\geq||\v_i||$ $\forall i$ are uniform rotations of the triplet of vectors $\{\v_1,\v_2,\v_3\}$.
\end{lemma}
\begin{lemma}\label{lemcoplanar2}
Let $\{\v_1,\v_2,\v_3\}$ be a triplet of coplanar vectors with scalar products $c_{ij}=\v_i\cdot\v_j$ in which $\v_i\|\v_j$ for some $i\neq j$ or $\v_i=0$ for at least one value of $i$. Then, any other triplet of coplanar vectors $\{\v'_1,\v'_2,\v'_3\}$ with the same scalar products ($\v'_i\cdot\v'_j=c_{ij}$) fulfilling $||\v'_i||\geq||\v_i||$ $\forall i$, must fulfill that $\v'_i\|\v'_j$ or $\v'_i=0$.
\end{lemma}
The proofs of Lemma \ref{lemcoplanar1} and Lemma \ref{lemcoplanar2} can be found in Appendix B. Using these lemmata and the fact that $r=1$ has already been shown to be a necessary condition for LOCC conversion (using the standard form and our assignment) we will now prove the following.
\begin{lemma}\label{thfamily1}
A LOCC conversion from $g|\Psi\rangle$ to $h|\Psi\rangle$ in these SLOCC classes is possible only if $\g_{i}\parallel\g_{j}$ and $\h_{i}\parallel\h_{j}$ for some $i\neq j$ or if $\g_{i}=\h_{i}=0$ for some $i$.
\end{lemma}
\begin{proof}
Notice that if $|\h_{i}|=|\g_{i}|$ $\forall i$, then Eqs.\ (\ref{scalar}) can only hold if $\h_i$ differs from $\g_i$ only by a rotation maintaining the angles among the three vectors. However, the standard form fixes one vector in the $z$ axis and another in the $xz$ plane, so this possibility needs to be discarded. Hence, $|\h_{i}|>|\g_{i}|$ for at least one $i$. In this case, however, one obtains the claim as otherwise Eqs.\ (\ref{scalar})-(\ref{vector}) cannot be simultaneously satisfied. To see this, notice that Eqs.\ (\ref{scalar}) can only hold if the angles between an increasing vector and the other two increase \footnote{Note that here and in the following we assume that the angle is smaller than $\pi/2$. Otherwise the angle can only decrease. However, the same arguments hold.}, while Eq.\ (\ref{vector}), which represents the volume of the parallelepiped having the three vector as sides, can only hold if at least one angle between a increasing vector and the other two decreases. Thus, the only possibility is that Eq.\ (\ref{vector}) is trivially fulfilled by being equal to zero,  which implies that the vectors $\g_{1}$, $\g_{2}$ and $\g_{3}$ have to be coplanar and the same holds true for the vectors $\h_{1}$, $\h_{2}$ and $\h_{3}$. Using Lemma \ref{lemcoplanar1} and Lemma \ref{lemcoplanar2} one obtains that in order for a transformation to be possible it has to hold that $\g_{i}\parallel\g_{j}$ and $\h_{i}\parallel\h_{j}$ for some $i\neq j$ or $\g_{i}=\h_{i}=0$ for some $i$.
\end{proof}

Notice that from this lemma we obtain as an immediate corollary that almost all states in these SLOCC classes are isolated. Furthermore, from the above proof (c.\ f.\ Eq.\ (\ref{scalar})), we see that the angle between $\h_{i}$ and $\h_{j}$ cannot decrease compared to that between $\g_{i}$ and $\g_{j}$. We will now show that LOCC conversions can actually always occur among states fulfilling the conditions of Lemma \ref{thfamily1} and identify the elements of the $MES$ for these families.

\begin{lemma}\label{MESfamily1}
The MES in these SLOCC classes is given by all states $g\ket{\Psi}$ for which neither $\g_{i}\parallel\g_{j}$  for some $i\neq j$ nor $\g_{i}=0$ for some $i$ (which are isolated) and by the LOCC convertible states for which one of the following conditions holds
\begin{itemize}
\item[(i)] $\g_{i}=0$ $\forall i$,
\item[(ii)] $\g_{i}\parallel\g_{j}$ and $\g_{k}=0$  with  $\g_{i}\neq 0$ and $\g_{j}\neq 0$ $\forall i\neq j\neq k\neq  i$,
\item[(iii)] $\g_{1}\parallel\g_{2}\parallel\g_{3}$ with $\g_{i}\neq 0$ $\forall i\in\{1,2,3\}$.
\end{itemize}
\end{lemma}
Recall that $\g_{1}=\g_{2}=\g_{3}$ is not contained in (iii), as in this case the standard form would be $\g_{1}=\g_{2}=\g_{3}=0$ and $\g_{4}\neq 0$, which will then be relabeled to a state with $\g_{1}\neq 0$.
\begin{proof}
Note that states of the form $h\ket{\Psi}$ for which neither $\h_{i}\parallel\h_{j}$  for some $i\neq j$ nor $\h_{i}=0$ for some $i$ are not fulfilling the premises of Lemma \ref{thfamily1} and are therefore in the MES.
From the analysis above, we have seen that LOCC transformations can only occur by increasing the length of the vectors and the angle between them (see Eqs.\ (\ref{scalar})). Therefore, necessarily, all states for which all vectors are zero and/or parallel are not LOCC reachable. Hence, the states given in (i)-(iii) must be in $MES_4$. It remains to prove that no other states are in $MES_4$. That is we show that no other state $h\ket{\Psi}$ for which  $\h_{i}\parallel\h_{j}$  for some $i\neq j$ or $\h_{i}=0$ for some $i$ (see Lemma \ref{thfamily1}) is in the MES. We will do so by providing LOCC protocols (sometimes implicitely) to transform a state in (i)-(iii) to any other state satisfying the premises of Lemma \ref{thfamily1} as listed below in (a)-(e).
\begin{itemize}
\item[(a)] States with $\h_{i}\neq0$, $\h_{j}=\h_{k}=0$ for some $i\neq j\neq k\neq i$:

It can be obtained from a state of the form (i) by party $i$ implementing the four-outcome POVM ($k=0,1,2,3$) $M_k\propto h_i\sigma_k$ and the other parties implementing $\sigma_k$ provided that outcome $k$ was obtained.
\item[(b)] States with $\h_{i}\perp\h_{j}$ and $\h_{k}=0$ for some $i\neq j\neq k\neq i$:

Take without loss of generality that $\h_{i}$ points in $z$ direction. Then, the state can be obtained from $\g_{i}=\h_{i}$ and $\g_{j}=\g_{k}=0$ by party $j$ implementing the two-outcome POVM $\{M_1,M_2\}\propto \{h_{j},h_{j}\sigma_z\}$ and the other parties implementing $\sigma_z$ in case of the second outcome.
\item[(c)] States with $\h_{i}\not\perp\h_{j}$ and $\h_{k}=0$ for some $i\neq j\neq k\neq i$:

Again, take without loss of generality that $\h_{i}$ points in $z$ direction. Then, the state can be obtained from a state of the form (ii) with $\g_{i}=\h_{i}$, $\g_{j}=(0,0,\bar{h}_{j}^3)$ and $\g_{k}=0$. It suffices that party $j$ implements the two-outcome POVM $\{M_1,M_2\}\propto \{h_{j}(g_{j})^{-1},h_{j}\sigma_z(g_{j})^{-1}\}$ and the other parties implement $\sigma_z$ in case of the second outcome.
\item[(d)] States with $\h_{i}\parallel\h_{j}\perp\h_{k}$ for some $i\neq j\neq k\neq i$:

We consider again that $\h_{i}$ (and $\h_{j}$) points in $z$ direction. Then, the state can be obtained from a state of the form (ii) with $\g_{i}=\h_{i}$, $\g_{j}=\h_{j}$ and $\g_{k}=0$. For that, party $k$ implements the two-outcome POVM $\{M_1,M_2\}\propto \{h_{k},h_{k}\sigma_z\}$ and the other parties implement $\sigma_z$ in case of the second outcome.
\item[(e)] States with $\h_{i}\parallel\h_{j}\not\perp\h_{k}$ for some $i\neq j\neq k\neq i$:

It suffices again to consider that $\h_{i}$ and $\h_{j}$ point in $z$ direction. The state can be obtained from a state of the form (iii) with $\g_{i}=\h_{i}$, $\g_{j}=\h_{j}$ and $\g_{k}=(0,0,\bar{h}_{k}^3)$. As before, party $k$ just needs to implement the two-outcome POVM $´\{M_1,M_2\}\propto \{h_{k}(g_{k})^{-1},h_{k}\sigma_z(g_{k})^{-1}\}$ and the other parties implement $\sigma_z$ in case of the second outcome.

\end{itemize}
\end{proof}
\subsection{A double-degenerate non-zero eigenvalue and two different non-degenerate eigenvalues}\label{family3}

In this case we have $D=diag(a^2,d^2,c^2,c^2)$ where $a^2\neq d^2\neq c^2\neq a^2$ and $c \neq 0$ and, therefore, Eq.\ (\ref{sym}) can only hold if $q=1$.  Notice that our seed states can be chosen to correspond to $b=c$, meaning
\begin{align}
|\Psi\rangle&=\frac{a+d}{2}(|0000\rangle+|1111\rangle)+\frac{a-d}{2}(|0011\rangle+|1100\rangle)+c(|0101\rangle+|1010\rangle).\label{seedfamily3}
\end{align}
The symmetries of these states are of the form $(\sigma_x^{\otimes4})^m(P_{z}\otimes P_{z^{-1}}\otimes P_{z^{-1}}\otimes P_{z})$, where $P_z$ is defined as in Eq.\ (\ref{pz}) but with $z$ being now an arbitrary nonzero complex number and $m\in\{0,1\}$. It will be more convenient to actually choose as a seed state $\ket{\Psi_{adcc}}=\one\otimes\sigma_x\otimes\sigma_x\otimes\one|\Psi\rangle,$ for which any element of the symmetry group can be written as $S_{z,m}=\otimes_i s_i^{z,m}=P_{z}^{\otimes4}X^m$ where $X=\sigma_x^{\otimes 4}$ and $m\in\{0,1\}$ \footnote{We have used here that $P_z\sigma_x=\sigma_xP_{z^{-1}}$.}. In the following, we will use the notation $G_i=\one/2+\sum_{j=1}^3 \bar{g}_i^j\sigma_j$ and similarly for $H_i$, as well as $\textrm{arg}(\bar{g}_j^1+i \bar{g}_j^2)=\phi_j$ with $\phi_j\in (-\pi,\pi]$ (in case $\bar{g}_j^1=\bar{g}_j^2=0$ we set $\phi_j=0$). Furthermore, we will use the notation $(g_i^{x})^\dagger g_i^{x}=\one/2+\bar{g}_i^1\sigma_x$.
We will also make use of the notation $g_i^{v,w}$ for $g_i$ in case $\textrm{tr}[\sigma_u(g_i^{v,w})^\dagger g_i^{v,w}]=0$ for $\{u,v,w\}=\{x,y,z\}$.\\
To fix the standard form, we will use that one can choose $|z|$ such that  $P_z^\dag G_iP_z\in\textrm{span}\{\one,\sigma_x,\sigma_y\}$ and $\alpha$ such that $Z(-\alpha) G_iZ(\alpha)\in\textrm{span}\{\one,\sigma_x,\sigma_z\}$. One might also choose $z$  such that $P_z^\dag G_iP_z\in\textrm{span}\{\one,\sigma_x\}$  for some $i$. Depending on the properties of the SLOCC operators, this symmetry allows to define the standard form as follows:
\begin{itemize}
\item $g_1^x\otimes g_2\otimes g_3\otimes g_4\ket{\Psi_{adcc}}$ (for $g_1$ diagonal, we choose  $ \phi_j=0$, where $j$ denotes the first party for which $g_j$ is non-diagonal) for states with $ \bar{g}_2^3\neq  \bar{g}_3^3$ and/or $ \bar{g}_3^3\neq  \bar{g}_4^3$ and $e^{i\phi_2}\neq \pm e^{i\phi_3}$  and/or $e^{i\phi_3}\neq \pm e^{i\phi_4}$ (excluding the case $e^{i\phi_j}=\pm e^{i\phi_k}$ for $\{j,k\}\in\{\{2,3\},\{2,4\},\{3,4\}\}$ and $\bar{g}_i^1=\bar{g}_i^2=0$ for $i\neq j,k$ and $i\in\{2,3,4\}$) or $\bar{g}_j^1=\bar{g}_j^2=0$ $\forall j\in\{2,3,4\}$,
\item $g_1^{x,z}\otimes g_2^{x,y}\otimes g_3^{x,y}\otimes g_4^{x,y}\ket{\Psi_{adcc}}$ (for $g_1$ diagonal, we choose  $ \phi_j=0$, where $j$ denotes the first party for which $g_j$ is non-diagonal) for states with $ \bar{g}_2^3= \bar{g}_3^3=  \bar{g}_4^3$ and $e^{i\phi_2}\neq \pm e^{i\phi_3}$ and/or $e^{i\phi_3}\neq \pm e^{i\phi_4}$ (excluding the cases $e^{i\phi_j}=\pm e^{i\phi_k}$ for $\{j,k\}\in\{\{2,3\},\{2,4\},\{3,4\}\}$ and $\bar{g}_i^1=\bar{g}_i^2=0$ for $i\neq j,k$ and $i\in\{2,3,4\}$ or $\bar{g}_j^1=\bar{g}_j^2=\bar{g}_k^1=\bar{g}_k^2=0$ for $\{j,k\}\in\{\{2,3\},\{2,4\},\{3,4\}\}$ and $\bar{g}_i^1\neq 0$ and/or $\bar{g}_i^2\neq 0$ for $i\neq j,k$ and $i\in\{2,3,4\}$) or $\bar{g}_j^1=\bar{g}_j^2=0$ $\forall j\in\{2,3,4\}$,
\item  $g_1^{x,y}\otimes g_2^{x,z}\otimes g_3^{x,z}\otimes g_4^{x,z}\ket{\Psi_{adcc}}$  for $ \bar{g}_2^3\neq  \bar{g}_3^3$ and/or $ \bar{g}_3^3\neq  \bar{g}_4^3$ and $e^{i\phi_2}=\pm e^{i\phi_3}$ and $e^{i\phi_2}=\pm e^{i\phi_4}$ (excluding the case $\bar{g}_j^1=\bar{g}_j^2=0$ $\forall j\in\{2,3,4\}$) or $e^{i\phi_j}=\pm e^{i\phi_k}$ for $\{j,k\}\in\{\{2,3\},\{2,4\},\{3,4\}\}$ and $\bar{g}_i^1=\bar{g}_i^2=0$ for $i\neq j,k$ and $i\in\{2,3,4\}$ and
\item $g_1\otimes g_2^{x}\otimes g_3^{x}\otimes g_4^{x}\ket{\Psi_{adcc}}$ for  $ \bar{g}_2^3= \bar{g}_3^3=  \bar{g}_4^3$ and  $e^{i\phi_2}=\pm e^{i\phi_3}$ and $e^{i\phi_2}=\pm e^{i\phi_4}$ or $e^{i\phi_j}=\pm e^{i\phi_k}$ for $\{j,k\}\in\{\{2,3\},\{2,4\},\{3,4\}\}$ and $\bar{g}_i^1=\bar{g}_i^2=0$ for $i\neq j,k$ and $i\in\{2,3,4\}$ or $\bar{g}_j^1=\bar{g}_j^2=\bar{g}_k^1=\bar{g}_k^2=0$ for $\{j,k\}\in\{\{2,3\},\{2,4\},\{3,4\}\}$ (excluding the case $\bar{g}_j^1=\bar{g}_j^2=0$ $\forall j\in\{2,3,4\}$).
\end{itemize}
For all these different cases we choose  $\bar{g}_j^1>0$, where we denote here by j the first party for which $\bar{g}_j^1\neq 0$.
Till now we only used the symmetry $P_z^{\otimes 4}$. Additionally making use of the symmetry $\sigma_x^{\otimes 4}$ allows to choose $\bar{g}_j^3>0$, where j denotes the first party for which $\bar{g}_j^3\neq 0$.
Using the standard form we will prove the following lemmata.
\begin{lemma}\label{lemmaadcc} The only states in the SLOCC classes $G_{abcd}$ where  $a^2\neq d^2\neq c^2\neq a^2$, $c^2= b^2$ and $c\neq 0$ that are reachable via LOCC are given by $h_i\otimes h_j^x\otimes h_k^x\otimes h_l^x \ket{\Psi_{adcc}}$ where $h_i\neq h_i^x$   and $h_i^{x,z}\otimes d_j\otimes d_k\otimes d_l \ket{\Psi_{adcc}}$ where $h_i^{x,z}\neq d_i$  for $\{i,j,k,l\} = \{1,2,3,4\}$.
Moreover, the only states  that are convertible via LOCC are of the form of the reachable states presented above without the condition that $h_i\neq h_i^x$ and $h_i^{x,z}\neq d_i$ .
\end{lemma}
That is the convertible states are given by $g_i\otimes g_j^x\otimes g_k^x\otimes g_l^x \ket{\Psi_{adcc}}$ where $g_i$ is arbitrary and $g_i^{x,z}\otimes d_j\otimes d_k\otimes d_l \ket{\Psi_{adcc}}$  where $\bar{g}_i^1, \bar{g}_i^3$ arbitrary for $\{i,j,k,l\} = \{1,2,3,4\}$.

\begin{proof}
Using the standard form given above we will first prove that for any possible transformation it must hold that $p_{z,m}=0$ for $|z|\neq 1$, i.e. only unitary symmetries can be used for transformations. We will first show this for states with $\bar{h}_1^3=0$ and then for states with $\bar{h}_2^3=\bar{h}_3^3=\bar{h}_4^3=0$. Note that according to our choice of the standard form either $\bar{h}_1^3=0$ or $\bar{h}_2^3=\bar{h}_3^3=\bar{h}_4^3=0$ has to be fulfilled. Then we will use an argument analogously to the one in the proof of Lemma \ref{thfamily1} to show that the reachable states are given by the ones of Lemma \ref{lemmaadcc}.\\
So let us first show that (with the definition of our standard form) only unitary symmetries allow for SEP transformations.
Inserting the symmetries $S_{z,m}=[P_z(\sigma_x)^m]^{\otimes 4}$ (with $z\in\C\backslash 0$ and $m\in\{0,1\}$) in Eq. (\ref{EqSep}) one obtains
\begin{eqnarray}\label{Eqsepadcc}&\sum_{z,m} p_{z,m} \otimes_j \left(
    \begin{array}{cccc}
       & (1/2+(-1)^m\bar{h}_j^{3}) |z|^{2(-1)^m}&  (\bar{h}_j^1-i (-1)^m \bar{h}_j^2) e^{-i2\phi(-1)^m} \\
     &(\bar{h}_j^1+i(-1)^m  \bar{h}_j^2) e^{i2\phi(-1)^m}  & (1/2+(-1)^{m+1}\bar{h}_j^{3})/|z|^{2(-1)^m}\\
     \end{array}
  \right)=
r \otimes_jG_j,\end{eqnarray}
where $z=|z|e^{i\phi}$. Considering the matrix elements $\ket{jklw}\bra{jklw}$ of this equation for $\{j,k,l,w\}=\{0,0,1,1\}$ leads to
\begin{eqnarray}\label{Eqzxadcc}
&\tilde{p}_0(\frac{1}{2}+(-1)^j\bar{h}_1^3)(\frac{1}{2}+(-1)^k\bar{h}_2^3)
(\frac{1}{2}+(-1)^l\bar{h}_3^3)(\frac{1}{2}+(-1)^w\bar{h}_4^3)+\\ \nonumber
&\tilde{p}_1(\frac{1}{2}+(-1)^{j+1}\bar{h}_1^3)(\frac{1}{2}+(-1)^{k+1}\bar{h}_2^3)
(\frac{1}{2}+(-1)^{l+1}\bar{h}_3^3)(\frac{1}{2}+(-1)^{w+1}\bar{h}_4^3)
\\ \nonumber
&=r(\frac{1}{2}+(-1)^{j}\bar{g}_1^3)(\frac{1}{2}+(-1)^{k}\bar{g}_2^3)
(\frac{1}{2}+(-1)^{l}\bar{g}_3^3)(\frac{1}{2}+(-1)^{w}\bar{g}_4^3)
\end{eqnarray}
where $\sum_{z}p_{z,m}=\tilde{p}_m$.
We will first consider the case $\bar{h}_1^3=0$. Thus, one obtains from Eq. (\ref{Eqzxadcc}) that
\begin{eqnarray}\label{1adcc}
&1/2[\tilde{p}_0(\frac{1}{2}+\bar{h}_i^3)
(\frac{1}{2}+\bar{h}_s^3)(\frac{1}{2}-\bar{h}_t^3)+\tilde{p}_1(\frac{1}{2}-\bar{h}_i^3)
(\frac{1}{2}-\bar{h}_s^3)(\frac{1}{2}+\bar{h}_t^3)]
=r(\frac{1}{2}-\bar{g}_1^3)(\frac{1}{2}+\bar{g}_i^3)
(\frac{1}{2}+\bar{g}_s^3)(\frac{1}{2}-\bar{g}_t^3)
\end{eqnarray}
 and
\begin{eqnarray}\label{2adcc}
&1/2[\tilde{p}_0(\frac{1}{2}-\bar{h}_i^3)
(\frac{1}{2}-\bar{h}_s^3)(\frac{1}{2}+\bar{h}_t^3)+\tilde{p}_1(\frac{1}{2}+\bar{h}_i^3)
(\frac{1}{2}+\bar{h}_s^3)(\frac{1}{2}-\bar{h}_t^3)]
=r(\frac{1}{2}+\bar{g}_1^3)(\frac{1}{2}-\bar{g}_i^3)
(\frac{1}{2}-\bar{g}_s^3)(\frac{1}{2}+\bar{g}_t^3),
\end{eqnarray}
where $\{i,s,t\}=\{2,3,4\}$.
Adding up these two equations for the same choice of $i,s$ and $t$ and using that $\tilde{p}_0+\tilde{p}_1=1$ leads to
\begin{eqnarray}
&1/2[(\frac{1}{2}+\bar{h}_i^3)
(\frac{1}{2}+\bar{h}_s^3)(\frac{1}{2}-\bar{h}_t^3)+(\frac{1}{2}-\bar{h}_i^3)
(\frac{1}{2}-\bar{h}_s^3)(\frac{1}{2}+\bar{h}_t^3)]
\\ \nonumber
&=r[(\frac{1}{2}-\bar{g}_1^3)(\frac{1}{2}+\bar{g}_i^3)
(\frac{1}{2}+\bar{g}_s^3)(\frac{1}{2}-\bar{g}_t^3)+(\frac{1}{2}+\bar{g}_1^3)(\frac{1}{2}-\bar{g}_i^3)
(\frac{1}{2}-\bar{g}_s^3)(\frac{1}{2}+\bar{g}_t^3)].
\end{eqnarray}
By adding up this equation for some choice of $i,s$ and $t$ and for the same choice but permuting $s$ and $t$ one obtains the relation
\begin{eqnarray}\label{reladcc}
&1/2[(\frac{1}{2}+\bar{h}_s^3)(\frac{1}{2}-\bar{h}_t^3)+(\frac{1}{2}-\bar{h}_s^3)(\frac{1}{2}+\bar{h}_t^3)]
\\ \nonumber
&=r[(\frac{1}{2}-\bar{g}_1^3)(\frac{1}{2}+\bar{g}_i^3)
+(\frac{1}{2}+\bar{g}_1^3)(\frac{1}{2}-\bar{g}_i^3)][(\frac{1}{2}+\bar{g}_s^3)(\frac{1}{2}-\bar{g}_t^3)+
(\frac{1}{2}-\bar{g}_s^3)(\frac{1}{2}+\bar{g}_t^3)].
\end{eqnarray}
Considering the matrix elements $\ket{jklw}\bra{jklw}$ of Eq. (\ref{Eqsepadcc}) for $j=0$ and $\{k,l,w\}=\{0,0,1\}$ and $j=1$ and $\{k,l,w\}=\{0,1,1\}$ leads to
\begin{eqnarray}
&1/2[\sum_z p_{z,0} |z|^4(\frac{1}{2}+\bar{h}_i^3)
(\frac{1}{2}+\bar{h}_s^3)(\frac{1}{2}-\bar{h}_t^3)+\sum_z p_{z,1}1/|z|^4(\frac{1}{2}-\bar{h}_i^3)
(\frac{1}{2}-\bar{h}_s^3)(\frac{1}{2}+\bar{h}_t^3)]
\\\nonumber & =r(\frac{1}{2}+\bar{g}_1^3)(\frac{1}{2}+\bar{g}_i^3)
(\frac{1}{2}+\bar{g}_s^3)(\frac{1}{2}-\bar{g}_t^3)
\end{eqnarray}
 and
\begin{eqnarray}
&1/2[\sum_z p_{z,0}1/|z|^4(\frac{1}{2}-\bar{h}_i^3)
(\frac{1}{2}-\bar{h}_s^3)(\frac{1}{2}+\bar{h}_t^3)+\sum_z p_{z,1} |z|^4(\frac{1}{2}+\bar{h}_i^3)
(\frac{1}{2}+\bar{h}_s^3)(\frac{1}{2}-\bar{h}_t^3)]
\\\nonumber & =r(\frac{1}{2}-\bar{g}_1^3)(\frac{1}{2}-\bar{g}_i^3)
(\frac{1}{2}-\bar{g}_s^3)(\frac{1}{2}+\bar{g}_t^3),
\end{eqnarray}
where $\{i,s,t\}=\{2,3,4\}$.
Adding up these two equations for the same choice of $i,s$ and $t$ leads to
\begin{eqnarray}
&1/2[\sum_{z,m} p_{z,m}|z|^4(\frac{1}{2}+\bar{h}_i^3)
(\frac{1}{2}+\bar{h}_s^3)(\frac{1}{2}-\bar{h}_t^3)+\sum_{z,m} p_{z,m}1/|z|^4(\frac{1}{2}-\bar{h}_i^3)
(\frac{1}{2}-\bar{h}_s^3)(\frac{1}{2}+\bar{h}_t^3)]
\\ \nonumber
&=r[(\frac{1}{2}+\bar{g}_1^3)(\frac{1}{2}+\bar{g}_i^3)
(\frac{1}{2}+\bar{g}_s^3)(\frac{1}{2}-\bar{g}_t^3)+(\frac{1}{2}-\bar{g}_1^3)(\frac{1}{2}-\bar{g}_i^3)
(\frac{1}{2}-\bar{g}_s^3)(\frac{1}{2}+\bar{g}_t^3)].
\end{eqnarray}
Adding up this equation for some choice of $i,s$ and $t$ and for the same choice but permuting $s$ and $t$ one obtains that
\begin{eqnarray}
&1/2[\sum_{z,m} p_{z,m}|z|^4(\frac{1}{2}+\bar{h}_i^3)+\sum_{z,m} p_{z,m}1/|z|^4(\frac{1}{2}-\bar{h}_i^3)][(\frac{1}{2}+\bar{h}_s^3)(\frac{1}{2}-\bar{h}_t^3)+(\frac{1}{2}-\bar{h}_s^3)(\frac{1}{2}+\bar{h}_t^3)]
\\ \nonumber
&=r[(\frac{1}{2}+\bar{g}_1^3)(\frac{1}{2}+\bar{g}_i^3)
+(\frac{1}{2}-\bar{g}_1^3)(\frac{1}{2}-\bar{g}_i^3)][(\frac{1}{2}+\bar{g}_s^3)(\frac{1}{2}-\bar{g}_t^3)+
(\frac{1}{2}-\bar{g}_s^3)(\frac{1}{2}+\bar{g}_t^3)].
\end{eqnarray}
Using now Eq. (\ref{reladcc}), as well as the fact that according to the standard form either $\bar{g}_1^3$ or $\bar{g}_i^3$ has to be zero, one obtains that
\bea\label{adcc}
\sum_{z,m} p_{z,m}(|z|^4-1)(\frac{1}{2}+\bar{h}_i^3)
+\sum_{z,m} p_{z,m}(1/|z|^4-1)(\frac{1}{2}-\bar{h}_i^3)=0.
\eea
Note that this equation has to hold $\forall i\in\{2,3,4\}$.
Thus, either $p_{z,m}=0$ for $|z|\neq 1$ or
\bea
(\frac{1}{2}-\bar{h}_i^3)/(\frac{1}{2}+\bar{h}_i^3)=-\sum_{z,m} p_{z,m}(|z|^4-1)/[\sum_{z,m} p_{z,m}(1/|z|^4-1)]\,\,\,\,\forall i\in\{2,3,4\},
\eea
i.e. $\bar{h}_2^3=\bar{h}_3^3=\bar{h}_4^3$. According to the standard form the latter case corresponds to $\bar{h}_2^3=\bar{h}_3^3=\bar{h}_4^3=0$. It can be easily seen from Eq. (\ref{adcc}) that for $\bar{h}_i^3=0$ it has to hold that $p_{z,m}=0$ for $|z|\neq 1$. Thus, states in standard form with $\bar{h}_1^3=0$ can only possibly reached by using unitary symmetries.\\
In order to show that also for states with $\bar{h}_2^3=\bar{h}_3^3=\bar{h}_4^3=0$ it has to hold that $p_{z,m}=0$ for $|z|\neq 1$ consider Eq. (\ref{Eqzxadcc}) for $j=0$.
One obtains that \begin{eqnarray}
&\frac{1}{8}[\tilde{p}_0(\frac{1}{2}+\bar{h}_1^3)+\tilde{p}_1(\frac{1}{2}-\bar{h}_1^3)]
=r(\frac{1}{2}+\bar{g}_1^3)(\frac{1}{2}+\bar{g}_i^3)
(\frac{1}{2}-\bar{g}_s^3)(\frac{1}{2}-\bar{g}_t^3)
=r(\frac{1}{2}+\bar{g}_1^3)(\frac{1}{2}-\bar{g}_i^3)
(\frac{1}{2}+\bar{g}_s^3)(\frac{1}{2}-\bar{g}_t^3),
\end{eqnarray}
where $\{i,s,t\}=\{2,3,4\}$.
Thus, it has to hold that $(\frac{1}{2}+\bar{g}_i^3)/(\frac{1}{2}-\bar{g}_i^3)=(\frac{1}{2}+\bar{g}_s^3)/(\frac{1}{2}-\bar{g}_s^3)$  $\forall \{i,s\}\in\{\{2,3\},\{2,4\},\{3,4\}\}$. Hence, we have that $\bar{g}_i^3=\bar{g}_s^3$ must hold. This implies according to our standard form that $\bar{g}_j^3=0$ $\forall j\in\{2,3,4\}$. Using this result in Eq. (\ref{Eqzxadcc}) for $j\in\{0,1\}$, as well as for the matrix elements $\ket{jklw}\bra{jklw}$ of Eq. (\ref{Eqsepadcc}) for $\{j,k,l,w\}=\{0,0,0,1\}$ and $\{j,k,l,w\}=\{0,1,1,1\}$ one obtains that
\bea
\tilde{p}_0(\frac{1}{2}+\bar{h}_1^3)+\tilde{p}_1(\frac{1}{2}-\bar{h}_1^3)
=r(\frac{1}{2}+\bar{g}_1^3),
\eea
\bea
\tilde{p}_0(\frac{1}{2}-\bar{h}_1^3)+\tilde{p}_1(\frac{1}{2}+\bar{h}_1^3)
=r(\frac{1}{2}-\bar{g}_1^3),
\eea
\begin{eqnarray}
&\sum_z p_{z,0} |z|^{4(-1)^v}(\frac{1}{2}+\bar{h}_1^3)
+\sum_z p_{z,1}1/|z|^{4(-1)^v}(\frac{1}{2}-\bar{h}_1^3)=r(\frac{1}{2}+\bar{g}_1^3)
\end{eqnarray}
and
\begin{eqnarray}
&\sum_z p_{z,0} |z|^{4(-1)^t}(\frac{1}{2}-\bar{h}_1^3)
+\sum_z p_{z,1}1/|z|^{4(-1)^t}(\frac{1}{2}+\bar{h}_1^3)=r(\frac{1}{2}-\bar{g}_1^3),
\eea
where $v,t\in\{0,1\}$.
Adding up the first two equations leads to $r=1$, whereas adding up the last two equations with $v+t=1$ leads to
\bea
\sum_{z,m} p_{z,m} |z|^{4(-1)^j}(\frac{1}{2}+\bar{h}_1^3)
+\sum_{z,m} p_{z,m}1/|z|^{4(-1)^j}(\frac{1}{2}-\bar{h}_1^3)
=1,
\eea
where $j\in\{0,1\}$.
Finally by adding up this equation for $j=0$ and $j=1$ one obtains that
\bea
\sum_{z,m} p_{z,m} (|z|^{2}-1/|z|^{2})^2=0
\eea
and therefore also for $\bar{h}_2^3=\bar{h}_3^3=\bar{h}_4^3=0$ it has to hold that $p_{z,m}=0$ for $|z|\neq 1$.
Thus, using the standard form as defined above only unitary symmetries can be used for SEP transformations in these SLOCC classes.\\
The fact that only unitary symmetries need to be considered implies two useful properties. One is that it must hold that $|\textbf{h}_i|\geq|\textbf{g}_i|$ $\forall i$ (Observation \ref{obsmon}) and the other is that we can trace out parties in Eq.\ (\ref{EqSep}) obtaining
\begin{equation}\label{EqSephere}
\sum_{z,m}p_{z,m}\bigotimes_{i\in I}(s_i^{z,m})^\dag H_i s_i^{z,m}=\bigotimes_{i\in I}G_i
\end{equation}
for whatever subset of parties $I$. Furthermore, as the symmetries always have the same unitary at every position it can be easily shown that Eqs.\ (\ref{scalar}) and (\ref{vector}) have to hold \footnote{\label{fn1} This can be easily seen as follows. Tracing Eq. (\ref{EqSep}) over two parties and considering  the projection onto $\ket{\Psi^-}$ of this equation leads to Eq. (\ref{scalar}). Moreover, multiplying Eq. (\ref{EqSep}) with $(P_{123}-P_{321})\otimes\one$ where $P_{jkl}$ denotes the three-party permutation in cyclic notation and considering the trace of this equation leads to Eq. (\ref{vector}). Analogously one obtains Eq. (\ref{vector}) for all other triplets of vectors $\{\textbf{g}_i\}$ and $\{\textbf{h}_i\}$.} and therefore the results of the previous section can be applied here provided it also holds in this case that uniform rotations of the vectors cannot be achieved by a nontrivial SEP transformation. To see that this is indeed the case, notice that, in case that would be possible, from Eq.\ (\ref{EqSephere}) it should hold that
\begin{equation}\label{unirot}
G_i=  U H_i U^\dagger=\sum_j p_j s_j^\dagger H_i s_j,
\end{equation}
for every party $i$ and for some fixed unitary $U$. Note that in order to simplify the notation we denote here the unitary symmetries by $S_j= s_j^{\otimes 4}$. From Eq. (\ref{unirot}) it follows that  $H_i =\sum_j p_j U^\dagger s_j^\dagger H_i s_j U$. In other words, $H_i$ is $\forall i$ a fixed point of the unital map with Kraus operators $\{\sqrt{p_j}U^\dagger s_j^\dagger\}$. It is shown in \cite{kribs} that for this to be the case $H_i$ has to commute with all Kraus operators. This implies that $H_i = U^\dagger s_j^\dagger H_i s_j U$ $\forall j$ such that $p_j\neq 0$ and consequently $s_j^\dagger H_i s_j=s_k^\dagger H_i s_k$ whenever $p_j,p_k\neq 0$. Inserting this back into Eq.\ (\ref{unirot}) and using that it has to hold for all $i$, we have $G=  S_j^\dagger H S_j$ (where $S_j$ is some symmetry acting on all parties such that the corresponding probability is non-zero). Hence, the states corresponding to $G$ and $H$, i.\ e.\ any pair of states for which a SEP transformation is possible and whose vectors are related by a uniform rotation, are LU-equivalent (i.\ e.\ they have the same standard form).

Thus, in conclusion, it not only needs to hold that $|\textbf{h}_i|\geq|\textbf{g}_i|$ $\forall i$, but also that the inequality is strict for at least one value of $i$. Moreover, every possible triplet of the vectors $\{\textbf{g}_i\}$ and $\{\textbf{h}_i\}$ must obey the conditions of  Eqs.\ (\ref{scalar}) and (\ref{vector}) (see \footnotemark[\value{footnote}]) and Lemma \ref{thfamily1}. This immediately shows that all states with vectors lying in more than two different directions (excluding here and in the following vanishing vectors) must be isolated. To finish the proof, it thus remains to show that, in order for LOCC conversions to be possible, in the case of states whose vectors lie in at most two different directions, these cannot be arbitrary but have to fulfill the premises of Lemma \ref{lemmaadcc}. To study this remaining case we distinguish four different possibilities: \textit{i)} There are two non-parallel vectors and two null vectors, \textit{ii)} There is more than one vector lying in one direction and one vector lying in a different direction, \textit{iii)} There are two vectors lying in one direction and two vectors lying in a different one. \textit{iv)} All nonzero vectors point in the same direction. As in previous proofs, we will first identify all isolated states among these classes and we will later provide explicit LOCC protocols for the remaining convertible and reachable states.

Considering first case $i)$, notice that the standard form in this case is such that one of the nonzero vectors points in the $x$ direction. Hence, these states are claimed to be reachable in the statement of the lemma and we will see so at the end of the proof by considering explicit protocols.
Moving now to case $ii)$, notice that the standard form for these states is such that the set with more than one vector in the same direction must lie in the $xz$ plane. We assume now that this is not the $x$ or $z$ direction and show that these states are isolated. Equating the appropriate entries of the matrix equation (\ref{EqSephere}) it is easy to see that
\begin{equation}
(\tilde{p}_0-\tilde{p}_1)\bar{h}_i^{3}=\bar{g}_i^{3},\quad (\tilde{p}_0-\tilde{p}_1)\bar{h}_j^{3}=\bar{g}_j^{3},\quad \bar{h}_i^{3}\bar{h}_j^{3}=\bar{g}_i^{3}\bar{g}_j^{3},\label{zcomps}
\end{equation}
\begin{align}
\bar{h}_i^{1}\bar{h}_j^{1}+\bar{h}_i^{2}\bar{h}_j^{2}=\bar{g}_i^{1}\bar{g}_j^{1}+\bar{g}_i^{2}\bar{g}_j^{2},\label{escxy}\\
(\tilde{p}_0-\tilde{p}_1)(\bar{h}_i^{2}\bar{h}_j^{1}-\bar{h}_i^{1}\bar{h}_j^{2})=\bar{g}_i^{2}\bar{g}_j^{1}-\bar{g}_i^{1}\bar{g}_j^{2}.\label{vecxy}
\end{align}
for any different parties $i$ and $j$ and where we use the notation $\tilde{p}_m=\sum_{z}p_{z,m}$. Since there are at least two parties with nonzero components in the $z$ directions, Eqs.\ (\ref{zcomps}) immediately imply that $\bar{h}_i^{3}=\bar{g}_i^{3}$ $\forall i$ and that $\tilde{p}_0-\tilde{p}_1=1$. On the other hand, since there are as well at least two parties with nonzero components in the $x$ direction, Eqs.\ (\ref{escxy})-(\ref{vecxy}) are not trivially fulfilled either. In the same spirit as in the proof of Lemma \ref{thfamily1}, Eq.\ (\ref{escxy}) shows that the scalar product in the $xy$ plane needs to be preserved while Eq.\ (\ref{vecxy}) with $\tilde{p}_0-\tilde{p}_1=1$ shows the same for the cross product in the $xy$ plane. Thus, analogously to the previous section and using again the monotonicity of the length of the vectors, this shows that the $x$ and $y$ coordinates of the vectors must stay the same too (recall that uniform rotations are not possible). Thus, these states are neither reachable nor convertible arriving at the desired conclusion that they are isolated.

Let us consider now case $iii)$ and see that all these states are isolated. They are clearly not reachable because there cannot exist LU-inequivalent states with the same pairwise scalar products and all vectors of smaller or equal size. Moreover, if one of the pairs of parallel vectors has nonvanishing $z$ component, the same reasoning as in case $ii)$ applies, so it only remains to see that when both pairs of parallel vectors lie in the $xy$ plane the states are not convertible. On the one hand, suppose that a transformation from these states to a state with a nonzero $z$ component for one party would be possible. In this case, Eq.\ (\ref{zcomps}) imposes that $\tilde{p}_0-\tilde{p}_1=0$; however, this leads to a contradiction with Eq.\ (\ref{vecxy}) when considered for a pair of originally non-parallel vectors. On the other hand, suppose now that a transformation from these states to a state with zero $z$ components for all parties would be possible. For the reason above, it must hold that $\tilde{p}_0-\tilde{p}_1\neq0$ and, hence, Eq.\ (\ref{vecxy}) tells us that parallel vectors must remain parallel. This together with Eq.\ (\ref{escxy}) implies that the sizes of all vectors must remain the same and, therefore, no conversion is possible in this case either.

Let us finally consider case $iv)$. It can be easily seen that the seed state is in the MES (as $|\textbf{h}_i|\geq|\textbf{g}_i|$ $\forall i$) and that it is convertible. States with exactly one non-zero vector are of the form given in Lemma \ref{lemmaadcc}. Thus, it only remains to show that states with more than one non-zero vector and all nonzero vectors in the same direction are not reachable (and are, hence, in the MES) which can be easily seen as there cannot exist LU-inequivalent states preserving the scalar products with all vectors of smaller or equal size.  In this case due to our choice of the standard form the vectors have to lie in the $xz$ plane. In case the vectors do not point in $x$ or $z$ direction we have already proven for case $ii)$ that no conversion is possible and therefore the states are isolated. If all vectors point in $x$ or in $z$ direction the corresponding states are convertible (see proof of Lemma \ref{MESfamily1}).\\
Thus, we finally just need to check that the convertible states in the MES allow to obtain all reachable states. However, as accounted in Lemma \ref{MESfamily1}, the possible transformations in the previous section only relied on the symmetries $\sigma_i^{\otimes4}$ which are also in the symmetry group of this class, and hence repeating the same protocols as therein leads to the desired transformations.
\end{proof}

Thus, the only convertible states in $MES_4$ in these classes are the seed states and those with at least two nonzero vectors $\textbf{g}_i$ and all of them being parallel and pointing in the $x$ or $z$ direction. Again, almost all states in these classes are isolated.

\subsection{Two non-zero double-degenerate eigenvalues}\label{family2}

We consider here the case $D=\textrm{diag}(a^2,a^2,c^2,c^2)$  where $a^2\neq c^2$ and $a ,c \neq 0$. We have to distinguish two cases here as $a^2=-c^2$ is a cyclic case.

\subsubsection{The case $a^2\neq \pm c^2$ and $a ,c \neq 0$\label{secaacc}}
In this case $D$ can be chosen as in Eq. (\ref{eigenvalues}) with $a=d$ and $b=c$. Hence, the seed states can be chosen to be the two-parameter family
\begin{equation}
|\Psi\rangle=a(|0000\rangle+|1111\rangle)+c(|0101\rangle+|1010\rangle).
\end{equation}
The elements of the symmetry group for these states can be shown to be of the form $S=P_{z_1}\otimes P_{z_2}\otimes P_{z_1^{-1}}\otimes P_{z_2^{-1}}$ together with $\sigma_x^{\otimes 4}$. As in the previous classes, we can redefine the seed state to be $|\Psi_{aacc}\rangle=\one\otimes\one\otimes\sigma_x\otimes\sigma_x|\Psi\rangle$, so that the symmetries are $S_{z_1,z_2,m}=X^m(P_{z_1}\otimes P_{z_2})^{\otimes 2}$ with $m\in\{0,1\}$, $z_1,z_2\in\C\backslash0$ and $X=\sigma_x^{\otimes 4}$. Using again that for any $G_i$ there exists a value of $z$ such that $P_z^\dag G_iP_z\in\textrm{span}\{\one,\sigma_x\}$, this symmetry allows to choose the following standard form $g_1^x\otimes g_2^x\otimes g_3\otimes g_4|\Psi_{aacc}\rangle$, where $g_1^2\geq 0$  and $g_2^2\geq 0$. In case $g_1^2=0$ ($g_2^2=0$) we choose $g_3^2\geq 0$ ($g_4^2\geq 0$) respectively. Note that till now we only used the symmetry $(P_{z_1}\otimes P_{z_2})^{\otimes 2}$ to  identify the standard form. In order to make the standard form unambiguous one would also need to take into account the symmetry $\sigma_x^{\otimes 4}$. However, in order to make the proof more readable we will not use this symmetry for the definition of the standard form.
In the following lemma we state which states in this SLOCC class can be reached via LOCC.
\begin{lemma}\label{lemmaaacc} The only states in the SLOCC classes $G_{abcd}$ where  $a^2= d^2\neq\pm c^2$, $c^2= b^2$ and $a, c\neq 0 $ that are reachable via LOCC are given by  $\otimes h_i |\Psi_{aacc}\rangle$ where $\otimes H_i$ obeys the following condition. There exists a unitary symmetry $S_{z_1,z_2,m}=\otimes_i s_i^{z_1,z_2,m}\in S(\Psi_{aacc})$ such that exactly
three out of the four operators $H_i$ commute with the corresponding
operator $s_i^{z_1,z_2,m}$ and one operator $H_j$ does not commute with $s_j^{z_1,z_2,m}$.
\end{lemma}
The proof of this lemma can be found in Appendix C.
In the following lemma we specify the convertible states for these SLOCC classes.
\begin{lemma}\label{convlemmaaacc} The only states  in the SLOCC classes $G_{abcd}$ where  $a^2= d^2\neq\pm c^2$, $c^2= b^2$ and $a, c\neq 0 $ that are convertible via LOCC are given by $\otimes g_i |\Psi_{aacc}\rangle$ where $\otimes G_i$ obeys the following condition. There exists a non-trivial unitary symmetry $S_{z_1,z_2,m}=\otimes_i s_i^{z_1,z_2,m}\in S(\Psi_{aacc})$ such that three out of the four operators $G_i$ commute with the corresponding
operator $s_i^{z_1,z_2,m}$ and one operator $G_j$ is arbitrary.
\end{lemma}
The proof of this lemma is given in Appendix C. Due to Lemma \ref{lemmaaacc} and Lemma \ref{convlemmaaacc} the only convertible states in $MES_4$ in these SLOCC classes are states of the form
$g_1^x\otimes g_2^x\otimes g_3^x\otimes g_4^x|\Psi_{aacc}\rangle$ where either
$g_1^2,g_3^2\neq 0$ or $g_1^2=g_3^2=0$  and either $g_2^2,g_4^2\neq 0$ or $g_2^2=g_4^2=0$, as well as $\one/2\otimes \one/2\otimes d_3\otimes d_4|\Psi_{aacc}\rangle$ where $d_i\neq \one/2$ for $i\in\{3,4\}$, $\one/2\otimes g_2^x\otimes d_3\otimes g_4|\Psi_{aacc}\rangle$ where $g_{4}\neq g_4^x$, $g_4\neq d_4,$ $g_2^x \neq \one/2$ and $d_3\neq \one/2$ and  $g_1^x\otimes \one/2 \otimes g_3\otimes d_4|\Psi_{aacc}\rangle$ where $g_{3}\neq g_3^x$, $g_{3}\neq d_3$, $g_1^x \neq \one/2$ and $d_4\neq \one/2$.

\subsubsection{The case $a^2= -c^2$ and $a \neq 0$}
Note that this case corresponds to the cyclic class analyzed in Sec. \ref{2cycles2} when $a=c$ and the eigenvalues are reordered appropriately.
The single seed state for this SLOCC class is given by
\begin{equation}\label{seedfamily2}
|\Psi_{aa(ia)(ia)}\rangle=a[|0000\rangle+|1111\rangle+i(|0101\rangle+|1010\rangle)].
\end{equation}
Note that due to the normalization there is no parameter left and therefore we will use in the following the notation $\ket{\Psi}\equiv\ket{\Psi_{aa(ia)(ia)}}$.
For the seed state one obtains the following non-compact symmetry group
\begin{eqnarray}
&S(\Psi)= \langle \sigma_x^{\otimes 4}, \sigma_x\left(
    \begin{array}{cc}
      y_1 & 0 \\
      0 & x_1 \\
    \end{array}
  \right)\otimes\left(
    \begin{array}{cc}
      y_2 & 0 \\
      0 & x_2\\
    \end{array}
  \right)\otimes\sigma_x\left(
    \begin{array}{cc}
      1/y_1 & 0\\
      0 & -1/x_1 \\
    \end{array}
  \right)
\otimes\left(
    \begin{array}{cc}
      1/y_2 & 0 \\
      0 & -1/x_2 \\
    \end{array}\right), \\ \nonumber&\left(
    \begin{array}{cc}
      z_1 & 0 \\
      0 & \eta_1 \\
    \end{array}
  \right)\otimes\left(
    \begin{array}{cc}
      z_2 & 0 \\
      0 & \eta_2\\
    \end{array}
  \right)\otimes\left(
    \begin{array}{cc}
      1/z_1 & 0\\
      0 & 1/\eta_1 \\
    \end{array}
  \right)
\otimes\left(
    \begin{array}{cc}
      1/z_2 & 0 \\
      0 & 1/\eta_2 \\
    \end{array}\right)\rangle,
\end{eqnarray}
where $x_1, x_2, y_1, y_2, z_1, z_2, \eta_1, \eta_2\in\C\backslash 0$. These symmetries allow to choose as standard form $g_1^x\otimes g_2^x\otimes g_3\otimes g_4\ket{\Psi}$ where $g_1^2\geq 0$ and $g_2^2\geq 0$. In case $g_1^2=0$ ($g_2^2=0$) we choose $g_3^2\geq 0$ ($g_4^2\geq 0$) respectively. Moreover, we choose $\Im(g_i^2)> 0$ ($\Re(g_i^2)> 0$) where $i\in\{3,4\}$ is the first party for which $\Im(g_i^2)\neq 0$ ($\Re(g_i^2)\neq 0$) respectively. In case $g_i^2=0$ $\forall i\in\{3,4\}$ we choose $g_i^1\geq g_i^3$ where $i\in\{3,4\}$ is the first party for which $g_i^1\neq g_i^3$.
The reachable states for this SLOCC class are given in the following lemma.
\begin{lemma}\label{lemmaaaaa} The only states in the SLOCC classes $G_{abcd}$ where  $a^2= d^2$, $c^2= b^2$, $a^2= -c^2$ and $a\neq 0$ that are reachable via LOCC are given by $\otimes_i h_i \ket{\Psi}$  where $\otimes H_i$ obeys the following condition. There exists a unitary symmetry $S=\otimes_i s_i\in S(\Psi)$ such that exactly
three out of the four operators $H_i$ commute with the corresponding
operator $s_i$ and one operator $H_j$ does not commute with $s_j$.
\end{lemma}
That is the reachable states are given by
\begin{itemize}
\item $g_1^x\otimes g_2^x\otimes g_3\otimes g_4^x|\Psi\rangle$  where $g_3\neq g_3^x$,
\item $g_1^x\otimes g_2^x\otimes g_3^x\otimes g_4|\Psi\rangle$ where $g_4\neq g_4^x$,
\item $\one/2\otimes g_2^x\otimes g_3\otimes g_4|\Psi\rangle$ where $g_3\neq d_3$,
\item $g_1^x\otimes \one/2\otimes g_3\otimes g_4|\Psi\rangle$ where $g_4\neq d_4$,
\item $g_1^x\otimes g_2^x\otimes d_3\otimes g_4|\Psi\rangle$ where $g_1^x\neq \one/2$ and /or $g_4\neq g_4^y$,
\item $g_1^x\otimes g_2^x\otimes g_3\otimes d_4|\Psi\rangle$ where $g_2^x\neq \one/2$ and/or $g_3\neq g_3^y$,
\item $g_1^x\otimes g_2^x\otimes g_3^y\otimes g_4|\Psi\rangle$ where $g_4\neq d_4$ and
\item $g_1^x\otimes g_2^x\otimes g_3\otimes g_4^y|\Psi\rangle$ where $g_3\neq d_3$.

\end{itemize}
The proof is lengthy and is presented in Appendix D. In the following lemma we state which states are convertible in this SLOCC class.
\begin{lemma}\label{convlemmaaaaa} The only states  in the SLOCC class $G_{abcd}$ where  $a^2= d^2$, $c^2= b^2$, $a^2= -c^2$ and $a \neq 0$ that are convertible via LOCC are given by $\otimes_i g_i \ket{\Psi}$ where  $\otimes G_i$ obeys the following condition. There exists a non-trivial unitary symmetry $S=\otimes_i s_i\in S(\Psi)$ such that
three out of the four operators $G_i$ commute with the corresponding
operator $s_i$ and one operator $G_j$ is arbitrary.
\end{lemma}
That is the convertible states are given by
\begin{itemize}
\item $g_1^x\otimes g_2^x\otimes g_3\otimes g_4^x|\Psi\rangle$,
\item $g_1^x\otimes g_2^x\otimes g_3^x\otimes g_4|\Psi\rangle$,
\item $\one/2\otimes g_2^x\otimes g_3\otimes g_4|\Psi\rangle$,
\item $g_1^x\otimes \one/2\otimes g_3\otimes g_4|\Psi\rangle$,
\item $g_1^x\otimes g_2^x\otimes d_3\otimes g_4|\Psi\rangle$,
\item $g_1^x\otimes g_2^x\otimes g_3\otimes d_4|\Psi\rangle$,
\item $g_1^x\otimes g_2^x\otimes g_3^y\otimes g_4|\Psi\rangle$ and
\item $g_1^x\otimes g_2^x\otimes g_3\otimes g_4^y|\Psi\rangle$.
\end{itemize}
The proof of this lemma can also be found in Appendix D. Combining Lemma \ref{lemmaaaaa} and Lemma \ref{convlemmaaaaa} we have that the non-isolated states in $MES_4$ for this SLOCC class are given by
\begin{itemize}
\item $|\Psi\rangle$,
\item $g_1^x\otimes g_2^x\otimes g_3^x\otimes g_4^x|\Psi\rangle$ where $g_i^2\neq 0$ $\forall i \in\{1,2,3,4\}$,
\item $\one/2\otimes g_2^x\otimes d_3\otimes g_4^y|\Psi\rangle$ where $g_2^2, g_4^2 \neq 0$ and $d_3\neq \one/2$ and
\item $g_1^x\otimes \one/2\otimes g_3^y\otimes d_4|\Psi\rangle$ where $g_1^2, g_3^2 \neq 0$ and $d_4\neq \one/2$.
\end{itemize}
\subsection{\label{secE}A double-degenerate zero eigenvalue and two different non-degenerate eigenvalues}

By choosing $b=c=0$ in Eq.\ (\ref{eigenvalues}), the seed states for these SLOCC classes can be chosen as
\begin{equation}
|\Psi_{ad00}\rangle=\frac{a+d}{2}(|0000\rangle+|1111\rangle)+\frac{a-d}{2}(|0011\rangle+|1100\rangle),
\end{equation}
where $a^2\neq d^2$.
Note that these states can be obtained from the seed states for the families discussed in Sec. \ref{family2} by permuting party $2$ and $3$. Thus, all the statements of  Sec. \ref{family2} extend to these families by taking into account the permutation.

\subsection{\label{secGHZ}A double-degenerate non-zero eigenvalue and a double-degenerate zero eigenvalue}

As the normalization and the global phase fix two parameters, there are no free parameters $\{a,b,c,d\}$ and this case corresponds to a single SLOCC class. Actually, one can choose $a=d$ and $b=c=0$ in Eq.\ (\ref{eigenvalues}) to see that the seed state is the GHZ state
\begin{equation}
|\Psi_{aa00}\rangle=|0000\rangle+|1111\rangle,
\end{equation}
and, hence, this is the GHZ SLOCC class. It can be readily found that the elements of the symmetry group are of the form $S_{z_1,z_2,z_3,m}=(P_{z_1}\otimes P_{z_2}\otimes P_{z_3}\otimes P_{(z_1z_2z_3)^{-1}})(\sigma_x^{\otimes4})^m$, where as before $P_z=diag(z,1/z)$, $0\neq z\in\C$ and $m=0,1$. Hence, by choosing $z_i$ such that $P_{z_i}^\dag G_iP_{z_i}\in span\{\one,\sigma_x\}$, we have the following standard form for the states in this class $g_1^x\otimes g_2^x\otimes g_3^x\otimes g_4|\Psi_{aa00}\rangle$, where $g_i^2\geq 0$ for $i\in\{1,2,3\}$. Moreover, we choose $\Im(g_4^2)\geq 0$ and in case $\Im(g_4^2)= 0$ one chooses $g_4^1\geq g_4^3$. In case $g_i^2= 0$ for at least one $i\in\{1,2,3\}$ we choose $g_4^2\geq 0$.  This is completely analogous to the 3-qubit GHZ class studied in \cite{MESus} and the possible LOCC conversions can be readily obtained from the analysis performed therein leading to the following lemma.

\begin{lemma} The MES in this class is given by states of the form $g_1^x\otimes g_2^x\otimes g_3^x\otimes g^x_4|\Psi_{aa00}\rangle$, where no $g_i^x\propto \one$ (except for the GHZ state). Moreover, all states in the MES for this SLOCC class are convertible.\end{lemma}

Thus this is one of the few SLOCC classes in which there are no isolated states. Notice that LOCC conversions in this class have been previously studied using other techniques in \cite{Turgut1}.

\section{The SLOCC classes $L_{abc_2}$}
\subsection{The SLOCC classes $L_{abc_2}$ for $c\neq 0$}
Let us now treat the SLOCC classes $L_{abc_2}$. The seed states for $c \neq 0$ are of the form
\begin{eqnarray}\label{abc2}\ket{\Psi_{abc_2}} &= a\ket{\Phi^+}_{12}\ket{\Phi^+}_{34} - b\ket{\Phi^-}_{12}\ket{\Phi^-}_{34} +c (\ket{0110}+\ket{1001})-i/(2 c)\ket{1010}.\end{eqnarray}
The corresponding $Z$ and $\tilde{Z}$ matrices coincide and have symmetric Jordan form.
In particular, they are of the form
\begin{eqnarray}
&Z_{abc_2}=\tilde{Z}_{abc_2}= \left(
    \begin{array}{cccc}
      a^2 & 0 & 0&0 \\
      0 & b^2  & 0&0\\
 0 & 0  & c^2+\frac{i}{2}&\frac{1}{2}\\
0 & 0  & \frac{1}{2}&c^2-\frac{i}{2}\\
    \end{array}
  \right).
\end{eqnarray}
When determining the corresponding symmetries, one observes that there are more symmetries if two or all eigenvalues are the same. Thus, we have to investigate these cases separately (see subsequent subsections). Let us briefly summarize the results here. One observes that except for the case when $a=b=0$ and $c\neq 0$ the states that are reachable are always of a specific form. This is due to the fact that the only non-trivial symmetry that will be used to construct LOCC protocols for the reachable and convertible states in these cases is $\sigma_z^{\otimes 4}$. In particular, respecting the standard form of the respective SLOCC class
the only states which are reachable are of the form $h_i\otimes d_j\otimes d_k\otimes d_l \ket{\Psi}$ where $i,j,k,l \in \{1,2,3,4\}$ and $h_i\neq d_i$. Here $\ket{\Psi}$ denotes the chosen representative of the SLOCC class. Furthermore, the states that are non-isolated and in $MES_4$ are given by $d_1\otimes d_2\otimes d_3\otimes d_4 \ket{\Psi}$ where, again, one has to take into account the standard form of the respective SLOCC class. In case $a=b=0$ and $c\neq 0$ the reachable states are given by $\one \otimes h_2\otimes h_3^x\otimes h_4 \ket{\Psi_{00c_2}}$ where $h_2\neq d_2$ and $h_1^x \otimes d_2\otimes h_3^x\otimes h_4 \ket{\Psi_{00c_2}}$ where $h_1^x\neq \one$ and $h_1^x \otimes h_2\otimes \one\otimes h_4 \ket{\Psi_{00c_2}}$ where $h_4\neq d_4$ and $h_1^x \otimes h_2\otimes h_3^x\otimes d_4 \ket{\Psi_{00c_2}}$ where $h_3^x\neq \one$.
\subsubsection{The case $a^2\neq b^2\neq c^2\neq a^2$ and $c\neq 0$}
We will first consider the case $a^2\neq b^2\neq c^2\neq a^2$ and $c\neq 0$.  The only non-trivial symmetry for these states is given by $\sigma_z^{\otimes 4}$.
Using this symmetry we define the following standard form for states in these SLOCC classes $\otimes_i g_i \ket{\Psi_{abc_2}}$, where the phase of $\tr(\ket{1}\bra{0}g_i^\dagger g_i)$ is in $[0,\pi)$, where $i\in \{1,2,3,4\}$ denotes the first party for which $\tr(\ket{1}\bra{0}g_i^\dagger g_i)$ is non-zero. In the following lemma we show which  states in these classes are reachable.
\begin{lemma}\label{lemmasz} The only states in the SLOCC classes $L_{abc_2}$ where  $a^2\neq b^2\neq c^2\neq a^2$ and $c\neq 0$ that are reachable via LOCC are given by $h_i\otimes d_j\otimes d_k\otimes d_l \ket{\Psi_{abc_2}}$ where $\{i,j,k,l\} = \{1,2,3,4\}$ and $h_i\neq d_i$.
\end{lemma}
\begin{proof}
Recall that the only non-trivial symmetry of $\ket{\Psi_{abc_2}}$ is given by $\sigma_z^{\otimes 4}$.
Using Eq. (\ref{EqSep}) one obtains that
\bea\sum_{m=0}^1 p_m \otimes_i \left(
    \begin{array}{cccc}
      & h_i^1&h_i^2 (-1)^m \\
     & h_i^{2 *}(-1)^m & h_i^3 \\
     \end{array}
  \right) =
r \otimes_iG_i.\eea
Since the symmetry is unitary and the operators are normalized, i.e. $\tr(G_i)=\tr(H_i)=1 \,\,\,\forall i$, it follows that $r=1$. Tracing over all but one party $l$ in the equation above, one obtains that $h_l^1=g_l^1$, $h_l^3=g_l^3$ and $h_l^2 (p_0 -p_1) =g_l^2$. Note here that the diagonal elements of $G_l$ can not be changed. Note further that in case $h_k^2=0$ for some $k\in\{1,2,3,4\}$ it follows that $g_k^2=0$ and therefore in this case $H_k=G_k$. Thus, states of the form $d_1\otimes d_2\otimes d_3\otimes d_4 \ket{\Psi_{abc_2}}$ are in $MES_4$.
Consider now the case where at least two of the matrices $H_i$ are non-diagonal, e.g. $j$ and $l$.
Tracing over all parties except of $j$ and $l$ and considering the matrix element $\ket{01}_{j,l}\bra{10}$ leads to $h_j^{2 *} h_l^2= g_j^{2 *} g_l^2$. Using now that $h_k^2 (p_0 -p_1) =g_k^2\,\,\, \forall k$ one obtains that in case $h_j^2,h_l^2\neq 0$ it has to hold that $(p_0-p_1)^2=1$ which corresponds to a trivial transformation.
Thus, the only states that can possibly be reached via SEP are of them form $h_i\otimes d_j\otimes d_k\otimes d_l \ket{\Psi_{abc_2}}$ where $\{i,j,k,l\} = \{1,2,3,4\}$. In particular, they can be reached via a LOCC protocol. For example the POVM $\{1/\sqrt{2} h_i d_i^{-1},1/\sqrt{2} h_i\sigma_z d_i^{-1}\}$ applied to party $i$ allows to obtain them from a state of the form $d_i\otimes d_j\otimes d_k\otimes d_l \ket{\Psi_{abc_2}}$ for a properly chosen $d_i$ (in case of the second outcome all other parties have to apply $\sigma_z$).
\end{proof}
Thus, the states in $MES_4$ are given by the ones that can not be written as in Lemma \ref{lemmasz}.\\
Note further that these results can be easily extended to the $n$-qubit case for SLOCC classes who contain a state $\ket{\Psi}$ with symmetry $S(\Psi)=\{\one^{\otimes n},\sigma_z^{\otimes n}\}$. From the proof of Lemma \ref{lemmasz} it follows directly that all states in these SLOCC  classes which are not of the form $h\ket{\Psi}$ where all $h_i$ but one are diagonal are in $MES_n$.
In the following lemma we show which states are convertible in these SLOCC classes.
\begin{lemma}\label{lemc1} The only states in the SLOCC classes $L_{abc_2}$ where  $a^2\neq b^2\neq c^2\neq a^2$ and $c\neq 0$ that are convertible via LOCC are given by $g_i\otimes d_j\otimes d_k\otimes d_l \ket{\Psi_{abc_2}}$ where $\{i,j,k,l\} = \{1,2,3,4\}$ and $g_i$ arbitrary.
\end{lemma}
\begin{proof}
This can be easily seen by using Eq. (\ref{EqSep}) and Lemma \ref{lemmasz}. Since the only possible final states are given by the reachable states one obtains that a necessary condition for a state to be convertible is that it is of the form $g_i\otimes d_j\otimes d_k\otimes d_l \ket{\Psi}$. In order to see that they are indeed convertible note that $\{\sqrt{p} (h_ig_i^{-1}\otimes\one^{\otimes 3}),\sqrt{1-p} (h_i\sigma_zg_i^{-1}\otimes \sigma_z^{\otimes 3)}\}$ constitutes a valid POVM for $h_i^1=g_i^1$, $h_i^3=g_i^3$ and $(2 p -1 ) h_i^2=g_i^2$. Since the only restriction is that $G_i$ ($H_i$) is a positive operator of rank 2, one can find for any $G_i$ a $H_i$ (which is not LU-equivalent) and a value of  $p$ such that these conditions are fulfilled.
\end{proof}
Combining  Lemma \ref{lemmasz} and Lemma \ref{lemc1} one obtains that the convertible states in the MES are given by $d_1\otimes d_2\otimes d_3\otimes d_4\ket{\Psi_{abc_2}}$.
Note that in the proofs of Lemma \ref{lemmasz} and Lemma \ref{lemc1} we just used the symmetries and were not referring to the standard form.
Since the symmetry  $S=\{\one^{\otimes 4},\sigma_z^{\otimes 4}\}$ also occurs in the context of other SLOCC classes we reformulate these lemmata as follows.
\begin{lemma}\label{lem0} Given a 4-qubit SLOCC class which contains a state $\ket{\Psi}$ with $S(\Psi)=\{\one^{\otimes 4},\sigma_z^{\otimes 4}\}$  the only states (in standard form) that are reachable via LOCC are given by $h_i\otimes d_j\otimes d_k\otimes d_l \ket{\Psi}$ where $\{i,j,k,l\} = \{1,2,3,4\}$ and $h_i\neq d_i$.
\end{lemma}

\begin{lemma} \label{lemc11} Given a 4-qubit SLOCC class which contains a state $\ket{\Psi}$ with $S(\Psi)=\{\one^{\otimes 4},\sigma_z^{\otimes 4}\}$  the only states (in standard form) that are convertible via LOCC are given by $g_i\otimes d_j\otimes d_k\otimes d_l \ket{\Psi_{abc_2}}$ where $\{i,j,k,l\} = \{1,2,3,4\}$ and $g_i$ arbitrary.
\end{lemma}

Note here that these results can also be used in the context of other SLOCC classes whose seed states have symmetries that include $\sigma_z^{\otimes 4}$ if one can show that the only non-trivial symmetry that can be used for SEP transformations is given by $\sigma_z^{\otimes 4}$.

\subsubsection{The case $a^2=c^2\neq b^2$ and $c\neq 0$}
Let us now consider the SLOCC classes $L_{abc_2}$ for which $a^2=c^2\neq b^2$ and $c\neq 0$. Note that the case $b^2=c^2\neq a^2$ is included in these classes, since the corresponding states are LU equivalent to each other (see Eq. (\ref{abc2}) and apply $Z(-\pi/4)^{\otimes 4}$). Note further that $\ket{\Psi_{abc_2}}$ with $a=c\neq 0$ is equal to $\sigma_z\otimes \sigma_z\otimes \one\otimes \one\ket{\Psi_{abc_2}}$ with $a=-c\neq 0$. Therefore, these two choices of parameters represent the same SLOCC class. Thus, we could take as seed states for all these parameter choices the states \begin{eqnarray}\ket{\Psi} &= a\ket{\Phi^+}_{12}\ket{\Phi^+}_{34} - b\ket{\Phi^-}_{12}\ket{\Phi^-}_{34} +a (\ket{0110}+\ket{1001})-i/(2 a)\ket{1010}.\end{eqnarray} Using the method described in Sec. II one obtains the following non-compact symmetry group
\begin{eqnarray}
&S_{x}= \langle \sigma_z^{\otimes 4}, \left(
    \begin{array}{cc}
      1 & 0 \\
      x & 1 \\
    \end{array}
  \right)\otimes\left(
    \begin{array}{cc}
      1 & x \\
      0 & 1 \\
    \end{array}
  \right)\otimes\left(
    \begin{array}{cc}
      1 & 0\\
      -x & 1 \\
    \end{array}
  \right)
\otimes\left(
    \begin{array}{cc}
      1 & -x \\
      0 & 1 \\
    \end{array}\right)\rangle ,
\end{eqnarray}
where $x\in\C$. In order to simplify the symmetry operations we will consider the states \bea\ket{\Psi_{aba_2}}=\sigma_x\otimes\one\otimes\sigma_y\otimes\sigma_z\ket{\Psi}\eea as seed states, whose symmetries are given by
\begin{eqnarray}
&\tilde{S}_{x,m}\equiv \otimes_i s_i(x,m)=& [(\sigma_z)^{m} \left(
    \begin{array}{cc}
      1 & x \\
      0 & 1 \\
    \end{array}
  \right)]^{\otimes 4},\\
\end{eqnarray}
where $m\in\{0,1\}$ and $x\in \C$.
As we will show the symmetry, $\tilde{S}_{x,m}$, allows to establish a certain standard form for states in these classes but LOCC transformations are only possible due to the symmetry $\sigma_z^{\otimes 4}$.
Denote by $\otimes_i g_i \ket{\Psi_{aba_2}}$ an arbitrary state in this SLOCC class.  We define the following standard form: If $g_i^2/g_i^1\neq g_j^2/g_j^1\,\,\, \forall i, j \in \{1,2,3,4\}$ with $i\neq j$, we choose as standard form  $d_1\otimes g_2\otimes g_3\otimes g_4 \ket{\Psi_{aba_2}}$. Note here that the symmetry allows us to choose this standard form as for $x=-g_i^2/g_i^1$ the matrix $ s_i(x,0)^{\dagger} G_i s_i(x,0)$ is diagonal.  For the states where $g_i^2/g_i^1= g_j^2/g_j^1$ for some $i, j \in \{1,2,3,4\}$ the standard form can be chosen to be $d_i\otimes d_j\otimes g_k\otimes g_l \ket{\Psi_{aba_2}}$ where $i,j,k,l \in \{1,2,3,4\}$ and for $g_i^2/g_i^1= g_j^2/g_j^1=g_k^2/g_k^1\neq0$ one obtains $d_i\otimes d_j\otimes d_k\otimes g_l \ket{\Psi_{aba_2}}$ etc.. Until now we just exploited the symmetries for $m=0$. The additional symmetry allows to restrict the phase of $g_i^2$ to $[0,\pi]$, where $i\in \{1,2,3,4\}$ denotes the first party for which $g_i^2$ is non-zero.
Using this standard form we will now prove the following lemma.
\begin{lemma}\label{lem2} The only states in the SLOCC classes $L_{abc_2}$ where  $a^2= c^2\neq b^2$ and $c\neq 0$ that are reachable via LOCC are given by $d_i\otimes h_j\otimes d_k\otimes d_l \ket{\Psi_{aba_2}}$ where $\{i,j,k,l\} = \{1,2,3,4\}$ and $h_j\neq d_j$.
\end{lemma}
\begin{proof}
Denote by $\otimes_i g_i \ket{\Psi_{aba_2}}$ an arbitrary initial state and by $\otimes_i h_i \ket{\Psi_{aba_2}}$ an arbitrary final state in standard form.
Using Eq. (\ref{EqSep}) one obtains that
\begin{eqnarray}\label{EQxm} &\sum_{m,x} p_{m,x} \otimes_i \left(
    \begin{array}{cccc}
       & h_i^1&x h_i^1+ h_i^2 (-1)^m \\
     &x^* h_i^1+ h_i^{2 *} (-1)^m &y_i\\
     \end{array}
  \right)=
r \otimes_iG_i,\end{eqnarray}
where $y_i=h_i^1|x|^2+(x^*h_i^2+x h_i^{2 *})(-1)^m+ h_i^3$. Note that according to the standard form at least one $H_i$ and $G_j$ has to be diagonal for some $i,j \in \{1,2,3,4\}$. For better readability we choose $i=1$. \\
First we will consider the case $j\neq 1$. Then according to the standard form there have to be at least two $G_j$ that are diagonal e.g. $j=2,4$. In this case considering the matrix element $\ket{0100}\bra{1000}$ in Eq. (\ref{EQxm}) leads to $ h_1^1 h_3^1 h_4^1\sum_{m,x} p_{m,x}(|x|^2 h_2^1+ h_2^2 x^*(-1)^m)=0$. Note that as $H_i$ has to be a positive operator of rank 2 it follows that $h_i^1\neq0$ and therefore $\sum_{m,x} p_{m,x}(|x|^2 h_2^1+ h_2^2 x^*(-1)^m)=0$. Considering now the matrix element $\ket{0001}\bra{1000}$, which leads to $\sum_{m,x} p_{m,x}(|x|^2 h_4^1+ h_4^2 x^*(-1)^m)=0$. Thus either $\sum_{m,x} p_{m,x}|x|^2=0$ and  $\sum_{m,x} p_{m,x} x^*(-1)^m=0$ or $h_2^2/h_2^1= h_4^2/h_4^1$. For the latter case the corresponding standard form would have been such that $h_4^2=h_2^2=0$. In this case $\sum_{m,x} p_{m,x}|x|^2=0$ would follow for example from considering the matrix element $\ket{0100}\bra{1000}$. Hence, in all these cases, the condition $\sum_{m,x} p_{m,x}|x|^2=0$, which can only be fulfilled if $p_{m,x}=0$ for $x\neq0$, has to be satisfied. Thus the only nontrivial symmetry that can be used in this case is $\sigma_z^{\otimes 4}$. As has been shown in Lemma \ref{lem0} the only reachable states using this symmetry are of the form $g_i\otimes d_j\otimes d_k\otimes d_l \ket{\Psi}$ where $i,j,k,l \in \{1,2,3,4\}$, $g_i\neq d_i$ and $\ket{\Psi}$ denotes the corresponding seed state in standard form. \\
Second and last, we will treat the case $j=1$. Considering the matrix elements $\ket{0100}_{1klm}\bra{1000}$ where $\{k,l,m\} = \{2,3,4\}$ and taking into account that $h_i^1\neq 0 $ $\forall i\in\{1,2,3,4\}$ one obtains that $\sum p_{m,x}(|x|^2 h_k^1+ h_k^2 x^*(-1)^m)=0$ for any $k\in\{2,3,4\}$. Thus either
$\sum_{m,x} p_{m,x}|x|^2=0$ and $\sum p_{m,x} x^*(-1)^m=0$, which has already been considered above, or $h_2^2/h_2^1= h_3^2/h_3^1= h_4^2/h_4^1$.  For $h_2^2/h_2^1= h_3^2/h_3^1= h_4^2/h_4^1$ the corresponding standard form is chosen such that  $H_j$ is diagonal $\forall j$ which implies similarly to before that $\sum p_{m,x}|x|^2=0$ (see Eq. (\ref{EQxm})). Therefore, also in this case the only symmetry that makes transformations possible is given by $\sigma_z^{\otimes 4}$ and according to Lemma \ref{lem0} the only states that can be reached are of the form $h_i\otimes d_j\otimes d_k\otimes d_l \ket{\Psi}$ where $\{i,j,k,l\} = \{1,2,3,4\}$ and $h_i\neq d_i$. \\
For both cases the corresponding LOCC protocol would be again $\{1/\sqrt{2} h_i d_i^{-1},1/\sqrt{2} h_i\sigma_z d_i^{-1}\}$ for a properly chosen $d_i$, where in case of the first (second) outcome the other parties do nothing (apply $\sigma_z$).
\end{proof}
Thus all states except the ones that are of the form $h_i\otimes d_j\otimes d_k\otimes d_l \ket{\Psi_{aba_2}}$ where $\{i,j,k,l\} = \{1,2,3,4\}$ and $h_i\neq d_i$ are in $MES_4$. Note that as has been seen in the proof of Lemma \ref{lem2} the only transformations that are possible are due to the symmetry $\sigma_z^{\otimes 4}$. Thus, according to Lemma \ref{lemc11} the states that are convertible are given by $g_i\otimes d_j\otimes d_k\otimes d_l \ket{\Psi_{aba_2}}$ with $\{i,j,k,l\} = \{1,2,3,4\}$ and $g_i$ arbitrary. Therefore, the subset of convertible MES states is given by \bea \{d_1\otimes d_2\otimes d_3\otimes d_4 \ket{\Psi_{aba_2}}, d_i \,\textrm{diagonal}\}.\eea
\subsubsection{The case $a^2=b^2\neq c^2$ and $a,c\neq 0$}
We will consider next the SLOCC classes $L_{abc_2}$ for which $a^2=b^2\neq c^2$ and $a,c\neq 0$. Note that the seed states for $a=-b$ coincide with those for $a=b$ after applying $\one\otimes\one\otimes\sigma_x\otimes\sigma_x$ and exchanging particles 3 and 4 (see Eq. (\ref{abc2})). Therefore, they represent the same SLOCC classes up to particle exchange. In the following we will take as seed states \begin{eqnarray}\ket{\Psi_{aac_2}} &= a(\ket{0011}+\ket{1100}) +c (\ket{0110}+\ket{1001})-i/(2 c)\ket{1010}.\end{eqnarray}
The corresponding symmetries are given by \bea S_z= P_z^{\otimes 4},\eea where $P_z= diag(z,1/z)$ and $z\in \C\backslash 0$. To derive the standard form we use that for a properly chosen value of $z$ the matrix $P_z^\dagger G_i P_z\propto G_i^x$ where $ G_i^x=g_i^{x \dagger} g_i^x=1/2\one+g_i^2\sigma_x$ and one can choose $g_i^2\in [0,1/2)$.  Thus, for $g_1^2\neq0$ the standard form is given by $g_1^x\otimes g_2\otimes g_3\otimes g_4 \ket{\Psi_{aac_2}}$, where $g_1^x=\sqrt{G_1^x}$ and $g_1^2>0$. In case $g_1^2=0$ the symmetry $Z(\alpha)^{\otimes 4}$ for $\alpha \in \R$ can be used to choose $g_2^2\geq 0$ \footnote{Note that in case $g_1^2=0$ the symmetry can not be used to choose $g_1\propto \one$ and $g_2=g_2^x$ since $P_z^\dagger \one P_z\not\propto \one$ for $|z|\neq 1$. Thus, in this case only the unitary symmetry can be used for a restriction on $g_2$.} and similarly if $g_1^2=g_2^2=0$ one can set $g_3^2\geq	0$ etc.. Similarly to the previous SLOCC classes we use the standard form to prove the following lemma.
\begin{lemma}\label{lem3} The only states in the SLOCC classes $L_{abc_2}$ where  $a^2= b^2\neq c^2$ and $a, c\neq 0$ that are reachable via LOCC are  given by  $h_1^x\otimes d_2\otimes d_3\otimes d_4 \ket{\Psi_{aac_2}}$ where  $h_1^x\neq \one/2$ and  $\one_1 \otimes d_j\otimes d_k\otimes h_l \ket{\Psi_{aac_2}}$ where  $\{j,k,l\} = \{2,3,4\}$ and $h_l\neq d_l$.
\end{lemma}
\begin{proof}
According to Eq. (\ref{EqSep}) an initial state $\otimes_i g_i \ket{\Psi}$ can be transformed via SEP into $\otimes_i h_i \ket{\Psi}$ iff
\begin{eqnarray}\label{Eq8}&\sum_{|z|,\phi} p_{|z|,\phi} \otimes_i \left(
    \begin{array}{cccc}
       & h_i^1 |z|^2&  h_i^2 e^{-i2\phi} \\
     &h_i^{2 *}e^{i2\phi}  &h_i^3 1/|z|^2\\
     \end{array}
  \right)=
r \otimes_iG_i,\end{eqnarray}
where $z=|z|e^{i\phi}$.
Recall that due to the definition of the standard form  $h_1^1=h_1^3=1/2$ and $h_1^2\in\R$.
Considering the matrix elements $\ket{ijkl}\bra{ijkl}$ of Eq. (\ref{Eq8}), where  $i,j,k,l \in \{0,1\}$ one obtains that \bea \label{EqdPz} h_2^{1+2j} h_3^{1+2k} h_4^{1+2l}  \sum_{|z|,\phi} p_{|z|,\phi} |z|^{4[2-(i+j+k+l)]} =r  g_2^{1+2j} g_3^{1+2k} g_4^{1+2l}.\eea Thus, from this equation for $i=j=0$ and $k=l=1$, as well as for $j=0$ and $i=k=l=1$ it follows that $\sum_{|z|,\phi} p_{|z|,\phi} 1/|z|^4=1$. Note that we used here that $h_i^1, g_i^1, h_i^3, g_i^3\neq 0 \forall i\in\{2,3,4\}$.
Considering now Eq. (\ref{EqdPz}) with $i=j=k=0$ and $l=1$ and with  $j=k=0$ and $i=l=1$ leads to $\sum_{|z|,\phi} p_{|z|,\phi} |z|^4=1$.
The condition $\sum_{|z|,\phi} p_{|z|,\phi} |z|^4=\sum_{|z|,\phi} p_{|z|,\phi} 1/|z|^4=1$ can only be fulfilled if $ p_{|z|,\phi}=0 $ $\forall |z|\neq 1$.
Thus, only unitary symmetries can be used for SEP transformations and it follows that $r=1$ from Observation \ref{obsmon}. Furthermore, it is now easy to see from Eq. (\ref{Eq8}) via tracing over all parties but party $k$ for $k \in \{2,3,4\}$ that $h_k^1=g_k^1$ and $h_k^3=g_k^3$. Thus, via SEP the diagonal elements can not be changed. \\
Consider now the case that there are at least two operators $H_l$ with non-zero off-diagonal components, say $H_i$ and $H_j$. Choose $i$ such that according to the standard form $h_i^2\in\R$. Tracing over the other parties and considering the matrix elements $\ket{10}_{i,j}\bra{00}$ and $\ket{01}_{i,j}\bra{00}$, as well as $\ket{10}_{i,j}\bra{01}$ leads to $h_k^2 \sum_{|z|,\phi} p_{|z|,\phi} e^{-i2\phi}= g_k^2$ for $k=i,j$ and $h_j^{2 *}h_i^2=g_j^{2 *}g_i^2$. Thus,  $h_j^{2}/h_i^2=g_j^{2}/g_i^2$ and $h_j^{2 *}h_i^2=g_j^{2 *}g_i^2$ \footnote{Note that in case $h_i^2,h_j^2\neq 0$ the solution $\sum_{|z|,\phi} p_{|z|,\phi} e^{-i2\phi}=0$ is not possible, because it would imply that $g_k^2=0$ for $k=i,j$ and thus $h_k^2=0$ for some $k=i,j$ due to $h_j^{2 *}h_i^2=g_j^{2 *}g_i^2$.} and therefore $|h_j^{2}|=|g_j^{2}|$ and $h_i^{2}=g_i^{2}$. Since $h_i^2, g_i^2\in\R$ it follows that $\sum_{|z|,\phi} p_{|z|,\phi} e^{-2i\phi} \in \R$ and therefore $h_j^{2}=g_j^{2}$. Note that in case there are more than these two $H_l$ with non-zero off-diagonal components a completely analogous argument can be applied. Note further that in case $h_k^{2}=0$ for some $k\in\{1,2,3,4\}$ it follows that $g_k^{2}=0$ which can be easily seen by considering the matrix element $\ket{0010}_{i,j,k,l}\bra{0000}$. Thus, whenever there exist at least two $H_l$ with non-zero off-diagonal components or all $H_l$ are diagonal it must hold that $H_l=G_l \,\,\,\forall l\in\{1,2,3,4\}$. Hence, unless there is only one non-diagonal $H_l$ the state is not reachable via SEP. For any other instance, the corresponding state, $h_i\otimes d_j\otimes d_k\otimes d_l \ket{\Psi_{aac_2}}$ where $\{i,j,k,l\} = \{1,2,3,4\}$ and $h_i\neq d_i$, can be obtained for example via the LOCC protocol $\{1/\sqrt{2} h_i\one^{\otimes 4}d_i^{-1},1/\sqrt{2} h_i\sigma_z^{\otimes 4}d_i^{-1}\}$ acting non-trivially on party $i$ from a state $d_i\otimes d_j\otimes d_k\otimes d_l \ket{\Psi_{aac_2}}$ for a properly chosen $d_i$.
\end{proof}
As before, we observe that the non-unitary symmetry allows to define the standard form but the existence of the non-trivial LOCC protocols (and SEP transformations) is solely due to the unitary symmetry. Note that the structure of the states that are reachable is the same as for all other previously discussed SLOCC classes. Again, the standard form of the SLOCC classes is crucial here.
In the following lemma the states that are convertible are presented.
\begin{lemma} \label{lem30} The only states  in the SLOCC classes $L_{abc_2}$ where  $a^2= b^2\neq c^2$ and $a, c\neq 0$ that are convertible via LOCC are  given by $g_1^x\otimes d_2\otimes d_3\otimes d_4 \ket{\Psi_{aac_2}}$ or $\one_1 \otimes d_j\otimes d_k\otimes g_l \ket{\Psi_{aac_2}}$ where $\{j,k,l\} = \{2,3,4\}$.
\end{lemma}
\begin{proof}
In the proof of Lemma \ref{lem3} we have shown that the only symmetries that can be useful for non-trivial LOCC transformations are phase gates. Thus, using Eq. (\ref{EqSep}) it can be easily seen that a state with $H_i= D_i$ can only be obtained from a state with $G_i= D_i$. Note that the only reachable states are given by the states in Lemma \ref{lem3}. Combining these results implies that the only possibly convertible states are given by $g_1^x\otimes d_2\otimes d_3\otimes d_4 \ket{\Psi_{aac_2}}$ or $\one_1 \otimes d_j\otimes d_k\otimes g_l \ket{\Psi_{aac_2}}$ where $\{j,k,l\} = \{2,3,4\}$. Let us now show that these states are indeed convertible. In order to see this consider $\{\sqrt{p} h_i\one^{\otimes 4}g_i^{-1},\sqrt{1-p} h_i\sigma_z^{\otimes 4}g_i^{-1}\}$ where party $i$ denotes the party for which $G_i$ is non-diagonal. In case all of them are diagonal the same protocol can be used where one can freely choose the party which is acting non-trivially. This protocol constitutes a valid POVM for $h_i^1=g_i^1$, $h_i^3=g_i^3$ and $(2 p -1 ) h_i^2=g_i^2$. Note that $p$ can again be chosen  such that $H_i$ is a positive operator of rank 2.
\end{proof}
Thus, in these SLOCC classes the subset of states in $MES_4$ which are not isolated is given by \bea \{\one_1 \otimes d_2\otimes d_3\otimes d_4 \ket{\Psi_{aac_2}}, d_i \,\textrm{diagonal}\}.\eea
Since the structure of the reachable and convertible states (as well as the standard form) only depends on the symmetries and this symmetry will occur for another SLOCC class too we generalize here Lemma \ref{lem3} and Lemma \ref{lem30} to the following lemmata.
\begin{lemma}\label{lem33} Given a 4-qubit SLOCC class which contains a state $\ket{\Psi}$ with symmetries $S(\Psi)=\{P_z^{\otimes 4}:z\in\C/0\}$  the only states (in standard form) that are reachable via LOCC are given by  $h_1^x\otimes d_2\otimes d_3\otimes d_4 \ket{\Psi}$ where  $h_1^x\neq \one/2$ and  $\one_1 \otimes d_j\otimes d_k\otimes h_l \ket{\Psi}$ where  $\{j,k,l\} = \{2,3,4\}$ and $h_l\neq d_l$.
\end{lemma}
\begin{lemma}\label{lem333} Given a 4-qubit SLOCC class which contains a state $\ket{\Psi}$ with symmetries $S(\Psi)=\{P_z^{\otimes 4}:z\in\C/0\}$  the only states (in standard form) that are convertible via LOCC are given by  $g_1^x\otimes d_2\otimes d_3\otimes d_4 \ket{\Psi}$  and  $\one_1 \otimes d_j\otimes d_k\otimes g_l \ket{\Psi}$ where  $\{j,k,l\} = \{2,3,4\}$.
\end{lemma}
Note that these lemmata can also be used in the context of SLOCC classes with seed states that have different symmetries if one can show that the only symmetries that can be used for SEP transformations are given by $P_z^{\otimes 4}$ and the standard form for these SLOCC classes is chosen such that it obeys the following conditions: It holds that $g_1=g_1^x$ and in case $g_1^2=0$ one chooses  $g_2^2\geq 0$ and similarly if $g_1^2=g_2^2=0$ one has $g_3^2\geq	0$ etc..

\subsubsection{The case $a^2=b^2=c^2\neq 0$}
The seed states for $L_{abc_2}$ with $a=b=-c$ and $-a=b=c$ can be transformed into the one with $a=b=c$ via LUs and in case $-a=b=c$ additionally exchanging party $3$ and $4$ (see Eq. (\ref{abc2})). Furthermore, one can transform the seed state for $a=-b=c$ into the one with $a=b=c$ by applying $\one\otimes\one\otimes\sigma_x\otimes\sigma_x$ to the state and exchanging party 3 and 4.
Thus, these choices of parameters are within the same SLOCC classes as $a=b=c$ (up to particle exchange), for which we consider as seed states the states \begin{eqnarray}\label{seedaaa2} \ket{\Psi_{aaa_2}} &= a(\ket{0011}+\ket{1100}) +a (\ket{0110}+\ket{1001})-i/(2 a)\ket{1010}.\end{eqnarray} The corresponding symmetries are given by
\begin{eqnarray}\nonumber
&S_{x,y,z}=\otimes s_i =\left(
    \begin{array}{cc}
      z & 0 \\
      x & 1/z \\
    \end{array}
  \right)\otimes\left(
    \begin{array}{cc}
      z & y \\
      0 & 1/z \\
    \end{array}
  \right)\otimes\left(
    \begin{array}{cc}
      z & 0\\
      -x & 1/z \\
    \end{array}
  \right)
\otimes\left(
    \begin{array}{cc}
      z & -y \\
      0 & 1/z \\
    \end{array}\right),
\end{eqnarray}
where  $x,y,z\in\C$ and $z\neq 0$. The standard form for any state $\otimes_i g_i\ket{\Psi_{aaa_2}}$ in these SLOCC classes can be chosen as $\one\otimes d_2\otimes g_3\otimes g_4\ket{\Psi_{aaa_2}}$, where $\tr(\ket{1}\bra{0}g_3^\dagger g_3)> 0$ if $\tr(\ket{1}\bra{0}g_3^\dagger g_3)\neq 0$ and $\tr(\ket{1}\bra{0}g_4^\dagger g_4)\geq 0$  otherwise. This can be easily seen, as $s_1^{ \dagger} g_1^\dagger g_1 s_1 \propto \one$ for properly chosen $x$ and $|z|$ and as $s_2^{ \dagger} g_2^\dagger g_2 s_2$ is diagonal by choosing $y$ correspondingly.  Exploiting additionally the $Z(\alpha)^{\otimes 4}$ symmetry allows to choose   $\tr(\ket{1}\bra{0}g_3^\dagger g_3)\geq 0$ or $\tr(\ket{1}\bra{0}g_4^\dagger g_4)\geq 0$. In the following lemma we characterize the reachable states in these SLOCC classes.
\begin{lemma}\label{lem4} The only states in the SLOCC classes $L_{abc_2}$ where  $a^2= b^2= c^2\neq 0$ that are reachable via LOCC are given by $\one \otimes d_2\otimes d_k\otimes h_l \ket{\Psi_{aaa_2}}$ where $\{k,l\} = \{3,4\}$ and $h_l\neq d_l$.
\end{lemma}
\begin{proof}
We will in the following divide all states into the subsequent three classes: (i) $h_4^2,h_3^2\neq 0$ (ii) $h_4^2=h_3^2= 0$ and (iii) $h_i^2=0$ and $h_j^2\neq 0$ for $\{i,j\}=\{3,4\}$. We will first show that states corresponding to class (i) and (ii) are not reachable, whereas we will construct for states corresponding to class (iii) LOCC protocols to show that these states are reachable, which proves the lemma.

Let us first consider the case (i). Considering the matrix elements $\ket{1010}\bra{0110}$ and $\ket{1010}\bra{0011}$ of Eq. (\ref{EqSep}) leads to $\sum p_{x,y,z} x/z=0$. Considering now the projection of Eq. (\ref{EqSep}) onto $\ket{\Psi^+}_{13}\ket{\Psi^+}_{24}$ leads to an equation of $r$ in terms of the entries of the matrices $H$ and $G$. Combining this equation, the previous one and the one resulting from projecting Eq. (\ref{EqSep}) onto $\ket{\Psi^+}_{24}$ leads to $x=0$ for any symmetry occurring in Eq. (\ref{EqSep}). Considering now the matrix element $\ket{1100}\bra{1010}$ of Eq. (\ref{EqSep}) leads to the condition $\sum_{y,z} p_{0,y,z}y^\ast z^\ast=0$. Combining it with the equation resulting from considering the matrix element $\ket{1100}\bra{1001}$ one obtains $\sum_{y,z} p_{0,y,z} |y|^2 |z|^2=0$, which leads to (as $z\neq 0$) $y=0$ for any symmetry occurring in Eq. (\ref{EqSep}). Hence, the only symmetries that allow for transformations are given by $P_z^{\otimes 4}$. Therefore, we can use Lemma \ref{lem33} which shows that states of the form $\one\otimes d_2\otimes h_3\otimes h_4\ket{\Psi}$ where $h_i\neq d_i$ for $i=3,4$ are not reachable via LOCC if just this symmetry is available \footnote{Note that $\ket{\Psi}$ denotes here the respective seed state as given in Eq. (\ref{seedaaa2}). Note further that the standard form that can be established via the symmetry $S_{x,y,z}$ obeys the constraints of the standard form given due to the symmetry $P_z^{\otimes 4}$.}.

Let us now consider the case (ii), for which it follows easily from Eq. (\ref{EqSep}) that the only symmetries that allow for transformations are given by $P_z^{\otimes 4}$. Using then again Lemma \ref{lem33} shows that states of the form $\one\otimes d_2\otimes d_3\otimes d_4\ket{\Psi}$ are not reachable via LOCC.

It remains to consider case (iii). As $\sigma_z^{\otimes 4}$ is included in $S_{x,y,z}$, the POVM $\{1/\sqrt{2} h_j(\one^{\otimes 4}) d_j^{-1},1/\sqrt{2} h_j(\sigma_z^{\otimes 4}) d_j^{-1}\}$ (acting non-trivially on party $j$), where $d_j$ is chosen properly, allows to obtain the states corresponding to case (iii). In particular, states of the form $\one_1 \otimes d_2\otimes d_i\otimes h_j \ket{\Psi_{aaa_2}}$, where $\{i,j\}=\{3,4\}$, can be obtained by using this LOCC protocol from states of the form $\one_1 \otimes d_2\otimes d_i\otimes d_j \ket{\Psi_{aaa_2}}$.\\
In summary, all reachable states in these SLOCC classes are given by $\one_1 \otimes d_2\otimes d_k\otimes h_l \ket{\Psi_{aaa_2}}$ where  $\{k,l\} = \{3,4\}$ and $h_l\neq d_l$.
\end{proof}

Again we observe the same structure for reachable states. In the following lemma we show which states are convertible in these SLOCC classes.
\begin{lemma}\label{lem4c} The only states in the SLOCC classes $L_{abc_2}$ where  $a^2= b^2= c^2\neq 0$ that are convertible via LOCC are given by $\one_1 \otimes d_2\otimes d_k\otimes g_l \ket{\Psi_{aaa_2}}$ with $g_l$ arbitrary and $\{k,l\} = \{3,4\}$.
\end{lemma}
 \begin{proof}
Let us show that the only states which can be converted into the reachable states given in Lemma \ref{lem4} for $l=4$ are of the form $\one_1 \otimes d_2\otimes d_3\otimes g_4 \ket{\Psi_{aaa_2}}$ and that they are indeed convertible. As the case $l=3$ can be proven analogously, this proves the lemma. Considering the matrix element $\ket{1000}\bra{0010}$ of Eq. (\ref{EqSep}) for $H_2$, $H_3$ diagonal (and $H_1=\one$) leads to $x=0$ for any symmetry occurring in the equation. The fact that $G_3$ also has to be diagonal in this case is then implied by considering the matrix element $\ket{1000}\bra{1010}$ of Eq. (\ref{EqSep}), which proves the first statement. It can be easily shown that these states are all indeed convertible using the symmetries $\one^{\otimes 4}$ and $\sigma_z^{\otimes 4}$ to construct a LOCC protocol (see also proof of Lemma \ref{lem30}).
\end{proof}
Combining Lemma \ref{lem4} and Lemma \ref{lem4c} we have that the non-isolated states in $MES_4$ in these SLOCC classes are given by \bea \{\one_1 \otimes d_2\otimes d_3\otimes d_4 \ket{\Psi_{aaa_2}}, d_i \,\textrm{diagonal}\}.\eea
\subsubsection{\label{secab0} The case $a=b=0$ and $c\neq 0$}
We will consider now the SLOCC classes $L_{abc_2}$ for which $a=b=0$ and $c\neq 0$. The corresponding seed state is \bea \ket{\Psi_{00c_2}} = c (\ket{0110}+\ket{1001})-i/(2 c)\ket{1010}\eea with the symmetries \bea S_{z,\tilde{z}}=P_z\otimes P_z\otimes P_{\tilde{z}}\otimes P_{\tilde{z}}\eea for $z,\tilde{z}\in\C\backslash 0$. This symmetry allows to choose the standard form $g_1^x\otimes g_2\otimes g_3^x\otimes g_4\ket\Psi_{00c_2}$ where $g_i^x$ is defined as $\sqrt{G_i^x}$ with  $ G_i^x=g_i^{x \dagger} g_i^x=1/2\one+g_i^2\sigma_x$. Moreover, one can choose $g_i^2\in [0,1/2)$ for $i\in\{1,3\}$. In case $g_1^2=0 (g_3^2=0)$ one can set  $g_2^2\geq 0 (g_4^2\geq 0)$ respectively. Note that  $\one\otimes\one\otimes\sigma_z\otimes\sigma_z, \sigma_z\otimes\sigma_z\otimes\one\otimes\one \in S_{z,\tilde{z}}$. As we will show now this fact leads to a different structure of the states in $MES_4$ compared to the previously discussed SLOCC classes.
\begin{lemma}\label{lem5} The only states in the SLOCC classes $L_{abc_2}$ where  $a=b=0$ and $c\neq 0$ that are reachable via LOCC are given by $\one \otimes h_2\otimes h_3^x\otimes h_4 \ket{\Psi_{00c_2}}$ with $h_2\neq d_2$ and $h_1^x \otimes d_2\otimes h_3^x\otimes h_4 \ket{\Psi_{00c_2}}$ with $h_1^x\neq \one/2$ and $h_1^x \otimes h_2\otimes \one\otimes h_4 \ket{\Psi_{00c_2}}$ with $h_4\neq d_4$ and $h_1^x \otimes h_2\otimes h_3^x\otimes d_4 \ket{\Psi_{00c_2}}$ with $h_3^x\neq \one/2$.
\end{lemma}
\begin{proof}
Due to  Eq. (\ref{EqSep}) a SEP transformation from an arbitrary initial state in standard form, $\otimes_i g_i \ket{\Psi}$,  to an arbitrary final state in standard form, $\otimes_i h_i \ket{\Psi}$, is possible iff
\begin{eqnarray}\label{SEPpz1pz2}
&\sum_{|z|,|\tilde{z}|,\phi,\tilde{\phi}} p_{|z|,|\tilde{z}|,\phi,\tilde{\phi}} \left(
    \begin{array}{cc}
      |z|^2/2 &   h_1^2e^{-i2\phi} \\
    h_1^{2}e^{i2\phi} & 1/(2|z|^2) \\
    \end{array}
  \right)\otimes\left(
    \begin{array}{cc}
      h_2^1|z|^2 &   h_2^2e^{-i2\phi} \\
    h_2^{2  *}e^{i2\phi} & h_2^3/|z|^2 \\
    \end{array}
  \right)\otimes\left(
    \begin{array}{cc}
      |\tilde{z}|^2 /2&   h_3^2e^{-i2\tilde{\phi}} \\
    h_3^{2}e^{i2\tilde{\phi}} & 1/(2|\tilde{z}|^2) \\
    \end{array}
  \right)\\ \nonumber&\otimes\left(
    \begin{array}{cc}
      h_4^1|\tilde{z}|^2 &   h_4^2e^{-i2\tilde{\phi}} \\
    h_4^{2  *}e^{i2\tilde{\phi}} & h_4^3/|\tilde{z}|^2\\
    \end{array}\right)=r G_1^x\otimes G_2\otimes G_3^x\otimes G_4,
\end{eqnarray}
where we used the notation $z=|z|e^{i\phi}$ and $\tilde{z}=|\tilde{z}|e^{i\tilde{\phi}}$.
Considering the matrix elements $\ket{ijkl}\bra{ijkl}$ of this equation, where $i,j,k,l \in \{0,1\}$, one obtains that
\bea \label{EqdPz1Pz2} h_2^{1+2j} h_4^{1+2l}  \sum_{|z|,|\tilde{z}|,\phi,\tilde{\phi}} p_{|z|,|\tilde{z}|,\phi,\tilde{\phi}} |z|^{4[1-(i+j)]} |\tilde{z}|^{4[1-(k+l)]}=r  g_2^{1+2j}  g_4^{1+2l}.\eea
Considering now this equation for $i=k=0$ and $j=l=1$ and $k=0$ and $i=j=l=1$, as well as for $i=j=k=0$ and $l=1$ and $j=k=0$ and $i=l=1$ leads to $\sum_{|z|,|\tilde{z}|,\phi,\tilde{\phi}} p_{|z|,|\tilde{z}|,\phi,\tilde{\phi}}|z|^4=\sum_{|z|,|\tilde{z}|,\phi,\tilde{\phi}} p_{|z|,|\tilde{z}|,\phi,\tilde{\phi}}1/|z|^4=1$. This implies that the only symmetries that can be used to transform a state via SEP obey $|z|=1$. Via an analogous argument one can show that $\sum_{|z|,|\tilde{z}|,\phi,\tilde{\phi}} p_{|z|,|\tilde{z}|,\phi,\tilde{\phi}}|\tilde{z}|^4=\sum_{|z|,|\tilde{z}|,\phi,\tilde{\phi}} p_{|z|,|\tilde{z}|,\phi,\tilde{\phi}}1/|\tilde{z}|^4=1$ and therefore $p_{|z|,|\tilde{z}|,\phi,\tilde{\phi}}=0$ for $|\tilde{z}|\neq 1$.
Thus, we only have to consider unitary symmetries and therefore $r=1$ (see Observation \ref{obsmon}). Furthermore, via tracing over all parties but party $i$ for $i\in\{2,4\}$ in Eq. (\ref{SEPpz1pz2}) one obtains that $h_i^1=g_i^1$ and $h_i^3=g_i^3$.
Note that via tracing over party $3$ and $4$ in Eq. (\ref{SEPpz1pz2}) one obtains the same conditions as in the proof of Lemma \ref{lem3} for the off-diagonal components. Namely that if $h_1^2=h_2^2=0$ or if  $h_1^2,h_2^2\neq 0$ it follows that  $h_1^2= g_1^2$ and $h_2^2= g_2^2$. Analogously we have that if $h_3^2=h_4^2=0$  or if $h_3^2,h_4^2\neq 0$ it has to hold that  $h_3^2= g_3^2$ and $h_4^2= g_4^2$. Combining these results one obtains that the only states that can possibly be reached are the ones given in Lemma \ref{lem5}. In particular, they can be reached via a LOCC protocol exploiting the symmetries $\one\otimes\one\otimes\sigma_z\otimes\sigma_z$ or $\sigma_z\otimes\sigma_z\otimes\one\otimes\one$. The idea here is again that $\{\sqrt{p}h_ig_i^{-1},\sqrt{1-p}h_i\sigma_zg_i^{-1}\}$ constitutes a valid POVM for $h_i^1=g_i^1$, $h_i^3=g_i^3$ and $(2 p -1 ) h_i^2=g_i^2$, which for example can be fulfilled by $p=1/2$ and $g_i=diag(h_i^1,h_i^3)$\footnote{Note here that the structure of the reachable states is such that the symmetry used for the LOCC protocol commutes with all $g_l$ for $l\neq i$.}.
\end{proof}
The states in $MES_4$ for the SLOCC classes $L_{abc_2}$ with $a=b=0$ and $c\neq 0$ are then given by all states that can not be written as in Lemma \ref{lem5}. Note that again the LOCC transformations which allowed to obtain all reachable states were due to symmetries that are elements of the Pauli group.
We will show in the following lemma which states in these SLOCC classes are convertible.
\begin{lemma} The only states in the SLOCC classes $L_{abc_2}$ where  $a=b=0$ and $c\neq 0$ that are convertible via LOCC are given by $\one \otimes g_2\otimes g_3^x\otimes g_4 \ket{\Psi_{00c_2}}$ and $g_1^x \otimes d_2\otimes g_3^x\otimes g_4 \ket{\Psi_{00c_2}}$ and $g_1^x \otimes g_2\otimes \one\otimes g_4 \ket{\Psi_{00c_2}}$  and $g_1^x \otimes g_2\otimes g_3^x\otimes d_4 \ket{\Psi_{00c_2}}$.
\end{lemma}
\begin{proof}
The only states which can be reached from a convertible state are given in Lemma \ref{lem5}.
Note that in the proof of Lemma \ref{lem5} we have already shown that only symmetries with $|z|=|\tilde{z}|=1$ can be used for transformations. Thus, the symmetries that allow for transformations are unitary and commute with a diagonal matrix. Combining these results and using Eq. (\ref{EqSep}) it is easy to show that the states with $h_i=d_i$ for some $i\in\{1,2,3,4\}$ can only be obtained from states with $g_i=d_i$.
Thus, the only possibly convertible states are given by $\one \otimes h_2\otimes h_3^x\otimes h_4 \ket{\Psi_{00c_2}}$, $h_1^x \otimes d_2\otimes h_3^x\otimes h_4 \ket{\Psi_{00c_2}}$, $h_1^x \otimes h_2\otimes \one\otimes h_4 \ket{\Psi_{00c_2}}$, and $h_1^x \otimes h_2\otimes h_3^x\otimes d_4 \ket{\Psi_{00c_2}}$.
It is easy to see that these states are indeed convertible using the symmetries $\one^{\otimes 4}$ and $\sigma_z\otimes\sigma_z\otimes \one\otimes\one$ or $\one\otimes\one\otimes\sigma_z\otimes\sigma_z $. Consider for example the states $\one \otimes g_2\otimes g_3^x\otimes g_4 \ket{\Psi_{00c_2}}$. In order to perform a LOCC transformation use $\{\sqrt{p} (\one\otimes h_2 g_2^{-1}\otimes \one\otimes\one), \sqrt{1-p} (\sigma_z\otimes h_2\sigma_zg_2^{-1}\otimes \one\otimes\one)\}$ which constitutes a valid POVM for $h_2^3=g_2^3$, $(2 p -1 ) h_2^2=g_2^2$ and $h_2^1=g_2^1$. As always one can find for any $G_2$ a value of  $p$ and a non-LU equivalent $H_2$ s. t. the POVM condition is fulfilled.
\end{proof}
Thus, in this SLOCC classes the states in $MES_4$ that are convertible are given by  $\one \otimes d_2\otimes g_3^x\otimes g_4 \ket{\Psi_{00c_2}}$  with $g_3^x,g_4$ nondiagonal, $g_1^x \otimes g_2\otimes \one\otimes d_4 \ket{\Psi_{00c_2}}$ with $g_1^x,g_2$ nondiagonal, and $\one \otimes d_2\otimes \one \otimes d_4 \ket{\Psi_{00c_2}}$.

\subsection{The SLOCC classes $L_{abc_2}$ for $c= 0$}
In this case the seed states are given by \bea\label{ab02}\ket{\Psi}&= a\ket{\Phi^+}\ket{\Phi^+}+b\ket{\Phi^-}\ket{\Phi^-}+\ket{0110}.\eea
The corresponding symmetric $Z$ and $\tilde{Z}$ matrices coincide and are diagonal (with eigenvalues $a^2$, $b^2$, $0$ and $0$), i.e.
\begin{eqnarray}
&Z=\tilde{Z}= \left(
    \begin{array}{cccc}
      a^2 & 0 & 0&0 \\
      0 & b^2  & 0&0\\
 0 & 0  & 0&0\\
0 & 0  & 0&0\\
    \end{array}
  \right).
\end{eqnarray}
 One observes that the corresponding seed states have different symmetries depending on whether the two eigenvalues are the same or not, as well as whether they are non-zero or not. Note that the case $a=b=0$ leads to a fully separable state. Furthermore, it can be shown that the choice of parameters $a=0$ and $b\neq 0$ corresponds to the SLOCC classes $L_{abc_2}$ for $a=b=c$ if one exchanges party 2 and 3. Thus, it remains to consider the following cases, for which also the reachable states are different:
\bi
\item[(1)] $a^2\neq  b^2$ and $a,b \neq 0$
\item[(2)] $a^2= b^2\neq 0$
\ei
\subsubsection{Case (1): $a^2\neq  b^2$ and $a,b \neq 0$}
In case $a^2\neq  b^2$ and $a,b \neq 0$ the symmetry of the state in Eq. (\ref{ab02}) is given by
\bea \nonumber
S_z=P_z\otimes P_{1/z}\otimes P_z\otimes P_{1/z},
\eea
where $z\in\C\backslash 0$.
In order to obtain the symmetry $P_z^{\otimes 4}$ that has already been  discussed, we will consider the states
\bea\ket{\Psi_{ab00}}= \one\otimes\sigma_x\otimes\one\otimes \sigma_x\ket{\Psi}\eea
as seed states.
Note that the case $a^2=-b^2$ could in principle lead to more symmetries, as then $q$ in Eq. (\ref{Eq:Zs}) can also be $-1$ (similar for $\tilde{q}$). However, it is straightforward to see that the corresponding seed states do not have more symmetries.\\
As for the other SLOCC classes containing a seed state with this symmetry the standard form is given by $g_1^x\otimes g_2\otimes g_3\otimes g_4 \ket{\Psi_{ab00}}$, where $g_1^2>0$. In case $g_1^2=0$ the symmetry allows to choose $g_2^2\geq 0$ and analogously if $g_1^2=g_2^2=0$ one can choose $g_3^2\geq	0$ etc..
We can apply Lemma \ref{lem33} and obtain that the reachable states have the form $h_1^x\otimes d_2\otimes d_3\otimes d_4 \ket{\Psi_{ab00}}$ where  $h_1^x\neq \one$ or  $\one_1 \otimes d_j\otimes d_k\otimes h_l \ket{\Psi_{ab00}}$ where $\{j,k,l\} = \{2,3,4\}$ and $h_l\neq d_l$. Furthermore, the only states that are convertible via LOCC are given by $g_1^x\otimes d_2\otimes d_3\otimes d_4 \ket{\Psi_{ab00}}$ or $\one_1 \otimes d_j\otimes d_k\otimes g_l \ket{\Psi_{ab00}}$ where $\{j,k,l\} = \{2,3,4\}$ (see Lemma \ref{lem333}).
\subsubsection{Case (2): $a^2= b^2$ and $a\neq 0$}
Let us discuss now the case $a^2= b^2$ and $a\neq 0$. We could consider as seed states the states
\bea\label{aa02}\ket{\Psi}&= a\ket{\Phi^+}\ket{\Phi^+}+a\ket{\Phi^-}\ket{\Phi^-}+\ket{0110}.\eea
Note that the choice $a=-b$ in Eq. (\ref{ab02}) represents the same SLOCC classes as $a=b$ up to particle exchange, as applying $\one\otimes\one\otimes \sigma_x\otimes \sigma_x$ and exchanging party 3 and 4 leads to Eq. (\ref{aa02}).
In order to obtain similar symmetries as before we will consider the states
\bea\ket{\Psi_{aa00}}=\one\otimes\one\otimes\sigma_x\otimes\sigma_x\ket{\Psi}\eea
as seed states,
whose symmetries are given by
\begin{eqnarray}\label{eqsymyz}
&S_{z,\tilde{z}}= P_{\tilde{z}} \otimes P_{z}\otimes
    P_{z}\otimes P_{\tilde{z}},
\end{eqnarray}
where $z,\tilde{z}\in\C\backslash 0$.
Note that up to permutation of the parties these symmetries correspond to the ones discussed in Sec. \ref{secab0}. Thus, the results obtained there can be extended to these SLOCC classes (taking into account the permutation of the parties).

\section{The SLOCC classes $L_{a_2b_2}$ }
\subsection{The SLOCC classes $L_{a_2b_2}$ for $a,b\neq0$}
The seed states of these classes are given by \begin{eqnarray}\label{a2b2}\ket{\Psi_{a_2b_2}}&= a(\ket{0011}+\ket{1100})+b(\ket{0110}+\ket{1001})+i/(2a) \ket{1111}-i/(2b) \ket{1010}.\end{eqnarray} The corresponding symmetric $Z_\Psi$ and $\tilde{Z}_\Psi$ matrices coincide and have symmetric Jordan form with two Jordan blocks of dimension 2 (with eigenvalues $a^2$ and $b^2$), i.e.
\begin{eqnarray}
&Z_{\Psi_{a_2b_2}}=\tilde{Z}_{\Psi_{a_2b_2}}= \left(
    \begin{array}{cccc}
      a^2+\frac{i}{2} & \frac{1}{2} & 0&0 \\
      \frac{1}{2} & a^2-\frac{i}{2}  & 0&0\\
 0 & 0  & b^2+\frac{i}{2}&\frac{1}{2}\\
0 & 0  & \frac{1}{2}&b^2-\frac{i}{2}\\
    \end{array}
  \right).
\end{eqnarray}

Again one observes that there are different symmetries depending on whether
\bi
\item[(1)] $a^2\neq \pm b^2$
\item[(2)] $a^2= -b^2$ or
\item[(3)] $a^2=b^2$.
\ei
As we will show the structure of the symmetries in these different SLOCC classes is very different, which is also reflected in a different structure of the reachable states.
\subsubsection{Case (1): $a^2\neq \pm b^2$}
In case $a^2\neq \pm b^2$ the symmetry of the seed states is given by \bea S=\{\one^{\otimes 4},\sigma_z^{\otimes 4}\}.\eea Thus, Lemma \ref{lem0} can be applied and the only states (in standard form) that are reachable via LOCC are  given by \bea h_i\otimes d_j\otimes d_k\otimes d_l \ket{\Psi_{a_2b_2}},\eea where $\{i,j,k,l\} = \{1,2,3,4\}$ and $h_i\neq d_i$. Furthermore, according to Lemma \ref{lemc11} the only convertible states are of the form \bea g_i\otimes d_j\otimes d_k\otimes d_l \ket{\Psi_{a_2b_2}},\eea where $\{i,j,k,l\} = \{1,2,3,4\}$ with $g_i$ arbitrary.
\subsubsection{Case (2): $a^2= - b^2$ and $a,b\neq0$}
We will consider next the SLOCC classes $L_{a_2b_2}$ for $a^2= - b^2$ and $a,b\neq 0$. The seed state for $a=-ib$ is LU equivalent to the state for $a=ib$ via the LU operation $\one\otimes\one\otimes\sigma_z\otimes\sigma_z$ (see Eq. (\ref{a2b2})). Thus, they represent the same SLOCC class. Therefore, we could choose the seed states of the form \begin{eqnarray}
\ket\Psi&= ib(\ket{0011}+\ket{1100})+b(\ket{0110}+\ket{1001})+1/(2b) \ket{1111}-i/(2b) \ket{1010}.\end{eqnarray}
They have the symmetries $ S=\{\one^{\otimes 4},\sigma_z^{\otimes 4}, \one\otimes\sigma_x\otimes\sigma_z\otimes\sigma_y, \sigma_z\otimes\sigma_y\otimes\one\otimes\sigma_x\}$. In order to obtain a simpler form of the symmetry, we will in the following consider the LU-equivalent states \bea\ket{\Psi_{(-b_2)b_2}}=\one\otimes\one\otimes\one\otimes Z(\pi/4)\ket\Psi\eea as seed states. For these states the symmetries are
\bea \label{SPauli}\nonumber S=\{\one^{\otimes 4},\sigma_z^{\otimes 4}, \one\otimes\sigma_x\otimes\sigma_z\otimes\sigma_x, \sigma_z\otimes\sigma_y\otimes\one\otimes\sigma_y\}.\eea
In the following we will denote by $S_k=\otimes_i s_i^k$ the elements of the symmetry group in the order given in Eq. (\ref{SPauli}).
As before we use the notation $G_i=g_i^\dagger g_i=1/2\one+\sum_j \bar{g}_i^j \sigma_j$ where $\bar{g}_i^j\in\R$ is arbitrary apart from the constraint $G_i>0$ and similar for $H_i$. \\
In order to define the standard form we proceed as follows.
\begin{itemize}
\item We choose $\bar{g}_2^2>0$ (for $\bar{g}_2^2=0$ one can choose $\bar{g}_4^2>0$) and $\bar{g}_2^1>0$ (for $\bar{g}_2^1=0$ we choose $\bar{g}_4^1>0$).
\item If $\bar{g}_2^2=\bar{g}_4^2=0$ we choose $\bar{g}_3^2> 0$ and if $\bar{g}_2^2=\bar{g}_4^2=\bar{g}_3^2=0$ we choose $\bar{g}_3^1\geq 0$.
\item In case $\bar{g}_2^1=\bar{g}_4^1=0$ we choose $\bar{g}_1^1 > 0$ and for $\bar{g}_2^1=\bar{g}_4^1=\bar{g}_1^1=0$ we choose $\bar{g}_1^2\geq 0$.
\item In case $\bar{g}_2^2=\bar{g}_4^2=\bar{g}_3^2=\bar{g}_3^1= 0$ and $\bar{g}_2^1=\bar{g}_4^1=\bar{g}_1^1=\bar{g}_1^2= 0$ we choose $\bar{g}_2^3>0$ (if it additionally holds that $\bar{g}_2^3=0$ we choose $\bar{g}_4^3\geq 0$).
\end{itemize}
As can be easily seen this form can always be made unique by choosing some entry positive.
As will be stated in the following lemma these symmetries allow for a very different structure of the reachable states compared to the previous classes. 
\begin{lemma} \label{lem10} The only states in the SLOCC classes $L_{a_2b_2}$ with $a^2= - b^2$ and $a,b\neq0$ that are reachable via LOCC are given by
$h_1\otimes h_2\otimes h_3\otimes h_4\ket{\Psi_{(-b_2)b_2}}$, where $\otimes H_i$ obeys the following condition. There exists a symmetry $S_k=\otimes_i s_i^k\in S(\Psi_{(-b_2)b_2}) $ for some $k\in\{1, 2, 3\}$ such that exactly
three out of the four operators $H_i$ commute with the corresponding operators $s_i^k$ and one operator $H_j$ does not commute with $s_j^k$.
\end{lemma}

The basic idea of the proof is to use the fact that the symmetries are all elements of the Pauli group, which allows us to identify the reachable states using Eq. (\ref{eq_Pauli}) (for the details see Appendix E). In the following lemma we give all states which are convertible in this SLOCC classes.
\begin{lemma}\label{lem10c} The only states  in the SLOCC classes $L_{a_2b_2}$ with $a^2= - b^2$ and $a,b\neq0$ that are convertible via LOCC are given by
$g_1\otimes g_2\otimes g_3\otimes g_4\ket{\Psi_{(-b_2)b_2}}$, where $\otimes G_i$ obeys the following condition. There exists a non-trivial symmetry $S_k=\otimes_i s_i^k\in S(\Psi_{(-b_2)b_2}) $ for some $k\in\{1, 2, 3\}$ such that three out of the four operators $G_i$ commute with the corresponding operators $s_i^k$ and one operator $G_j$ is arbitrary.
\end{lemma}
\begin{proof}
Recall that according to Lemma \ref{lem10} the reachable states are given by
$h_1\otimes h_2\otimes h_3\otimes h_4\ket{\Psi_{(-b_2)b_2}}$, where $\otimes H_i$ has the property that there exists a symmetry $S_k=\otimes_i s_i^k \in S(\Psi_{(-b_2)b_2})$ such that exactly
three out of the four operators $H_i$ commute with the corresponding operators $s_i^k$.
Tracing over all parties except one in  Eq. (\ref{EqSep}) and using Observation \ref{obsmon} (in particular $r=1$) one obtains [see Eq. (\ref{eq_symLUsingle})]
\bea\label{Eqsep1new}
{\cal E}_i(H_i)=G_i \quad \forall i.
\eea
Thus, in case a component $\bar{h}_i^j$ is zero, it follows that also $\bar{g}_i^j=0$. Combining these results one obtains that the states $g_1\otimes g_2\otimes g_3\otimes g_4\ket{\Psi_{(-b_2)b_2}}$ that are possibly convertible have the property that there exists a symmetry $S_k=\otimes_i s_i^k \in S(\Psi_{(-b_2)b_2})$ such that at least
three out of the four operators $G_i$ commute with the corresponding operators $s_i^k$.  Denote by $j$ the party for which $[G_j, s_j^k]\neq 0$ in case there exist a $G_j$ that does not commute with the symmetry, otherwise $j$ can be chosen arbitrarily under the constraint that $s_j^k\neq\one$. Moreover, we use $s_j^k=\sigma_l$ for some $l\in\{1,2,3\}$. In order to show that these states are indeed convertible consider $\{\sqrt{p} h_j\one^{\otimes 4}g_j^{-1},\sqrt{1-p} h_jS_kg_j^{-1}\}$ which corresponds to a POVM for $\bar{h}_j^l=\bar{g}_j^l$ and $(2 p -1 ) \bar{h}_j^i=\bar{g}_j^i$ for $i\neq l$ and $i\in\{1,2,3\}$. As one can always find a value of $p$ and a positive $H_j$ that does not lead to a LU-equivalent state, these states are convertible.
\end{proof}
Due to Lemma \ref{lem10} and Lemma \ref{lem10c} we have that the states $g_1\otimes g_2\otimes g_3\otimes g_4\ket{\Psi_{(-b_2)b_2}}$ that are convertible and are in $MES_4$ obey the following condition. There exists a symmetry $S_k\in S(\Psi_{(-b_2)b_2}) $ for some $k\in\{1, 2, 3\}$ such that $\otimes G_i$ commutes with $S_k$  and there exist no symmetry $S_l=\otimes_i s_i^l\in S(\Psi_{(-b_2)b_2}) $ for some $l\in\{1, 2, 3\}$ such that exactly
three out of the four operators $G_i$ commute with the corresponding operators $s_i^l$ and one operator $G_j$ does not commute with $s_j^l$.

\subsubsection{\label{seca2a2} Case (3) $a^2= b^2$ and $a,b\neq0$}
The seed states for the SLOCC classes $L_{a_2b_2}$ with $a^2= b^2$ and $a,b\neq0$ can be chosen as \footnote{Note that the seed state for $a=-b$ is LU equivalent to the state for $a=b$ via the local unitary $\one\otimes\one\otimes\sigma_z\otimes\sigma_z$ (see Eq. (\ref{a2b2})).} $\ket{\Psi}= a(\ket{0011}+\ket{1100}+
\ket{0110}+\ket{1001})+i/(2a) (\ket{1111}- \ket{1010})$ with symmetries
\begin{eqnarray}
&S_{x,u,v,m}=(\sigma_z^{\otimes 4})^{m} \left(
    \begin{array}{cc}
      1 & 0 \\
      x & 1 \\
    \end{array}
  \right)\otimes\left(
    \begin{array}{cc}
      u & v \\
      -v & u\\
    \end{array}
  \right)\otimes\left(
    \begin{array}{cc}
      1 & 0\\
      -x & 1 \\
    \end{array}
  \right)
\otimes\left(
    \begin{array}{cc}
     u & -v \\
     v & u \\
    \end{array}\right),
\end{eqnarray}
where $m\in\{0,1\}$ and $x,u,v\in\C$. In order to work with symmetries that are similar to the ones that already occured in the context of other SLOCC classes we will in the following work with the seed states \bea \ket{\Psi_{a_2a_2}}=\one\otimes X(-\pi/4)\otimes\one\otimes X(-\pi/4)\ket\Psi,\eea which then have symmetries
\begin{eqnarray}
\{\tilde{S}_{x,z,m}\}&=\langle(\sigma_z\otimes\sigma_y\otimes\sigma_z\otimes\sigma_y)^{m}, \left(
    \begin{array}{cc}
      1 & 0 \\
      x & 1 \\
    \end{array}
  \right)\otimes P_{1/z}\otimes\left(
    \begin{array}{cc}
      1 & 0\\
      -x & 1 \\
    \end{array}
  \right)
\otimes P_{z}\rangle,
\end{eqnarray}
where $m\in\{0,1\}$, $x,z\in\C$ and $z\neq 0$. Note that any operator in the symmetry group can be written as
\begin{eqnarray}\nonumber
(\sigma_z\otimes\sigma_y\otimes\sigma_z\otimes\sigma_y)^{m}\left[\left(
    \begin{array}{cc}
      1 & 0 \\
      x & 1 \\
    \end{array}
  \right)\otimes P_{1/z}\otimes\left(
    \begin{array}{cc}
      1 & 0\\
      -x & 1 \\
    \end{array}
  \right)
\otimes P_{z}\right]
\end{eqnarray}
for some $m\in\{0,1\}$, $x,z\in\C$ and $z\neq 0$.
Using this symmetry the standard form can be chosen to be $d_1 \otimes g_2^x\otimes g_3\otimes g_4 \ket{\Psi_{a_2a_2}}$ where $g_2^2\geq 0$. In case $g_2^2=0$ one can choose $g_4^2\geq 0$.  Moreover, we choose $\Re(g_3^2)\geq 0$ and in case $\Re(g_3^2)= 0$ we choose $\Im(g_3^2)\geq 0$. For $g_3^2=0$ one has to impose some restriction $G_4$ in order to make the standard form unique. In the following lemma we state which states in this SLOCC class can be reached via LOCC.
\begin{lemma} \label{lem77} The only states in the SLOCC classes $L_{a_2b_2}$ with $a^2= b^2$ and $a,b\neq0$ that are reachable via LOCC are given by $d_1 \otimes h_2^x\otimes h_3\otimes h_4^x \ket{\Psi_{a_2a_2}}$ with $h_3\neq d_3$, $d_1 \otimes \one\otimes h_3\otimes h_4 \ket{\Psi_{a_2a_2}}$ with $h_4\neq d_4$, $d_1 \otimes h_2^x\otimes h_3\otimes d_4 \ket{\Psi_{a_2a_2}}$ with $h_2^x\neq \one/2$, and $d_1 \otimes h_2^x\otimes d_3\otimes h_4 \ket{\Psi_{a_2a_2}}$ with $h_4\neq h_4^x$ and $d_1 \otimes \one /2\otimes d_3\otimes h_4 \ket{\Psi_{a_2a_2}}$ where $h_4\neq \one/2$.
\end{lemma}
The proof is straightforward but tedious and can be found in Appendix F.

In the following lemma we state which states in these SLOCC classes are convertible.
\begin{lemma}  \label{lemmab2} The only states  in the SLOCC classes $L_{a_2b_2}$ with $a^2=  b^2$ and $a,b\neq0$ that are convertible via LOCC are given by $d_1 \otimes g_2^x\otimes g_3\otimes g_4^x \ket{\Psi_{a_2a_2}}$, $d_1 \otimes \one\otimes g_3\otimes g_4 \ket{\Psi_{a_2a_2}}$, $d_1 \otimes g_2^x\otimes g_3\otimes d_4 \ket{\Psi_{a_2a_2}}$, and $d_1 \otimes g_2^x\otimes d_3\otimes g_4 \ket{\Psi_{a_2a_2}}$.
\end{lemma}
The proof of this lemma can be found in Appendix F. Due to Lemma \ref{lem77} and Lemma \ref{lemmab2} the states in $MES_4$ that are convertible are given by  $d_1 \otimes g_2^x\otimes d_3\otimes g_4^x \ket{\Psi_{a_2a_2}}$, where $g_2^2,g_4^2\neq 0$, $d_1 \otimes \one\otimes g_3\otimes d_4 \ket{\Psi_{a_2a_2}}$, where $g_3\neq d_3$ and $d_4\neq\one$ and $d_1 \otimes \one\otimes d_3\otimes \one \ket{\Psi_{a_2a_2}}$.
\subsection{The SLOCC classes $L_{a_2b_2}$ for $a=0$ and $b\neq0$}
We will consider now the SLOCC classes $L_{a_2b_2}$ for $a=0$ and $b\neq0$. Note here that the block structure of $Z_\Psi$ and $\tilde{Z}_\Psi$ changes compared to the one for $a,b\neq 0$. In particular one obtains two Jordan blocks of dimension 1 (with eigenvalues 0) and one Jordan block of dimension 2 (with eigenvalue $b^2$), i.e.
\begin{eqnarray}
&Z_\Psi=\tilde{Z}_\Psi= \left(
    \begin{array}{cccc}
      0 & 0 & 0&0 \\
      0 & 0  & 0&0\\
 0 & 0  & b^2+\frac{i}{2}&\frac{1}{2}\\
0 & 0  & \frac{1}{2}&b^2-\frac{i}{2}\\
    \end{array}
  \right).
\end{eqnarray}
In \cite{slocc4} the representative for the classes $L_{a_2b_2}$ is given by $\ket\Psi= a(\ket{0000}+\ket{1111})+b(\ket{0101}+\ket{1010})+\ket{0110}+\ket{0011}$. Thus, the choice $b=0$ and $a\neq 0$ is LU-equivalent to $a=0$ and $b\neq 0$ via $\one\otimes \sigma_x\otimes \one\otimes \sigma_x$. The seed state of these SLOCC classes could be chosen as  $\ket\Psi= b(\ket{0110}+\ket{1001})+ \ket{1111}-i/(2b) \ket{1010}$ with the only non-trivial symmetries $P_z\otimes P_z\otimes P_{1/z}\otimes P_{1/z}$ for $z\in\C\backslash 0$. Hence, the seed states \bea\ket{\Psi_{00b_2}}=\one\otimes\one\otimes\sigma_x\otimes\sigma_x\ket{\Psi}\eea have symmetries $P_z^{\otimes 4}$. Due to Lemma \ref{lem33} we have that the states that are in $MES_4$ are all those that cannot be written as $h_1^x\otimes d_2\otimes d_3\otimes d_4 \ket{\Psi_{00b_2}}$ where  $h_1^x\neq \one$ or  $\one_1 \otimes d_j\otimes d_k\otimes h_l \ket{\Psi_{00b_2}}$ with $\{j,k,l\} = \{2,3,4\}$ and $h_l\neq d_l$. Furthermore, according to Lemma \ref{lem333} the states that are convertible can be written in the form $g_1^x\otimes d_2\otimes d_3\otimes d_4 \ket{\Psi_{00b_2}}$ or $\one_1 \otimes d_j\otimes d_k\otimes g_l \ket{\Psi_{00b_2}}$ with $\{j,k,l\} = \{2,3,4\}$.
\subsection{The SLOCC class $L_{a_2b_2}$ for $a=b=0$}
A representative of the class $L_{a_2b_2}$ for $a=b=0$ is given by $\ket\Psi= \ket{0110}+\ket{0011}$, i.e. the state is separable and its possible transformations are known \cite{Nielsen}.
\section{The SLOCC classes $L_{ba_3}$}
\subsection{\label{secLba3aneq0}The SLOCC classes $L_{ba_3}$ for $a\neq 0$}
Let us consider next the classes $L_{ba_3}$ with $a\neq 0$. Note that the corresponding seed state can be chosen as \footnote{Notice that the seed state is in the same SLOCC class as the state $\ket{\phi}=a(\ket{0000}+\ket{1111})+ a\ket{\Psi^+}\ket{\Psi^+}+b \ket{\Psi^-}\ket{\Psi^-}+ i/\sqrt{2}(\ket{0111}+\ket{1011}-\ket{0001}-\ket{0010})$, which differs up to some signs from the representative for this SLOCC class presented in  \cite{slocc4}.}

\begin{eqnarray}\label{ab3}
&\ket{\Psi_{ab_3}} =1/(32 a^3)[(-i + (8\,i - 8) a^2)\ket{1111} -(i + (8\,i - 8 ) a^2)\ket{0000} -(1 + (4 + 4 i) a^2)\sqrt{2}(\ket{00}\ket{\Psi^-}\\ \nonumber
&+\ket{\Psi^-}\ket{00})-(i-32a^4) (\ket{0011}+\ket{1100})+(i+16 a^3(a - b))
(\ket{0101}+\ket{1010})-(i + 16 a^3 (a + b))\\ \nonumber &
(\ket{0110}+\ket{1001})-(1 - (4 + 4 i) a^2)
\sqrt{2}(\ket{\Psi^-}\ket{11}+\ket{11}\ket{\Psi^-})]
\end{eqnarray}
The corresponding $Z$ and $\tilde{Z}$ matrices coincide and have symmetric Jordan form.
In particular, they are of the form
\begin{eqnarray}
&Z_{\Psi_{ab_3}}=\tilde{Z}_{\Psi_{ab_3}}= \left(
    \begin{array}{cccc}
      a^2 & \frac{1}{2} +\frac{i}{2} & 0&0 \\
     \frac{1}{2} +\frac{i}{2} & a^2  & \frac{1}{2} -\frac{i}{2}&0\\
 0 &\frac{1}{2} -\frac{i}{2} & a^2&0\\
0 & 0  & 0&b^2\\
    \end{array}
  \right).
\end{eqnarray}
As before one obtains completely different symmetries for $a^2\neq b^2$ and $a^2= b^2$ which we will distinguish in the following. Nevertheless, one finds that no state is reachable in both cases.
\subsubsection{The case $a^2\neq b^2$ and $a\neq 0$}
For the classes $L_{ba_3}$ with $a^2\neq b^2$ and $a\neq 0$ the corresponding seed state has no non-trivial symmetry. Thus, no transformations via SEP is possible and, therefore, all states in these SLOCC classes are in $MES_4$.
\subsubsection{\label{secaa3} The case $a^2= b^2$ and $a\neq 0$}
Let us consider next the classes $L_{ba_3}$ with $a= b$ and $a\neq 0$. Note here that the parameter choice $a=b$ and $a=-b$ for the seed states leads to the same SLOCC classes up to exchange of particles. In order to simplify the symmetries we will not consider the state $\ket{\Psi_{ab_3}}$ [see Eq. (\ref{ab3})] as seed state but rather \bea\ket{\Psi_{aa_3}}=X(-\pi/4)\otimes X(\pi/4)\otimes X(-\pi/4)\otimes X(\pi/4)\ket{\tilde{\Psi}_{aa_3}},\eea where we set on the righthand side of Eq. (\ref{ab3}) $a=b$.
The symmetries of these states are given by
\begin{eqnarray}\label{eqsx}
S_{x}&=\left(
    \begin{array}{cc}
      1 & x \\
      0 & 1 \\
    \end{array}
  \right)^{\otimes 4}\equiv \tilde{s}_x^{\otimes 4},
\end{eqnarray}
where $x\in\C$.
It is easy to see that one can choose the standard form $d_1\otimes g_2\otimes g_3\otimes g_4 \ket{\Psi_{aa_3}}$ as for an appropriate value of $x$
the matrix $ \tilde{s}_x^\dagger G_i \tilde{s}_x$ is diagonal.
With all that, it is now straightforward to derive the states in $MES_4$ of these classes.
\begin{lemma} \label{lem07} The SLOCC classes $L_{ba_3}$ with $a^2= b^2$ and $a\neq 0$ contains no state that is reachable via LOCC.
\end{lemma}
\begin{proof}
Inserting the symmetries $S_x$ into Eq. (\ref{EqSep}) one obtains
\begin{eqnarray}\label{sepab3}&\sum_x p_{x} \otimes_i \left(
    \begin{array}{cccc}
       & h_i^1&x h_i^1+ h_i^2  \\
     &x^* h_i^1+ h_i^{2 *} &y\\
     \end{array}
  \right)=
r \otimes_iG_i,\end{eqnarray}
where $y=h_i^1|x|^2+(x^*h_i^2+x h_i^{2 *})+ h_i^3$. Note that $h_1^2=g_1^2=0$. Thus by considering the matrix element $\ket{0000}\bra{1000}$ and $\ket{0100}\bra{1000}$ of this equation one obtains $\sum_x p_{x} x=0$ and $\sum_x p_{x} |x|^2=0$. This can only be fulfilled if $p_x=0$ for $x\neq 0$.
Inserting this result in Eq.(\ref{sepab3}) and using the normalization condition $\tr(G_i)=\tr(H_i)=1$ (which leads to $r=1$) one obtains that $\otimes_iG_i=\otimes_iH_i$.
Therefore, no states in this SLOCC class are reachable via SEP.
\end{proof}
As we have shown all states in these SLOCC classes are in $MES_4$. Trivially it follows that there are also no convertible states in these SLOCC classes.
\subsection{The SLOCC classes $L_{ba_3}$ with $a= 0$}
Let us discuss now the classes $L_{ba_3}$ with $a= 0$ and $b$ arbitrary. The corresponding $Z_\Psi$ and $\tilde{Z}_\Psi$ matrices coincide and are in symmetric Jordan form with a Jordan block of dimension 2 (with eigenvalue 0) and two Jordan blocks of dimension 1 (with eigenvalues 0 and $b^2$), i.e.
\begin{eqnarray}
&Z_\Psi=\tilde{Z}_\Psi= \left(
    \begin{array}{cccc}
     0 & 0 & 0&0 \\
     0 & b^2  &  0&0\\
 0 &0 & \frac{i}{2}&\frac{1}{2}\\
0 & 0  & \frac{1}{2}&-\frac{i}{2}\\
    \end{array}
  \right).
\end{eqnarray}
One could choose as a seed states the states \bea\label{b02}\ket\Psi=-b\ket{\Phi^-}\ket{\Phi^-}-e^{i\pi/4}(\ket{10}\ket{\Phi^+}-\ket{\Phi^+}\ket{10}).\eea
The symmetries are different depending on whether $b$ vanishes or not. In particular, for the SLOCC classes for $b\neq 0$ no state is reachable, whereas for the case $b=0$ transformations are possible.
\subsubsection{\label{secLba3a0bneq0}The case $a= 0$ and $b\neq 0$}
The corresponding seed states could be chosen as in Eq. (\ref{b02}).
These states have symmetries
\begin{eqnarray}
S_{x}&=(\sigma_x\otimes\one\otimes\sigma_x\otimes\one) \tilde{s}_x^{\otimes 4}(\sigma_x\otimes\one\otimes\sigma_x\otimes\one) ,
\end{eqnarray}
where $\tilde{s}_x$ is defined as in Eq. (\ref{eqsx}). It follows that the states \bea\ket{\Psi_{0b0_2}}=(\sigma_x\otimes\one\otimes\sigma_x\otimes\one)\ket\Psi\eea have symmetries \bea S_x= \tilde{s}_x^{\otimes 4}.\eea Note that these symmetries correspond to the ones for the classes $L_{ba_3}$ with $a= b$ and $a\neq 0$ (see subsection \ref{secaa3}). Therefore, the results for these classes apply. In particular, the standard forms for these classes coincide and all states in the SLOCC classes $L_{ba_3}$ with $a= 0$ and $b\neq 0$ are in $MES_4$ and no state is convertible (see proof of Lemma \ref{lem07}).\\
\subsubsection{\label{secW}The case $a=b= 0$}
The SLOCC class $L_{ab_3}$ with $a=b= 0$ contains the $4$-particle W state. One could choose as seed state $\ket\Psi=e^{i\pi/4}(\ket{10}\ket{\Phi^+}-\ket{\Phi^+}\ket{10})$, which is LU--equivalent (with the LU transformation $\sigma_y\otimes\sigma_z\otimes\sigma_x\otimes\one$) to the $W$--state, $\ket{W}=\ket{0001}+\ket{0010}+\ket{0100}+\ket{1000}$. Hence, we will consider the state \bea \ket{\Psi_{000_2}}=\ket{W}\eea as seed state. The symmetry is given by
\begin{eqnarray}
&S_{w,x,y,z}= \frac{1}{z^2}\left(
    \begin{array}{cc}
      z & w \\
      0 & 1/z \\
    \end{array}
  \right)\otimes\left(
    \begin{array}{cc}
      z & x \\
      0 & 1/z \\
    \end{array}
  \right)\otimes\left(
    \begin{array}{cc}
      z & y\\
      0 & 1/z \\
    \end{array}
  \right)
\otimes\left(
    \begin{array}{cc}
      z & -y-w-x \\
      0 & 1/z \\
    \end{array}\right).
\end{eqnarray}\\
This symmetry allows to choose the standard form $d_1\otimes d_2\otimes g_3\otimes \one\ket{W}$, where $d_1=\mbox{diag}(1,x_1/x_4)$, $d_2= \mbox{diag}(1,x_2/x_4)$ and $g_3=[x_4,x_0;0,x_3]$, with $x_i\geq 0$\footnote{To ease the notation we consider here unnormalized local operators.}. That is, any state in this class is of the form
$x_0\ket{0000}+x_1\ket{1000}+x_2\ket{0100}+x_3\ket{0010}+x_4\ket{0001}$. As we will show now, the reachable states in the 3-qubit W-class (see e.g. \cite{MESus}) can be generalized to four qubits, as stated in the following lemma. Note that in contrast to \cite{MESus} we provide here a rigorous proof of the fact that only unitary symmetries can be used for transformations.
\begin{lemma}  The only states in the SLOCC classes  $L_{ab_3}$ with $a=b= 0$ that are reachable via LOCC are given by $d_1\otimes d_2\otimes g_3\otimes \one\ket{W}$ where $g_3\neq d_3$.
\end{lemma}

Note that this means that the only states which are reachable are of the form $x_0\ket{0000}+x_1\ket{1000}+x_2\ket{0100}+x_3\ket{0010}+x_4\ket{0001}$ with $x_0\neq 0$.
\begin{proof}
That the states presented in the lemma can be reached via LOCC can be easily seen as follows. Applying the POVM $\{1/\sqrt{2}\one^{\otimes 2} \otimes h_3 g_3^{-1}\otimes \one,1/\sqrt{2}\sigma_z^{\otimes 2}\otimes h_3 \sigma_zg_3^{-1}\otimes \sigma_z\}$, with $G_3 = diag[(x_4)^2,(x_0)^2 + (x_3)^2]$ to the state $d_1\otimes d_2\otimes g_3\otimes \one\ket{W}$ leads deterministically to the state $d_1\otimes d_2\otimes h_3\otimes \one\ket{W}=x_0\ket{0000}+x_1\ket{1000}+x_2\ket{0100}+x_3\ket{0010}+x_4\ket{0001}$ where $h_3=[x_4,x_0;0,x_3]$ with $x_0\neq 0$.

To see that non of the remaining states, i.e. states of the form $d_1\otimes d_2\otimes d_3\otimes \one\ket{W}$ are reachable, we consider Eq. (\ref{EqSep}) for $H_1, H_2, H_3,G_1,G_2$ diagonal, $H_4=G_4\propto \one$ and $G_3=g_3^\dagger g_3$ with $g_3=[x_4,x_0;0,x_3]$, i.e.

\bea \label{Eq_WSep} \sum p_{w,x,y,z} S_{w,x,y,z}^\dagger (H_1\otimes H_2\otimes H_3\otimes \one) S_{w,x,y,z}=r G_1\otimes G_2\otimes G_3 \otimes \one.\eea

Considering now the matrix elements $\ket{0000}\bra{0000}\sigma_x^{i}$, where $\sigma_x^{i}$ is acting on system $i$ for $i\in \{1,2,3,4\}$ of Eq. (\ref{Eq_WSep}), one easily obtains that $G_3$ has to be diagonal for the equation to be satisfied. Next, we consider the matrix elements $\sigma_x^j \ket{0000}\bra{0000}\sigma_x^{i}$ for $j,i\in \{1,2,3,4\}$, $i\neq j$ to obtain that the symmetries used in this equations, must all be diagonal, i.e. $p_{w,x,y,z}=0$ unless $w=x=y=0$. Combining now the diagonal matrix element $\proj{ijkl}$ with those of $\proj{ijkl^\prime}$ leads to $p_{w,x,y,z}=0$ for $|z|\neq 1$, which implies that there exist no non--trivial transformation and hence, the states are not reachable.

\end{proof}

Thus, one obtains that the states in this SLOCC class which are contained in $MES_4$ are given by $d_1\otimes d_2\otimes d_3\otimes \one\ket{W}=
x_1\ket{1000}+x_2\ket{0100}+x_3\ket{0010}+x_4\ket{0001}$.\\
As the symmetry contains $\sigma_z^{\otimes 4}$ any state ($d_1\otimes d_2\otimes g_3\otimes \one\ket{W}$) in this SLOCC class is convertible via the LOCC protocol $\{\sqrt{p}\one^{\otimes 2} \otimes h_3g_3^{-1}\otimes \one,\sqrt{1-p}\sigma_z^{\otimes 2}\otimes h_3\sigma_zg_3^{-1}\otimes \sigma_z\}$ for a appropriately chosen $h_3$ and $p$. Note that all these results can be straightforward generalized to an arbitrary number of qubits. Note further that the W--class and the GHZ--class are the only classes where no state is isolated.
\section{The SLOCC classes $L_{a_4}$}
\subsection{\label{secLa4aneq0}The SLOCC classes $L_{a_4}$ for $a\neq 0$}
Let us consider next the classes $L_{a_4}$ (see \cite{slocc4}) with $a\neq 0$. The  $Z$ and $\tilde{Z}$ matrices for states in this class can be brought to Jordan form with a single Jordan block of dimension 4 (with eigenvalues $a^2$). Since the states in these SLOCC classes have no non-trivial symmetry, all states are in $MES_4$ and no state is convertible.\\
\subsection{The SLOCC class $L_{a_4}$ for $a= 0$}
The seed state for the class $L_{a_4}$ with $a= 0$ could be chosen as \begin{eqnarray}
&\ket{\Psi} =i \ket{00}(\ket{10}+i\ket{11})
+i\ket{01}(-\ket{11}+i\ket{10})
+i\ket{10}(-i\ket{11}+\ket{10}+i\ket{01}+\ket{00})\\ \nonumber
&+i\ket{11}(-\ket{11}-i\ket{10}-\ket{01}+i\ket{00}).
\end{eqnarray}
Note that the corresponding $Z_\Psi$ and $\tilde{Z}_\Psi$ matrices have both two Jordan blocks of dimension 2 (with eigenvalues 0), i.e.
\begin{eqnarray}
&Z_\Psi=\tilde{Z}_\Psi= \left(
    \begin{array}{cccc}
     \frac{i}{2} & \frac{1}{2} & 0&0 \\
    \frac{1}{2}&-\frac{i}{2}  &  0&0\\
 0 &0 & \frac{i}{2}&\frac{1}{2}\\
0 & 0  & \frac{1}{2}&-\frac{i}{2}\\
    \end{array}
  \right).
\end{eqnarray}
In order to simplify the corresponding symmetries we will use \bea\ket{\Psi_{0_20_2}}=\one\otimes X(-\pi/4)\otimes \one \otimes X(-\pi/4)\ket\Psi\eea as seed state. Then the symmetries are given by \begin{eqnarray}
&S_{x,y,z}=\otimes_i s_i =\frac{1}{y^3 z^3}\left(
    \begin{array}{cc}
      1 & 0 \\
      x & y^2z^2 \\
    \end{array}
  \right)\otimes\left(
    \begin{array}{cc}
      z &0\\
      0 & 1/z \\
    \end{array}
  \right)\otimes\left(
    \begin{array}{cc}
      1 & 0\\
      -x & y^2z^2 \\
    \end{array}
  \right)
\otimes\left(
    \begin{array}{cc}
      y & 0 \\
      0 & 1/y \\
    \end{array}\right),
\end{eqnarray}
where $x,y,z\in \C$ and $y,z\neq 0$. This symmetry allows to choose the standard form as $\one\otimes g_2^x\otimes g_3\otimes g_4 \ket{\Psi_{0_20_2}}$ where $g_2^2\geq 0$ and $g_4^2\geq 0$ (for $g_4^2= 0$ and/or $g_2^2= 0$ choose $g_3^2\geq 0$). Here we used that $s_2^{\dagger} g_2^{\dagger}g_2s_2\propto G_2^x$ for a properly chosen $z$ and $s_1^{\dagger} g_1^{\dagger}g_1s_1\propto \one$ for a properly chosen $x$ and $|y|$. Furthermore it is easy to see that the phase of $y$ can be chosen such that the symmetry $S_{0,y,1}$ allows to choose $g_4^2\geq 0$ (or $g_3^2\geq 0$). In the following lemma we characterize the states which are in $MES_4$ for this SLOCC class.
\begin{lemma}  The only states in the SLOCC class  $L_{a_4}$ with $a=0$ that are reachable via LOCC are given by $\one\otimes h_2^x\otimes h_3\otimes h_4 \ket{\Psi_{0_20_2}}$ with $h_i^2=0$ for at least one $i\in\{2,3,4\}$ and $h_j^2\neq 0$ for at least one $j\in\{2,3,4\}$.
\end{lemma}
\begin{proof}
We split all states into the following three cases
\begin{itemize}
\item[(i)] $h_2^2,h_3^2,h_4^2\neq 0$
\item[(ii)] $h_2^2=h_3^2=h_4^2=0$ and
\item[(iii)]  $h_j^2=0$, $h_i^2\neq 0$ and $h_k^2$ arbitrary for $\{i,j,k\}=\{2,3,4\}$.
\end{itemize}

We will show now that states in case (i) and (ii) are not reachable, whereas we will present a LOCC protocol reaching states in class (iii).
Let us start with case (i). Considering the $\ket{1111}\bra{0010}$ matrix element of Eq. (\ref{EqSep}) leads in this case to $\sum p_{x,y,z} x=0$.
The matrix element $\ket{1000}\bra{0010}$ implies then that $p_{x,y,z}=0$ whenever $x\neq 0$. Hence, $s_1 \propto s_3$ for any symmetry which can be used for the transformation. Projecting then Eq. (\ref{EqSep}) onto $\ket{\Psi^-}_{13}$ leads to a simple equation for party $2$ and $4$. Considering in this resulting equation the matrix element $\ket{00}\bra{00}$ leads to $h_4^1=r g_4^1$. Using this equation together with the $\ket{1000}\bra{1000}$ matrix element of Eq. (\ref{EqSep}) one obtains that the diagonal elements of $H_3$ coincide with the diagonal elements of $G_3$. It is then easy to show that $|y|=1/|z|$ has to hold for any symmetry occurring in Eq. (\ref{EqSep}) by considering further diagonal elements of that equation. Hence $s_1$ is unitary, which implies (using that in standard form $H_1=\one$) that we only have to consider the equation for parties $2$ to $4$. More precisely we have,
\begin{eqnarray} \label{EqnewBK}
&\sum_{|z|,\phi,\tilde{\phi}} p_{|z|,\phi,\tilde{\phi}} \left(
    \begin{array}{cc}
   |z|^2/2 &   h_2^2e^{-i2\phi} \\
    h_2^{2}e^{i2\phi} & 1/(2|z|^2) \\
    \end{array}
  \right)\otimes\left(
    \begin{array}{cc}
      h_3^1&   h_3^2e^{i2(\tilde{\phi}+\phi)} \\
    h_3^{2   *}e^{-i2(\tilde{\phi}+\phi)} & h_3^3 \\
    \end{array}
  \right)\otimes\left(
    \begin{array}{cc}
      h_4^1/|z|^2 &   h_4^2e^{-i2\tilde{\phi}} \\
    h_4^{2}e^{i2\tilde{\phi}} & h_4^3|z|^2\\
    \end{array}\right)\\ \nonumber
&=r  G_2^x\otimes G_3\otimes G_4,
\end{eqnarray}
where we adapted the notation for the probabilities and where $z=|z|e^{i\phi}$ and $y=e^{i\tilde{\phi}}/|z|$.
Considering the matrix elements $\ket{ijk}\bra{ijk}$ of this equation for $i,j,k\in\{0,1\}$ one obtains that
\bea \label{Eqdxyz}h_3^{1+2j} h_4^{1+2k}  \sum_{|z|,\phi,\tilde{\phi}} p_{|z|,\phi,\tilde{\phi}} |z|^{4[k-i]} =r  g_3^{1+2j} g_4^{1+2k}.\eea Recall that the diagonal elements of $G_3$ coincide with those of $H_3$. From this equation for $i=j=k=0$ and for $i=1$ and $j=k=0$, as well as for  $i=k=1$ and $j=0$ and for $k=1$ and $i=j=0$ it follows that $\sum_{|z|,\phi,\tilde{\phi}} p_{|z|,\phi,\tilde{\phi}}  |z|^4=\sum_{|z|,\phi,\tilde{\phi}} p_{|z|,\phi,\tilde{\phi}}  1/|z|^4=1$ which can only be fulfilled if $p_{|z|,\phi,\tilde{\phi}}=0$ for $|z|\neq 1$. Thus, only the symmetries where $|z|= 1$ and therefore $|y|= 1$, i.e. unitary symmetries can be used for transformations. Thus, we can use Observation \ref{obsmon} and therefore $r=1$. Moreover, using Eq. (\ref{Eqdxyz}) and the normalization condition $\tr(G_i)=\tr(H_i)=1\,\,\,\forall i$ it is easy to show that $ h_4^1= g_4^1$ and  $ h_4^3= g_4^3$. Thus, the diagonal elements of both $G_3$ and $G_4$ can not be changed.
From the matrix element $\ket{111}\bra{000}$,$\ket{110}\bra{000}$ of Eq. (\ref{Eqdxyz}) and $\ket{000}\bra{001}$ one obtains
\begin{eqnarray}\nonumber
&h_2^2h_3^2h_4^2=g_2^2g_3^2g_4^2\\\nonumber&
 h_2^2h_3^2\sum_{|z|,\phi,\tilde{\phi}} p_{|z|,\phi,\tilde{\phi}} e^{i2\tilde{\phi}}=g_2^2g_3^2\end{eqnarray}
and
\bea \nonumber h_4^2\sum_{|z|,\phi,\tilde{\phi}} p_{|z|,\phi,\tilde{\phi}} e^{i2\tilde{\phi}}=g_4^2.\eea
Due to the fact that the sum in the last equation cannot be zero, as $g_4^2\neq 0$ (the left--hand side of the first equation cannot vanish), the last two equations lead to $ g_4^2/(g_2^2g_3^2)=h_4^2/(h_2^2h_3^2)$ and then it is easy to see that $h_4^2=g_4^2$ and $ g_2^2g_3^2=h_2^2h_3^2$. Note that one uses here that in standard form $h_4^2=|h_4^2|$ and $g_4^2=|g_4^2|$. Via a similar argument one obtains $h_2^2=g_2^2$ and so $h_3^2=g_3^2$. Hence, states with $h_2^2,h_3^2,h_4^2\neq 0$ can only be reached from states with $h_2^2=g_2^2$, $h_3^2=g_3^2$ and $h_4^2=g_4^2$ and are therefore in $MES_4$.

Let us next consider the case (ii). The fact that $x=0$ has to hold for any symmetry occurring in Eq. (\ref{EqSep}) follows from considering the matrix element $\ket{1000}\bra{0010}$ of that equation. That $|y|=1/|z|=1$ has to hold follows via the same argument as above. From Eq. (\ref{Eqdxyz}) it follows then immediately that the diagonal elements of $H_i$ coincide with those of $G_i$ for $i=3,4$. Considering the matrix element $\ket{100}_{i,j,k}\bra{000}$ where $i,j,k\in\{2,3,4\}$ of Eq. (\ref{EqnewBK}) leads to $g_i^2=0$. Thus, this class of states is also in $MES_4$. \\

We still need to show that for case (iii) the corresponding states are reachable. In order to see this note that $\one\otimes\sigma_z\otimes\one\otimes\sigma_z$, $\sigma_z\otimes\sigma_z\otimes\sigma_z\otimes\one$ and $\sigma_z\otimes\one\otimes\sigma_z\otimes\sigma_z$ are symmetries of the seed state. Thus, for any $\{i,j,k\}=\{2,3,4\}$ there exists a symmetry that is acting on parties $i$, $j$ and $k$ via $Z_i\otimes Z_j\otimes \one_k$. In the following we will denote this symmetry by $s^{ij}$. Note that the operator acting on party $1$ is determined by $s_1\otimes s^{ij}$ being a symmetry. In order to reach a state with $h_j^2=0$, $h_i^2\neq 0$ and $h_k^2$ arbitrary one can use  $\{1/\sqrt{2}h_i\one^{\otimes 4} d_i^{-1},1/\sqrt{2}h_i s^{ij} d_i^{-1}\}$, where $d_i$ is chosen such that it corresponds to a valid POVM. Note that since $h_j^2=0$ one can easily apply $\sigma_z$ on party $1$ and/or $j$ in case of the second outcome.
\end{proof}
Note that again any LOCC protocol that we constructed exploits only the symmetries that are an element of the Pauli group. The states in this SLOCC class that are in $MES_4$ are given by $\one\otimes \one\otimes d_3\otimes d_4 \ket{\Psi_{0_20_2}}$ and $\one\otimes h_2^x\otimes h_3\otimes h_4 \ket{\Psi_{0_20_2}}$ where $h_2^x,h_3$ and $h_4$ are not diagonal.
In the following lemma we show which states in this SLOCC class are convertible.
\begin{lemma}  The only states in the SLOCC class  $L_{a_4}$ with $a=0$ that are convertible are given by $\one\otimes g_2^x\otimes g_3\otimes g_4 \ket{\Psi_{0_20_2}}$ where at least one $g_j^2=0$ for $j\in\{2,3,4\}$.
\end{lemma}
\begin{proof}
Note first that the only states that are reachable obey that at least one $h_j^2=0$ for $j\in\{2,3,4\}$. Via Eq. (\ref{EqSep}) one obtains that this is also a necessary condition for all states which might be convertible to them. More precisely for $j=2,4$ it only requires to consider one particular matrix element of Eq. (\ref{EqSep}) to show that also $G_j$ has to be diagonal. For $j=3$, one first shows that $x=0$ using the same argument as in the proof above [case (ii)] and then proceeds as before. In order to see that indeed any state obeying this constraint is convertible consider the POVM $\{\sqrt{p}h_i\one^{\otimes 4} g_i^{-1},\sqrt{1-p}h_i s^{ij} g_i^{-1}\}$. Here, $s^{ij}$ is acting as $Z_i\otimes Z_j\otimes \one_k$ on parties $i,j,k\in\{2,3,4\}$ and the operator acting on the first qubit is defined by $s^{ij}$ being a symmetry of the seed state. One obtains that this is a valid POVM for $h_i^1=g_i^1$, $h_i^3=g_i^3$ and $(2p-1)h_i^2=g_i^2$.  As the only constraints on $G_i$ and $H_i$ is that they are positive operators of rank 2, for any $G_i$ one can find a $H_i\neq G_i$ and a value for $p$ such that these conditions are fulfilled.
\end{proof}
Thus, the non-isolated states in $MES_4$ in this SLOCC class are given by \bea\one\otimes \one\otimes d_3\otimes d_4 \ket{\Psi_{0_20_2}}.\eea
\section{The SLOCC classes $L_{a_20_{3\oplus1}}$}
Let us treat next the SLOCC classes $L_{a_20_{3\oplus1}}$ (see \cite{slocc4}). As for $a=0$ the corresponding states are biseperable we will not consider them here.
We chose the seed states as \bea\ket{\Psi_{a_20_{3\oplus1}}}=a(\ket{1001}+\ket{0110})-i/(2a)\ket{1010}+e^{i\pi/4}\ket{11}\ket{\Phi^+}.\eea
The corresponding $Z$ and $\tilde{Z}$ matrices are in symmetric Jordan form and are given by
\begin{eqnarray}
&Z_\Psi= \left(
    \begin{array}{cccc}
     \frac{i}{2} & \frac{1}{2} & 0&0 \\
    \frac{1}{2}&-\frac{i}{2}  &  0&0\\
 0 &0 & a^2+\frac{i}{2}&\frac{1}{2}\\
0 & 0  & \frac{1}{2}&a^2-\frac{i}{2}\\
    \end{array}
  \right)
\end{eqnarray}
and
\begin{eqnarray}
&\tilde{Z}_\Psi= \left(
    \begin{array}{cccc}
     0 & 0 & 0&0 \\
    0&0  &  0&0\\
 0 &0 & a^2+\frac{i}{2}&\frac{1}{2}\\
0 & 0  & \frac{1}{2}&a^2-\frac{i}{2}\\
    \end{array}
  \right).
\end{eqnarray}
The symmetries of these states are given by \bea S_m =(\sigma_z^{\otimes 4})^m,\eea for $m\in\{0,1\}$. Thus, we can apply Lemma \ref{lem0}. One obtains that the only states that are reachable in this SLOCC classes can be written as $h_i\otimes d_j\otimes d_k\otimes d_l \ket{\Psi_{a_20_{3\oplus1}}}$ where $\{i,j,k,l\} =\{1,2,3,4\}$ and $h_i\neq d_i$, whereas all other states are in $MES_4$. The convertible states are given by $g_i\otimes d_j\otimes d_k\otimes d_l \ket{\Psi_{a_20_{3\oplus1}}}$ (see Lemma \ref{lemc11}).\\
\section{The SLOCC class $L_{0_{5\oplus 3}}$}
We will treat next the SLOCC class which is denoted by $L_{0_{5\oplus 3}}$ in \cite{slocc4}. The corresponding seed state can be chosen as \begin{eqnarray}\nonumber
\ket{\Psi}&=1/\sqrt{2}\ket{01}(\ket{\Phi^-}-i\ket{\Psi^+})
+e^{-i\pi/4}\ket{10}\ket{\Psi^-}+e^{i\pi/4}\ket{11}\ket{\Psi^+}-\sqrt{2}e^{-i\pi/4}\ket{1100}.\end{eqnarray}
The $Z$ matrix and $\tilde{Z}$ matrix are of the form
\begin{eqnarray}
&Z_\Psi= \left(
    \begin{array}{cccc}
     \frac{i}{2} & \frac{1}{2} & 0&0 \\
    \frac{1}{2}&-\frac{i}{2}  &  0&0\\
 0 &0 & \frac{i}{2}&\frac{1}{2}\\
0 & 0  & \frac{1}{2}&-\frac{i}{2}\\
    \end{array}
  \right)
\end{eqnarray}
and
\begin{eqnarray}
&\tilde{Z}_\Psi= \left(
    \begin{array}{cccc}
     0 & 0 & 0&0 \\
    0&0  &  \frac{1}{2}+\frac{i}{2}&0\\
 0 &\frac{1}{2}+\frac{i}{2} & 0&\frac{1}{2}-\frac{i}{2}\\
0 & 0  & \frac{1}{2}-\frac{i}{2}&0\\
    \end{array}
  \right).
\end{eqnarray}

In order to simplify the notation for the symmetry consider the state \bea\ket{\Psi_{0_{5\oplus 3}}}=\one\otimes\one\otimes H Z(\pi/4)\otimes H Z(\pi/4)\ket\Psi\eea as seed state. The symmetries of this state are then given by
\begin{eqnarray}
&\tilde{S}_{x,y}= \left(
    \begin{array}{cc}
      1 & 0 \\
      \frac{(1-i) (1+x-y)}{y} & 1/y \\
    \end{array}
  \right)\otimes\one\otimes\left(
    \begin{array}{cc}
      y & 0\\
      x & 1 \\
    \end{array}
  \right)
\otimes\left(
    \begin{array}{cc}
      y & 0\\
      x & 1 \\
    \end{array}\right),
\end{eqnarray}
where $x,y\in \C$, $y\neq 0$. As we will see below, it is possible to choose a different seed state, which is of the form $A\otimes \one^{\otimes 3}\ket{\Psi_{0_{5\oplus 3}}}$, such that all symmetries which can be used in all possible transformations will be unitary. However, for the sake of simplicity we continue with this choice of the seed state.

Using the notation $G_i=[g_i^1,g_i^2;(g_i^2)^\ast,g_i^3]$, as before, we choose the standard form in this case as $G_1\otimes G_2\otimes \one/2 \otimes G_4$, with $g_4^2\geq 0$ unless one of the following conditions is fulfilled. If  $\Re(g_1^2)=\Im(g_1^2)=g_1^3(1+g_4^2/g_4^3)$ and $g_4^2 \neq 0$ we choose $G_1\otimes G_2\otimes G_3\otimes \one/2$, where $g_3^2\geq 0$ and
\bea \label{eq_G1} G_1=\left(
    \begin{array}{cc}
      g_1^1 & (1+i)g_1^3 \\
      (1-i)g_1^3 & g_1^3 \end{array}\right).
\eea
Note that $G_1$ has the property that $(S^{(1)}_{0,e^{i \phi_y}})^\dagger G_1 S^{(1)}_{0,e^{i \phi_y}}=G_1$ for any $\phi_y \in \R$, which will be a crucial property in the following. In case the previous conditions are not satisfied, but $g_4^2=0$ \footnote{Note that if $g_4^2=0$ and $\Re(g_1^2)=\Im(g_1^2)=g_1^3(1+g_4^2/g_4^3)=g_1^3$, the standard form (resulting from the first case) is given by
$G_1\otimes G_2\otimes \one/2 \otimes G_4$, with $G_1$ given in Eq. (\ref{eq_G1}).} one can easily seen that the standard form can be chosen as
$G_1\otimes G_2\otimes \one/2 \otimes G_4$, where $G_4$ is diagonal and
\bea \label{eq_G11} G_1=\left(
    \begin{array}{cc}
      g_1^1 & (1+i)g_1^2 \\
      (1-i)g_1^2 & g_1^3 \end{array}\right),
\eea
where $g_1^2\geq 0$ \footnote{Note that this standard form is not unique if $2g_1^3-g_1^2\geq 0$. However, in this case it can be easily chosen uniquely by applying $S_{0,-1}$ iff $g_1^3<g_1^2$.} Using this standard form we will now show the following lemma.
\begin{lemma} \label{lemma12L05} In the SLOCC class  $L_{0_{5\oplus 3}}$ the reachable states (in standard form) are of the form
\bi
\item[i)]$h_1\otimes h_2\otimes \one \otimes d_4\ket{\Psi_{0_{5\oplus 3}}}$ with $h_2$ arbitrary and $H_1=h_1^\dagger h_1$ is not of the form given in Eq. (\ref{eq_G1}) or
\item[ii)] $h_1\otimes h_2 \otimes h_3\otimes \one \ket{\Psi_{0_{5\oplus 3}}}$, with $h_2$ arbitrary, $h_3$ not diagonal, and $H_1=h_1^\dagger h_1$ is of the form given in Eq. (\ref{eq_G1})or
\item[iii)] same as case ii) with party $3$ and $4$ exchanged.
\ei

\end{lemma}

\begin{proof}

Let us first consider states for which the standard form is given by $H_1\otimes H_2\otimes \one \otimes H_4$. Note that the fact that all symmetries act trivially on the second system implies that $H_2=G_2$. Considering now the matrix elements $\ket{1001}\bra{1011}$, $\ket{1001}\bra{1010}$ and $\ket{1001}\bra{0011}$  of Eq. (\ref{EqSep})
it can be seen that for any symmetry which is used it must either hold that $x=0$ or that $\Re(h_1^2)=\Im(h_1^2)=h_1^3(1+h_4^2/h_4^3)$. Note that $h_4^2=0$ implies that $x=0$. This is why this instance will be treated in the first case. Note further that in the second case, the standard form is chosen differently, as mentioned above. Using then Eq. (\ref{EqSep}), one can show that also in the standard form of $G$, $G_1$ has to be of the form given in Eq. (\ref{eq_G1}). In order to do so, one uses
some matrix elements of Eq. (\ref{EqSep}) in combination with the equation, $r=h_1^3/g_1^3$, resulting from the projection of Eq. (\ref{EqSep}) onto $\ket{\Psi^-}_{34}$. Using then again Eq. (\ref{EqSep}) it can be easily seen that for any used symmetry it must hold that $x=0$. Hence, we have that whenever the standard form is used in Eq. (\ref{EqSep}), all symmetries occurring in this equation fulfill that $x=0$. Using similar methods as in e.g. the proof of Lemma \ref{lem3} and the equation resulting from the projection of Eq. (\ref{EqSep}) onto $\ket{\Psi^-}_{34}$ it is then easy to see that in this case it must furthermore hold that $|y|=1$ for any symmetry occurring in Eq. (\ref{EqSep}).
Note that this implies that any symmetry, $S_{x,y}$, occurring in Eq. (\ref{EqSep}) (using the standard form defined above) is of the form $S_{0,y}$, with $|y|=1$. In particular, the symmetries act as phase gates on party $3$ and $4$.

Let us now consider the different standard forms separately. For state with standard form $H_1\otimes H_2\otimes \one \otimes H_4$ we show now that the only states which are reachable are of the form given in (i) and (iii). More precisely, we show that states where (a) $H_1$ is not of the form given in Eq. (\ref{eq_G1}) and $H_4$ is not diagonal and (b) states where $H_1$ is of the form given in Eq. (\ref{eq_G1}) and $H_4$ is diagonal cannot be reached and present a LOCC protocol to reach the others. To prove that states in (a) cannot be reached one has to show that the only solution to Eq. (\ref{EqSep}) is that $G=H$ (using for both the standard form). Note that due to the symmetries, $H_2=G_2$ and $H_3=G_3=\one/2$. Moreover, as the symmetries are only phase gates on party $4$ it follows by tracing over all other parties that the diagonal elements of $H_4$ and $G_4$ have to coincide. Using then that $r=h_1^3/g_1^3$ and that $H_1$ is not of the form given in Eq. (\ref{eq_G1}) one obtains that $H_4=G_4$. It is then straightforward to show that also $H_1=G_1$ which implies that the states in case (a) are not reachable. For states in case (b), it is easy to see that all symmetries occuring in Eq. (\ref{EqSep}) commute with $H$, which implies that $H=G$ and therefore the states are not reachable. Let us now construct the LOCC protocol which reaches the states which do neither belong to case (a) nor to case (b), i.e. those which are given in Lemma \ref{lemma12L05} (case (i) and (iii)).
Any state corresponding to $H_1\otimes H_2\otimes \one \otimes D_4$, with $H_1$ not of the form given in Eq. (\ref{eq_G1}) [case (i)] [for which the standard form is as given above Eq. (\ref{eq_G11})] can be reached from a state $G_1\otimes H_2\otimes \one \otimes D_4$, where $G_1$ is of the form given in Eq. (\ref{eq_G1}). This can be easily seen by noting that,
as $H_4$ is diagonal, it commutes with all symmetries and therefore Eq. (\ref{EqSep}) reduces to $\sum_{y_i:|y_i|=1} p_{y_i} (S_{0,y_i}^{(1)})^\dagger H_1 S_{0,y_i}^{(1)}=r G_1$, which can always be satisfied by choosing $G_1$ of the form given in Eq. (\ref{eq_G1}).
The corresponding POVM elements are $\{1/\sqrt{2r}h_1 g_1^{-1}\otimes \one^{\otimes 3}, 1/\sqrt{2r} h_1 S_{0,-1}^{(1)} g_1^{-1}\otimes \one \otimes (S_{0,-1}^{(3)})^{\otimes 2}\}$, where $S_{0,-1}^{(1)}=[1,0;-2(1-i),-1]$ and $S_{0,-1}^{(3)}=\sigma_z$ and $r=h_1^3/g_1^3=1-4h_1^2+4h_1^3$ \footnote{Let us remark here that this SLOCC class is the only class where we could not find a standard form such that $r$ in Eq. (\ref{EqSep}) can be chosen to be $1$.}. Moreover, any state corresponding to $H_1\otimes H_2\otimes \one \otimes H_4$, with $H_1$ of the form given in Eq. (\ref{eq_G1}) and $H_4$ not diagonal [case (iii)] can be reached from a state $H_1\otimes H_2\otimes \one \otimes D_4$, where $D_4$ has the same diagonal elements as $H_4$. The corresponding POVM elements are $\{1/\sqrt{2} \one^{\otimes 3}\otimes h_4 d_4^{-1}, 1/\sqrt{2} U^{(1)}\otimes  \one \otimes S_{0,-1}^{(3)} \otimes h_4 S_{0,-1}^{(3)} d_4^{-1}\}$, where $U^{(1)}=h_1S_{0,-1}^{(1)} h_1^{-1}$ is unitary and  $S_{0,-1}^{(1)}$, $S_{0,-1}^{(3)}$ as above. Note that the crucial property that $U^{(1)}$ is unitary (which implies that this POVM can be implemented via LOCC) results from the special form of $h_1$.

Hence, it remains to consider states whose standard form is given by $H_1\otimes H_2\otimes H_3 \otimes \one $, where $H_1$ is of the form given in Eq. (\ref{eq_G1}). We consider the two cases $h_3^2=0$ and $h_3^2\neq 0$ separately. In the first case, we have similarly to before that the symmetries commute with $H_3$ and $H_4\propto \one$. Moreover, as the symmetries also commute with $H_1$ we have that $H=G$ has to hold, which implies that these states are not reachable. If, however, $h_3^2\neq 0$, the state can be reached from the state $H_1\otimes H_2\otimes G_3 \otimes \one$, where $G_3$ is diagonal. The corresponding POVM elements are as given in case (iii) replacing party 3 and 4.
\end{proof}

Let us note here that this case seems very different from all the other SLOCC classes because first of all the seed state is reachable [see Lemma \ref{lemma12L05} (i)] and moreover, as seen in the proof, the symmetries used for the transformation are, in contrast to all other cases not unitary. However, this is only an artifact of the choice of the seed state, as we will explain in the following. It is easy to see that choosing the seed state $A\otimes \one^{\otimes 3}\ket{\Psi_{0_{5\oplus 3}}}$ with $A=[\sqrt{2/3}, 1/(2\sqrt{6}) (1+i);1/(2\sqrt{6}) (1-i), 1/\sqrt{6}]$ leads to results resembling all the other cases. More precisely, choosing this seed state, the symmetries used in all the transformations are unitary and therefore [see Eq. (\ref{EqSep})] the seed state is in MES, as it is not reachable. Similarly to the previous lemma we can now prove

\begin{lemma} \label{lemma12} The only states in the SLOCC class  $L_{0_{5\oplus 3}}$ that are convertible are (in standard form) of the form
\bi
\item[i)] $g_1\otimes g_2\otimes \one \otimes d_4\ket{\Psi_{0_{5\oplus 3}}}$ with $g_1,g_2$ arbitrary \footnote{Note that in this case $G_1$ can always be chosen as in Eq. (\ref{eq_G11}) (as imposed by the standard form).} or
\item[ii)] $g_1\otimes g_2 \otimes g_3\otimes \one\ket{\Psi_{0_{5\oplus 3}}}$ with $G_1=g_1^\dagger g_1$ of the form given in Eq. (\ref{eq_G1}) and $g_2,g_3$ arbitrary or
\item[iii)]  same as case ii) with party $3$ and $4$ exchanged.
\ei

\end{lemma}

\begin{proof}
Due to the proof of Lemma \ref{lemma12L05} we have that the only symmetries occurring in Eq. (\ref{EqSep}) are of the form $S_{0,y}$, with $|y|=1$. Using then that the reachable states are given in Lemma \ref{lemma12L05} one can easily see that any convertible state has to be of one of the forms given in Lemma \ref{lemma12}. It hence remains to show that any state given in this lemma is convertible. Whereas states corresponding to case (ii) and (iii) can be easily converted using a POVM (constructed from the symmetries identity and $S_{0,-1}$) which acts non--trivially on party 3 and 4 respectively, it is more lengthy, however straightforward to show that any state in case (i) can be converted. In order to do so, one considers a transformation from states $g_1\otimes g_2\otimes \one \otimes d_4\ket{\Psi_{0_{5\oplus 3}}}$ to states $h_1\otimes g_2\otimes \one \otimes d_4\ket{\Psi_{0_{5\oplus 3}}}$. Let us recall here that in this case $G_1$ can always be chosen as in Eq. (\ref{eq_G11}) (as imposed by the standard form). Hence, we assume in the following that $G_1$ is of this form. As parties 2, 3 and 4 apply in this case only phase gates (or the identity) the necessary and sufficient condition for LOCC convertibility is $\sum_{y_i:|y_i|=1} p_{y_i} (S_{0,y_i}^{(1)})^\dagger H_1 S_{0,y_i}^{(1)}=r G_1$. Choosing $p_{y=1},p_{y=-1}\neq 0$ and all other probabilities equal to zero, it can be shown that for any $G_1$ there exists a positive $H_1$ and corresponding probabilities such that the equation above is satisfied, which proves the statement.
\end{proof}

Combining Lemma \ref{lemma12L05} and Lemma \ref{lemma12} it follows that the non-isolated states in the MES are given by $g_1\otimes g_2 \otimes d_3\otimes \one\ket{\Psi_{0_{5\oplus 3}}}$ and $g_1\otimes g_2 \otimes \one\otimes d_4\ket{\Psi_{0_{5\oplus 3}}}$ where in both cases $G_1=g_1^\dagger g_1$ is of the form given in Eq. (\ref{eq_G1}) and $g_2$ is arbitrary.

\section{The SLOCC class $L_{0_{7\oplus 1}}$}
Let us proceed with the SLOCC class $L_{0_{7\oplus 1}}$ in \cite{slocc4}.
One could choose as seed state \begin{eqnarray}&\ket{\Psi}= (1-i)/(2\sqrt{2}) (\ket{00}\ket{\Phi^-}-\ket{01}\ket{\Phi^-})+i/2(\ket{0001}+i\ket{0010}
-i\ket{0101}-\ket{0110})\\\nonumber
&-(1+i)/(2\sqrt{2})(\ket{10}\ket{\Phi^-}-\ket{11}\ket{\Phi^-})-(1+i)/2(\ket{1001}-\ket{1110})
\end{eqnarray}
The corresponding $Z$ and $\tilde{Z}$ matrix are given by
\begin{eqnarray}
&Z_\Psi= \left(
    \begin{array}{cccc}
     0 & \frac{1}{2} & \frac{i}{2}&0 \\
    \frac{1}{2}&\frac{i}{2}  & \frac{1}{2}&-\frac{i}{2}\\
\frac{i}{2} &\frac{1}{2} & -\frac{i}{2}&\frac{1}{2}\\
0 & -\frac{i}{2}  & \frac{1}{2}&0\\
    \end{array}
  \right)
\end{eqnarray}
and
\begin{eqnarray}
&\tilde{Z}_\Psi= \left(
    \begin{array}{cccc}
     0 & 0 & 0&0 \\
    0&0  &  \frac{1}{2}+\frac{i}{2}&0\\
 0 &\frac{1}{2}+\frac{i}{2} & 0&\frac{1}{2}-\frac{i}{2}\\
0 & 0  & \frac{1}{2}-\frac{i}{2}&0\\
    \end{array}
  \right).
\end{eqnarray}

In order to work with a similar kind of symmetry as before we will consider the state \bea\ket{\Psi_{0_{7\oplus 1}}} = H^{\otimes 4} (\sigma_z\otimes \sigma_z\otimes Z(-\pi/4)\otimes Z(-\pi/4))\ket{\Psi}\eea as seed state where $H$ denotes the Hadamard gate.
The corresponding symmetries are \bea S_z=\tfrac{1}{z} P_{z^2}\otimes P_z^{\otimes 3} \eea with $z\in\C\backslash 0$. This symmetry group allows to choose the following standard form
\begin{itemize}
\item $g_1^x\otimes g_2\otimes g_3\otimes g_4\ket{\Psi_{0_{7\oplus 1}}}$ (with $g_1^2>0$) for $g_1^2\neq 0$ and
\item $\one\otimes g_2\otimes g_3\otimes g_4\ket{\Psi_{0_{7\oplus 1}}}$ with $g_2^2\geq 0$ for $g_1^2= 0$.
\item In case  $g_1^2=g_2^2= 0$  we choose $g_3^2\geq 0$ and for $g_1^2=g_2^2=g_3^2= 0$ we set $g_4^2\geq 0$.
\end{itemize}
\begin{lemma} \label{lem13} The only states in the SLOCC class  $L_{0_{7\oplus 1}}$ that are reachable are given by $h_1^x\otimes h_j\otimes d_k\otimes d_l \ket{\Psi_{0_{7\oplus 1}}}$ with $\{j,k,l\} = \{2,3,4\}$ and $h_j\neq d_j$ and $h_1^x\otimes d_2\otimes d_3\otimes d_4 \ket{\Psi_{0_{7\oplus 1}}}$ with $h_1^x\neq \one/2$.
\end{lemma}
\begin{proof}
We will first show that all states different than the ones given in Lemma \ref{lem13} are in $MES_4$. Then we will provide for the states in Lemma \ref{lem13} the corresponding LOCC protocols that allow to reach them.\\
Writing $z=|z|e^{i\phi}$ Eq. (\ref{EqSep}) reads
\begin{eqnarray}\label{Eq_pzBK}
& \otimes\sum_z p_{z}\left(
    \begin{array}{cc}
     |z|^2/2&   h_1^2e^{-4i\phi}/|z|^2\\
    h_1^{2}e^{4i\phi}/|z|^2  &1/(2|z|^6)\\
    \end{array}
  \right)\otimes_{i=2}^4\left(
    \begin{array}{cc}
      h_i^1|z|^2&   h_i^2 e^{-2i\phi}\\
     h_i^{2 *}e^{i2\phi}  & h_i^3/|z|^2\\
    \end{array}
  \right)=r  G_1^x\otimes G_2\otimes G_3\otimes G_4.
\end{eqnarray}
Considering the matrix elements  $\ket{1000}\bra{1000}$ and  $\ket{0011}_{1,i,j,k}\bra{0011}$ of this equation where $\{i,j,k\}=\{2,3,4\}$ one obtains $h_j^1h_k^1/(h_j^3h_k^3)=g_j^1g_k^1/(g_j^3g_k^3)$. Considering this equation for various choices of $j$ and $k$ and using the normalization condition it follows that $h_l^1=g_l^1$ and $h_l^3=g_l^3$ $\forall l\in\{2,3,4\}$, as well as $r=1$. Thus, the diagonal components of $\otimes_i H_i$ can not be changed.\\
In order to structure the proof we  will discuss the following 6 different cases which cover all possibilities: (i) $h_j^2 \neq 0 $ $\forall j\in\{1,2,3,4\}$, (ii) $h_j^2 = 0 $ $\forall j\in\{1,2,3,4\}$ (iii) $h_1^2\neq 0$, $h_j^2=0$ for one $j\in\{2,3,4\}$ and $h_k^2\neq 0$ $\forall k\neq j$ (iv) $h_1^2=0$ and $h_j,h_k\neq 0$ but $h_l$ arbitrary  for  $\{j,k,l\} = \{2,3,4\}$, (v) $h_1^2$ arbitrary, $h_j^2 \neq 0$ and $h_k^2=h_l^2=0$ where $\{j,k,l\} = \{2,3,4\}$ and (vi) $h_1^2\neq 0$ and $h_2^2=h_3^2=h_4^2=0$. We will show that for the cases (i)-(iv) the corresponding states are in $MES_4$. As we will show case (v) and case (vi) correspond to reachable states.\\
Let us start with case (i). In order to see that the corresponding states are in $MES_4$ consider the matrix elements  $\ket{1010}_{1,i,j,k}\bra{1000}$ and $\ket{0101}_{1,i,j,k}\bra{0110}$ for several $i,j,k$. It follows that $h_j^{2 *}h_k^2=g_j^{2 *}g_k^2$ and $h_j^2/h_k^2=g_j^2/g_k^2$. Note that the first equation implies that for $h_j^{2},h_k^{2}\neq 0$  it has to hold that $g_j^{2},g_k^{2}\neq 0$.   We will use the notation $h_j^2=|h_j^2|e^{i \phi_j}$ and $g_j^2=|g_j^2|e^{i \tilde{\phi}_j}$. It follows that $|h_j^2|=|g_j^2| \,\,\, \forall j\in\{2,3,4\}$ and $\phi_j-\phi_k= \tilde{\phi}_j-\tilde{\phi}_k\,\,\, \forall j,k\in\{2,3,4\}$.  Considering now the matrix elements  $\ket{1001}_{1,2,3,4}\bra{0001}$ and  $\ket{0111}_{1,i,j,k}\bra{0100}$ one obtains that
\bea \label{Eq_auxB1}
h_1^2\sum_z p_z e^{-i4\phi}= g_1^2
\eea
and
\bea \label{Eq_auxB2}
h_j^2h_k^2\sum_z p_z e^{-i4\phi}= g_j^2g_k^2.
\eea
Recall that according to the standard form $h_1^2,g_1^2\geq 0$. Hence, the equations above imply that $\sum_z p_z  e^{-i4\phi}> 0$ \footnote{As already mentioned in the case considered here it has to hold that $g_j^2\neq 0$ for $j\in\{2,3,4\}$ and therefore it follows from Eq. (\ref{Eq_auxB2}) that $\sum_z p_z  e^{-i4\phi}\neq 0$.}  and therefore $\phi_j+\phi_k= \tilde{\phi}_j+\tilde{\phi}_k$. As the differences of these phases also coincide, we have $\phi_j=\tilde{\phi}_j$ and $\phi_k=\tilde{\phi}_k$. Thus, $h_j^2=g_j^2$ $\forall j\in\{2,3,4\}$ and so also $h_1^2=g_1^2$. Hence, we have shown in case (i) that $H_i=G_i$ $\forall i\in\{1,2,3,4\}$ and therefore the corresponding states are in $MES_4$.\\
In case (ii) one can easily show that the corresponding states are in $MES_4$. In order to see this consider the matrix elements $\ket{0000}_{i,j,k,l}\bra{1000}$ for various choices of $i,j,k,l$. \\
In case (iii) (and using the same notation as before) it is easy to see that $g_j^2=0$. Furthermore, one can use the same argument as in case (i) to show that $h_k^2=g_k^2$ $\forall k\neq j$ and $h_1^2=g_1^2$. Thus, $H_i=G_i$ $\forall i\in\{1,2,3,4\}$ and the corresponding states are in $MES_4$.\\
For case (iv), i.e. $h_1^2=0$ and $h_j,h_k\neq 0$ but $h_l$ arbitrary, it can be shown that the corresponding states are not reachable by considering  the matrix elements $\ket{1100}_{1,j,k,l}\bra{1000}$, $\ket{1010}_{1,j,k,l}\bra{1000}$ and $\ket{0011}_{1,j,k,l}\bra{0101}$ in Eq. (\ref{Eq_pzBK}). It follows that $h_j^2/h_k^2=g_j^2/g_k^2$ and $h_j^{2 *}h_k^2=g_j^{2 *}g_k^2$. Assume without loss of generality $j<k,l$. Due to the standard form we have $h_j^2> 0$ and therefore it follows that $h_j^2=g_j^2$ and hence $h_k^2=g_k^2$. Analogously one obtains $h_l^2=g_l^2$ in case $h_l^2\neq 0$. In case $h_l^2= 0$ it can be shown as before that $g_l^2=0$. Thus, this kind of states is not reachable.\\
For case (v) and case (vi) we show that the corresponding states are reachable by providing the corresponding LOCC protocols. For case (vi) one can use the POVM  $\{1/\sqrt{2}(g_1^x\otimes \one^{\otimes 3} ),1/\sqrt{2} (g_1^x\sigma_z\otimes Z(\pi/4)\otimes Z(\pi/4)\otimes Z(\pi/4)) \}$ which is acting non-trivially on party $1$ to reach them. States which correspond to case (v), i.e. $h_1^x\otimes h_j\otimes d_k\otimes d_l \ket{\Psi_{0_{7\oplus 1}}}$ where $\{j,k,l\} = \{2,3,4\}$ and $h_j^2\neq 0$, can be reached for example from states of the form $h_1^x\otimes d_j\otimes d_k\otimes d_l \ket{\Psi_{0_{7\oplus 1}}}$ with a properly chosen $d_j$ via $\{1/\sqrt{2}(\one_1\otimes g_jd_j^{-1}\otimes\one^{\otimes 2}) ,1/\sqrt{2}(\one_1\otimes g_j \sigma_z d_j^{-1}\otimes \sigma_z^{\otimes 2})\}$. Thus, we have shown that all states in this SLOCC class apart from the ones given in this lemma are in $MES_4$.
\end{proof}
Note that in contrast to most SLOCC classes before we used in the proof a LOCC protocol that exploits a symmetry that is not an element
of the Pauli group. Note further that nevertheless this symmetry acts with $\sigma_z$ on the qubit on which the LOCC protocol is acting non-trivially. In the following lemma we show which states are convertible in this SLOCC class.
\begin{lemma}  The only states in the SLOCC class  $L_{0_{7\oplus 1}}$ that are convertible are given by $g_1^x\otimes g_j\otimes d_k\otimes d_l \ket{\Psi_{0_{7\oplus 1}}}$ where $j,k,l \in \{2,3,4\}$.
\end{lemma}
\begin{proof}
Inserting the reachable states of Lemma \ref{lem13} in Eq. (\ref{EqSep}) one obtains as a necessary condition for convertible states that they have to be of the form $g_1^x\otimes g_j\otimes d_k\otimes d_l \ket{\Psi_{0_{7\oplus 1}}}$ where $j,k,l \in \{2,3,4\}$. In order to see that one can indeed transform them consider
$\{\sqrt{p}(\one_1\otimes h_jg_j^{-1}\otimes\one^{\otimes 2} ),\sqrt{1-p} (\one_1\otimes h_j\sigma_z g_j^{-1}\otimes\sigma_z^{\otimes 2})\}$ which is a valid POVM for $h_j^1=g_j^1$, $h_j^3=g_j^3$ and $(2p-1)h_j^2=g_j^2$. Thus, one can find for any state $g_1^x\otimes g_j\otimes d_k\otimes d_l \ket{\Psi}$ a value of $p$ such that the state can be transformed to a LU-inequivalent state via the POVM given above.
\end{proof}
Thus, the convertible states in the MES are given by $\one/2\otimes d_2\otimes d_3\otimes d_4\ket{\Psi_{0_{7\oplus 1}}}$.

\section{Conclusions}

The study of LOCC transformations is of paramount importance in entanglement theory. First, this is because these operations give rise to all possible protocols to manipulate entangled states which are distributed among spatially separated parties. Second, LOCC induces the most natural ordering in the set of entangled states. Hence, its study allows to classify entangled states according to their usefulness. However, although LOCC transformations among pure states were characterized in the bipartite case more than 10 years ago \cite{Nielsen}, they have remained mostly unexplored in the multipartite regime with only a few exceptions \cite{Turgut1,Turgut2,Tajima}. In \cite{MESus}, we recently provided general techniques to address the problem of LOCC convertibility for general multipartite pure states. Our main goal was to determine the MES, that we introduced therein in order to establish maximal multipartite entanglement on solid theoretical and operational grounds and to identify relevant classes of multipartite entangled states for the applications of quantum information theory. In particular, in \cite{MESus} we characterized the MES of 3-qubit states, $MES_3$, and the generic subset of the MES for 4-qubit states, $MES_4$ (i.\ e.\ the MES corresponding to the generic SLOCC families known as $G_{abcd}$ up to certain particular values of the parameters $\{a,b,c,d\}$). In \cite{MESus} we showed that almost all generic 4-qubit states are isolated. Thus, perhaps, surprisingly $MES_4$ is full measure in the set of 4-qubit states. This implies that LOCC induces a trivial ordering (almost all states are incomparable) and LOCC entanglement manipulation is almost never possible already for systems of only 4 constituents. However, there still exists a zero-measure set of the most relevant class of states in this context. This is given by the subset of LOCC convertible states inside the MES, which we also characterized in the generic case. These findings encouraged further the need to study LOCC transformations and to determine the subset of the MES inside the non-generic families, which is the aim of the current paper. Most paradigmatic states belong to these classes and it could happen that, contrary to the generic case, isolation would be rare in some of these cases, unraveling classes with a rich LOCC structure.

In this paper we investigated all possible SLOCC classes (including the remaining cases of the $G_{abcd}$ families). We identified classes with very distinct transformation properties. However, for most classes it holds true that any state which can be reached via a separable state transformation can be reached via a very simple LOCC protocol, where only one party needs to measure. Moreover, for most classes it is as in the generic class true that most states are isolated. Clear exceptions of this are the SLOCC classes $L_{ba_3}$ for $a\neq 0$, $L_{ba_3}$ for $a= 0$ and $b\neq 0$, $L_{a_4}$ for $a\neq 0$, the L-state class, and the GHZ and W--class. Whereas the first 3 classes contain only isolated states, the L-state is special due to the fact that different SEP transformations become possible, prohibiting a simple characterization of reachable states (via SEP) from the standard form.
Interestingly, it turns out from the analysis performed here that only two very particular classes have no isolated states: the generalization of the GHZ and W class to 4-qubit states. Note that this result holds also true for $n$--partite GHZ and W--classes. Moreover, these are the only classes in which a zero-measure subset inside the particular class is in the MES. In all the other cases almost all states inside the particular class are isolated. Hence, the generic impossibility of LOCC entanglement manipulation is not just a feature of the generic SLOCC families. Nevertheless, we characterized here all elements of $MES_4$ (with the exception of the L--state) which are indeed LOCC convertible. This constitutes a zero-measure class of multipartite entangled states which is operationally most relevant for LOCC manipulation. In the future it would be interesting to study further the physical and mathematical properties of this class of states and see whether one can establish connections among their LOCC relevance and other applications of quantum information theory such as quantum computation and quantum communication.

Moreover, our investigation shows that the ideas we introduced in \cite{MESus} are powerful enough to decide LOCC convertibility for SLOCC classes independently of their very different mathematical structure. In particular, the generic 4-qubit family has a very simple stabilizer; however, in this paper we have provided a systematic technique to identify the set of symmetries for an arbitrary SLOCC class, which plays a key role in the characterization of state transformations. Furthermore, we have shown in detail how to construct LU standard forms for each family that are very suitable for the characterization of LOCC transformations (in analogy to the role played by the Schmidt coefficients in the bipartite case).

The tools used here allow in principle to study LOCC convertibility for systems of an arbitrary number of parties and dimensions. However, in these cases, despite recent advances \cite{GourSLOCC}, their SLOCC structure is not as well established as in the 3 and 4-qubit case. Moreover, as can be seen from the length of this paper, the fact that each SLOCC class is studied separately suggests that this is a formidable task. It would be then interesting to find tools that allow to decide LOCC convertibility from general properties of the symmetry group without the need of using its very particular form in each case. More modest problems, such as proving that isolation is a generic feature for $n$-partite states for $n\geq4$, as intuition suggests, could be perhaps addressed with the techniques used here.

The generic triviality of the LOCC ordering also suggests to study state manipulation beyond single-copy LOCC transformations. It would be interesting to see whether isolation can be removed by considering multiple-copy LOCC manipulation and/or the use of catalyzers. As in the bipartite case, this seems to be a natural continuation of our work. Other possible directions are the study of maximal probabilities of success under SLOCC operations and the derivation of operational multipartite entanglement measures according to LOCC usefulness. Another alternative is to study maximally useful states for LOCC conversions from $n$-partite states to $m$-partite states with $m<n$. For instance, the 3-qubit GHZ state allows to obtain any 2-qubit state by LOCC and we recently presented a 6-qubit state that allows to obtain any 3-qubit state by LOCC \cite{REP}. Finally, very little is still known about multipartite LOCC transformations involving mixed states. Apart from that, we intend to also study $\epsilon$--convertability, i.e. LOCC conversions that allow to reach the target state up to a deviation of $\epsilon$ (measured with some suitably chosen norm). In \cite{SaSc15} and \cite{ScSa15} we used the results on LOCC transformation presented here and in \cite{MESus} to introduce and compute new entanglement measures.

\section*{ACKNOWLEDGEMENTS}

The research of CS and BK was funded by the Austrian Science Fund (FWF): Y535-N16. JIdV acknowledges financial support from the Spanish MINECO grants MTM2010-21186-C02-02, MTM2011-26912, MTM2014-54692 and MTM2014-54240-P and CAM regional research consortium QUITEMAD+CM S2013/ICE-2801.

\section*{\label{appA}Appendix A: Proof of Lemma \ref{lemcycle2}}

In this appendix the proof of Lemma \ref{lemcycle2}, i.e. the characterization of the reachable states in the SLOCC classes $G_{abcd}$ with $a^2=-d^2$, $b^2=-c^2$, $a^2 \neq \pm c^2$   and $a,c\neq 0$ is presented. The symmetries for these SLOCC classes are given by
\begin{align}\label{symcycle20}
S&=\{\one^{\otimes 4},\sigma_x^{\otimes 4},\sigma_y^{\otimes 4},\sigma_z^{\otimes 4},\one\otimes\sigma_z\otimes\sigma_x\otimes\sigma_y, \sigma_z\otimes\one\otimes\sigma_y\otimes\sigma_x,\sigma_x\otimes\sigma_y\otimes\one\otimes\sigma_z,\sigma_y\otimes\sigma_x\otimes\sigma_z\otimes\one\}.
\end{align}
We denote the elements of the symmetry group $S$ by $S_j=\otimes_i s_i^j$ for $j=0,\ldots, 7$ and choose the same numeration as in Eq. (\ref{symcycle20}), i.e. $S_0=\one^{\otimes 4}$, $S_1=\sigma_x^{\otimes 4}$ etc.. In order to prove Lemma \ref{lemcycle2} we will make use of the following lemma.
\begin{lemma}\label{twooutcome0}
The only states in the SLOCC classes $G_{abcd}$ with $a^2=-d^2$, $b^2=-c^2$, $a^2 \neq \pm c^2$   and $a,c\neq 0$ that are reachable via using only the symmetries $S_0$ and one $S_j$ for some $j\in\{1,2,3,4,5,6,7\}$ are given by $\otimes h_i \ket{\Psi_{a(ia)c(-ic)}}$ where $\otimes H_i$ obeys the condition that exactly
three out of the four operators $H_i$ commute with the corresponding
operator $s_i^j$ and one operator $H_k$ does not commute with $s_k^j$.
\end{lemma}
\begin{proof}
Since the symmetries are elements of the Pauli group we will use Eq. (\ref{eq_Pauli}). We will first investigate in detail which transformations are possible with the symmetries $S_0$ and $S_4$ and then generalize this result. Considering Eq. (\ref{eq_Pauli}) with P=$\one\otimes \one \otimes\sigma_i\otimes \sigma_j$ where $i=2,3$ and $j=1,3$, as well as with $P=\one\otimes\sigma_i\otimes \one\otimes \sigma_j$ where $i=1,2$ and $j=1,3$ and $P=\one\otimes \sigma_i\otimes \sigma_j\otimes\one$ where $i=1,2$ and $j=2,3$
one obtains that the corresponding equations are always of the form
\bea
\bar{h}_k^i\bar{h}_l^j=\bar{h}_k^i\bar{h}_l^j(p_0-p_4)^2.
\eea
In order to obtain a non-trivial transformation it follows that $\bar{h}_k^i\bar{h}_l^j=0$ for all the above mentioned components. That is, one obtains that either $H_2=H_2^{00A}$ and $H_3=H_3^{A00}$ or that $H_2=H_2^{00A}$ and $H_4=H_4^{0A0}$ or that $H_3=H_3^{A00}$ and $H_4=H_4^{0A0}$. These states are reachable from some other state in case $\otimes_i H_i$ is such that $S_4(\otimes_i H_i) S_4 \neq \otimes_i H_i$. In particular one can reach them via the POVM $\{1/\sqrt{2} h_jS_0g_j^{-1},1/\sqrt{2} h_jS_4g_j^{-1}\}$ (for a properly chosen $g_j$) where $j$ corresponds to the party for which it holds that $s_j^4H_js_j^4\neq H_j$. Using that $S_i$ coincides with $S_4$ for $i\in\{5,6,7\}$ up to permutations, the states reachable via these symmetries follow straightforwardly. For instance, the reachable states that can be obtained via LOCC using $S_0$ and $S_5$ are given by $g_1^{\neq 3}\otimes g_2^{AAA}\otimes g_3^{0A0}\otimes g_4^{A00}\ket{\Psi_{a(ia)c(-ic)}}$ and $g_1^{00A}\otimes g_2^{AAA}\otimes g_3^{\neq 2}\otimes g_4^{A00}\ket{\Psi_{a(ia)c(-ic)}}$ and $g_1^{00A}\otimes g_2^{AAA}\otimes g_3^{0A0}\otimes g_4^{\neq 1}\ket{\Psi_{a(ia)c(-ic)}}$. Let us consider next which states are reachable using the symmetries $S_0$ and one $S_j$ for some $j\in\{1,2,3\}$.  As we have seen in the generic case (see \cite{MESus}) the symmetries $\sigma_l^{\otimes 4}$ for $l=0,1,2,3$ allow to reach only states of the form  $h_i^{\neq m}\otimes h_j^{W_m=A00}\otimes h_k^{W_m=A00}\otimes h_l^{W_m=A00}\ket{\Psi_{a(ia)c(-ic)}}$ for $m\in\{1,2,3\}$ and $h_i^{\neq 0}\otimes \one/2 \otimes \one/2 \otimes \one/2 \ket{\Psi_{a(ia)c(-ic)}}$.
Let us now characterize those states within this set which are reachable using  only the symmetries $S_0$ and one $S_j$ for some $j\in\{1,2,3\}$. We will first consider $S_0$ and $S_1$ and then generalize the results.  Considering Eq. (\ref{eq_Pauli}) with P such that we have $\sigma_i$ on party $k$, $\sigma_j$ on party $l$ and the identity on the remaining parties  for $i\in\{2,3\}$, $j\in\{2,3\}$, $k\neq l$ and $k,l\in\{1,2,3,4\}$ we obtain
\bea
\bar{h}_k^i\bar{h}_l^j=\bar{h}_k^i\bar{h}_l^j(p_0-p_1)^2.
\eea
Hence, states of the form $h_i^{\neq m}\otimes h_j^{W_m=A00}\otimes h_k^{W_m=A00}\otimes h_l^{W_m=A00}\ket{\Psi_{a(ia)c(-ic)}}$ for $m\in\{2,3\}$ and $\{i,j,k,l\}=\{1,2,3,4\}$ are not reachable (except if $h_t^{W_m=A00}=\one/2$ $\forall t\in\{j,k,l\}$ or if they can be written also in the form $h_s^{100}\otimes h_t^{W_m=A00}\otimes \one_p/2\otimes \one_r/2\ket{\Psi_{a(ia)c(-ic)}}$ where $\{s,t,p,r\}=\{1,2,3,4\}$).
Note that for reachable states it has to hold that $S_1^\dagger(\otimes_i H_i) S_1 \neq \otimes_i H_i$ as otherwise Eq. (\ref{EqSep}) implies that $H_i=G_i$ $\forall i\{1,2,3,4\}$.
Thus, the only states that are reachable via $S_0$ and $S_1$ are given by $h_i^{\neq 1}\otimes h_j^{A00}\otimes h_k^{A00}\otimes h_l^{A00}\ket{\Psi_{a(ia)c(-ic)}}$. Via an analogous argumentation one can show that the states that are reachable via $S_0$ and one $S_m$ for some $m\in\{2,3\}$ are given by  $h_i^{\neq m}\otimes h_j^{W_m=A00}\otimes h_k^{W_m=A00}\otimes h_l^{W_m=A00}\ket{\Psi_{a(ia)c(-ic)}}$ where $\{i,j,k,l\}=\{1,2,3,4\}$, which completes the proof.
\end{proof}
We will prove now Lemma \ref{lemcycle2} by showing that any state that is reachable can be obtained via a two outcome POVM using only the symmetries $S_0$ and one $S_j$ for some $j\in\{1,2,3,4,5,6,7\}$. Lemma \ref{twooutcome0} is then used to identify the reachable states.  \\ \\
\textit{Proof of Lemma \ref{lemcycle2}: }
As the symmetries are elements of the Pauli group we use Eq. (\ref{eq_Pauli}) to identify the reachable states. In the following we will use the notation $\eta_{ijkl}=p_0+p_m+p_n+p_s+p_t-p_i-p_j-p_k-p_l$ where $\{m,n,s,t,i,j,k,l\}=\{1,\ldots,7\}$. Note that due to $\sum_{w=0}^7 p_w=1$ it holds that $|\eta_{i,j,k,l}|\leq 1$. Note further that due to the fact that the symmetries form a group we can choose w.l.o.g. that in any transformation $S_0$ occurs, i.e. $p_0\neq 0$, as long as we do not impose any constraint on the standard form. In order to see this consider for simplicity a two-outcome transformation which transforms $g\ket{\Psi}$ to $h\ket{\Psi}$, i.e. $M_1=\sqrt{p_1}hS_ig^{-1}$ and $M_2=\sqrt{p_2}hS_jg^{-1}$, which is a valid POVM if $p_1 S_iHS_i+ p_2 S_jHS_j =G$. Multiplying both sides of this equation with $S_i$ one obtains that this transformation can be done iff $p_1 S_0HS_0+ p_2 S_iS_jHS_jS_i =S_iGS_i=\tilde{G}$. The symmetries form a group and so $S_iS_j=S_k\in S$ and  $\tilde{G}$ defines the same state as $G$, which can be easily seen using the ambiguity due to the symmetries of the seed states. Note that if we would specify here a specific standard form and investigate transformations between states in standard form, either $G$ or $\tilde{G}$ would not be in standard form (except for $S_i=S_0$) and therefore we would not be looking at transformations that would allow to reach the corresponding state. Thus, we are allowed to either restrict to a specific standard form or to transformation that involve $S_0$. In order to make the proof simpler we choose in the following that all allowed transformations involve $S_0$ but we will not use any restrictions set by a standard form.
\\Considering now Eq. (\ref{eq_Pauli}) with $P=S_k$, as well as $P=(\sigma_x\otimes \one\otimes \sigma_x\otimes \one)S_k$, $P=(\sigma_y\otimes \one \otimes \one\otimes\sigma_y) S_k$, and $P=(\sigma_z\otimes \sigma_z \otimes\one\otimes \one) S_k$ where $k \in\{4,5,6,7\}$ one obtains the following equations
\begin{eqnarray}
\bar{h}_2^3\bar{h}_3^1\bar{h}_4^2&=\bar{h}_2^3\bar{h}_3^1\bar{h}_4^2 \eta_{1267}\eta_{2357}\eta_{1356}\label{eqnarr1}\\\nonumber
\bar{h}_1^3\bar{h}_3^2\bar{h}_4^1&=\bar{h}_1^3\bar{h}_3^2\bar{h}_4^1 \eta_{1267}\eta_{1347}\eta_{2346}\\ \nonumber
\bar{h}_1^1\bar{h}_2^2\bar{h}_4^3&=\bar{h}_1^1\bar{h}_2^2\bar{h}_4^3 \eta_{2357}\eta_{1347}\eta_{1245}\\ \nonumber
\bar{h}_1^2\bar{h}_2^1\bar{h}_3^3&=\bar{h}_1^2\bar{h}_2^1\bar{h}_3^3 \eta_{1356}\eta_{2346}\eta_{1245}\\ \nonumber
\bar{h}_1^2\bar{h}_3^3\bar{h}_4^1&=\bar{h}_1^2\bar{h}_3^3\bar{h}_4^1 \eta_{1356}\eta_{2346}\eta_{1245}\\ \nonumber
\bar{h}_1^1\bar{h}_2^3\bar{h}_4^2&=\bar{h}_1^1\bar{h}_2^3\bar{h}_4^2 \eta_{1267}\eta_{2357}\eta_{1356}\\ \nonumber
\bar{h}_2^2\bar{h}_3^1\bar{h}_4^3&=\bar{h}_2^2\bar{h}_3^1\bar{h}_4^3 \eta_{2357}\eta_{1347}\eta_{1245}\\ \nonumber
\bar{h}_1^3\bar{h}_2^1\bar{h}_3^2&=\bar{h}_1^3\bar{h}_2^1\bar{h}_3^2 \eta_{1267}\eta_{1347}\eta_{2346}\\ \nonumber
\bar{h}_1^3\bar{h}_2^2\bar{h}_4^1&=\bar{h}_1^3\bar{h}_2^2\bar{h}_4^1 \eta_{1267}\eta_{1347}\eta_{2346}\\ \nonumber
\bar{h}_1^2\bar{h}_2^3\bar{h}_3^1&=\bar{h}_1^2\bar{h}_2^3\bar{h}_3^1 \eta_{1267}\eta_{2357}\eta_{1356}\\ \nonumber
\bar{h}_1^1\bar{h}_3^2\bar{h}_4^3&=\bar{h}_1^1\bar{h}_3^2\bar{h}_4^3 \eta_{2357}\eta_{1347}\eta_{1245}\\ \nonumber
\bar{h}_2^1\bar{h}_3^3\bar{h}_4^2&=\bar{h}_2^1\bar{h}_3^3\bar{h}_4^2 \eta_{1356}\eta_{2346}\eta_{1245}\\ \nonumber
\bar{h}_2^3\bar{h}_3^2\bar{h}_4^1&=\bar{h}_2^3\bar{h}_3^2\bar{h}_4^1 \eta_{1267}\eta_{1347}\eta_{2346}\\ \nonumber
\bar{h}_1^3\bar{h}_3^1\bar{h}_4^2&=\bar{h}_1^3\bar{h}_3^1\bar{h}_4^2 \eta_{1267}\eta_{2357}\eta_{1356}\\ \nonumber
\bar{h}_1^2\bar{h}_2^1\bar{h}_4^3&=\bar{h}_1^2\bar{h}_2^1\bar{h}_4^3 \eta_{1356}\eta_{2346}\eta_{1245}\\ \nonumber
\bar{h}_1^1\bar{h}_2^2\bar{h}_3^3&=\bar{h}_1^1\bar{h}_2^2\bar{h}_3^3 \eta_{2357}\eta_{1347}\eta_{1245}\\ \nonumber
\end{eqnarray}
Let us now distinguish the following 2 cases:
\begin{itemize}
\item[(i)] There exist a product of the $\bar{h}_i^j$ occuring in Eq. (\ref{eqnarr1}) which does not vanish.
\item[(ii)] All these products vanish.
\end{itemize}
In case (i) it follows immediately that only two symmetries,  $S_0$ and one $S_j$ for some $j\in\{4,5,6,7\}$, can be used for LOCC transformations. Thus, in this case the reachable states are characterized by Lemma \ref{twooutcome0}. Consider for example $\bar{h}_2^3\bar{h}_3^1\bar{h}_4^2\neq 0$. Then the first equation of (\ref{eqnarr1}) implies that $|\eta_{1267}|=|\eta_{2357}|=|\eta_{1356}|=1$ (as $|\eta_{i,j,k,l}|\leq 1$) and therefore all probabilities except $p_0$ and $p_4$ have to be zero. Case (ii) is more involved. However, using again Eq. (\ref{eq_Pauli}) for P being an operator acting non-trivially on two parties one can again show that for another class of states the possible non-trivial LOCC protocols can only use two symmetries. For the remaining set of states one can then show by using Eq. (\ref{eq_Pauli}) with P acting non-trivially on all parties that any reachable state can be obtained from some other state via LOCC by using only $S_0$ and one $S_j$ for some $j\in\{1,2,3,4,5,6,7\}$. Let us first consider Eq. (\ref{eq_Pauli}) for P acting non-trivially on two parties.
Considering first Eq. (\ref{eq_Pauli}) with $P= \sigma_x\otimes\one\otimes\sigma_x\otimes \one$ and $P= \one\otimes\sigma_x\otimes\one\otimes\sigma_x$ one obtains that
\bea\label{e1}
\bar{h}_1^1\bar{h}_3^1=\bar{h}_1^1\bar{h}_3^1 (\eta_{2357})^2
\eea
and
\bea\label{e2}
\bar{h}_2^1\bar{h}_4^1=\bar{h}_2^1\bar{h}_4^1 (\eta_{2346})^2.
\eea
Thus, for $\bar{h}_j^1\neq 0$ $\forall j\in\{1,2,3,4\}$ the only possible transformations are due to the symmetries $S_0$ and $S_1$. Similar one obtains that if $\bar{h}_j^2\neq 0 $ $\forall j$ ($\bar{h}_j^3\neq 0 \forall j$), all probabilities except $p_0$ and $p_2$ ($p_0$ and $p_3$) have to be zero respectively. More precisely, for particular choices of $P$ in Eq. (\ref{eq_Pauli}) one obtains that
\begin{eqnarray}
\label{eqnarr2}
\bar{h}_1^2\bar{h}_4^2&=\bar{h}_1^2\bar{h}_4^2 (\eta_{1356})^2\\ \label{eqnarr21}
\bar{h}_2^2\bar{h}_3^2&=\bar{h}_2^2\bar{h}_3^2 (\eta_{1347})^2\\ \label{eqnarr22}
\bar{h}_1^3\bar{h}_2^3&=\bar{h}_1^3\bar{h}_2^3 (\eta_{1267})^2\\ \label{eqnarr23}
\bar{h}_3^3\bar{h}_4^3&=\bar{h}_3^3\bar{h}_4^3 (\eta_{1245})^2.
\end{eqnarray}
Note here that the case $\bar{h}_i^k\bar{h}_j^k\neq 0$ for two different choices of $i,j,k$ occurring in  Eq. (\ref{e1}), Eq. (\ref{e2}) and the Eqs. (\ref{eqnarr2}) - (\ref{eqnarr23}) leads to the condition that only two of the probabilities can be different from zero. Thus, all transformations that need to be considered in this case only involve $S_0$ and one $S_i$ for some $i\in\{1,2,3,4,5,6,7\}$. These transformations have already been discussed in Lemma \ref{twooutcome0} and therefore the states that are reachable in this case are given in Lemma \ref{twooutcome0}. \\
Let us now consider as an example the case where $\bar{h}_1^1\bar{h}_3^1\neq 0$ and that there is no further restriction on the probabilities from Eq. (\ref{e2}) and the Eqs. (\ref{eqnarr2}) - (\ref{eqnarr23}), as well as from the set of equations given in Eq. (\ref{eqnarr1}), i.e. $\bar{h}_2^1\bar{h}_4^1=0$, $\bar{h}_1^2\bar{h}_4^2=0$ etc..
The cases where some other product of the $\bar{h}_i^j$ occurring in Eq. (\ref{e2}) or in the Eqs. (\ref{eqnarr2}) - (\ref{eqnarr23}) is non-zero and all others are zero can be treated analogously.
Considering Eq. (\ref{eq_Pauli}) for $P=\sigma_y\otimes \sigma_z\otimes \sigma_y\otimes\sigma_z$, $P=\sigma_z\otimes \sigma_y\otimes \sigma_z\otimes\sigma_y$, $P=\sigma_x\otimes \sigma_z\otimes \sigma_z\otimes\sigma_x$, $P=\sigma_z\otimes \sigma_x\otimes \sigma_x\otimes\sigma_z$, $P=\sigma_x\otimes \sigma_x\otimes \sigma_y\otimes\sigma_y$ and $P=\sigma_y\otimes \sigma_y\otimes \sigma_x\otimes\sigma_x$ leads to the conditions
\begin{eqnarray}
\label{eqnarr0}
\bar{h}_1^2\bar{h}_2^3\bar{h}_3^2\bar{h}_4^3&=\bar{h}_1^2\bar{h}_2^3\bar{h}_3^2\bar{h}_4^3 \eta_{1356}\eta_{1267}\eta_{1347}\eta_{1245}\\ \nonumber
\bar{h}_1^3\bar{h}_2^2\bar{h}_3^3\bar{h}_4^2&=\bar{h}_1^3\bar{h}_2^2\bar{h}_3^3\bar{h}_4^2\eta_{1267}\eta_{1347}\eta_{1245}\eta_{1356}\\ \nonumber
\bar{h}_1^1\bar{h}_2^3\bar{h}_3^3\bar{h}_4^1&=\bar{h}_1^1\bar{h}_2^3\bar{h}_3^3\bar{h}_4^1\eta_{2357}\eta_{1267}\eta_{1245}\eta_{2346}\\ \nonumber
\bar{h}_1^3\bar{h}_2^1\bar{h}_3^1\bar{h}_4^3&=\bar{h}_1^3\bar{h}_2^1\bar{h}_3^1\bar{h}_4^3\eta_{2357}\eta_{1267}\eta_{1245}\eta_{2346}\\ \nonumber
\bar{h}_1^1\bar{h}_2^1\bar{h}_3^2\bar{h}_4^2&=\bar{h}_1^1\bar{h}_2^1\bar{h}_3^2\bar{h}_4^2\eta_{2357}\eta_{1347}\eta_{1356}\eta_{2346}\\ \nonumber
\bar{h}_1^2\bar{h}_2^2\bar{h}_3^1\bar{h}_4^1&=\bar{h}_1^2\bar{h}_2^2\bar{h}_3^1\bar{h}_4^1\eta_{2357}\eta_{1347}\eta_{1356}\eta_{2346}
\end{eqnarray}
 Thus, one obtains that for $\bar{h}_1^2\bar{h}_2^3\bar{h}_3^2\bar{h}_4^3\neq 0$ all $p_i$ have to be equal to zero except for either $p_1$ or $p_0$ which corresponds to a LU or the trivial transformation. Thus, in this case no transformation to another (LU-inequivalent) state is possible. Analogous, one obtains that in order to allow for a non-trivial transformation $\bar{h}_1^3\bar{h}_2^2\bar{h}_3^3\bar{h}_4^2=\bar{h}_1^1\bar{h}_2^3\bar{h}_3^3\bar{h}_4^1=\bar{h}_1^3\bar{h}_2^1\bar{h}_3^1\bar{h}_4^3=\bar{h}_1^1\bar{h}_2^1\bar{h}_3^2\bar{h}_4^2=\bar{h}_1^2\bar{h}_2^2\bar{h}_3^1\bar{h}_4^1=0$.
Equation (\ref{e1}) implies that for $\bar{h}_1^1\bar{h}_3^1\neq 0$ the only symmetries that can contribute to a transformation are given by $S_0$, $S_1$, $S_4$ and $S_6$. Thus, the conditions given by Eq. (\ref{e2}) and the Eqs. (\ref{eqnarr2}) - (\ref{eqnarr23}), as well as by those given in Eqs. (\ref{eqnarr1}) and Eqs. (\ref{eqnarr0})  have to hold too, when one replaces $\bar{h}_1^2$  with $\bar{h}_1^3$ (and vice versa) and/or $\bar{h}_3^2$ with $\bar{h}_3^3$ (and vice versa). Solving then $\tr(HP)=0$ for all operators $P$ leading to the equations mentioned above, and to Eq. (\ref{e2}), Eqs. (\ref{eqnarr2}) - (\ref{eqnarr23}), Eqs. (\ref{eqnarr1}) and Eqs. (\ref{eqnarr0}) one obtains several possible sets of conditions on $\otimes_i H_i$.
Here we will just discuss one example of a solution which is given by $H_1=H_1^{100}$, $H_2=\one/2$, $H_4=H_4^{AA0}$ and $H_3=H_3^{1AA}$. In case $\bar{h}_4^1\neq 0$ and/or $\bar{h}_4^2\neq 0$ the corresponding states are reachable via the POVM $\{1/\sqrt{2}h_4S_0,1/\sqrt{2}h_4S_6\}$, which is acting non-trivially on party $4$, from states with $H_1=H_1^{A00}$, $H_2, H_4=\one/2$ and $H_3=H_3^{1AA}$. For $\bar{h}_4^1,\bar{h}_4^2= 0$ and $H_3\neq H_3^{A00}$ one can obtain the corresponding states via $\{1/\sqrt{2}h_3S_0g_3^{-1},1/\sqrt{2}h_3S_1g_3^{-1}\}$ which is a valid two-outcome POVM for $\bar{g}_3^2, \bar{g}_3^3=0$ and $\bar{g}_3^1=\bar{h}_3^1$ and is acting non-trivially on party $3$. Note that here one could also use a POVM that involves all $S_i$ for $i\in\{0,1,4,6\}$ to reach these states. The states with $H_j=H_j^{100}$ for $j=1,3$ and $H_k\propto \one$ for $k=2,4$ are in $MES_4$ as the only possible symmetries commute with $\otimes_i H_i$ and therefore one obtains from Eq. (\ref{EqSep}) that $H_j=G_j$ $\forall j\in\{1,2,3,4\}$.\\
All other solutions can be treated analogously. In particular, one can show that all the corresponding states except the ones where $H_j=H_j^{100}$ for $j=1,3$ and $H_k\propto \one$ for $k=2,4$ can be reached via a two-outcome POVM using $S_0$ and one $S_i$ for some $i\in\{1,4,6\}$. The states that can be reached via a POVM involving only $S_0$ and one $S_i$ for some $i\in\{1,4,6\}$ have already been discussed in Lemma \ref{twooutcome0}. \\
As the whole set of symmetries is invariant under the simultaneous exchange of party 1 and 2 and of party 3 and 4 one immediately gets the case where $\bar{h}_2^1\bar{h}_4^1\neq 0$ \footnote{More precisely, note that this implies that there are no further conditions on the probabilities from the Eq. (\ref{e1}), the Eqs. (\ref{eqnarr2}) - (\ref{eqnarr23}), as well as from the sets of equations given in  Eqs. (\ref{eqnarr1}) and Eqs. (\ref{eqnarr0}), i.e. $\tr(H P)=0$ for the corresponding $P$.}, if one applies this exchange of parties to $\otimes H_i$ and renumbers the used symmetries correspondingly. \\
As mentioned before, the other cases, e.g. $\bar{h}_1^2\bar{h}_4^2\neq 0$ can be treated analogously. One observes that also for these cases any state that is reachable can be obtained via a two-outcome POVM and therefore Lemma  \ref{twooutcome0} characterizes the reachable states for these cases.\\
Let us now treat the case that the parameter of $\otimes H_i$ are chosen such that $\tr(H P)=0$ for all operators $P$ which were used to derive Eqs. (\ref{eqnarr1}),the Eqs. (\ref{e1}) - (\ref{eqnarr23}) and the Eqs. (\ref{eqnarr0}).  One obtains several possible solutions.
For each of the solutions it can be shown that the corresponding states are reachable via a two-outcome POVM (using $S_0$ and one $S_i$ for some $i\in\{1,2,3,4,5,6,7\}$) except for the case $H_i\propto \one$ $\forall i$. In particular, one constructs the  two-outcome POVMs that allow to reach the corresponding states. It is easy to see that the state with $H_i\propto \one$ $\forall i$ is in $MES_4$. Since the reachable states can be obtained via $S_0$ and $S_i$ they are already captured by Lemma \ref{twooutcome0}. In order to illustrate how one can construct the corresponding two-outcome POVMs we will discuss the following solution.
\bea \nonumber \bar{h}_1^1 = 0, \bar{h}_2^1 = 0,  \bar{h}_3^1= 0, \bar{h}_1^2= 0, \bar{h}_3^2= 0, \bar{h}_1^3= 0 \,\,\textrm{and}\,\,
   \bar{h}_4^3 = 0.\eea
Note that all other solutions can be treated analogously.
In case $H_2\neq \one/2$ one can use the POVM  $\{1/\sqrt{2}h_2S_0,1/\sqrt{2}h_2S_7\}$ which is acting non-trivially on party $2$ to reach the corresponding states from states with $G_1, G_2=\one/2$, $G_3=G_3^{00A}$ and $G_4=G_4^{AA0}$. States with $H_2=\one/2$ and $H_4\neq \one/2$ can be reached via the two-outcome POVM $\{1/\sqrt{2}h_4S_0,1/\sqrt{2}h_4S_6\}$ which is acting non-trivially on party $4$. Note that here one could also construct a successful 4-outcome POVM that uses $S_0$, $S_3$, $S_6$ and $S_7$. States with $H_2, H_4=\one/2$ and $H_3\neq \one/2$ are reachable from the corresponding seed state via the POVM $\{1/\sqrt{2}h_3S_0,1/\sqrt{2}h_3S_1\}$ which is acting non-trivially on party $3$. Note here that one could also construct a POVM that uses all the symmetries $S_i$ for $i\in\{0,1,2,3,4,5,6,7\}$ and that allows to obtain these states from some other state.\\
Note that hereby we have discussed for various choices of $P$ all possible solution of Eq. (\ref{eq_Pauli}) , i.e. we considered the case that either $\prod \eta_{i,j,k,l}=1$ or $\tr(H P)=0$, which completes the proof. $\square$

\section*{\label{appB}Appendix B: Proof of Lemmas \ref{lemcoplanar1} and \ref{lemcoplanar2}}

In this section we prove Lemma \ref{lemcoplanar1} and Lemma \ref{lemcoplanar2}. In order to improve readability we also repeat them here.
\\ \\
\noindent\textit{ {\bf Lemma \ref{lemcoplanar1}.} Let $\{\v_1,\v_2,\v_3\}$ be a triplet of coplanar pairwise-nonparallel vectors with scalar products $c_{ij}=\v_i\cdot\v_j$. Then, the only other triplets of coplanar vectors $\{\v'_1,\v'_2,\v'_3\}$ with the same scalar products ($\v'_i\cdot\v'_j=c_{ij}$) fulfilling $||\v'_i||\geq||\v_i||$ $\forall i$ are uniform rotations of the triplet of vectors $\{\v_1,\v_2,\v_3\}$.}

\begin{proof}
Since the vectors are nonparallel (and nonzero) it must hold that $\v_3=\alpha\v_1+\beta\v_2$ for some fixed real numbers $\alpha,\beta\neq0$ for every possible triplet. Thus, the length of the vectors for fixed $c_{ij}$ can be expressed as a function of $\alpha$ and $\beta$ as $||\v_i||^2=f_i(\alpha,\beta)$, where
\begin{align}
f_1(\alpha,\beta)&=\frac{c_{13}-\beta c_{12}}{\alpha},\nonumber\\
f_2(\alpha,\beta)&=\frac{c_{23}-\alpha c_{12}}{\beta},\nonumber\\
f_3(\alpha,\beta)&=\alpha c_{13}+\beta c_{23}.
\end{align}
The gradient of these functions is then given by \footnote{Notice that by the constraints of our setting these gradients never vanish.}
\begin{align}
\nabla f_1(\alpha,\beta)&=\frac{1}{\alpha^2}(\beta c_{12}-c_{13},-\alpha c_{12}),\nonumber\\
\nabla f_2(\alpha,\beta)&=\frac{1}{\beta^2}(-\beta c_{12},\alpha c_{12}-c_{23}),\nonumber\\
\nabla f_3(\alpha,\beta)&=(c_{13},c_{23}).
\end{align}
Let $(\alpha^*,\beta^*)$ be the point corresponding to the initial triplet $\{\v_1,\v_2,\v_3\}$ and consider now an arbitrary point $(\alpha',\beta')$ corresponding to any other triplet $\{\v'_1,\v'_2,\v'_3\}$ \footnote{Since we are requiring that $||\v'_i||\geq||\v_i||$ $\forall i$ and scalar products must be preserved, it is clear that if such a triplet exists, it must also have pairwise-nonparallel nonzero vectors and hence there must exist nonvanishing $\alpha',\beta'$ such that $\v'_3=\alpha'\v'_1+\beta'\v'_2$.}. Let moreover ${\bf x}=(x,y)$ be the the direction connecting $(\alpha^*,\beta^*)$ and $(\alpha',\beta')$. To prove that it is impossible that $||\v'_i||\geq||\v_i||$ $\forall i$, we will show in the following that it cannot hold that $f_i(\alpha',\beta')\geq f_i(\alpha^*,\beta^*)$ $\forall i$ (unless $f_i(\alpha',\beta')= f_i(\alpha^*,\beta^*)$ $\forall i$, i.\ e.\ $||\v'_i||=||\v_i||$ $\forall i$, which accounts for uniform rotations of the triplet of vectors).
  Since $\nabla f_3$ is constant, to obtain a vector $\v'_3$ such that $||\v'_3||\geq||\v_3||$ imposes that
\begin{equation}
{\bf x}\cdot\nabla f_3=xc_{13}+yc_{23}\geq0.
\end{equation}
On the other hand, for every point $(\alpha,\beta)$ it holds that
\begin{equation}
\alpha^2{\bf x}\cdot\nabla f_1(\alpha,\beta)+\beta^2{\bf x}\cdot\nabla f_2(\alpha,\beta)=-xc_{13}-yc_{23}\leq0,
\end{equation}
where we have used the previous equation to obtain the inequality. This inequality is fulfilled if $f_1$ and $f_2$ are monotonously decreasing in the line connecting $(\alpha^*,\beta^*)$ and $(\alpha',\beta')$, which agrees with our claim. If the functions are not constant, the only remaining possibility to satisfy the inequality is then that in every point in the line from $(\alpha^*,\beta^*)$ to $(\alpha',\beta')$ in which $f_1$ increases, then $f_2$ must decrease and viceversa. Thus, our claim is also fulfilled if the functions are monotonous in this direction. In fact, it can be seen that this must be indeed the case \footnote{This can also be seen using that in every fixed direction, $\beta=m\alpha+k$, $f_1$ and $f_2$ correspond to
 
$$f_1(\alpha)=\frac{c_{13}-m\alpha c_{12}-kc_{12}}{\alpha},\quad f_2(\alpha)=\frac{c_{23}-c_{12}\alpha}{m\alpha+k},$$
which leads to
$$f'_1(\alpha)=\frac{kc_{12}-c_{13}}{\alpha^2},\quad f'_2(\alpha)=-\frac{kc_{12}+mc_{23}}{(m\alpha+k)^2}.$$
}. Therefore, the only possible way to avoid that one of the functions decreases is to keep the three of them constant, i.\ e.\ $xc_{13}+yc_{23}=0={\bf x}\cdot\nabla f_i(\alpha,\beta)$ $\forall i$, which means that $||\v'_i||=||\v_i||$ $\forall i$. However, it is straightforward to see that this can only hold under uniform rotations of the three vectors, i.\ e.\ $\alpha$ and $\beta$ cannot change.
\end{proof}
In the following we prove Lemma \ref{lemcoplanar2}.
\\ \\
\noindent\textit{ {\bf Lemma \ref{lemcoplanar2}.} Let $\{\v_1,\v_2,\v_3\}$ be a triplet of coplanar vectors with scalar products $c_{ij}=\v_i\cdot\v_j$ in which $\v_i\|\v_j$ for some $i\neq j$ or $\v_i=0$ for at least one value of $i$. Then, any other triplet of coplanar vectors $\{\v'_1,\v'_2,\v'_3\}$ with the same scalar products ($\v'_i\cdot\v'_j=c_{ij}$) fulfilling $||\v'_i||\geq||\v_i||$ $\forall i$, must fulfill that $\v'_i\|\v'_j$ or $\v'_i=0$.}

\begin{proof}
The case in which one vector is zero, say $\v_1=0$, is trivial as from the fact that $c_{12}=c_{13}=0$ one readily obtains the conclusion. Thus, we consider the case in which at least two vectors are parallel, which we take to be $\v_1$ and $\v_2$ without loss of generality. The parametrization used in the proof of the previous lemma does not hold in this case, so we are forced to use a different one. Let us consider in principle arbitrary coplanar vectors and denote by $\theta_{ij}$ the angles formed by the vectors $\v_i$ and $\v_j$. The coplanarity condition imposes that $\theta_{13}=\theta_{12}+\theta_{23}$ and since $c_{ij}=||\v_i||||\v_j||\cos\theta_{ij}$, we obtain that $||\v_i||^2=g_i(\theta_{12},\theta_{23})$, where
\begin{align}
g_1(\theta_{12},\theta_{23})&=c_1\frac{\cos\theta_{23}}{\cos\theta_{12}\cos(\theta_{12}+\theta_{23})},\nonumber\\
g_2(\theta_{12},\theta_{23})&=c_2\frac{\cos(\theta_{12}+\theta_{23})}{\cos\theta_{12}\cos\theta_{23}},\nonumber\\
g_3(\theta_{12},\theta_{23})&=c_3\frac{\cos\theta_{12}}{\cos\theta_{23}\cos(\theta_{12}+\theta_{23})},
\end{align}
where $c_i=c_{ij}c_{ik}/c_{jk}$. These functions are well defined except when one vector is orthogonal to another. However, in our case it is clear that when $\v_3\perp\v_1\|\v_2$, then, in order to fulfill $c_{13}=c_{23}=0$, it must also hold that $\v'_3\perp\v'_1\|\v'_2$ in agreement with our claim. Therefore, in the cases that remain to be studied the functions $\{g_i\}$ are all continuous (and differentiable). Our initial triplet must correspond to a point $(\theta_{12},\theta_{23})=(0,\theta^*_{23})$ since $\v_1\|\v_2$. In order to find another triplet $\{\v'_1,\v'_2,\v'_3\}$ corresponding to the point $(\theta'_{12},\theta'_{23})$ with $\theta_{12}'\neq0$, we know from the previous proof that from any point in the neighbourhood of $(0,\theta^*_{23})$ such that $\theta_{12}\neq0$, there must exist a path towards $(\theta'_{12},\theta'_{23})$ (keeping $\theta_{12}\neq0$) in which the functions $\{g_i\}$ are all monotonous and one must be decreasing (say $g_i$). By continuity of $g_i$ and its partial derivatives, it must hold then that $g_i(0,\theta^*_{23})>g_i(\theta'_{12},\theta'_{23})$. Hence, the only possibility to find a point $(\theta'_{12},\theta'_{23})$ compatible with the condition $||\v'_i||\geq||\v_i||$ $\forall i$ is to stay in the line $\theta_{12}=0$ (or $\theta_{23}=0$ if $\theta_{23}^*=0$ too) \footnote{Indeed,
$$\nabla g_1(0,\theta^*_{23})=\frac{c_1}{2\cos^2\theta_{23}^*}(\sin(2\theta_{23}^*),0),$$
$$\nabla g_2(0,\theta^*_{23})=-\frac{c_2}{2\cos^2\theta_{23}^*}(\sin(2\theta_{23}^*),0).$$
Since sign$(c_1)$=sign$(c_2)$, it follows that sign$({\bf x}\cdot\nabla g_1(0,\theta^*_{23}))=-$sign$({\bf x}\cdot\nabla g_1(0,\theta^*_{23}))$ unless ${\bf x}$ is proportional to $(0,1)$.}.
\end{proof}

\section*{\label{appC}Appendix C: Proof of Lemmas \ref{lemmaaacc} and \ref{convlemmaaacc}}

In this section we prove Lemma \ref{lemmaaacc}, i.e. we show which states in the SLOCC classes $G_{abcd}$ (see \cite{slocc4}) with  $a^2= d^2\neq\pm c^2$, $c^2= b^2$ and $a, c\neq 0 $ are reachable. Furthermore, we proof which states in these SLOCC classes are convertible.
As already stated in the main text the corresponding symmetries are given by $S_{z_1,z_2,m}=X^m(P_{z_1}\otimes P_{z_2})^{\otimes 2}$ where $m\in\{0,1\}$, $z_1,z_2\in\C\backslash0$ and $X=\sigma_x^{\otimes 4}$ and the standard form is given by $g_1^x\otimes g_2^x\otimes g_3\otimes g_4|\Psi_{aacc}\rangle$, where $g_1^2\geq 0$  and $g_2^2\geq 0$. In case $g_1^2=0$ ($g_2^2=0$) we choose $g_3^2\geq 0$ ($g_4^2\geq 0$) respectively. In order to make the proof more readable we did not use the symmetry $\sigma_x^{\otimes 4}$ for the definition of the standard form and therefore this standard form is not unambiguously.
Using this standard form we will proof the following lemma.\\ \\
\noindent\textit{ {\bf Lemma \ref{lemmaaacc}.} The only states in the SLOCC classes $G_{abcd}$ where  $a^2= d^2\neq\pm c^2$, $c^2= b^2$ and $a, c\neq 0 $ that are reachable via LOCC are given by  $\otimes h_i |\Psi_{aacc}\rangle$ where $\otimes H_i$ obeys the following condition. There exists a unitary symmetry $S_{z_1,z_2,m}=\otimes_i s_i^{z_1,z_2,m}\in S(\Psi_{aacc})$ such that exactly
three out of the four operators $H_i$ commute with the corresponding
operator $s_i^{z_1,z_2,m}$ and one operator $H_j$ does not commute with $s_j^{z_1,z_2,m}$.}
\begin{proof}

We first show that the states given in Lemma \ref{lemmaaacc} are reachable by providing the corresponding LOCC protocols.  Then, we show that the remaining states can not be reached via SEP and therefore also not via LOCC. In particular, we show that for $h_j^1=h_j^3=1/2$ for $j\in\{3,4\}$ only unitary symmetries can possibly be used for transformations and that by using only unitary symmetries states of the form $h_1^x\otimes h_2^x\otimes h_3^x\otimes h_4^x|\Psi_{aacc}\rangle$ with either $h_1^2,h_3^2\neq 0$ or $h_1^2=h_3^2=0$  and where either $h_2^2,h_4^2\neq 0$ or $h_2^2=h_4^2=0$  can not be reached.
Furthermore, we show that states of the form $h_1^x\otimes h_2^x\otimes h_3\otimes h_4|\Psi_{aacc}\rangle$ for which it holds that $h_j\neq h_j^x$ for $j\in\{3,4\}$ and where either $h_1^2,h_3^2\neq 0$ or $h_1^2=h_3^2=0$, as well as either $h_2^2,h_4^2\neq 0$ or $h_2^2=h_4^2=0$ are in $MES_4$.
States of the form $h_1^x\otimes h_2^x\otimes h_j\otimes h_k^x|\Psi_{aacc}\rangle$ where $h_j\neq h_j^x$ and $\{j,k\}=\{3,4\}$ can be reached from states of the form  $g_1^x\otimes g_2^x\otimes g_3^x\otimes g_4^x|\Psi_{aacc}\rangle$ via the POVM $\{1/\sqrt{2} (h_j(g_j^x)^{-1}\otimes \one^{\otimes 3}),1/\sqrt{2} (h_j\sigma_x(g_j^x)^{-1}\otimes \sigma_x^{\otimes 3})\}$ where $\Re(h_j^2)=g_j^2$.
States of the form $h_1^x\otimes h_2^x\otimes h_3\otimes h_4|\Psi_{aacc}\rangle$ where $h_j\neq d_j$ and $h_k=d_k$ for $\{j,k\}=\{1,3\}$ can be obtained via LOCC from states of the form $\one/2\otimes g_2^x\otimes d_3\otimes g_4|\Psi_{aacc}\rangle$. In order to see this consider  $\{1/\sqrt{2} h_j(d_j)^{-1},1/\sqrt{2} h_j\sigma_z(d_j)^{-1}\}$ with $d_j=diag(\sqrt{h_j^1},\sqrt{h_j^3})$. In case of the first (second) outcome, party $k$ does nothing (applies $\sigma_z$) respectively, whereas for both outcomes party $2$ and $4$ do nothing. Analogously, it can be shown that states of the form $h_1^x\otimes h_2^x\otimes h_3\otimes h_4|\Psi_{aacc}\rangle$ where $h_j\neq d_j$ and $h_k=d_k$ for $\{j,k\}=\{2,4\}$ can be reached via LOCC from states of the form $g_1^x\otimes \one/2\otimes g_3\otimes d_4|\Psi_{aacc}\rangle$. In particular, one can use the symmetries $\one^{\otimes 4}$ and $\one\otimes \sigma_z\otimes \one\otimes \sigma_z$ to construct a POVM.
States of the form  $\one/2\otimes h_2^x\otimes h_3^x\otimes h_4^x|\Psi_{aacc}\rangle$ where $h_3^x\neq \one/2$ are reachable from states of the form  $\one/2\otimes g_2^x\otimes \one/2 \otimes g_4^x|\Psi_{aacc}\rangle$. In order to see this consider the POVM $\{1/\sqrt{2}(\one^{\otimes 2}\otimes  h_3^x\otimes \one),1/\sqrt{2} (\sigma_z\otimes\one \otimes h_3^x\sigma_z\otimes \one)\}$ which is acting non-trivially  on party $3$. Analogously, one can use the symmetries $\one^{\otimes 4}$ and $\one\otimes \sigma_z\otimes \one\otimes \sigma_z$ to construct a POVM that allows to reach states of the form $h_1^x\otimes \one/2\otimes h_3^x\otimes h_4^x |\Psi_{aacc}\rangle$ where $h_4^x\neq \one/2$ from states of the form $g_1^x\otimes \one/2\otimes g_3^x \otimes \one/2|\Psi_{aacc}\rangle$.\\
As we will show in the following states of the form $h_1^x\otimes h_2^x\otimes h_3^x\otimes h_4^x|\Psi_{aacc}\rangle$, as well as states of the form $h_1^x\otimes h_2^x\otimes h_3\otimes h_4|\Psi_{aacc}\rangle$ with $h_j\neq h_j^x$ for $j\in\{3,4\}$ are not reachable if either $h_1^2,h_3^2\neq 0$ or $h_1^2=h_3^2=0$  and either $h_2^2,h_4^2\neq 0$ or $h_2^2=h_4^2=0$.
Inserting the symmetries $S_{z_1,z_2,m}$ in Eq. (\ref{EqSep}) one obtains
\begin{eqnarray}\label{Eqsepaacc}&\sum_{z_1,z_2,m} p_{z_1,z_2,m}  \left(
    \begin{array}{cccc}
       & |z_1|^{2(-1)^m}/2&  h_1^2e^{-i2\phi_1(-1)^m} \\
     &h_1^2e^{i2\phi_1(-1)^m}   & 1/(2|z_1|^{2(-1)^m})\\
     \end{array}
  \right)\otimes
\left(
    \begin{array}{cccc}
       &  |z_2|^{2(-1)^m}/2&  h_2^2e^{-i2\phi_2(-1)^m}  \\
     & h_2^2e^{i2\phi_2(-1)^m}  &1/(2|z_2|^{2(-1)^m})\\
     \end{array}
  \right)
\\\nonumber &\otimes\left(
    \begin{array}{cccc}
       & h_3^{1+2m} |z_1|^{2(-1)^m}&   |h_3^2|e^{-i(2\phi_1-\theta_3)(-1)^m} \\
     &|h_3^2|e^{i(2\phi_1-\theta_3)(-1)^m}& h_3^{3-2m} 1/|z_1|^{2(-1)^m}\\
     \end{array}
  \right)\otimes
\left(
    \begin{array}{cccc}
       &  h_4^{1+2m} |z_2|^{2(-1)^m}&   |h_4^2|e^{-i(2\phi_2-\theta_4)(-1)^m} \\
     &|h_4^2|e^{i(2\phi_2-\theta_4)(-1)^m}   & h_4^{3-2m} 1/|z_2|^{2(-1)^m}\\
     \end{array}
  \right)\\\nonumber &=
r G_1^x\otimes G_2^x \otimes G_3\otimes G_4,\end{eqnarray}
where $z_j=|z_j|e^{i\phi_j}$ for $j\in\{1,2\}$ and $h_l^2=|h_l^2|e^{i\theta_l}$ for $l\in\{3,4\}$. Let us first show that states of the form $h_1^x\otimes h_2^x\otimes h_3^x\otimes h_4^x|\Psi_{aacc}\rangle$ with either $h_1^2,h_3^2\neq 0$ or $h_1^2=h_3^2=0$  and either $h_2^2,h_4^2\neq 0$ or $h_2^2=h_4^2=0$ are in $MES_4$. Considering the matrix elements $\ket{jklw}\bra{jkw}$ of Eq. (\ref{Eqsepaacc}) for $\{j,k,l,w\}=\{0,0,1,1\}$ with $j\neq l$ and $k\neq w$ leads to
\bea
1/4=rg_3^{1+2s}g_4^{1+2t} \,\,\, \forall s,t\in\{0,1\}.
\eea
Considering this equation for different values of $s$ and $t$ and using the normalization condition, i.e. $\tr(G_i)=1$ it is easy to show that $g_3^1=g_3^3=g_4^1=g_4^3=1/2$. Note that this implies that $r=1$. Considering the matrix elements $\ket{1110}\bra{1110}$ and $\ket{0100}\bra{0100}$ of Eq. (\ref{Eqsepaacc}) one obtains that $\sum_{z_1,z_2,m} p_{z_1,z_2,m} |z_1|^{4(-1)^m}=\sum_{z_1,z_2,m} p_{z_1,z_2,m} 1/|z_1|^{4(-1)^m}=1$ and therefore $\sum_{z_1,z_2,m} p_{z_1,z_2,m} (|z_1|^{2}-1/|z_1|^{2})^2=0$ which can only be fulfilled if  $p_{z_1,z_2,m}=0$ for $|z_1|\neq 1$. Via an analogous argument one obtains that $p_{z_1,z_2,m}=0$ for $|z_2|\neq 1$. Thus, for these states only unitary symmetries could be possibly used in order to reach them via SEP transformations. Considering the matrix elements $\ket{1000}\bra{0000}$, $\ket{0010}\bra{0000}$ and $\ket{0010}\bra{1000}$ of Eq. (\ref{Eqsepaacc}) one obtains that
\bea
h_i^2\sum_{z_1,z_2,m} p_{z_1,z_2,m} e^{-i2\phi_1(-1)^m}= g_i^2\,\,\, \textrm{for}\,\,\, i\in\{1,3\}
\eea
and
\bea\label{aacc3}
h_1^2h_3^2= g_1^2g_3^2.
\eea
Hence, for $h_1^2=h_3^2=0$ it follows that $g_1^2=g_3^2=0$. For the case $h_1^2,h_3^2\neq 0$\footnote{Note that in this case $\sum_{z_1,z_2,m} p_{z_1,z_2,m} e^{-i2\phi_1(-1)^m}\neq 0$ as otherwise $g_i^2=0$ for $i\in\{1,3\}$ which would imply $h_1^2h_3^2=0$ due to Eq. (\ref{aacc3}).} one obtains that $h_1^2/h_3^2=g_1^2/g_3^2$ which together with Eq. (\ref{aacc3}) implies that $h_1^2=g_1^2$ and $h_3^2=g_3^2$. Analogously, one can show that for states of the form $h_1^x\otimes h_2^x\otimes h_3^x\otimes h_4^x|\Psi_{aacc}\rangle$ with either $h_2^2,h_4^2\neq 0$ or $h_2^2=h_4^2=0$ it has to hold that $h_2^2=g_2^2$ and $h_4^2=g_4^2$. Combining these results one obtains that $H_i^x=G_i^x$ $\forall i\in\{1,2,3,4\}$ and thus states of the form $h_1^x\otimes h_2^x\otimes h_3^x\otimes h_4^x|\Psi_{aacc}\rangle$ with either $h_1^2,h_3^2\neq 0$ or $h_1^2=h_3^2=0$  and either $h_2^2,h_4^2\neq 0$ or $h_2^2=h_4^2=0$ are in $MES_4$.\\
Let us now show that states of the form $h_1^x\otimes h_2^x\otimes h_3\otimes h_4|\Psi_{aacc}\rangle$ with $h_j\neq h_j^x$ for $j\in\{3,4\}$ and either $h_1^2,h_3^2\neq 0$ or $h_1^2=h_3^2=0$  and either $h_2^2,h_4^2\neq 0$ or $h_2^2=h_4^2=0$ are not reachable.
Considering the matrix elements $\ket{jklw}\bra{jkw}$ of Eq. (\ref{Eqsepaacc}) for $\{j,k,l,w\}=\{0,0,1,1\}$ with $j\neq l$ and $k\neq w$ leads to
\bea\label{aacc4}
 \tilde{p}_0 h_3^{1+2s}h_4^{1+2t}+\tilde{p}_1 h_3^{3-2s}h_4^{3-2t}=rg_3^{1+2s}g_4^{1+2t},
\eea
where $s,t\in\{0,1\}$ and $\sum_{z_1,z_2} p_{z_1,z_2,m}=\tilde{p}_m$. Adding the equations for all possible combinations of $s$ and $t$  one obtains that $r=1$. From Eq. (\ref{aacc4}) for $s=0$ ($t=0$) and $t\in\{0,1\}$ ($s\in\{0,1\}$) resp. one obtains that
\bea\label{aacc6}
 \tilde{p}_0 h_3^{1}+\tilde{p}_1 h_3^{3}=g_3^{1}
\eea
and
\bea\label{aacc7}
\tilde{p}_0 h_4^{1}+\tilde{p}_1 h_4^{3}=g_4^{1}.
\eea
Using these equations, as well as Eq. (\ref{aacc4}) for $s=t=0$ and $\tilde{p}_0 +\tilde{p}_1=1$ it is easy to show that either $\tilde{p}_0\tilde{p}_1=0$ (solution (i)) or $h_3^{1}= h_3^{3}$ (solution (ii)) or $h_4^{1}= h_4^{3}$ (solution (iii)).
For solution (i) either $\tilde{p}_0=0$ or $\tilde{p}_1=0$. Note that these two possibilities coincide up to LUs and as we did not use $\sigma_x^{\otimes 4}$ to fix the standard form we only have to consider  transformation with $\tilde{p}_1=0$, i.e. transformations using the symmetries $(P_{z_1}\otimes P_{z_2})^{\otimes 2}$.  For $\tilde{p}_1=0$ it immediately follows from Eq. (\ref{aacc6}) and Eq. (\ref{aacc7}) that $h_i^1=g_i^1$ for $i\in\{3,4\}$. Due to the normalization condition, i.e. $\tr(H_i)=\tr(G_i)$ we also have that $h_i^3=g_i^3$ for $i\in\{3,4\}$. Analogously to before one can show that in this case $p_{z_1,z_2,m}=0$ for $|z_1|\neq 1$ and/or $|z_2|\neq 1$.
Considering the matrix elements $\ket{1000}\bra{0000}$, $\ket{0010}\bra{0000}$ and $\ket{0010}\bra{1000}$ of Eq. (\ref{Eqsepaacc}) one obtains that
\bea
h_i^2\sum_{z_1,z_2} p_{z_1,z_2,0} e^{-i2\phi_1}= g_i^2\,\,\, \textrm{for} \,\,\,i\in\{1,3\}
\eea
and
\bea
h_1^{2}h_3^2= g_1^{2}g_3^2.
\eea
Using that $h_1^2, g_1^2\geq 0$ and a similar argument as before, one can show that $h_i^2=g_i^2$ for $i\in\{1,3\}$ and analogously for party $2$ and $4$. Thus, the symmetries $(P_{z_1}\otimes P_{z_2})^{\otimes 2}$ do not allow to reach states of the form $h_1^x\otimes h_2^x\otimes h_3\otimes h_4|\Psi_{aacc}\rangle$ with $h_j\neq h_j^x$ for $j\in\{3,4\}$ and either $h_1^2,h_3^2\neq 0$ or $h_1^2=h_3^2=0$  and either $h_2^2,h_4^2\neq 0$ or $h_2^2=h_4^2=0$.
Let us consider next solution (ii), i.e. $h_3^{1}= h_3^{3}$. Note that due to the normalization condition we have that $h_3^{1}= h_3^{3}=1/2$. Thus, from Eq. (\ref{aacc6}) one obtains that $g_3^{1}=1/2$ and therefore also that $g_3^{3}=1/2$. Considering the matrix elements $\ket{1110}\bra{1110}$, $\ket{0100}\bra{0100}$, $\ket{1011}\bra{1011}$ and $\ket{0001}\bra{0001}$ of Eq. (\ref{Eqsepaacc}) one obtains that
\begin{eqnarray}
\sum_{z_1,z_2} p_{z_1,z_2,0} |z_1|^{4}h_4^{1+kl}+\sum_{z_1,z_2} p_{z_1,z_2,1} 1/|z_1|^{4}h_4^{3-2k}=g_4^{1+2k}
\end{eqnarray}
and
\begin{eqnarray}
\sum_{z_1,z_2} p_{z_1,z_2,0} 1/|z_1|^{4}h_4^{1+2l}+\sum_{z_1,z_2} p_{z_1,z_2,1} |z_1|^{4}h_4^{3-2l}=g_4^{1+2l},
\end{eqnarray}
where $k,l\in\{0,1\}$. Adding these equations for $l=0$, $l=1$, $k=0$ and  $k=1$ one obtains that $\sum_{z_1,z_2,m} p_{z_1,z_2,m} (|z_1|^{2}-1/|z_1|^{2})^2=0$ and therefore $p_{z_1,z_2,m}=0$ for $|z_1|\neq 1$. Note that $h_3^{1}= h_3^{3}=1/2$ and $h_3\neq h_3^x$ implies that $\Im(h_3^2)\neq 0$. Note further that since for the class of states we are discussing here it also has to hold that either $h_1^2,h_3^2\neq 0$ or $h_1^2=h_3^2=0$, we also have that $h_1^2\neq 0$.
Considering the matrix elements $\ket{1001}\bra{0001}$ and $\ket{1100}\bra{0100}$ of Eq. (\ref{Eqsepaacc}) one obtains that
\bea\label{aacc8}
h_1^2(\sum_{z_1,z_2} p_{z_1,z_2,0}e^{-i2\phi_1}h_4^1+p_{z_1,z_2,1}e^{i2\phi_1}h_4^3)=g_1^2h_4^1
\eea
and
\bea\label{aacc9}
h_1^2(\sum_{z_1,z_2} p_{z_1,z_2,0}e^{-i2\phi_1}h_4^3+p_{z_1,z_2,1}e^{i2\phi_1}h_4^1)=g_1^2h_4^3.
\eea
Taking the imaginary part of these equations one obtains that $\sum_{z_1,z_2} p_{z_1,z_2,0}\sin(2\phi_1)h_4^1=\sum_{z_1,z_2} p_{z_1,z_2,1}\sin(2\phi_1)h_4^3$ and $\sum_{z_1,z_2} p_{z_1,z_2,0}\sin(2\phi_1)h_4^3=\sum_{z_1,z_2} p_{z_1,z_2,1}\sin(2\phi_1)h_4^1$. Thus, either $h_4^3=h_4^1=1/2$ (solution (a)) or $\sum_{z_1,z_2} p_{z_1,z_2,0}\sin(2\phi_1)=\sum_{z_1,z_2} p_{z_1,z_2,1}\sin(2\phi_1)=0$ (solution (b)). For solution (a) one can show analogously to before that $g_4^3=g_4^1=1/2$ and $p_{z_1,z_2,m}=0$ for $|z_2|\neq 1$. Thus, in case $h_3^3=h_3^1=h_4^3=h_4^1=1/2$ only unitary symmetries can be used for SEP transformations. Let us first discuss solution (a) in more detail. Note that since we are discussing states here with the property that $h_4\neq h_4^x$ and either $h_2^2,h_4^2\neq 0$ or $h_2^2=h_4^2=0$ it follows from $h_4^3=h_4^1=1/2$ that $\Im(h_4^2)\neq 0$ and therefore also $h_2^2\neq 0$.  Considering the matrix elements $\ket{0010}\bra{1000}$ and $\ket{0001}\bra{0100}$ of Eq. (\ref{Eqsepaacc}) and taking the imaginary part one obtains that
\begin{eqnarray}\label{aacc10}
h_i^2 \Im(h_j^2) (\tilde{p}_0-\tilde{p}_1)=g_i^2 \Im(g_j^2).
\end{eqnarray}
where $\{i,j\}\in\{\{1,3\},\{2,4\}\}$. Using the matrix elements $\ket{0110}\bra{1001}$, $\ket{0011}\bra{1100}$, $\ket{1001}\bra{0110}$, and $\ket{1100}\bra{0011}$ of Eq. (\ref{Eqsepaacc}) one can easily show that
\bea
h_1^2 h_2^2 \Im(h_3^2) \Im(h_4^2)=g_1^2g_2^2 \Im(g_3^2)  \Im(g_4^2).
\eea
As $h_1^2,h_2^2,\Im(h_3^2),\Im(h_4^2)\neq 0$ one obtains that $(\tilde{p}_0-\tilde{p}_1)^2=1$ and therefore either $\tilde{p}_0=0$ or $\tilde{p}_1=0$. As we have already pointed out these two cases coincide up to LUs. Moreover, we have already shown that transformations using just symmetries with $m=0$ do not allow to reach states of the form $h_1^x\otimes h_2^x\otimes h_3\otimes h_4|\Psi_{aacc}\rangle$ with $h_j\neq h_j^x$ for $j\in\{3,4\}$ and either $h_1^2,h_3^2\neq 0$ or $h_1^2=h_3^2=0$  and either $h_2^2,h_4^2\neq 0$ or $h_2^2=h_4^2=0$ (see solution (i)). \\
Let us now proceed with solution (b), i.e. $\sum_{z_1,z_2} p_{z_1,z_2,0}\sin(2\phi_1)=\sum_{z_1,z_2} p_{z_1,z_2,1}\sin(2\phi_1)=0$. Adding  Eq. (\ref{aacc8}) and Eq. (\ref{aacc9}) one obtains that
\begin{equation}\label{aacc14}
h_1^2\sum_{z_1,z_2,m} p_{z_1,z_2,m}\cos(2\phi_1)=g_1^2.\end{equation}
Considering the matrix elements $\ket{0110}\bra{0100}$ and $\ket{0110}\bra{1100}$,  $\ket{0011}\bra{0001}$, and $\ket{0011}\bra{1001}$ of Eq. (\ref{Eqsepaacc}) leads to
\begin{eqnarray}
\label{aacc11}
&\Im(h_3^2)\sum_{z_1,z_2} (p_{z_1,z_2,0} h_4^{1+2l}-p_{z_1,z_2,1}h_4^{3-2l})\cos(2\phi_1)=\Im(g_3^2)g_4^{1+2l}
\end{eqnarray}
and
\begin{eqnarray}\label{aacc12}
h_1^2 \Im(h_3^2) (\tilde{p}_0h_4^{1+2l}-\tilde{p}_1h_4^{3-2l})=g_1^2 \Im(g_3^2)g_4^{1+2l},
\end{eqnarray}
where $l\in\{0,1\}$. Using Eq. (\ref{aacc7}), $\tilde{p}_0+\tilde{p}_1=1$ and the normalization condition it is easy to show that
\bea\label{aacc13}
\tilde{p}_0 h_4^{3}+\tilde{p}_1 h_4^{1}=g_4^{3}.
\eea
Using this equation and Eq. (\ref{aacc11}) for $l=0$ and $l=1$, as well as Eq. (\ref{aacc7}) one obtains that for $\sum_{z_1,z_2} (p_{z_1,z_2,0} h_4^{3}-p_{z_1,z_2,1}h_4^{1})\cos(2\phi_1)\neq 0$
\bea\label{aacc15}
\frac{\bar{p}_0h_4^{1}-\bar{p}_1h_4^{3}}{\bar{p}_0h_4^{3}-\bar{p}_1h_4^{1}}=\frac{\tilde{p}_0h_4^{1}+\tilde{p}_1h_4^{3}}{\tilde{p}_0h_4^{3}+\tilde{p}_1h_4^{1}},
\eea
where here and in the following $\sum_{z_1,z_2} p_{z_1,z_2,m} \cos(2\phi_1)=\bar{p}_m$. Note that here we used that $\Im(g_3^2)=0$ would imply that $\bar{p}_0h_4^{3}-\bar{p}_1h_4^{1}=0$, since $\Im(h_3^2)\neq 0$. Note further that by using Eq. (\ref{aacc11}) it is easy to show that for $\bar{p}_0h_4^{3}-\bar{p}_1h_4^{1}=0$ it also has to hold that $\bar{p}_0h_4^{1}-\bar{p}_1h_4^{3}=0$  which can only be fulfilled by $\bar{p}_0=\bar{p}_1=0$ or $h_4^{1}=h_4^{3}=1/2$ and $\bar{p}_0=\bar{p}_1$. Note that $h_4^{1}=h_4^{3}=1/2$ corresponds to solution (a) and has already been discussed. Using Eq. (\ref{aacc14})  one obtains for $\bar{p}_0=\bar{p}_1=0$ that $g_1^2=0$ and thus Eq. (\ref{aacc12}) leads to
\begin{eqnarray}
\tilde{p}_0h_4^{1+2l}-\tilde{p}_1h_4^{3-2l}=0,
\end{eqnarray}
for $l=0,1$. These equations can only be fulfilled for $h_4^{1}=h_4^{3}=1/2$ and $\tilde{p}_0=\tilde{p}_1$. Note that we have already shown that for $h_4^{1}=h_4^{3}=1/2$ either $\tilde{p}_0=0$ or $\tilde{p}_1=0$. Thus, $\bar{p}_0=\bar{p}_1=0$ is not a possible solution. So let us proceed with the case $\bar{p}_0h_4^{3}-\bar{p}_1h_4^{1} \neq 0$. Using Eq. (\ref{aacc15}) one can show that either $h_4^{1}=h_4^{3}=1/2$  or \bea \label{aacc16}
\bar{p}_0\tilde{p}_1+\bar{p}_1\tilde{p}_0=0.\eea
Note that $h_4^{1}=h_4^{3}=1/2$ corresponds to solution (a) and we will therefore proceed with the latter case.
Note further that $\bar{p}_1=0$ implies that either $\tilde{p}_1=0$ or $\bar{p}_0=0$. These two solutions have been already discussed. Thus, we only have to consider here the case $\bar{p}_1, \tilde{p}_1\neq 0$ which leads to
\bea \label{aacc17}
\frac{\bar{p}_0}{\bar{p}_1}=-\frac{\tilde{p}_0}{\tilde{p}_1}.\eea
Using Eqs. (\ref{aacc11}) and (\ref{aacc12}) for $l=0$ and $l=1$, as well as Eq. (\ref{aacc14}) one can show that for $\Im(g_3^2),g_1^2\neq 0$ \footnote{ Note that for $\Im(g_3^2)=0$ and/or $g_1^2= 0$ it has to hold that $\tilde{p}_0h_4^{1+2l}-\tilde{p}_1h_4^{3-2l}=0$ for $l=0,1$, which we have already discussed.}
\bea
\tilde{p}_0^2-\tilde{p}_1^2=\bar{p}_0^2-\bar{p}_1^2.
\eea
From this equation and Eq. (\ref{aacc17}) it follows that $\tilde{p}_1=\mp \bar{p}_1$ and $\tilde{p}_0=\pm \bar{p}_0$. In order to get rid of the ambiguity of the sign we use Eq. (\ref{aacc14}), as well as the fact that according to our standard form  $g_1^2,h_1^2 \geq 0 $ and we choose $\tilde{p}_0\geq \tilde{p}_1$  \footnote{Note that we can always choose this order because $\tilde{p}_0\leq \tilde{p}_1$ lead to transformations that are LU-equivalent to the ones obtained for $\tilde{p}_0\geq \tilde{p}_1$  as we did not use $\sigma_x$ to fix our standard form.}.  Thus, with our choice we obtain $\tilde{p}_1=- \bar{p}_1$ and $\tilde{p}_0= \bar{p}_0$.
Using now Eq. (\ref{aacc8}) and Eq. (\ref{aacc9}), as well as Eqs. (\ref{aacc7}) and (\ref{aacc13})  one obtains for $\tilde{p}_0h_4^{3}-\tilde{p}_1h_4^{1}\neq 0$ that either $\tilde{p}_0\tilde{p}_1=0$ or $h_4^{3}=h_4^{1}=1/2$. Note that both cases have already been discussed. Note further that for $\tilde{p}_0h_4^{3}-\tilde{p}_1h_4^{1}= 0$ it has to hold that $g_1^2=0$ (see Eq. (\ref{aacc9})) and therefore it follows from Eq. (\ref{aacc8}) that $\tilde{p}_0h_4^{1}-\tilde{p}_1h_4^{3}= 0$. Note that we have discussed already the case $\tilde{p}_0h_4^{1+2l}-\tilde{p}_1h_4^{3-2l}= 0$ for $l=0,1$. Thus, states of the form $h_1^x\otimes h_2^x\otimes h_3\otimes h_4|\Psi_{aacc}\rangle$ with $h_j\neq h_j^x$ for $j\in\{3,4\}$, $h_3^1=h_3^2=1/2$ and either $h_1^2,h_3^2\neq 0$ or $h_1^2=h_3^2=0$  and either $h_2^2,h_4^2\neq 0$ or $h_2^2=h_4^2=0$ (which correspond to solution (ii)) can not be reached via LOCC.
Solution (iii), i.e. $h_4^1=h_4^3=1/2$ can be treated completely analogous to solution (ii).\\
To summarize, we have shown that for states of the form $h_1^x\otimes h_2^x\otimes h_3\otimes h_4|\Psi_{aacc}\rangle$ with $h_j\neq h_j^x$ for $j\in\{3,4\}$ and either $h_1^2,h_3^2\neq 0$ or $h_1^2=h_3^2=0$  and either $h_2^2,h_4^2\neq 0$ or $h_2^2=h_4^2=0$ only symmetries with $m=0$ (or $m=1$) could have been possibly used for SEP transformations, but these symmetries do not allow to obtain these states via SEP as has been shown for solution (i). Thus, these states are in $MES_4$.
\end{proof}
In the following we will prove Lemma \ref{convlemmaaacc}. In order to improve the readability we repeat the lemma here again.\\ \\
\noindent\textit{ {\bf Lemma \ref{convlemmaaacc}.} The only states  in the SLOCC classes $G_{abcd}$ where  $a^2= d^2\neq\pm c^2$, $c^2= b^2$ and $a, c\neq 0 $ that are convertible via LOCC are given by $\otimes g_i |\Psi_{aacc}\rangle$ where $\otimes G_i$ obeys the following condition. There exists a non-trivial unitary symmetry $S_{z_1,z_2,m}=\otimes_i s_i^{z_1,z_2,m}\in S(\Psi_{aacc})$ such that three out of the four operators $G_i$ commute with the corresponding
operator $s_i^{z_1,z_2,m}$ and one operator $G_j$ is arbitrary.}
\begin{proof}
The only states that can be reached are given by the ones in Lemma \ref{lemmaaacc}.  We will first show that states of the form $h_1^x\otimes h_2^x\otimes h_3\otimes h_4 \ket{\Psi_{adcc}}$ where $h_i\neq h_i^x$ and $h_j=h_j^x$ for $\{i,j\} = \{3,4\}$ and either $h_k^2=h_j^2=0$ or $h_k^2, h_j^2\neq 0$ for $\{j,k\} \in \{\{1,3\},\{2,4\}\}$ can only be possibly reached by states of the form $g_1^x\otimes g_2^x\otimes g_3\otimes g_4 \ket{\Psi_{adcc}}$ where $g_j=g_j^x$ and either $g_k^2=g_j^2=0$ or $g_k^2, g_j^2\neq 0$. In order to see this recall that in the proof of Lemma \ref{lemmaaacc} we have shown that states with $h_k=h_k^x$ and $h_l=h_l^x$ for some $\{k,l\} \in \{\{1,3\},\{2,4\}\}$  and either  $h_k^2=h_l^2=0$ or $h_k^2, h_l^2\neq 0$ can only be obtained from states with $h_k=g_k$ and $h_l=g_l$. In particular, we have shown that $h_i^1=h_i^3=1/2$ for some $i\in\{3,4\}$ implies that $g_i^1=g_i^3=1/2$ (see Eqs. (\ref{aacc6})  and (\ref{aacc7})). Moreover, we have shown that for $h_k^2,h_l^2\in\R$ and either $h_k^2=h_l^2=0$ or $h_k^2, h_l^2\neq 0$ for some $\{k,l\} \in \{\{1,3\},\{2,4\}\}$ it has to hold that $h_k^2=g_k^2$ and $h_l^2=g_l^2$.
To convert these states one can use the POVM $\{\sqrt{p} (h_ig_i^{-1}\otimes \one^{\otimes 3}),\sqrt{1-p} (h_i\sigma_x g_i^{-1}\otimes \sigma_x^{\otimes 3})\}$, where $\Re (h_i^2)=\Re (g_i^2)$, $ (2p-1)\Im (h_i^2)=\Im (g_i^2)$ and $(2 p-1) (h_i^1-1/2)=(g_i^1-1/2)$ for a properly chosen $p\geq 0$.\\
Let us next show that states with $h_w^2=0$ for some $w\in \{1,2,3,4\}$  can only  be reached from states with  $g_w^2=0$.
In order to see this consider the matrix element $\ket{1000}_{w,j,k,l}\bra{0000}$ of Eq. (\ref{Eqsepaacc}) for $j\neq k \neq l\neq j$, $j,k,l\neq w$ and $j,k,l\in\{1,2,3,4\}$ and use that $g_i^1\neq 0$ $\forall i\in\{1,2,3,4\}$.
Thus, states of the form  $h_1^{x}\otimes h_2^x\otimes h_3\otimes h_4 \ket{\Psi_{aacc}}$ with $h_i\neq d_i$ and $h_j=d_j$ for some $\{i,j\} \in \{\{1,3\},\{2,4\}\}$ can only be obtained from states of the form $g_1^{x}\otimes g_2^x\otimes g_3\otimes g_4 \ket{\Psi_{aacc}}$ with $h_j=d_j$.
Except for the states of the form $h_1^x\otimes h_2^x\otimes h_3\otimes h_4|\Psi_{aacc}\rangle$ where $h_i=d_i\neq \one/2$, $h_k=h_k^x$ and $h_j=\one/2$ for $\{i,j\}, \{k,l\} \in \{\{1,3\},\{2,4\}\}$, $k,l\neq i,j$, $k\in\{3,4\}$ and either $h_k^2=h_l^2=0$ or $h_k^2, h_l^2\neq 0$, they can be reached by $\{\sqrt{p} h_ig_i^{-1},\sqrt{1-p} h_i\sigma_z g_i^{-1}\}$ which corresponds to a valid POVM for $h_i^3=g_i^3$,  $h_i^1=g_i^1$ and $(2 p -1 ) h_i^2=g_i^2$. Note that in case of the second outcome party $j$ has to apply $\sigma_z$ in order to make the transformation deterministic. Note further that again for any $G_i$ one can find a $H_i$ (which is not LU-equivalent) and a value of  $p$ such that these conditions are fulfilled. States of the form $h_1^x\otimes h_2^x\otimes h_3\otimes h_4|\Psi_{aacc}\rangle$ where $h_i=d_i\neq \one/2$, $h_k=h_k^x$ and $h_j=\one/2$ for $\{i,j\}, \{k,l\} \in \{\{1,3\},\{2,4\}\}$, $k,l\neq i,j$, $k\in\{3,4\}$ and either $h_k^2=h_l^2=0$ or $h_k^2, h_l^2\neq 0$ can be obtained from states of the form $g_1^x\otimes g_2^x\otimes g_3\otimes g_4|\Psi_{aacc}\rangle$ where $g_i=\tilde{d}_i$, $g_k=g_k^x$, $g_j=\one/2$ and either $g_k^2=g_l^2=0$ or $g_k^2, g_l^2\neq 0$\footnote{Note that we used here again that states with $h_k=h_k^x$ and $h_l=h_l^x$ for some $\{k,l\} \in \{\{1,3\},\{2,4\}\}$  and either  $h_k^2=h_l^2=0$ or $h_k^2, h_l^2\neq 0$ can only be reached from states with $h_k=g_k$ and $h_l=g_l$.}.
To this end one can use the POVM $\{\sqrt{p} d_i\tilde{d}_i^{-1},\sqrt{1-p} d_i\sigma_y \tilde{d}_i^{-1}\}$ where $d_i=diag (\sqrt{1/2+\bar{h}_i^1},\sqrt{1/2-\bar{h}_i^1})$,  $\tilde{d}_i=diag (\sqrt{1/2+(2 p -1 ) \bar{h}_i^1},\sqrt{1/2-(2 p -1 ) \bar{h}_i^1})$  and in case of outcome 1 (2) all other parties do nothing (party $j$ aplies $\sigma_y$ and all parties except parties $i,j$ apply $\sigma_x$) respectively.
\end{proof}

\section*{\label{appD} Appendix D: Proof of Lemmas \ref{lemmaaaaa} and \ref{convlemmaaaaa}}

In this section we prove Lemma \ref{lemmaaaaa} and Lemma \ref{convlemmaaaaa}, i.e. we show which are the reachable and convertible states in the SLOCC class $G_{abcd}$ where  $a^2= d^2$, $c^2= b^2$, $a^2= -c^2$ and $a \neq 0$.
\\ \\
\textit{ Proof of Lemma \ref{lemmaaaaa}.}
The structure of the proof is the following. We will first show that the states given in this lemma are indeed reachable by constructing a LOCC protocol that allows to reach them. Then we we will use Eq. (\ref{EqSep}) to show that all remaining states can not be reached via SEP transformations which implies that they can not be obtained via LOCC from some other state. \\ Note first that $X^m[(\sigma_z)^k\otimes(\sigma_z)^l\otimes(\sigma_z)^k\otimes(\sigma_z)^l]\in S(\Psi)$, where $X=\sigma_x^{\otimes 4}$ and $m, k, l\in\{0,1\}$. As this symmetry is also contained in the symmetry of the SLOCC classes $G_{abcd}$ where  $a^2= d^2\neq\pm c^2$ and $c^2= b^2$ we have already provided the LOCC protocols that use this symmetry in the proof of Lemma \ref{lemmaaacc} (see Appendix C). In particular, we used these symmetries to construct LOCC protocols that allow to reach states of the form  $h_1^x\otimes h_2^x\otimes h_3\otimes h_4^x|\Psi\rangle$  where $h_3\neq h_3^x$ and $h_1^x\otimes h_2^x\otimes h_3^x\otimes h_4|\Psi\rangle$ where $h_4\neq h_4^x$ and $\one/2\otimes h_2^x\otimes h_3\otimes h_4|\Psi\rangle$ where $h_3\neq d_3$ and $h_1^x\otimes \one/2\otimes h_3\otimes h_4|\Psi\rangle$ where $h_4\neq d_4$ and $h_1^x\otimes h_2^x\otimes d_3\otimes h_4|\Psi\rangle$ where $h_1^x\neq \one/2$  and $h_1^x\otimes h_2^x\otimes h_3\otimes d_4|\Psi\rangle$ where $h_2^x\neq \one/2$.\\
Note that $\one\otimes \sigma_x\otimes \sigma_z\otimes \sigma_y\equiv S_1$ and $\sigma_x\otimes \one \otimes \sigma_y\otimes \sigma_z\equiv S_2$ are (up to some proportionality factor) elements of $S(\Psi)$. We will use these symmetries to construct LOCC protocols that allow to reach the remaining states given in Lemma \ref{lemmaaaaa}. In the following we will use the notation $S_0\equiv \one^{\otimes 4}$. States of the form $h_1^x\otimes h_2^x\otimes d_3\otimes h_4|\Psi\rangle$ where $h_4\neq h_4^y$  can be obtained from states of the form $g_1^x\otimes g_2^x\otimes d_3\otimes g_4^y|\Psi\rangle$ via the POVM $\{1/\sqrt{2} h_4S_0(g_4^y)^{-1},1/\sqrt{2} h_4S_1(g_4^y)^{-1}\}$ (where $\Im(h_4^2)=g_4^2$) which is acting non-trivially on party $4$. Analogously, using $S_0$ and $S_2$ one can construct a POVM acting non-trivially on party $3$ that allows to reach states of the form $h_1^x\otimes h_2^x\otimes h_3\otimes d_4|\Psi\rangle$ where $h_3\neq h_3^y$ from states of the form $g_1^x\otimes g_2^x\otimes g_3^y\otimes d_4|\Psi\rangle$.
States of the form $h_1^x\otimes h_2^x\otimes h_3^y\otimes h_4|\Psi\rangle$ where $h_4\neq d_4$  can for example be reached from states of the form $g_1^x\otimes g_2^x\otimes g_3^y\otimes d_4|\Psi\rangle$. In particular, the POVM $\{1/\sqrt{2} h_4S_0(d_4)^{-1}, 1/\sqrt{2}h_4S_2(d_4)^{-1}\}$ (with $d_4=diag(\sqrt{h_4^1},\sqrt{h_4^3})$) which is acting
non-trivially on party $4$ performs this task. Similarly one can construct a LOCC protocol that allows to reach states of the form $h_1^x\otimes h_2^x\otimes h_3\otimes h_4^y|\Psi\rangle$ where $h_3\neq d_3$ using the symmetries $S_0$ and $S_1$.
\\
We will next show that states of the form $h_1^x\otimes h_2^x\otimes h_3^x\otimes h_4^x|\Psi\rangle$  where $h_i^2\neq 0$ $\forall i\in\{1,2,3,4\}$ and $\one/2\otimes h_2^x\otimes d_3\otimes h_4^y|\Psi\rangle$  where $h_j^2\neq 0$ for $j\in\{2,4\}$ and $d_3\neq \one$ and $h_1^x\otimes \one/2\otimes g_3^y\otimes d_4|\Psi\rangle$  where $h_k^2\neq 0$ for $k\in\{1,3\}$ and $d_4\neq \one$ and $h_1^x\otimes h_2^x\otimes h_3\otimes h_4|\Psi\rangle$  where $h_i^2\neq 0$ $\forall i\in\{1,2,3,4\}$ and $h_l\neq h_l^x, h_l^y$ for $l\in\{3,4\}$ and $|\Psi\rangle$, i.e. all states that are not given by Lemma \ref{lemmaaaaa}, are not reachable via SEP.\\
Note first that every element of $S(\Psi)$ can be written as
\begin{eqnarray}
&[ \left(
    \begin{array}{cc}
      y_1 & 0 \\
      0 & x_1 \\
    \end{array}
  \right)\sigma_x^n\otimes\left(
    \begin{array}{cc}
      y_2 & 0 \\
      0 & x_2\\
    \end{array}
  \right)\otimes\left(
    \begin{array}{cc}
      1/y_1 & 0\\
      0 & (-1)^n/x_1 \\
    \end{array}
  \right)\sigma_x^n
\otimes\left(
    \begin{array}{cc}
      1/y_2 & 0 \\
      0 & (-1)^n/x_2 \\
    \end{array}\right)](\sigma_x^{\otimes 4})^m,
\end{eqnarray}
where $x_1, x_2, y_1, y_2\in\C\backslash 0$ and $m, n\in\{0,1\}$.
Inserting the symmetries into Eq. (\ref{EqSep}) one obtains that
\begin{eqnarray}\label{Eqsepaaaa}&\sum_{x_1,x_2,y_1,y_2,m,n} p_{m,n,x_1,x_2,y_1,y_2} \left(
    \begin{array}{cccc}
       & |y_1|^{2(-1)^{(m+n)}}/2&  h_1^2|y_1||x_1|e^{i(\alpha_1-\beta_1)(-1)^{(m+n)}} \\
     &h_1^2|y_1||x_1|e^{-i(\alpha_1-\beta_1)(-1)^{(m+n)}}    & |x_1|^{2(-1)^{(m+n)}}/2\\
     \end{array}
  \right)\\\nonumber &\otimes
\left(
    \begin{array}{cccc}
       & |y_2|^{2(-1)^{m}}/2&  h_2^2|y_2||x_2|e^{i(\alpha_2-\beta_2)(-1)^{m}} \\
     &h_2^2|y_2||x_2|e^{-i(\alpha_2-\beta_2)(-1)^{m}}    & |x_2|^{2(-1)^{m}}/2\\
     \end{array}
  \right)
\\\nonumber &\otimes\left(
    \begin{array}{cccc}
       & h_3^{[1+2(m-n)^2]}/|y_1|^{2(-1)^{(m+n)}}&   |h_3^2|(-1)^n/(|y_1||x_1|)e^{-i(\alpha_1-\beta_1-\theta_3)(-1)^{(m+n)}} \\
     &|h_3^2|(-1)^n/(|y_1||x_1|)e^{i(\alpha_1-\beta_1-\theta_3)(-1)^{(m+n)}}& h_3^{[3- 2(m-n)^2]}/|x_1|^{2(-1)^{(m+n)}}\\
     \end{array}
  \right)\\\nonumber &\otimes
\left(
    \begin{array}{cccc}
       &  h_4^{1+2m}/|y_2|^{2(-1)^m}&   |h_4^2|(-1)^n/(|y_2||x_2|) e^{-i(\alpha_2-\beta_2-\theta_4)(-1)^m} \\
     &|h_4^2|(-1)^n/(|y_2||x_2|) e^{i(\alpha_2-\beta_2-\theta_4)(-1)^m}   & h_4^{3-2m}/|x_2|^{2(-1)^m}\\
     \end{array}
  \right)\\\nonumber &=
r G_1^x\otimes G_2^x \otimes G_3\otimes G_4,\end{eqnarray}
where $m, n\in\{0,1\}$ $x_j=|x_j|e^{i\alpha_j}$ and  $y_j=|y_j|e^{i\beta_j}$ for $j\in\{1,2\}$ and $h_l^2=|h_l^2|e^{i\theta_l}$ for $l\in\{3,4\}$.
We will first consider states of the form $h_1^x\otimes h_2^x\otimes h_3^x\otimes h_4^x|\Psi\rangle$  where $h_i^2\neq 0$ $\forall i\in\{1,2,3,4\}$. Considering the matrix elements $\ket{jkjk}\bra{jkjk}$ of Eq. (\ref{Eqsepaaaa}) for $j,k\in\{0,1\}$ leads to
\bea \label{eqdiagaaaa}
1/4=rg_3^{1+2s}g_4^{1+2t} \,\,\, \forall s,t\in\{0,1\}.
\eea
Additionally using the normalization condition, i.e. $\tr(G_i)=1$ $\forall i$, one obtains that $g_3^1=g_3^3=1/2$ and  $g_4^1=g_4^3=1/2$. Note that this implies that $r=1$. Considering the matrix elements $\ket{j0k0}\bra{j0k0}$ of Eq. (\ref{Eqsepaaaa}) for $j,k\in\{0,1\}$ and $j\neq k$ leads to
\begin{eqnarray}
&\sum_{x_1,x_2,y_1,y_2} (p_{0,0,x_1,x_2,y_1,y_2}+p_{1,1,x_1,x_2,y_1,y_2}) |y_1|^2/|x_1|^2+ (p_{0,1,x_1,x_2,y_1,y_2}+p_{1,0,x_1,x_2,y_1,y_2}) |x_1|^2/|y_1|^2=1
\end{eqnarray}
and
\begin{eqnarray}
&\sum_{x_1,x_2,y_1,y_2} (p_{0,0,x_1,x_2,y_1,y_2}+p_{1,1,x_1,x_2,y_1,y_2}) |x_1|^2/|y_1|^2+ (p_{0,1,x_1,x_2,y_1,y_2}+p_{1,0,x_1,x_2,y_1,y_2}) |y_1|^2/|x_1|^2=1.
\end{eqnarray}
Adding these equations and using that $\sum_{x_1,x_2,y_1,y_2,m,n} p_{m,n,x_1,x_2,y_1,y_2}=1$ one obtains that
\bea
\sum_{x_1,x_2,y_1,y_2,m,n} p_{m,n,x_1,x_2,y_1,y_2} (|x_1|/|y_1|-|y_1|/|x_1|)^2=0.
\eea
This condition can only be fulfilled if $p_{m,n,x_1,x_2,y_1,y_2} =0$ for $|y_1| \neq |x_1|$. Note that this implies that  Eq. (\ref{Eqsepaaaa}) is independent of $|y_1|$ and $|x_1|$. Analogously, one can show that $p_{m,n,x_1,x_2,y_1,y_2} =0$ for $|y_2| \neq |x_2|$ and therefore Eq. (\ref{Eqsepaaaa}) does not depend on  $|y_2|$ and $|x_2|$.
Considering the matrix elements $\ket{1010}\bra{0000}$, $\ket{0101}\bra{0000}$ and $\ket{1111}\bra{0000}$ of Eq. (\ref{Eqsepaaaa}) and using that for the here considered states $h_i^2\in\Re$ $\forall i\in\{1,2,3,4\}$ leads to
\begin{eqnarray}
&h_1^2h_3^2(\tilde{p}_{00}+\tilde{p}_{10}-\tilde{p}_{01}-\tilde{p}_{11})=g_1^2g_3^2\\ \nonumber
&h_2^2h_4^2(\tilde{p}_{00}+\tilde{p}_{10}-\tilde{p}_{01}-\tilde{p}_{11})=g_2^2g_4^2
\end{eqnarray}
and
\bea
h_1^2h_3^2h_2^2h_4^2=g_1^2g_3^2g_2^2g_4^2,
\eea
where here and in the following we use \bea\tilde{p}_{mn}=\sum_{x_1,x_2,y_1,y_2} p_{m,n,x_1,x_2,y_1,y_2}.\eea Note that as $g_1^2, g_2^2\in\R$ (according to our standard form) and $h_i^2\in\Re$ $\forall i\in\{1,2,3,4\}$ these equations imply that $g_3^2,g_4^2 \in \R$. As $h_i^2\neq 0$ $\forall i\in\{1,2,3,4\}$ one obtains that $(\tilde{p}_{00}+\tilde{p}_{10}-\tilde{p}_{01}-\tilde{p}_{11})^2=1$ and therefore either $\tilde{p}_{00}=\tilde{p}_{10}=0$ or $\tilde{p}_{01}=\tilde{p}_{11}=0$. Note that the first case leads to $h_1^2<0$ or $h_3^2<0$ which does not match our standard form. Thus, we have that $\tilde{p}_{01}=\tilde{p}_{11}=0$ and therefore
\bea\label{aaaa1}
h_1^2h_3^2=g_1^2g_3^2
\eea
and
\bea\label{aaaa2}
h_2^2h_4^2=g_2^2g_4^2.
\eea
Considering the matrix elements $\ket{1000}\bra{0000}$, $\ket{0000}\bra{0010}$, $\ket{0100}\bra{0000}$  and $\ket{0000}\bra{0001}$ of Eq. (\ref{Eqsepaaaa}) and using that $h_i^2\neq 0$ $\forall i\in\{1,2,3,4\}$ one can easily show that
\bea
h_i^2/h_j^2=g_i^2/g_j^2\,\,\, \textrm{for}\,\,\, \{i,j\}\in\{\{1,3\},\{2,4\}\}.
\eea
Using this equation, as well as Eq. (\ref{aaaa1}), Eq. (\ref{aaaa2}), and  $g_1^2, g_2^2>0$ one obtains that $h_i^2=g_i^2$ $\forall i\in\{1,2,3,4\}$. Thus, we have shown that $g_i=h_i^x$ $\forall i \in\{1,2,3,4\}$ and therefore states of the form $h_1^x\otimes h_2^x\otimes h_3^x\otimes h_4^x|\Psi\rangle$  are in $MES_4$.\\
Let us next show that states of the form $\one/2\otimes h_2^x\otimes d_3\otimes h_4^y|\Psi\rangle$  where $h_j^2\neq 0$ for $j\in\{2,4\}$ and $d_3\neq \one$ are not reachable via LOCC. Considering the matrix elements $\ket{jkjk}\bra{jkjk}$ of Eq. (\ref{Eqsepaacc}) for $\{j,k\}=\{0,1\}$ leads to
\bea \label{aaaa6}
1/2[h_3^{1+2j}(\tilde{p}_{00}+\tilde{p}_{11})+h_3^{3-2j}(\tilde{p}_{10}+\tilde{p}_{01})]=rg_3^{1+2j}g_4^{1+2k}.
\eea
Considering this equation for $k=0$ and $k=1$ (and fixed j) it is easy to see that $g_4^1=g_4^3$. Additionally using the normalization condition, i.e. $\tr(G_i)=1$ one obtains that $g_4^1=g_4^3=1/2$. Adding the equations for $j=0$ and $j=1$ and using again the normalization condition leads to $r=1$. Considering the matrix elements $\ket{ljlk}\bra{ljlk}$ of Eq. (\ref{Eqsepaaaa}) for $j,k,l\in\{0,1\}$ and $j\neq k$ leads to
\begin{eqnarray}\label{aaaa3}
&\sum_{x_1,x_2,y_1,y_2} [(p_{0,0,x_1,x_2,y_1,y_2}h_3^{1+2l} +p_{0,1,x_1,x_2,y_1,y_2}h_3^{3-2l}) |y_2|^2/|x_2|^2\\\nonumber & + (p_{1,0,x_1,x_2,y_1,y_2}h_3^{3-2l}+p_{1,1,x_1,x_2,y_1,y_2}h_3^{1+2l}) |x_2|^2/|y_2|^2]=g_3^{1+2l}
\end{eqnarray}
and
\begin{eqnarray}\label{aaaa4}
&\sum_{x_1,x_2,y_1,y_2} [(p_{0,0,x_1,x_2,y_1,y_2}h_3^{1+2l} +p_{0,1,x_1,x_2,y_1,y_2}h_3^{3-2l}) |x_2|^2/|y_2|^2\\ \nonumber &+ (p_{1,0,x_1,x_2,y_1,y_2}h_3^{3-2l}+p_{1,1,x_1,x_2,y_1,y_2}h_3^{1+2l}) |y_2|^2/|x_2|^2]=g_3^{1+2l}.
\end{eqnarray}
Adding these four equations leads to
\bea
\sum_{x_1,x_2,y_1,y_2,m,n} p_{m,n,x_1,x_2,y_1,y_2} (|x_2|/|y_2|-|y_2|/|x_2|)^2=0,
\eea
which can only be fulfilled if $p_{m,n,x_1,x_2,y_1,y_2} =0$ for $|y_2| \neq |x_2|$. Thus,  Eq. (\ref{Eqsepaaaa}) is independent of $|y_2|$ and $|x_2|$. Considering the matrix elements $\ket{1000}\bra{0000}$ and $\ket{0010}\bra{0000}$ of Eq. (\ref{Eqsepaaaa}) it is easy to see that for $h_1^2=h_3^2=0$ we have that $g_1^2=g_3^2=0$. Considering the matrix elements $\ket{k0k1}\bra{k1k0}$ for $k\in\{0,1\}$ of Eq. (\ref{Eqsepaaaa}) and using that $h_4^2\in i\R$ leads to
\bea \label{aaaa5}
h_2^2h_4^2(\tilde{p}_{00}h_3^{1+2k}-\tilde{p}_{10}h_3^{3-2k}-\tilde{p}_{01}h_3^{3-2k}+\tilde{p}_{11}h_3^{1+2k})=g_2^2g_4^2g_3^{1+2k}.
\eea
Note that since $h_2^2, g_2^2\in \R$ and $h_4^2\in i\R$ we have that $g_4^2\in i\R$. Note further that we consider here states with $h_2^2,h_4^2\neq 0$. Thus, for $g_2^2g_4^2= 0$ one obtains that
\bea
(\tilde{p}_{00}+\tilde{p}_{11})h_3^{1+2k}=(\tilde{p}_{10}+\tilde{p}_{01})h_3^{3-2k},
\eea
for $l\in\{0,1\}$. As $\tilde{p}_{mn}=0$ $\forall m,n\in\{0,1\}$ is no solution because $\sum_{m,n=0}^1\tilde{p}_{mn}=1$ and  as $h_3^1, h_3^3\neq 0$, one obtains for $g_2^2g_4^2= 0$ that $h_3^3=h_3^1=1/2$. As we consider here states with $d_3\neq \one/2$ it has to hold that $g_2^2,g_4^2\neq 0$. Note that the same argument shows that for the states considered here it has to hold that $\tilde{p}_{00}h_3^{1+2k}-\tilde{p}_{10}h_3^{3-2k}-\tilde{p}_{01}h_3^{3-2k}+\tilde{p}_{11}h_3^{1+2k}\neq 0$ for $k\in\{0,1\}$. Considering Eq. (\ref{aaaa5}), as well as Eq. (\ref{aaaa6}) for $k=0,1$ one can easily show that
\begin{eqnarray}
&\frac{(\tilde{p}_{00}+\tilde{p}_{11})h_3^{1}-(\tilde{p}_{10}+\tilde{p}_{01})h_3^{3}}
{(\tilde{p}_{00}+\tilde{p}_{11})h_3^{3}-(\tilde{p}_{10}+\tilde{p}_{01})h_3^{1}}=
\frac{h_3^{1}(\tilde{p}_{00}+\tilde{p}_{11})+h_3^{3}(\tilde{p}_{10}+\tilde{p}_{01})}
{h_3^{3}(\tilde{p}_{00}+\tilde{p}_{11})+h_3^{1}(\tilde{p}_{10}+\tilde{p}_{01})}.
\end{eqnarray}
Using this equation it is easy to show that either $\tilde{p}_{00}=\tilde{p}_{11}=0$ or $\tilde{p}_{01}=\tilde{p}_{10}=0$ or $h_3^1=h_3^3=1/2$. The latter solution is not possible as we consider here states with $d_3\neq \one/2$. For the case $\tilde{p}_{00}=\tilde{p}_{11}=0$ one obtains from Eq. (\ref{aaaa5}) that either $g_2^2<0$ or $\Im(g_4^2)<0$ which does not match our standard form. Thus, we have that $\tilde{p}_{01}=\tilde{p}_{10}=0$. Using Eq. (\ref{aaaa6}) for $j=0,1$ it is easy to see that $h_3^1=g_3^1$ and $h_3^3=g_3^3$. Considering the matrix elements  $\ket{0000}\bra{0100}$  and $\ket{0001}\bra{0000}$ of Eq. (\ref{Eqsepaaaa}) and using that $h_i^2\neq 0$ for $ i\in\{2,4\}$ one can easily show that
\bea
h_2^2/h_4^2=g_2^2/g_4^2.
\eea
Moreover, from Eq. (\ref{aaaa5}) for $k=0,1$ one obtains that $g_2^2g_4^2=g_2^2g_4^2$. As according to our standard form it holds that $h_2^2,g_2^2\geq 0$ one obtains that $g_i^2=h_i^2$ for $i\in\{2,4\}$. To summarize these results we have shown that for states of the form $\one/2\otimes h_2^x\otimes d_3\otimes h_4^y|\Psi\rangle$  where $h_j^2\neq 0$ for $j\in\{2,4\}$ and $d_3\neq \one$ Eq. (\ref{EqSep}) can only hold if $H_i=G_i$ $\forall i\in\{1,2,3,4\}$ and therefore these states are in $MES_4$. Analogously, one can show that states of the form  $h_1^x\otimes \one/2\otimes g_3^y\otimes d_4|\Psi\rangle$  where $h_k^2\neq 0$ for $k\in\{1,3\}$ and $d_4\neq \one$  are not reachable via SEP.\\
We will show next that states of the form $h_1^x\otimes h_2^x\otimes h_3\otimes h_4|\Psi\rangle$  where $h_i^2\neq 0$ $\forall i\in\{1,2,3,4\}$ and $h_l\neq h_l^x, h_l^y$ for $l\in\{3,4\}$ are not reachable either. Considering the matrix elements $\ket{jkjk}\bra{jkjk}$ of Eq. (\ref{Eqsepaacc}) for $\{j,k\}=\{0,1\}$ leads to
\begin{eqnarray} \label{aaaa8}
&h_3^{1+2j}h_4^{1+2k}\tilde{p}_{00}+h_3^{1+2j}h_4^{3-2k}\tilde{p}_{11}+h_3^{3-2j}h_4^{3-2k}\tilde{p}_{10}
+h_3^{3-2j}h_4^{1+2k}\tilde{p}_{01}=rg_3^{1+2j}g_4^{1+2k}.
\end{eqnarray}
Using these equations, as well as the normalization condition, i.e. $\tr(G_i)=\tr(H_i)=1$, it is easy to show that $r=1$,
\bea\label{aaaa14}
g_3^{1+2j}=h_3^{1+2j}(\tilde{p}_{00}+\tilde{p}_{11})+h_3^{3-2j}(\tilde{p}_{01}+\tilde{p}_{10}),
\eea
and
\bea\label{aaaa15}
g_4^{1+2k}=h_4^{1+2k}(\tilde{p}_{00}+\tilde{p}_{01})+h_4^{3-2k}(\tilde{p}_{10}+\tilde{p}_{11}),
\eea
where $j,k\in\{0,1\}$
Using these equations, as well as Eq. (\ref{aaaa8}) for $j=0$ and $k=0$ it can be shown that either $h_4^1=h_4^3=1/2$ (solution (a)) or $h_3^1=h_3^3=1/2$ (solution (b)) or $\tilde{p}_{11}\tilde{p}_{10}=\tilde{p}_{00}\tilde{p}_{01}$ (solution (c)). We will first consider solution (a). As shown above, for $h_4^1=h_4^3=1/2$ it has to hold that
$p_{m,n,x_1,x_2,y_1,y_2} =0$ for $|y_2| \neq |x_2|$ and therefore  Eq. (\ref{Eqsepaaaa}) is independent of $|y_2|$ and $|x_2|$. Moreover, we have already shown that states with $h_4^1=h_4^3=1/2$ can only be reached from states with $g_4^1=g_4^3=1/2$. Considering the imaginary part of the matrix elements $\ket{k0k1}\bra{k1k0}$ for $k\in\{0,1\}$ of Eq. (\ref{Eqsepaaaa})  leads to
\begin{eqnarray}\label{aaaa9}
&h_2^2\Im(h_4^2)(\tilde{p}_{00}h_3^{1+2k}-\tilde{p}_{10}h_3^{3-2k}
-\tilde{p}_{01}h_3^{3-2k}+\tilde{p}_{11}h_3^{1+2k})=g_2^2\Im(g_4^2)g_3^{1+2k}.
\end{eqnarray}
Note that as $h_4\neq h_4^x, h_4^y$ we have that $\Im(h_4^2)\neq 0$ and $\Re(h_4^2)\neq 0$. Moreover, we consider states with $h_i^2\neq 0$ $\forall i\in\{1,2,3,4\}$ and so in particular $h_2^2\neq 0$.
Thus, we can use the same argument as before and obtain that for $\tilde{p}_{00}h_3^{1+2k}-\tilde{p}_{10}h_3^{3-2k}-\tilde{p}_{01}h_3^{3-2k}+\tilde{p}_{11}h_3^{1+2k}\neq 0$ for $k\in\{0,1\}$ (which implies $g_2^2,g_4^2\neq 0$) that $\tilde{p}_{00}=\tilde{p}_{11}=0$ (case (i)) or $\tilde{p}_{01}=\tilde{p}_{10}=0$ (case (ii)) or $h_3^1=h_3^3=1/2$ (case (iii)). Note that for $\tilde{p}_{00}h_3^{1+2k}-\tilde{p}_{10}h_3^{3-2k}-\tilde{p}_{01}h_3^{3-2k}+\tilde{p}_{11}h_3^{1+2k}= 0$ for $k\in\{0,1\}$ we have already shown that it has to hold that $h_3^1=h_3^3=1/2$. We will discuss first case (iii), i.e. $h_3^1=h_3^3=1/2$. Analogously to the discussion for states of the form $h_1^x\otimes h_2^x\otimes h_3^x\otimes h_4^x|\Psi\rangle$  where $h_i^2\neq 0$ $\forall i\in\{1,2,3,4\}$ one can show that for $h_4^1=h_4^3=h_3^1=h_3^3=1/2$ it has to hold that $p_{m,n,x_1,x_2,y_1,y_2} =0$ for $|y_1| \neq |x_1|$ and therefore  Eq. (\ref{Eqsepaaaa}) does not depend on $|y_1|$ and $|x_1|$. Moreover, we have shown that for $h_4^1=h_4^3=h_3^1=h_3^3=1/2$ Eq. (\ref{Eqsepaaaa}) can only be fulfilled if $g_4^1=g_4^3=g_3^1=g_3^3=1/2$. Considering the matrix elements $\ket{1111}\bra{0000}$, $\ket{1010}\bra{0101}$, $\ket{0101}\bra{1010}$  and $\ket{0000}\bra{1111}$ of Eq. (\ref{Eqsepaaaa}) and adding them one obtains that
\bea\label{aaaa110}
h_1^2h_2^2\Re(h_3^2)\Re(h_4^2)=g_1^2g_2^2\Re(g_3^2)\Re(g_4^2).
\eea
Considering the matrix elements $\ket{1010}\bra{0000}$ and $\ket{0101}\bra{0000}$ of Eq. (\ref{Eqsepaaaa}) leads to
\begin{eqnarray}\label{aaaa90}
&h_1^2\Re(h_3^2)(\tilde{p}_{00}+\tilde{p}_{10}-\tilde{p}_{01}-\tilde{p}_{11})=g_1^2\Re(g_3^2)\\
&h_1^2\Im(h_3^2)(\tilde{p}_{00}-\tilde{p}_{10}+\tilde{p}_{01}-\tilde{p}_{11})=g_1^2\Im(g_3^2)\label{aaaa12}\\
&h_2^2\Re(h_4^2)(\tilde{p}_{00}+\tilde{p}_{10}-\tilde{p}_{01}-\tilde{p}_{11})=g_2^2\Re(g_4^2)\label{aaaa10}
\end{eqnarray}
and
\bea\label{aaaa13}
h_2^2\Im(h_4^2)(\tilde{p}_{00}-\tilde{p}_{10}-\tilde{p}_{01}+\tilde{p}_{11})=g_2^2\Im(g_4^2).
\eea
Recall that we consider here states with $h_i^2\neq 0$ for $i\in\{1,2,3,4\}$ and $h_l\neq h_l^x, h_l^y$ for $l\in\{3,4\}$. Thus, for $h_4^1=h_4^3=h_3^1=h_3^3=1/2$ it must hold that neither $\Re(h_l^2)=0$ nor $\Im(h_l^2)=0$.
Thus, inserting Eq. (\ref{aaaa90}) and Eq. (\ref{aaaa10}) into Eq. (\ref{aaaa110}) leads to $(\tilde{p}_{00}+\tilde{p}_{10}-\tilde{p}_{01}-\tilde{p}_{11})^2=1$. Thus, either $\tilde{p}_{00}=\tilde{p}_{10}=0$ or $\tilde{p}_{01}=\tilde{p}_{11}=0$. Note that the first case leads to either $g_1^2<0$ or $\Re(g_3^2)<0$ which does not match our standard form. Thus, we have that $\tilde{p}_{01}=\tilde{p}_{11}=0$. Considering the matrix elements $\ket{1111}\bra{0000}$, $\ket{1010}\bra{0101}$, $\ket{0101}\bra{1010}$  and $\ket{0000}\bra{1111}$ of Eq. (\ref{Eqsepaaaa}) one can show that \bea\label{aaaa11}
h_1^2h_2^2\Im(h_3^2)\Im(h_4^2)=g_1^2g_2^2\Im(g_3^2)\Im(g_4^2).
\eea
Using this equation, as well as Eq. (\ref{aaaa12}) and Eq. (\ref{aaaa13}) it can be easily seen that either $\tilde{p}_{00}=0$ or $\tilde{p}_{10}=0$. Using our definition of the standard form one can show that it has to hold that $\tilde{p}_{10}=0$. Similar to before one can show that
\bea
h_i^2/h_j^2=g_i^2/g_j^2\,\,\, \textrm{for}\,\,\, \{i,j\}\in\{\{1,3\},\{2,4\}\}.
\eea
Note that $g_i\neq 0$ $\forall i\in\{1,2,3,4\}$ due to Eq. (\ref{aaaa11}) and $h_i^2\neq 0$.
Using these equations, as well as Eqs. (\ref{aaaa90}), (\ref{aaaa12}), (\ref{aaaa10}) and (\ref{aaaa13}) one can show analogously to before that $h_j^2=g_j^2$ $\forall j\in\{1,2,3,4\}$. Thus, states of the form $h_1^x\otimes h_2^x\otimes h_3\otimes h_4|\Psi\rangle$  where $h_i^2\neq 0$ $\forall i\in\{1,2,3,4\}$, $h_l\neq h_l^x, h_l^y$ for $l\in\{3,4\}$ and $h_4^1=h_4^3=h_3^1=h_3^3=1/2$ are not reachable.\\
We will discuss next solution (a) together with the condition of case (ii), i.e. $h_4^1=h_4^3=1/2$ and $\tilde{p}_{01}=\tilde{p}_{10}=0$. Using Eq. (\ref{aaaa14}) for $k=0,1$ it is easy to see that $h_3^1=g_3^1$ and $h_3^3=g_3^3$. Considering the matrix elements $\ket{1010}\bra{0000}$, $\ket{0101}\bra{0000}$,$\ket{1111}\bra{0000}$ and $\ket{1111}\bra{1010}$ of Eq. (\ref{Eqsepaaaa}) one can show that
\begin{eqnarray}
&h_1^2h_3^2(\tilde{p}_{00}-\tilde{p}_{11})=g_1^2g_3^2\\
&h_2^2(h_4^2\tilde{p}_{00}-(h_4^{2})^*\tilde{p}_{11})=g_2^2g_4^2
\end{eqnarray}
and
\bea
h_1^2h_2^2h_3^2(h_4^2\tilde{p}_{00}+(h_4^{2})^*\tilde{p}_{11})=g_1^2g_2^2g_3^2g_4^2.
\eea
Thus, it follows that
\bea
(h_4^2\tilde{p}_{00}+(h_4^{2})^*\tilde{p}_{11})=(\tilde{p}_{00}-\tilde{p}_{11})
(h_4^2\tilde{p}_{00}-(h_4^{2})^*\tilde{p}_{11}).
\eea
Taking the real part of this equation and using that for the states we consider here $\Re(h_4^2)\neq 0$ one obtains that $(\tilde{p}_{00}-\tilde{p}_{11})^2=1$. As it has to hold that $\tilde{p}_{00}+\tilde{p}_{11}=1$ one obtains that either $\tilde{p}_{00}=0$ or $\tilde{p}_{11}=0$. Note that the first case does not match our definition of the standard form. Thus, we have that $\tilde{p}_{11}=0$ and therefore $h_1^2h_3^2=g_1^2g_3^2$ and $h_2^2h_4^2=g_2^2g_4^2$. For $\tilde{p}_{01}=\tilde{p}_{10}=\tilde{p}_{11}=0$ it can be easily shown that $p_{m,n,x_1,x_2,y_1,y_2} =0$ for $|y_1| \neq |x_1|$. Considering the matrix elements  $\ket{0000}\bra{1000}$, $\ket{0010}\bra{0000}$,$\ket{0000}\bra{0100}$ and $\ket{0001}\bra{0000}$ of Eq. (\ref{Eqsepaaaa}) one can show analogously to before that $h_i^2=g_i^2$ $\forall i\in\{1,2,3,4\}$. Thus,  transformations with $p_{m,n,x_1,x_2,y_1,y_2} =0$ for $m\neq n$ and $m,n\in\{0,1\}$ (case (ii)) do not allow to reach states of the form $h_1^x\otimes h_2^x\otimes h_3\otimes h_4|\Psi\rangle$  where $h_i^2\neq 0$ $\forall i\in\{1,2,3,4\}$, $h_l\neq h_l^x, h_l^y$ for $l\in\{3,4\}$ and $h_4^1=h_4^3=1/2$. \\
We will next consider solution (a) together with the condition of case (i), i.e. $\tilde{p}_{00}=\tilde{p}_{11}=0$ and $h_4^1=h_4^3=1/2$. It can be easily shown that $h_3^3=g_3^1$ and $h_3^1=g_3^3$.
Furthermore, one can show analogously to case (ii) that either $\tilde{p}_{01}=0$ or $\tilde{p}_{10}=0$. In particular, one considers the matrix elements $\ket{0000}\bra{1010}$, $\ket{0101}\bra{0000}$,$\ket{0101}\bra{1010}$ and $\ket{1111}\bra{1010}$ of Eq. (\ref{Eqsepaaaa}). Moreover, one can show similar to before that for $\tilde{p}_{10}=0$ $h_3^2=-g_3^{2\,*}$ and $h_4^2=-g_4^2$ and for $\tilde{p}_{01}=0$ $h_3^2=g_3^{2\,*}$ and $h_4^2=g_4^{2\,*}$. Note that therefore case (i) does not match our standard form. Thus, we have shown that neither of the cases (i)-(iii) allows to reach states of the form $h_1^x\otimes h_2^x\otimes h_3\otimes h_4|\Psi\rangle$  where $h_i^2\neq 0$ $\forall i\in\{1,2,3,4\}$, $h_l\neq h_l^x, h_l^y$ for $l\in\{3,4\}$ and $h_4^1=h_4^3=1/2$, i.e. solution (a).\\
Note that in case of solution (b), i.e. $h_3^1=h_3^3=1/2$ a completely analogous argument can be used to show that states of the form $h_1^x\otimes h_2^x\otimes h_3\otimes h_4|\Psi\rangle$  where $h_i^2\neq 0$ $\forall i\in\{1,2,3,4\}$, $h_l\neq h_l^x, h_l^y$ for $l\in\{3,4\}$ and $h_3^1=h_3^3=1/2$ are in $MES_4$. \\
Thus, it remains to discuss solution (c), i.e. $\tilde{p}_{11}\tilde{p}_{10}=\tilde{p}_{00}\tilde{p}_{01}$. Considering the matrix elements $\ket{k0k1}\bra{k1k0}$ for $k\in\{0,1\}$ of Eq. (\ref{Eqsepaaaa}) leads to
\begin{eqnarray}
&h_2^2\Im(h_4^2)(\tilde{p}_{00}h_3^{1+2k}-\tilde{p}_{10}h_3^{3-2k}
-\tilde{p}_{01}h_3^{3-2k}+\tilde{p}_{11}h_3^{1+2k})=g_2^2\Im(g_4^2)g_3^{1+2k}
\end{eqnarray}
and
\begin{eqnarray} \label{aaaa20}
&h_2^2\Re(h_4^2)(\tilde{p}_{00}h_3^{1+2k}+\tilde{p}_{10}h_3^{3-2k}
-\tilde{p}_{01}h_3^{3-2k}-\tilde{p}_{11}h_3^{1+2k})=g_2^2\Re(g_4^2)g_3^{1+2k}.
\end{eqnarray}
Note that for $\Im(h_4^2)\neq 0$ one can use the same argument as before to show that $\tilde{p}_{00}=\tilde{p}_{11}=0$  or $\tilde{p}_{01}=\tilde{p}_{10}=0$  or $h_3^1=h_3^3=1/2$. Via a similar argument one can show that for  $\Im(h_3^2)\neq 0$ it has to hold that $\tilde{p}_{01}=\tilde{p}_{11}=0$  or $\tilde{p}_{00}=\tilde{p}_{10}=0$ or $h_4^1=h_4^3=1/2$. Analogously, one can show by using Eq. (\ref{aaaa20}) and Eq. (\ref{aaaa14}) for $j=0,1$ that for $\Re(h_4^2)\neq 0$ it has to hold that either $h_3^1=h_3^3=1/2$ or $\tilde{p}_{11}\tilde{p}_{01}=\tilde{p}_{00}\tilde{p}_{10}$.
For the latter case one obtains that $\tilde{p}_{00}=\tilde{p}_{11}=\tilde{p}_{01}=\tilde{p}_{10}$ or $\tilde{p}_{01}=\tilde{p}_{10}=0$ or $\tilde{p}_{00}=\tilde{p}_{11}=0$ as we also have that $\tilde{p}_{11}\tilde{p}_{10}=\tilde{p}_{00}\tilde{p}_{01}$.
Similarly one can show that for  $\Re(h_3^2)\neq 0$ it has to hold that either $h_4^1=h_4^3=1/2$ or $\tilde{p}_{11}\tilde{p}_{00}=\tilde{p}_{01}\tilde{p}_{10}$. As we are considering here solution (c), i.e. $\tilde{p}_{11}\tilde{p}_{10}=\tilde{p}_{00}\tilde{p}_{01}$ it has to hold for the latter case that either $\tilde{p}_{00}=\tilde{p}_{11}=\tilde{p}_{01}=\tilde{p}_{10}$ or $\tilde{p}_{01}=\tilde{p}_{11}=0$ or $\tilde{p}_{00}=\tilde{p}_{10}=0$.
Note that the case $h_4^1=h_4^3=1/2$ ($h_3^1=h_3^3=1/2$) corresponds to solution (a) (solution (b)) respectively which has been already discussed. Note further that since $h_i^2\neq 0$ for $i\in\{1,2,3,4\}$ we have that $\Im(h_j^2)\neq 0$ and/or $\Re(h_j^2)\neq 0$ for $j\in\{3,4\}$. For $h_4^1,h_4^3,h_3^1,h_3^3\neq 1/2$ and either $\Im(h_4^2),\Im(h_3^2)\neq 0$, or $\Re(h_4^2),\Im(h_3^2)\neq 0$, or $\Im(h_4^2),\Re(h_3^2)\neq 0$ one obtains that $\tilde{p}_{mn}=0$ except for exactly one pair $(m, n)\in\{(0, 0),(0, 1),(1, 0),(1, 1)\}$. Note here that $\tilde{p}_{00}=\tilde{p}_{11}=\tilde{p}_{01}=\tilde{p}_{10}=0$ does not correspond to a valid transformation. It can be shown that the only possibility to match our definition of the standard form is that $\tilde{p}_{00}\neq 0$. Furthermore, one can show that in this case it holds that $H_i=G_i$ $\forall i\in\{1,2,3,4\}$ and therefore these kind of transformations to not allow to reach these states from some other states. For $\Re(h_4^2),\Re(h_3^2)\neq 0$ and $h_4^1,h_4^3,h_3^1,h_3^3\neq 1/2$ we have either that $\tilde{p}_{mn}=0$ except for exactly one pair $(m, n)\in\{(0, 0),(0, 1),(1, 0),(1, 1)\}$ which does not allow to reach these states or $\tilde{p}_{00}=\tilde{p}_{11}=\tilde{p}_{01}=\tilde{p}_{10}=1/4$. Note that we used here that $\sum_{m,n}\tilde{p}_{mn}=1$. Considering Eqs. (\ref{aaaa14}) and (\ref{aaaa15}) for $j,k=0,1$ one obtains for the latter case that $h_i^1=g_i^1$ and $h_i^3=g_i^3$ for $i\in\{3,4\}$. Considering the matrix elements $\ket{k0k1}\bra{k1k0}$ and $\ket{0k1k}\bra{1k0k}$ for $k\in\{0,1\}$ of Eq. (\ref{Eqsepaaaa}) one can easily show that
$g_1^2g_3^2=0$ and $g_2^2g_4^2=0$. Considering now the matrix elements $\ket{0011}\bra{1100}$, $\ket{1100}\bra{0011}$, $\ket{0101}\bra{1010}$ and $\ket{1010}\bra{0101}$ of Eq. (\ref{Eqsepaaaa}) and adding the corresponding equations one obtains that
\bea
h_1^2h_2^2\Re(h_3^2)\Re(h_4^2)=g_1^2g_2^2\Re(g_3^2)\Re(g_4^2)=0.
\eea
As $h_i^2\neq 0$ for $i\in\{1,2\}$ for the states we are discussing here and $\Re(h_4^2),\Re(h_3^2)\neq 0$ this equation can not be fulfilled. Thus, states of the form  $h_1^x\otimes h_2^x\otimes h_3\otimes h_4|\Psi\rangle$  where $h_i^2\neq 0$ $\forall i\in\{1,2,3,4\}$ and $h_l\neq h_l^x, h_l^y$ for $l\in\{3,4\}$ are in $MES_4$.\\
It can be easily seen that the seed states for these SLOCC classes are not reachable. In particular, using Eq. (\ref{eqdiagaaaa}) one obtains that $g_i^1=g_i^3=1/2$ for $i\in\{3,4\}$. Considering then the matrix elements $\ket{1000}_{i,j,k,l}\bra{0000}$ of Eq. (\ref{Eqsepaaaa}) where $\{i,j,k,l\}=\{1,2,3,4\}$ and using that $h_m^2=0$ $\forall m\in\{1,2,3,4\}$ leads to $g_m^2=0$ $\forall m$. Thus, one obtains $H_i=G_i=\one/2$ $\forall i\in\{1,2,3,4\}$ which implies that the seed states are in $MES_4$.
$\square$
\\ \\
In the following we prove Lemma \ref{convlemmaaaaa}, i.e. which states in the SLOCC class $G_{abcd}$ where  $a^2= d^2$, $c^2= b^2$, $a^2= -c^2$ and $a \neq 0$ are convertible.
\\ \\
\textit{ Proof of Lemma \ref{convlemmaaaaa}.}
Note first that the states that can be reached from some convertible state are given by the ones in Lemma \ref{lemmaaaaa}. Note further that using the same techniques as in the proof of Lemma \ref{lemmaaaaa} it can be shown that states with  $h_i^2=0$ ($h_i=h_i^x$ or $h_i=h_i^y$)  for some $i\in\{3,4\}$ can only be obtained from states with $g_i^2=0$ ($g_i=g_i^x$ or $g_i=g_i^y$) respectively. Combining these results one obtains that the only potentially convertible states are the ones given by Lemma \ref{convlemmaaaaa}. Hence, it remains to show that they are indeed convertible via LOCC. In order to do so, we first note that  $X^m[(\sigma_z)^k\otimes(\sigma_z)^l\otimes(\sigma_z)^k\otimes(\sigma_z)^l]\in S(\Psi)$ (where $m, k, l\in\{0,1\}$) which is element of the symmetry group of the SLOCC classes $G_{abcd}$ where  $a^2= d^2\neq\pm c^2$ and $c^2= b^2$ (see Sec. \ref{secaacc}). There we have already provided the LOCC protocols that allow to convert states of the form $g_1^x\otimes g_2^x\otimes g_3\otimes g_4^x|\Psi\rangle$ and $g_1^x\otimes g_2^x\otimes g_3^x\otimes g_4|\Psi\rangle$ and $\one/2\otimes g_2^x\otimes g_3\otimes g_4|\Psi\rangle$ and $g_1^x\otimes \one/2\otimes g_3\otimes g_4|\Psi\rangle$ and  $g_1^x\otimes g_2^x\otimes d_3\otimes g_4|\Psi\rangle$ and $g_1^x\otimes g_2^x\otimes g_3\otimes d_4|\Psi\rangle$ in the proof of Lemma \ref{convlemmaaacc} (see Appendix C). In the following we will use the notation $\one\otimes \sigma_x\otimes \sigma_z\otimes \sigma_y\equiv S_1$, $\sigma_x\otimes \one \otimes \sigma_y\otimes \sigma_z\equiv S_2$ and $\one^{\otimes 4}\equiv S_0$. States of the form $g_1^x\otimes g_2^x\otimes g_3^y\otimes g_4|\Psi\rangle$ can be converted to states of the form $h_1^x\otimes h_2^x\otimes h_3^y\otimes h_4|\Psi\rangle$ where $h_4\neq d_4$  via the POVM $\{\sqrt{p} h_4S_0g_4^{-1}, \sqrt{1-p}h_4S_2g_4^{-1}\}$ (with $(2p-1)h_4^2=g_4^2$, $h_4^1=g_4^1$ and $h_4^3=g_4^3$) which is acting non-trivially on party 4. Note that for any $G_4$ one can always find a value of $p$ and $H_4$ which is a positive operator of rank 2 such that these conditions are fulfilled. Analogously, using $S_0$ and $S_1$ one can construct a POVM that is acting non-trivially on party $3$ which allows to convert states of the form $g_1^x\otimes g_2^x\otimes g_3\otimes g_4^y|\Psi\rangle$ to states of the form $h_1^x\otimes h_2^x\otimes h_3\otimes h_4^y|\Psi\rangle$ where $h_3\neq d_3$.
$\square$

\section*{\label{appE}Appendix E: Proof of Lemma \ref{lem10}}

In this section we prove Lemma \ref{lem10} and hereby show which states in the SLOCC classes $L_{a_2b_2}$ (see \cite{slocc4}) for $a^2= - b^2$ and $a,b\neq0$ are reachable.
As in the main text we will use the notation $G_i=g_i^\dagger g_i=1/2\one+\sum_j \bar{g}_i^j \sigma_j$ where $\bar{g}_i^j\in \R$ and $0\leq\sqrt{\sum_j[\bar{g}_i^j]^2}<1/2$ and similar for $H_i$. In order to compactify the presentation in this appendix we use the following unconventional notation. That is, we will use $G_i^{W}=1/2\one+\sum_j \bar{g}_i^j \sigma_j$ for $W\in\{W_1 W_2 W_3,\neq j\}$ and   $W_j \in \{j,0,A\}$ whenever we need to give information on the precise structure of $G_i$. Here $W_j=j$ ($W_j=0$) denotes $\bar{g}_i^j\neq 0$($\bar{g}_i^j= 0$) whereas $W_j=A$ means that $\bar{g}_i^j$ is arbitrary. Whenever $W$ corresponds to $\neq j$  this means that $\bar{g}_i^k\neq 0$ and/or $\bar{g}_i^l\neq 0$ where $k\neq l\neq j\neq k$ and $j,k,l\in\{1,2,3\}$.
In case $W$ corresponds to $\neq 0$ this denotes the constraint that $G_i^{\neq 0}\neq \one/2$.  We will use the same notation for $H_i^{W}$. In order to illustrate this notation let us consider two examples. By $G_i^{A20}$ we denote the case that $\bar{g}_i^1$ is arbitrary, $\bar{g}_i^2\neq 0$ and $\bar{g}_i^3=0$. By $G_i^{\neq 3}$ we denote that $\bar{g}_i^1\neq 0$ and /or $\bar{g}_i^2\neq 0$ (i.e. $\sigma_zG_i^{\neq 3}\sigma_z\neq G_i^{\neq 3}$) and $\bar{g}_i^3$ is arbitrary.  \\ In the following we will denote  by $S_k=\otimes_i s_i^k$ the elements of the symmetry group $S(\Psi_{(-b_2)b_2})=\{\one^{\otimes 4},\one\otimes\sigma_x\otimes\sigma_z\otimes\sigma_x,\sigma_z\otimes\sigma_y\otimes\one\otimes\sigma_y,\sigma_z^{\otimes 4}\}$ which are numbered in the same order as they are given here in the set.\\
We will use the following lemma to prove Lemma \ref{lem10}.
\begin{lemma}\label{twooutcome}
The only states in the SLOCC classes $L_{a_2b_2}$ with $a^2= - b^2$ and $a,b\neq 0$ that are reachable via using solely the symmetries $S_0$ and $S_j=\otimes_i s_j^i$ for one $j\in\{1,2,3\}$ are given by $\otimes h_i \ket{\Psi_{(-b_2)b_2}}$ where $\otimes H_i$ obeys the condition that exactly
three out of the four operators $H_i$ commute with the corresponding
operator $s_j^i$ and one operator $H_k$ does not commute with $s_j^k$.
\end{lemma}
\begin{proof}
Let us first consider which states are reachable using only $S_0$ and $S_1$. Since these symmetries are elements of the Pauli group we will use Eq. (\ref{eq_Pauli}) for $P=\one\otimes \sigma_i\otimes\sigma_j\otimes \one$ for $i\in\{2,3\}$ and $j\in\{1,2\}$ which leads to
\bea
\bar{h}_2^i\bar{h}_3^j=(p_0-p_1)^2\bar{h}_2^i\bar{h}_3^j.
\eea
Thus, either $\bar{h}_2^i\bar{h}_3^j=0$ for $i\in\{2,3\}$ and $j\in\{1,2\}$ or $(p_0-p_1)^2=1$. The latter case corresponds to a trivial transformation. Note that by exchanging party 2 and 4 in this argument one obtains analogous conditions involving party 4.  Therefore, one obtains that either $H_3=H_3^{00A}$ or the case that $H_2=H_2^{A00}$ and $H_4=H_4^{A00}$. In case $H_2=H_2^{A00}$ and $H_4=H_4^{A00}$ one can reach the corresponding state iff $H_3=H_3^{\neq 3}$. In order to see this note that if  $H_3=H_3^{00A}$ then $S_0$ and $S_1$ commute with $\otimes H_i$ and therefore using Eq. (\ref{EqSep}) (with symmetries $S_0$ and $S_1$) leads to $H_l=G_l$ $\forall l\in\{1,2,3,4\}$. Thus, these states are not reachable via $S_0$ and $S_1$. In order to see that states with $H_3= H_3^{\neq 3}$ are reachable consider $\{1/\sqrt{2}h_3S_0g_3^{-1},1/\sqrt{2}h_3S_1g_3^{-1}\}$ which is a valid POVM for $\bar{g}_3^2=0$, $\bar{g}_3^1=0$ and $\bar{h}_3^3=\bar{g}_3^3$. This POVM allows to obtain these states from states of the form $g_1^{AAA}\otimes g_2^{A00}\otimes g_3^{00A}\otimes g_4^{A00}\ket{\Psi_{(-b_2)b_2}}$. Let us proceed with the case $H_3=H_3^{00A}$.
Considering Eq. (\ref{eq_Pauli}) for $P=\one\otimes \sigma_i\otimes\one\otimes \sigma_j$ for $i\in\{2,3\}$ and $j\in\{2,3\}$ one obtains that
\bea
\bar{h}_2^i\bar{h}_4^j=(p_0-p_1)^2\bar{h}_2^i\bar{h}_4^j.
\eea
Thus, in order to allow for a non-trivial transformation one obtains that either $H_2=H_2^{A00}$ or $H_4=H_4^{A00}$. States of the form $h_1^{AAA}\otimes h_2^{A00}\otimes h_3^{00A}\otimes h_4^{\neq 1}\ket{\Psi_{(-b_2)b_2}}$ can be obtained for example from states of the form $g_1^{AAA}\otimes g_2^{A00}\otimes g_3^{00A}\otimes g_4^{A00}\ket{\Psi_{(-b_2)b_2}}$ via the POVM $\{1/\sqrt{2}h_4S_0g_4^{-1},1/\sqrt{2}h_4S_1g_4^{-1}\}$ where $\bar{g}_4^2=0$, $\bar{g}_4^3=0$ and $\bar{h}_4^1=\bar{g}_4^1$. Note that analogously one can construct a POVM that allows to reach states with $H_4=H_4^{A00}$, $H_3=H_3^{00A}$ and $H_2=H_2^{\neq 1}$. Note further that  we have already shown before that states of the form $h_1^{AAA}\otimes h_2^{A00}\otimes h_3^{00A}\otimes h_4^{A00}\ket{\Psi_{(-b_2)b_2}}$ are not reachable by using solely the symmetries $S_0$ and $S_1$. Thus, the only states, $\otimes h_i \ket{\Psi_{(-b_2)b_2}}$, that are reachable via $S_0$ and $S_1$ have the property that three out of the four operators $h_j$ commute with the corresponding
operator $s_i^j$ and one does not commute.\\
The states that can be reached by LOCC transformations that only use the symmetries $S_0$ and $S_2$ can be obtained via a completely analogous argument. One just has to exchange party 1 and 3, as well as the $\sigma_x$- and $\sigma_y$-components for party 2 and 4.\\
Let us now consider which states are reachable using only the symmetries $S_0$ and $S_3$.
Considering Eq. (\ref{eq_Pauli}) with $P$ denoting the operator such that $\sigma_i$ is acting on party $k$, $\sigma_j$ is acting on party $l$ and the identity is acting on the remaining parties  for $i\in\{1,2\}$, $j\in\{1,2\}$, $k\neq l$ and $k,l\in\{1,2,3,4\}$  one obtains that
\bea
\bar{h}_k^i\bar{h}_l^j=(p_0-p_3)^2\bar{h}_k^i\bar{h}_l^j.
\eea
Thus, in order to allow for a non-trivial transformation either $H_k=H_k^{00A}$ or $H_l=H_l^{00A}$. Considering this condition for all different pairs $(k,l)$ one obtains that at least for three parties we have that the corresponding operator is $H_i=H_i^{00A}$. For $H_i=H_i^{00A}$ $\forall i\in\{1,2,3,4\}$ it is easy to see by considering  Eq. (\ref{EqSep}) (with symmetries $S_0$ and $S_3$) that $H_i=G_i$ $\forall i\in\{1,2,3,4\}$. Therefore, these states are not reachable via the symmetries $S_0$ and $S_3$. In order to complete the proof we note that states of the form $h_i^{\neq 3}\otimes h_j^{00A}\otimes h_k^{00A}\otimes h_l^{00A}\ket{\Psi_{(-b_2)b_2}}$ where $\{i,j,k,l\}=\{1,2,3,4\}$ can be reached for example from states of the form $g_1^{00A}\otimes g_2^{00A}\otimes g_3^{00A}\otimes g_4^{00A}\ket{\Psi_{(-b_2)b_2}}$ by using the POVM $\{1/\sqrt{2}h_iS_0g_i^{-1},1/\sqrt{2}h_iS_3g_i^{-1}\}$ where $\bar{g}_i^1=0$, $\bar{g}_i^2=0$ and $\bar{h}_i^3=\bar{g}_i^3$.
\end{proof}
As we will show in the proof of Lemma \ref{lem10} for any state that is reachable there exists a two-outcome POVM using solely the symmetries $S_0$ and some $S_i$ for $i\in\{1,2,3\}$ that allows to obtain it.
In order to improve the readability we repeat here Lemma \ref{lem10}.\\ \\
 \noindent\textit{ {\bf Lemma \ref{lem10}.} The only states  in the SLOCC classes $L_{a_2b_2}$ with $a^2= - b^2$ and $a,b\neq0$ that are reachable via LOCC are given by
$h_1\otimes h_2\otimes h_3\otimes h_4\ket{\Psi_{(-b_2)b_2}}$ where $\otimes H_i$ obeys the following condition. There exists a symmetry $S_k=\otimes_i s_i^k\in S(\Psi_{(-b_2)b_2}) $ for some $k\in\{1, 2, 3\}$ such that exactly
three out of the four operators $H_i$ commute with the corresponding operators $s_i^k$ and one operator $H_j$ does not commute with $s_j^k$.}
\begin{proof}
Note first that in this proof we will not use the standard form, since then it is easier to deal with transformations that are equivalent up to conjugation with $S_i$ where $S_i\in S(\Psi_{(-b_2)b_2})$. In particular, these transformations allow to obtain the same states if one does not use the standard form. To clarify that note that if $\sum_j p_j S_j\otimes_k H_k S_j =\otimes_k G_k$ then conjugating with $S_i$ leads to $\sum_j p_{f(i,j)} S_j\otimes_k H_k S_j =\otimes_k S_iG_kS_i$ where $f(i,j)$ is a function defined by $f(i,j)=l$ for $S_iS_j \propto S_l$. Note that $\otimes_k S_iG_kS_i$ corresponds to the same state as $\otimes_k G_k$ and if one does not use the restrictions set by the standard form both versions can occur on the right-hand-side of Eq. (\ref{EqSep}). Thus, conjugation with $S_i$ where $S_i\in S(\Psi_{(-b_2)b_2})$ simply corresponds to a reordering of the probabilities. \\
Note further that since the symmetry is unitary, we can use Eq. (\ref{eq_symLU}) which gives us necessary conditions for a state to be reachable. Moreover, since the symmetry is an element of the Pauli group we can make use of Eq. (\ref{eq_Pauli}). In the following we use the notation $\eta_i = p_0+p_i-p_l-p_k$ where $i\neq l\neq k\neq i\in\{1,2,3\}$ and $\eta_0=\sum_i p_i=1$.
We will distinguish here the 4 different cases:
\begin{itemize}
\item[(i)] $H_1= H_1^{00A}$ and $H_3=H_3^{00A}$
\item[(ii)] $H_1= H_1^{\neq 3}$ and $H_3=H_3^{\neq 3}$
\item[(iii)] $H_1= H_1^{00A}$  and $H_3=H_3^{\neq 3}$
\item[(iv)] $H_1= H_1^{\neq 3}$  and $H_3=H_3^{00A}$.
\end{itemize}
Let us start with case (i). Using Eq. (\ref{EqSep}) and the fact that for this case $H_i$ for $i\in\{1,3\}$ commutes with $S_k$ $\forall k\in\{0,1,2,3\}$ we have that $H_i=G_i$ for $i\in\{1,3\}$. Tracing over pary $1$ and $3$ Eq. (\ref{eq_symLU}) is equivalent to \bea
\label{Eq_sep2} \h_{(2)}(\h_{(4)})^T \bigodot (N_1-N_2)=0,\eea where by
$\odot$ we denote the Hadamard product (i.\ e.\ entry-wise matrix
multiplication), $\h_{(i)}=(h_i^{1},h^2_{i},h^3_{i})^T$ for
$i\in \{2,4\}$ and $N_1=\vec{\e} \vec{\e}^T$, with
$\vec{\e}=(\e_1,\e_2,\e_3)^T$
and \bea N_2=\begin{pmatrix} \e_0 &\e_3& \e_2\\ \e_3& \e_0& \e_1\\
\e_2 & \e_1& \e_0\end{pmatrix}.\eea
This set of equations also occurs in the generic case and thus has already been discussed in \cite{MESus}. Hence, one obtains that for case (i) the only reachable states are given by $h_1^{0 0 A}\otimes h_k^{\neq j}\otimes h_3^{0 0 A}\otimes h_l^{j 0 0}\ket{\Psi_{(-b_2)b_2}}$ where $\{k,l\}=\{2,4\}$ and $j\in\{1,2,3\}$, as well as $h_1^{0 0 A}\otimes h_k^{\neq 0}\otimes h_3^{00A}\otimes \one_l/2 \ket{\Psi_{(-b_2)b_2}}$ where $\{k,l\}=\{2,4\}$. Note that these states are reachable via a two-outcome POVM.\\
We will proceed with case (ii). Considering Eq. (\ref{eq_Pauli}) with $P=\sigma_i\otimes \sigma_z\otimes\sigma_j\otimes \one$ for $i\in\{1,2\}$ and $j\in\{1,2\}$ one obtains that
\bea
\bar{h}_1^i\bar{h}_2^3\bar{h}_3^j=\eta_1\eta_2\eta_3\bar{h}_1^i\bar{h}_2^3\bar{h}_3^j.
\eea
Since $\eta_j\leq 1$ $\forall j\in\{1,2,3\}$ the solution $\eta_1\eta_2\eta_3=1$ implies that $\eta_1,\eta_2,\eta_3\in\{1,-1\}$ and therefore $p_m\neq 0$ for exactly one $m\in\{0,1,2,3\}$. Thus, $\eta_1\eta_2\eta_3=1$ corresponds to a trivial transformation. \\
Note that as $H_1= H_1^{\neq 3}$ and $H_3=H_3^{\neq 3}$ there exists at least one $i\in\{1,2\}$ and at least one $j\in\{1,2\}$ such that $\bar{h}_1^i\bar{h}_3^j\neq 0$. Thus, in order to allow for a non-trivial transformation we also have that $\bar{h}_2^3=0$. Analogously, one obtains that $\bar{h}_4^3=0$ (for a non-trivial transformation). Considering  Eq. (\ref{eq_Pauli}) for $P=\one\otimes \sigma_y\otimes\sigma_i\otimes \one$ and $P=\one\otimes\one \otimes\sigma_i\otimes \sigma_y$ for $i\in\{1,2\}$, as well as for $P=\sigma_j\otimes \sigma_x\otimes\one\otimes \one$ and $P=\sigma_j\otimes\one \otimes\one\otimes \sigma_x$ for $j\in\{1,2\}$ leads to
\begin{eqnarray}
&\bar{h}_2^1\bar{h}_1^j=\eta_1^2\bar{h}_2^1\bar{h}_1^j,\\
&\bar{h}_4^1\bar{h}_1^j=\eta_1^2\bar{h}_4^1\bar{h}_1^j, \\ \label{appC1}
&\bar{h}_2^2\bar{h}_3^i=\eta_2^2\bar{h}_2^2\bar{h}_3^i
\end{eqnarray}
and
\bea\label{appC2}
\bar{h}_4^2\bar{h}_3^i=\eta_2^2\bar{h}_4^2\bar{h}_3^i.
\eea
Thus, it is required that either $\eta_2^2=1$ or $\bar{h}_2^2=\bar{h}_4^2=0$, as $\bar{h}_3^i\neq 0$ for at least one $i\in\{1,2\}$ and in addition it has to hold that either $\bar{h}_2^1=\bar{h}_4^1=0$ or $\eta_1^2=1$, as $\bar{h}_1^j\neq  0$ for at least one $j\in\{1,2\}$. In case $\eta_l^2=1$ for $l\in\{1,2\}$ one obtains that either $p_0=p_l=0$ or $p_k=p_m=0$ where $k\neq m\neq l\neq k$ and $k,m\in\{1,2,3\}$. Note that the transformations that correspond to these two solutions are equivalent up to conjugation with $S_k$ and therefore they allow to obtain the same states \footnote{Recall that if one does not use the restrictions set by the standard form, transformations that only differ up to conjugation with one of the symmetries allow to reach the same states.}.  In Lemma \ref{twooutcome} we have already shown which states are reachable by using the symmetries $S_0$ and $S_l$ for $l\in\{1,2\}$. Thus, it remains to consider the case $\bar{h}_2^1=\bar{h}_4^1=\bar{h}_2^2=\bar{h}_4^2=0$. Recall that we have already shown that for a non-trivial transformation it has to hold that $\bar{h}_2^3=\bar{h}_4^3=0$. Combining these conditions we have that $H_2=H_4=\one/2$. States of the form  $h_1^{\neq 3}\otimes \one/2\otimes h_3^{\neq 3}\otimes \one/2\ket{\Psi_{(-b_2)b_2}}$ can for example  be reached from states of the form $g_1^{\neq 3}\otimes\one/2\otimes g_3^{00A}\otimes \one/2 \ket{\Psi_{(-b_2)b_2}}$ via the POVM $\{1/\sqrt{2}h_3S_0g_3^{-1},1/\sqrt{2}h_3S_1g_3^{-1}\}$ where $\bar{g}_3^2=0$, $\bar{g}_3^1=0$ and $\bar{h}_3^3=\bar{g}_3^3$.\\
Let us proceed with case (iii). Note that $H_3=H_3^{\neq 3}$ implies that there exists at least one $i\in\{1,2\}$ such that $h_3^i\neq 0$.
Thus, it follows from  Eqs. (\ref{appC1}) and (\ref{appC2}) that either $\eta_2^2=1$ or $\bar{h}_2^2=\bar{h}_4^2=0$. As has been already discussed for $\eta_2^2=1$ the corresponding reachable states can be obtained by using only the symmetries $S_0$ and $S_2$ and therefore they are given by Lemma \ref{twooutcome}. Hence, it remains to consider the case $\bar{h}_2^2=\bar{h}_4^2=0$.
Considering Eq. (\ref{eq_Pauli}) for $P=\one\otimes \sigma_x\otimes\one\otimes \sigma_x$ and $P=\one\otimes \sigma_z\otimes\one\otimes \sigma_z$, as well as for $P=\one\otimes \sigma_x\otimes\sigma_i\otimes \sigma_z$ and $P=\one\otimes \sigma_x\otimes\sigma_i\otimes \sigma_z$ where $i\in\{1,2\}$ leads to
\begin{eqnarray}
&\bar{h}_2^1\bar{h}_4^1=\eta_1^2\bar{h}_2^1\bar{h}_4^1,\\
&\bar{h}_2^3\bar{h}_4^3=\eta_3^2\bar{h}_2^3\bar{h}_4^3, \\
&\bar{h}_2^1\bar{h}_3^i\bar{h}_4^3=\eta_1\eta_2\eta_3\bar{h}_2^1\bar{h}_3^i\bar{h}_4^3
\end{eqnarray}
and
\bea
\bar{h}_2^3\bar{h}_3^i\bar{h}_4^1=\eta_1\eta_2\eta_3\bar{h}_2^3\bar{h}_3^i\bar{h}_4^1.
\eea
Recall that $\eta_1\eta_2\eta_3=1$ corresponds to a trivial transformations. Thus, in order to allow for a non-trivial transformation it has to hold that $\bar{h}_2^3\bar{h}_4^1=0$ and $\bar{h}_2^1\bar{h}_4^3=0$.
As before for the case $\eta_l^2=1$ for $l\in\{1,3\}$ one obtains that either $p_0=p_l=0$ or $p_k=p_m=0$ where $k\neq m\neq l\neq k$ and $k,m\in\{1,2,3\}$. As $p_0=p_l=0$ and $p_k=p_m=0$ correspond to transformations that are equivalent up to conjugation with $S_k$ and we do not use the restrictions set by the standard form here the corresponding reachable states are given by Lemma \ref{twooutcome}. For $\eta_l^2\neq 1$ for $l\in\{1,3\}$ we have that $\bar{h}_j^1=0$ for at least one $j\in\{2,4\}$ and $\bar{h}_k^3$ for at least one $k\in\{2,4\}$. Combining this with the conditions that $\bar{h}_2^2=\bar{h}_4^2=0$, $\bar{h}_2^3\bar{h}_4^1=0$ and $\bar{h}_2^1\bar{h}_4^3=0$ one obtains that $H_i=\one/2$ for some $i\in\{2,4\}$ and $H_j=H_j^{A0A}$ for $i\neq j$ and $j\in\{2,4\}$. W.l.o.g. we choose $i=2$
 and $j=4$. In case $H_4\neq \one/2$, one can use the POVM $\{1/\sqrt{2}h_4S_0,1/\sqrt{2}h_4S_2\}$ to obtain these states from states of the form $g_1^{00A}\otimes\one/2\otimes g_3^{\neq 3}\otimes \one/2 \ket{\Psi_{(-b_2)b_2}}$. In case $H_4= \one/2$ the POVM $\{1/\sqrt{2}h_3S_0g_3^{-1},1/\sqrt{2}h_3S_1g_3^{-1}\}$ where $\bar{g}_3^2=0$, $\bar{g}_3^1=0$ and $\bar{h}_3^3=\bar{g}_3^3$ allows to reach the corresponding states from states of the form $g_1^{00A}\otimes\one/2\otimes g_3^{00A}\otimes \one/2 \ket{\Psi_{(-b_2)b_2}}$.\\
Note that case (iv) can be treated completely analogously to case (iii). In particular, one only has to exchange party 1 and 3 and simultaneously the x- and y- components for party 2 and 4, as well as the role of $S_1$ and $S_2$ (and therefore also $\eta_1$ and $\eta_2$).
\end{proof}

\section*{\label{appF}Appendix F: Proof of Lemmas \ref{lem77} and \ref{lemmab2}}

This section is concerned with the proofs of  Lemma \ref{lem77}  and Lemma \ref{lemmab2} (see subsection \ref{seca2a2}), which specifies the states that are reachable  and convertible in  the SLOCC classes $L_{a_2b_2}$ (see \cite{slocc4}) for $a^2= b^2$ and $a,b\neq0$. We will first prove Lemma \ref{lem77}. In order to improve the readability we repeat the lemma here again.\\ \\
\noindent\textit{ {\bf Lemma \ref{lem77}.} The only states in the SLOCC classes $L_{a_2b_2}$ with $a^2= b^2$ and $a,b\neq0$ that are reachable via LOCC are given by $d_1 \otimes h_2^x\otimes h_3\otimes h_4^x \ket{\Psi_{a_2a_2}}$ where $h_3\neq d_3$ and $d_1 \otimes \one/2\otimes h_3\otimes h_4 \ket{\Psi_{a_2a_2}}$ where $h_4\neq d_4$  and  $d_1 \otimes h_2^x\otimes h_3\otimes d_4 \ket{\Psi_{a_2a_2}}$ where $h_2^x\neq \one/2$ and $d_1 \otimes h_2^x\otimes d_3\otimes h_4 \ket{\Psi_{a_2a_2}}$ where $h_4\neq h_4^x$ and $d_1 \otimes \one /2\otimes d_3\otimes h_4 \ket{\Psi_{a_2a_2}}$ where $h_4\neq \one/2$.}
\begin{proof} Recall that the corresponding symmetry is given by \begin{eqnarray}
\tilde{S}_{x,z,m}&=(\sigma_z\otimes\sigma_y\otimes\sigma_z\otimes\sigma_y)^{m} \left(
    \begin{array}{cc}
      1 & 0 \\
      x & 1 \\
    \end{array}
  \right)\otimes P_{1/z}\otimes\left(
    \begin{array}{cc}
      1 & 0\\
      -x & 1 \\
    \end{array}
  \right)
\otimes P_{z},
\end{eqnarray}
where $m\in\{0,1\}$, $x,z\in\C$ and $z\neq 0$.
This symmetry allows to fix the  standard form as $d_1 \otimes g_2^x\otimes g_3\otimes g_4 \ket{\Psi_{a_2a_2}}$ where $g_2^2\geq 0$.  In case $g_2^2=0$ one  chooses $g_4^2\geq 0$.  Moreover, we have that $\Re(g_3^2)\geq 0$ and for $\Re(g_3^2)= 0$ we choose $\Im(g_3^2)\geq 0$. Inserting the symmetries into Eq. (\ref{EqSep}) leads to
\begin{eqnarray}\label{eq7}
&\sum_{x,|z|,\phi,m} p_{x,|z|,\phi,m}\left(
    \begin{array}{cc}
      h_1^1+h_1^3|x|^2&  h_1^3 x^* \\
    h_1^3 x  & h_1^3 \\
    \end{array}
  \right)\otimes \left(
    \begin{array}{cc}
     1/(2|z|^2)&  (-1)^mh_2^2e^{2i\phi} \\
   (-1)^mh_2^2e^{-2i\phi} & |z|^2/2\\
    \end{array}
  \right)\\ \nonumber&\otimes\left(
    \begin{array}{cc}
     a_3(x) &  -h_3^3 x^*+(-1)^mh_3^2 \\
    -h_3^3 x^*+(-1)^mh_3^{2 *}   & h_3^3 \\
    \end{array}
  \right)\otimes\left(
    \begin{array}{cc}
     (h_4^3)^m|z|^2/(h_4^1)^{m-1}&  (-1)^m|h_4^2|e^{-i(2\phi-(-1)^m\theta_4)} \\
   (-1)^m|h_4^{2}|e^{i(2\phi-(-1)^m\theta_4)} &(h_4^1)^m /[|z|^2(h_4^3)^{m-1}]\\
    \end{array}
  \right)\\ \nonumber&=r \tilde{D}_1\otimes G_2^x\otimes G_3\otimes G_4,
\end{eqnarray}
where $h_4^2=|h_4^{2}|e^{i\theta_4}$ and $a_3(x)=h_3^1+h_3^3|x|^2-(-1)^mh_3^2x-(-1)^mh_3^{2 *}x^*$.
We will distinguish in the following 2 different cases: (i) $h_3^2\neq 0$ and (ii) $h_3^2=0$. For case (i) we will show that either  $ p_{x,|z|,\phi,m}=0$ for $m=1$ (solution (a)) or for $m=0$ (solution (b)) or $h_4^1=h_4^3=1/2$ (solution (c)). For solution (c) we show that  states with $H_2=H_4=\one/2$ or $h_i^2=0$ and $h_j^2\neq 0$ for $\{i,j\}=\{2,4\}$ are reachable via providing the corresponding LOCC protocol. In case $h_2^2,h_4^2 \neq 0$ we show that either $h_4^2\in\R$ which corresponds to reachable states or  $ p_{x,|z|,\phi,m}=0$ for $m=0$ or for $m=1$. Then, we show that using only symmetries with $m=0$ (Solution (a)) only allows to reach states with $h_i^2=0$ and $h_j^2\neq 0$ for $\{i,j\}=\{2,4\}$. This result applies independently of the form of $H_3$.\\
For case (ii) it is easy to show that if $H_4\neq H_4^x$ the corresponding states are reachable via LOCC. Thus, we only have to investigate the cases (I) $H_2,H_4\neq \one/2$ and $H_4= H_4^x$ (II) $H_2=H_4= \one/2$ and (III) $H_i=\one/2$ and $H_j=H_j^x\neq \one/2$ for $\{i,j\}=\{2,4\}$. We show that the cases (ii) (I) and (ii) (II) correspond to states in $MES_4$. Note that case (ii) (II) includes the seed states. Furthermore, we show that the states corresponding to case (ii) (III) are reachable.
Let us start with case (i), i.e. $h_3^2\neq 0$.
Considering the matrix elements $\ket{1010}\bra{1010}$ and $\ket{1111}\bra{1111}$ of Eq. (\ref{eq7}) and additionally using the normalization condition one easily obtains that
\begin{eqnarray}
&r=h_1^3h_3^3/g_1^3g_3^3,\label{lem70}\\
&g_4^1=\tilde{p}_0h_4^1+\tilde{p}_1h_4^3\label{lem71}
\end{eqnarray} and
\bea g_4^3=\tilde{p}_0h_4^3+\tilde{p}_1h_4^1\label{lem72},
\eea where here and in the following we use the definition $\tilde{p}_m=\sum_{x,|z|,\phi} p_{x,|z|,\phi,m}$. Considering the matrix elements $\ket{1010}\bra{0010}$ and $\ket{1111}\bra{0111}$ leads to $\sum_{x,|z|,\phi} p_{x,|z|,\phi,0}x^*h_4^1=-\sum_{x,|z|,\phi} p_{x,|z|,\phi,1}x^*h_4^3$ and $\sum_{x,|z|,\phi} p_{x,|z|,\phi,0}x^*h_4^3=-\sum_{x,|z|,\phi} p_{x,|z|,\phi,1}x^*h_4^1$. Using these equations as well as Eq. (\ref{lem70}), (\ref{lem71}) and (\ref{lem72}) for the matrix elements $\ket{1010}\bra{1000}$ and $\ket{1111}\bra{1101}$ of Eq. (\ref{eq7}) one obtains that
\bea\label{h3}
\frac{h_3^2}{h_3^3}\frac{(\tilde{p}_0h_4^1-\tilde{p}_1h_4^3)}{(\tilde{p}_0h_4^1+\tilde{p}_1h_4^3)}
=\frac{h_3^2}{h_3^3}\frac{(\tilde{p}_0h_4^3-\tilde{p}_1h_4^1)}{(\tilde{p}_0h_4^3+\tilde{p}_1h_4^1)}=\frac{g_3^2}{g_3^3}.
\eea
From this equation it follows that $\tilde{p}_0\tilde{p}_1[(h_4^3)^2-(h_4^1)^2]=0$, which can only be fulfilled if $\tilde{p}_1=0$ (solution (a)), $\tilde{p}_0=0$ (solution (b)) or $h_4^3=h_4^1=1/2$ (solution (c)).  Note that we used here the normalization condition $\tr (H_4)=1$, i. e. $h_4^3=h_4^1$ implies  $h_4^3=h_4^1=1/2$. \\
Let us first consider solution (c), i.e. $h_4^3=h_4^1=1/2$. From Eq. (\ref{lem71}) and (\ref{lem72}) it follows that $g_4^3=g_4^1=1/2$.
When investigating the implications of solution (c) in more detail we will distinguish in the following the 3 different subcases (I) $H_2=H_4=\one/2$, (II) $h_2^2,h_4^2\neq 0$ and (III) $h_i^2=0$ and $h_j^2\neq 0$ for $\{i,j\}=\{2,4\}$.
States that correspond to case (i) (I), i.e. $H_2=H_4=\one/2$ and $h_3^2\neq 0$, can be reached via a LOCC protocol using the symmetries $\one^{\otimes 4}$ and $\sigma_z\otimes\sigma_y\otimes \sigma_z\otimes \sigma_y$. \\
We will now treat case (i) (II), i.e. $h_2^2,h_4^2\neq 0$ and $h_3^2\neq 0$, using solution (c). Considering the matrix elements $\ket{1110}\bra{1110}$ and $\ket{1011}\bra{1011}$ of Eq. (\ref{eq7}) leads to $\sum_{x,|z|,\phi} p_{x,|z|,\phi,m}|z|^4=\sum_{x,|z|,\phi} p_{x,|z|,\phi,m}1/|z|^4=1$ which can only be fulfilled if $p_{x,|z|,\phi,m}=0$ for $|z|\neq 1$. Considering the matrix element $\ket{1111}\bra{0010}$  of Eq. (\ref{eq7}) and using that  $\sum_{x,|z|,\phi} p_{x,|z|,\phi,0}x^*=-\sum_{x,|z|,\phi} p_{x,|z|,\phi,1}x^*$ one obtains that $h_2^2\sum_{x,|z|,\phi} p_{x,|z|,\phi,0}x^* (h_4^2-h_4^{2 *})=0$. Thus, either $h_4^2\in\R$ or $\sum_{x,|z|,\phi} p_{x,|z|,\phi,0}x^*=0$. In case $h_4^2\in\R$ one can use the POVM $\{1/\sqrt{2} (\one^{\otimes 2}\otimes h_3d_3^{-1}\otimes \one) ,1/\sqrt{2} (\sigma_z\otimes\sigma_x\otimes h_3 \sigma_zd_3^{-1}\otimes\sigma_x ) \}$ (for a correspondingly chosen $d_3$) to reach these states. Let us now investigate what happens for $h_4^2\notin\R$. In this case we have already shown that  $\sum_{x,|z|,\phi} p_{x,|z|,\phi,0}x^*=-\sum_{x,|z|,\phi} p_{x,|z|,\phi,1}x^*=0$. Recall that $h_3^2/h_3^3(\tilde{p}_0-\tilde{p}_1)=g_3^2/g_3^3$ for $h_4^3=h_4^1=1/2$ (see Eq. (\ref{h3})). The matrix element $\ket{1111}\bra{1010}$ of Eq. (\ref{eq7}) corresponds to $h_2^2 (h_4^2\tilde{p}_0+h_4^{2 *}\tilde{p}_1)=g_2^2 g_4^2$. From the matrix element $\ket{1111}\bra{1000}$ one obtains that
\bea
h_2^2 h_3^2/h_3^3(h_4^2\tilde{p}_0-h_4^{2 *}\tilde{p}_1)=g_2^2g_4^2 g_3^2/g_3^3.
\eea
Inserting the expressions for $g_3^2/g_3^3$ and $g_2^2 g_4^2$ and taking the imaginary part shows that for $h_4^2\notin\R$ it has to hold that $\tilde{p}_0=1$ or $\tilde{p}_1=1$. Note that this corresponds to solution (a) or solution (b) which we will discuss later on.  \\
Let us consider now the case (i) (III), i.e $h_i^2=0$ and $h_j^2\neq 0$ where $i,j\in\{2,4\}$. The POVM $\{1/\sqrt{2} h_j\one^{\otimes 4}d_j^{-1} ,1/\sqrt{2} h_j \one\otimes\sigma_z\otimes\one\otimes\sigma_zd_j^{-1} \}$ with a correspondingly chosen $d_j$ allows to reach these states.  Note that these states are reachable independent of whether $h_3^2$ is zero or not.\\
Let us now consider solution (a) for case (i), i.e. $p_{x,|z|,\phi,1}=0$ $\forall x,z\in\C$ and $h_3^2\neq 0$. Note that we will treat this solution without using any constraints that we obtained in the discussion of solution (c). Considering first Eq. (\ref{lem71}) and (\ref{lem72}) one obtains that $h_4^3=g_4^3$ and $h_4^1=g_4^1$, i.e. the diagonal elements of $H_4$ can not be changed. Furthermore, it can be easily shown that $p_{x,|z|,\phi,0}=0$ for $|z|\neq 1$. From the matrix elements $\ket{1010}\bra{0010}$ and $\ket{1000}\bra{0010}$ of Eq. (\ref{eq7}) it follows that $p_{x,|z|,\phi,0}=0$ for $x\neq 0$. Using that only symmetries with $x=0$ and $m=0$, i.e. symmetries which act trivially on party $1$ and $3$, can contribute to a transformation in this case it is easy to see that $H_i=G_i$ for $i\in\{1,3\}$. \\
Thus, when proceeding with our discussion of which states are reachable using only symmetries with $m=0$ we will distinguish the 3 different subcases: (A) $h_2^2=h_4^2=0$, (B)$h_2^2,h_4^2\neq 0$ and (C) $h_i^2=0$ and $h_j^2\neq 0$ for $\{i,j\}=\{2,4\}$.
Let us show first that for case (i) (A) and (i) (B) the corresponding states are in $MES_4$. It can be easily seen that $h_2^2=h_4^2=0$ implies $g_2^2=g_4^2=0$ and thus a transformation using only symmetries with $m=0$ is not possible. In order to show that states corresponding to case (i) (B) are not reachable consider the matrix element $\ket{1111}\bra{1010}$, which leads to $h_2^2 h_4^2=g_2^2 g_4^2$.
Using the matrix elements $\ket{1110}\bra{1010}$ and $\ket{1010}\bra{1011}$ of Eq. (\ref{eq7}) it can be easily seen that it has to hold that $h_4^{2 *}/h_2^2=g_4^{2 *}/g_2^2$. Using additionally that according to our standard form $h_2^2\geq 0$ leads to  $h_2^2=g_2^2$ and therefore  $h_4^2=g_4^2$.
Note that in our discussion of which states are reachable using only symmetries with $m=0$ we never used any constraint on $h_3^2$. Thus, these results hold true independent of the form of $H_3$.\\
Let us treat next the case (ii), i.e. $h_3^2=0$. In case $H_4\neq H_4^x$ the states are reachable via the POVM $\{1/\sqrt{2} h_4\one^{\otimes 4}g_4^{-1} ,1/\sqrt{2} h_4 \sigma_z\otimes\sigma_x\otimes\sigma_z\otimes\sigma_xg_4^{-1}  \}$ where $g_4^1=g_4^3=1/2$ and $\Re(h_4^2)=g_4^2$. Thus, we only have to consider the cases (I) $H_2,H_4\neq \one/2$ and $H_4= H_4^x$ (II) $H_2=H_4= \one/2$ and (III) $H_i=\one/2$ and $H_j=H_j^x\neq \one/2$ for $\{i,j\}=\{2,4\}$. We will first show that for case (ii) (I), i.e. $H_2, H_4\neq \one/2$, $H_4= H_4^x$ and $h_3^2=0$, the states are in $MES_4$. It can be easily seen that for $h_3^2=0$ one obtains that $p_{x,|z|,\phi,m}=0$ for $x\neq 0$ and $g_3^2=0$. Thus, it follows that $H_i=G_i$ for $i\in\{1,3\}$.  Considering Eq. (\ref{lem71}) and (\ref{lem72}) leads to $g_4^1=g_4^3=1/2$. In order to show that $h_2^2=g_2^2$ and  $h_4^2=g_4^2$ consider the matrix elements  $\ket{1110}\bra{1010}$ and $\ket{1110}\bra{1111}$ of Eq. (\ref{eq7}), as well as $\ket{1111}\bra{1010}$ and recall that according to our standard form $g_2^2> 0$. Thus, in this case $H_i=G_i$ for $i\in\{1,2,3,4\}$  and therefore the corresponding states are in $MES_4$. \\
It can be shown using our choice of the standard form that for case (i) solution (b) is no solution to Eq. (\ref{eq7}).\\
The case (ii) (II), i.e. $h_3^2=0$ and $H_2=H_4=\one/2$, also corresponds to states in $MES_4$. As we have discussed before it is easy to show that for $h_3^2=0$ one obtains that $p_{x,|z|,\phi,m}=0$ for $x\neq 0$ and $g_3^2=0$. Thus, $H_i=G_i$ for $i\in\{1,3\}$. Moreover, one can show that $p_{x,|z|,\phi,m}=0$ for $|z|\neq 1$, which implies that only unitary symmetries can be used for transformations and therefore $G_2=G_4=\one/2$. Thus, the states which correspond to case (ii) (II) are in $MES_4$. Note that the seed states are contained in this case.\\
For the remaining case (ii) (III), i.e. $h_3^2=0$ and  $H_i=\one/2$ and $H_j=H_j^x\neq \one/2$ for $\{i,j\}=\{2,4\}$ one can construct a LOCC protocol which allows to reach the corresponding states. In particular, one uses the symmetries $\one^{\otimes 4}$ and $\sigma_z\otimes \sigma_y\otimes \sigma_z\otimes \sigma_y$.
 \end{proof}
Let us now show Lemma \ref{lemmab2} which specifies the states that are convertible in this SLOCC classes.\\ \\
\noindent\textit{ {\bf Lemma \ref{lemmab2}.}  The only states in the SLOCC classes $L_{a_2b_2}$ with $a^2=  b^2$ and $a,b\neq0$ that are convertible via LOCC are given by $d_1 \otimes g_2^x\otimes g_3\otimes g_4^x \ket{\Psi_{a_2a_2}}$ and $d_1 \otimes \one\otimes g_3\otimes g_4 \ket{\Psi_{a_2a_2}}$  and  $d_1 \otimes g_2^x\otimes g_3\otimes d_4 \ket{\Psi_{a_2a_2}}$ and $d_1 \otimes g_2^x\otimes d_3\otimes g_4 \ket{\Psi_{a_2a_2}}$.}
\begin{proof}
The only states that can be reached from a convertible state are given by the ones in Lemma \ref{lem77}. From Eq. (\ref{eq7}) it can be easily seen that from $h_i^2=0$ it follows that $g_i^2=0$ (see also proof of Lemma \ref{lem77}). Thus, the only states that allow to reach the states $d_1 \otimes \one/2\otimes h_3\otimes h_4 \ket{\Psi_{a_2a_2}}$ where $h_4\neq d_4$  and  $d_1 \otimes h_2^x\otimes h_3\otimes d_4 \ket{\Psi_{a_2a_2}}$ where $h_2^x\neq \one/2$ and $d_1 \otimes h_2^x\otimes d_3\otimes h_4 \ket{\Psi_{a_2a_2}}$ where $h_4\neq h_4^x$ and $d_1 \otimes \one /2\otimes d_3\otimes h_4 \ket{\Psi_{a_2a_2}}$ where $h_4\neq \one/2$ are of the form $d_1 \otimes \one/2\otimes g_3\otimes g_4 \ket{\Psi_{a_2a_2}}$  and  $d_1 \otimes g_2^x\otimes g_3\otimes d_4 \ket{\Psi_{a_2a_2}}$ and $d_1 \otimes g_2^x\otimes d_3\otimes g_4 \ket{\Psi_{a_2a_2}}$. In order to show that the states $d_1 \otimes \one/2\otimes g_3\otimes g_4 \ket{\Psi_{a_2a_2}}$ are indeed convertible consider  $\{\sqrt{p} (\one^{\otimes 3}\otimes h_4g_4^{-1}),\sqrt{1-p} (\one\otimes\sigma_z\otimes\one\otimes h_4\sigma_zg_4^{-1})\}$ which is a valid POVM for $h_4^3=g_4^3$, $h_4^1=g_4^1$ and $(2 p -1 ) h_4^2=g_4^2$. Since it is always possible to find a value of $p$ and a $H_4\neq G_4$, these states are convertible. Analogously one can show that the states $d_1 \otimes g_2^x\otimes g_3\otimes d_4 \ket{\Psi_{a_2a_2}}$ ($d_1 \otimes g_2^x\otimes d_3\otimes g_4 \ket{\Psi_{a_2a_2}}$) are convertible via constructing a POVM acting non-trivially on party $2$ (party $4$) via using the symmetries $\one^{\otimes 4}$ and $\one\otimes\sigma_z\otimes\one\otimes \sigma_z$ ($\sigma_z\otimes \sigma_x\otimes\sigma_z\otimes \sigma_x$) respectively. The states $\tilde{d}_1 \otimes h_2^x\otimes h_3\otimes h_4^x \ket{\Psi_{a_2a_2}}$ where $h_3,h_2^x,h_4^x$ are non-diagonal can only be obtained from states of the form $d_1 \otimes g_2^x\otimes g_3\otimes g_4^x \ket{\Psi_{a_2a_2}}$ (see proof of Lemma \ref{lem77} where it has been shown that $h_4^1=h_4^3=1/2$ implies $g_4^1=g_4^3=1/2$ and consider the matrix elements $\ket{1110}\bra{1010}$ and $\ket{1110}\bra{1111}$, as well as $\ket{1111}\bra{1010}$ of Eq. (\ref{eq7}) to show that $g_4^2=h_4^2\in\R$). Note that in case $g_i^x=\one/2$ for $i\in\{2,4\}$ we would obtain the same convertible states as already discussed. In order to see that the states $d_1 \otimes g_2^x\otimes g_3\otimes g_4^x \ket{\Psi_{a_2a_2}}$ are indeed convertible consider $\{\sqrt{p} (\one^{\otimes 2}\otimes h_3g_3^{-1}\otimes \one),\sqrt{1-p} (\sigma_z\otimes \sigma_x\otimes h_3\sigma_zg_3^{-1}\otimes \sigma_x)\}$ which forms a POVM for $h_3^3=g_3^3$, $h_3^1=g_3^1$ and $(2 p -1 ) h_3^2=g_3^2$.
\end{proof}

\end{document}